\documentclass[12pt]{iopart}
\pdfoutput=1 

\usepackage{iopams}
\usepackage{amstext} 
\usepackage{amssymb} 
\usepackage{bbold} 
\usepackage{bbm} 
\usepackage{amsthm} 
\usepackage{stmaryrd} 
\usepackage{stix}
\usepackage{etoolbox} 
\usepackage{xparse}
\usepackage{pgffor}

\usepackage[utf8]{inputenc} 
\usepackage{textgreek} 
\usepackage{relsize} 
\usepackage{microtype} 
\usepackage{ulem} 
\usepackage[compress]{cite} 
\usepackage{csquotes} 
\usepackage{setspace} 
\usepackage{scalerel}

\usepackage{graphicx} 
\usepackage[usenames,dvipsnames]{xcolor} 
\usepackage{tikzit} 
\usepackage{stackengine}
\usepackage{circuitikz} 
\usetikzlibrary{
    arrows,
    shapes,
    decorations,
    intersections,
    backgrounds,
    positioning,
    circuits.ee.IEC
    }

\tikzstyle{inline text}=[text height=1.5ex, text depth=0.25ex, yshift=0.5mm]
\tikzstyle{upground}=[circuit ee IEC, thick, ground, rotate=90, scale=2]
\tikzstyle{downground}=[circuit ee IEC, thick, ground, rotate=-90, scale=1.5]
\tikzstyle{point}=[regular polygon, regular polygon sides=3, draw, scale=0.75, inner sep=-0.5pt, minimum width=9mm, fill=white, regular polygon rotate=180, tikzit fill={rgb,255: red,242; green,255; blue,92}]
\tikzstyle{wide copoint}=[fill=white, draw, shape=isosceles triangle, shape border rotate=90, isosceles triangle stretches=true, inner sep=0pt, minimum width=1.5cm, minimum height=6.12mm]
\tikzstyle{wide point}=[fill=white, draw, shape=isosceles triangle, shape border rotate=-90, isosceles triangle stretches=true, inner sep=0pt, minimum width=1.5cm, minimum height=6.12mm, yshift=-0.0mm]
\tikzstyle{wide dpoint}=[wide point, doubled]
\tikzstyle{copoint}=[regular polygon, regular polygon sides=3, draw, scale=0.75, inner sep=-0.5pt, minimum width=9mm, fill=white, tikzit fill={rgb,255: red,255; green,128; blue,0}, tikzit draw={rgb,255: red,255; green,128; blue,0}]
\tikzstyle{dot}=[inner sep=0mm, minimum width=2mm, minimum height=2mm, draw, shape=circle]
\tikzstyle{black dot}=[dot, fill={gray!30}, text depth=-0.2mm]
\tikzstyle{white dot}=[dot, fill=white, text depth=-0.2mm]
\tikzstyle{small box}=[rectangle, inline text, fill=white, draw, minimum height=5mm, yshift=-0.5mm, minimum width=5mm, font={\small}]
\tikzstyle{small gray box}=[small box, fill={gray!30}]
\tikzstyle{medium box}=[rectangle, inline text, fill=white, draw, minimum height=5mm, yshift=-0.5mm, minimum width=10mm, font={\small}]
\tikzstyle{square box}=[small box]
\tikzstyle{medium gray box}=[small box, fill={gray!30}]
\tikzstyle{semilarge box}=[rectangle, inline text, fill=white, draw, minimum height=5mm, yshift=-0.5mm, minimum width=12.5mm, font={\small}]
\tikzstyle{large box}=[rectangle, inline text, fill=white, draw, minimum height=5mm, yshift=-0.5mm, minimum width=15mm, font={\small}]
\tikzstyle{large gray box}=[small box, fill={gray!30}]
\tikzstyle{dpoint}=[point, doubled]
\tikzstyle{dcopoint}=[copoint, doubled]
\tikzstyle{boldedge}=[doubled, shorten <=-0.17mm, shorten >=-0.17mm]
\tikzstyle{normal}=[line width=0.9pt]
\tikzstyle{doubled}=[line width=1pt]
\tikzstyle{boldedge}=[doubled, shorten <=-0.17mm, shorten >=-0.17mm]
\tikzstyle{small dbox}=[small box, doubled]
\tikzstyle{white ddot}=[white dot, doubled]
\tikzstyle{black ddot}=[black dot, doubled, tikzit fill=black]
\tikzstyle{map}=[draw, shape=NEbox, inner sep=2pt, minimum height=6mm, fill=white]
\tikzstyle{box}=[draw, shape=rectangle, inner sep=2pt, minimum height=6mm, minimum width=6mm, fill=white]
\tikzstyle{dbox}=[draw, doubled, shape=rectangle, inner sep=2pt, minimum height=6mm, minimum width=6mm, fill=white]
\tikzstyle{dmap}=[draw, doubled, shape=NEbox, inner sep=2pt, minimum height=6mm, fill=white]
\tikzstyle{dmapdag}=[draw, doubled, shape=SEbox, inner sep=2pt, minimum height=6mm, fill=white]
\tikzstyle{dmapadj}=[draw, doubled, shape=SEbox, inner sep=2pt, minimum height=6mm, fill=white]
\tikzstyle{dmaptrans}=[draw, doubled, shape=SWbox, inner sep=2pt, minimum height=6mm, fill=white]
\tikzstyle{dmapconj}=[draw, doubled, shape=NWbox, inner sep=2pt, minimum height=6mm, fill=white]
\tikzstyle{map}=[draw, shape=NEbox, inner sep=2pt, minimum height=6mm, fill=white]
\tikzstyle{dashedmap}=[draw, dashed, shape=NEbox, inner sep=2pt, minimum height=6mm, fill=white]
\tikzstyle{mapdag}=[draw, shape=SEbox, inner sep=2pt, minimum height=6mm, fill=white]
\tikzstyle{mapadj}=[draw, shape=SEbox, inner sep=2pt, minimum height=6mm, fill=white]
\tikzstyle{maptrans}=[draw, shape=SWbox, inner sep=2pt, minimum height=6mm, fill=white]
\tikzstyle{mapconj}=[draw, shape=NWbox, inner sep=2pt, minimum height=6mm, fill=white]
\tikzstyle{semilarge map}=[draw, shape=NEbox, inner sep=2pt, minimum height=6mm, fill=white, minimum width=9.5mm]
\tikzstyle{semilarge dmap}=[draw, doubled, shape=NEbox, inner sep=2pt, minimum height=6mm, fill=white, minimum width=9.5mm]
\tikzstyle{kpointdag}=[kpoint adjoint]
\tikzstyle{kpointadj}=[kpoint adjoint]
\tikzstyle{kpointconj}=[kpoint conjugate]
\tikzstyle{kpointtrans}=[kpoint transpose]
\tikzstyle{kpoint common}=[draw, fill=white, inner sep=1pt, minimum height=4mm]
\tikzstyle{kpoint sc}=[shape=cornerpoint, kpoint common]
\tikzstyle{kpoint adjoint sc}=[shape=cornercopoint, kpoint common]
\tikzstyle{kpoint}=[shape=cornerpoint, shorten left=5pt, kpoint common, tikzit fill={rgb,255: red,255; green,128; blue,0}]
\tikzstyle{kpoint adjoint}=[shape=cornercopoint, shorten left=5pt, kpoint common, tikzit fill={rgb,255: red,255; green,128; blue,0}]
\tikzstyle{kpoint conjugate}=[shape=cornerpoint, shorten right=5pt, kpoint common]
\tikzstyle{kpoint transpose}=[shape=cornercopoint, shorten right=5pt, kpoint common]
\tikzstyle{kpoint symm}=[shape=cornerpoint, shorten left=5pt, shorten right=5pt, kpoint common]
\tikzstyle{wide kpoint}=[kpoint, minimum width=1 cm, inner sep=2pt]
\tikzstyle{wide kpointdag}=[kpointdag, minimum width=1 cm, inner sep=2pt]
\tikzstyle{wide kpointconj}=[kpointconj, minimum width=1 cm, inner sep=2pt]
\tikzstyle{wide kpointtrans}=[kpointtrans, minimum width=1 cm, inner sep=2pt]
\tikzstyle{wider kpoint}=[kpoint, minimum width=1.25 cm, inner sep=2pt]
\tikzstyle{wider kpointdag}=[kpointdag, minimum width=1.25 cm, inner sep=2pt]
\tikzstyle{wider kpointconj}=[kpointconj, minimum width=1.25 cm, inner sep=2pt]
\tikzstyle{wider kpointtrans}=[kpointtrans, minimum width=1.25 cm, inner sep=2pt]
\tikzstyle{dkpoint}=[kpoint, doubled, tikzit fill={rgb,255: red,255; green,85; blue,210}]
\tikzstyle{wide dkpoint}=[wide kpoint, doubled, tikzit fill={rgb,255: red,68; green,255; blue,0}]
\tikzstyle{dkpointdag}=[kpoint adjoint, doubled]
\tikzstyle{wide dkpointdag}=[wide kpointdag, doubled]
\tikzstyle{label}=[fill=white, draw=white, shape=circle, tikzit draw={rgb,255: red,10; green,26; blue,255}, tikzit fill={rgb,255: red,0; green,12; blue,255}, font={\small}]
\tikzstyle{squarelabel}=[fill=white, draw=white, shape=rectangle, tikzit draw=black]
\tikzstyle{eslabel}=[tikzit draw={rgb,255: red,255; green,191; blue,191}, tikzit fill={rgb,255: red,255; green,191; blue,191}, font={\tiny}]
\tikzstyle{large dmap}=[draw, doubled, shape=NEbox, inner sep=2pt, minimum height=6mm, fill=white, minimum width=12mm]
\tikzstyle{gray point}=[point, fill={gray!40!white}]
\tikzstyle{gray dpoint}=[gray point, doubled, tikzit draw={rgb,255: red,128; green,128; blue,128}, tikzit fill={rgb,255: red,128; green,128; blue,128}]
\tikzstyle{gray copoint}=[copoint, fill={gray!40!white}, tikzit fill={rgb,255: red,128; green,128; blue,128}]
\tikzstyle{gray dcopoint}=[gray copoint, doubled, tikzit fill={rgb,255: red,128; green,128; blue,128}]
\tikzstyle{circlenew}=[draw=black, shape=circle, inner sep=1pt]
\tikzstyle{blue label}=[text=NavyBlue, tikzit draw={rgb,255: red,0; green,96; blue,167}, tikzit fill={rgb,255: red,35; green,68; blue,255}]
\tikzstyle{big dot}=[fill=white, draw=black, shape=circle, minimum width=6mm, minimum height=6mm]
\tikzstyle{3d box}=[fill=white, draw=black, shape=trapezium, trapezium left angle=-70, trapezium right angle=70, rotate=10]
\tikzstyle{slant red box}=[fill={rgb,255: red,191; green,0; blue,64}, draw={rgb,255: red,191; green,0; blue,64}, shape=rectangle, xslant=0.5, font={\tiny}, text={rgb,255: red,191; green,0; blue,64}, fill opacity=0.5, line width=1pt]
\tikzstyle{slant point}=[regular polygon, regular polygon sides=3, draw, scale=0.75, inner sep=-0.5pt, minimum width=9mm, fill white, regular polygon rotate=180, yslant=-0.3]
\tikzstyle{tiny orange label}=[font={\tiny}, text={rgb,255: red,255; green,128; blue,0}, tikzit draw={rgb,255: red,255; green,128; blue,0}]
\tikzstyle{tiny red label}=[font={\tiny}, text={rgb,255: red,191; green,0; blue,64}, tikzit draw={rgb,255: red,191; green,0; blue,64}, draw=none]
\tikzstyle{red label}=[text={rgb,255: red,191; green,0; blue,64}, tikzit draw={rgb,255: red,191; green,0; blue,64}]
\tikzstyle{slant label black}=[font={\tiny}, xslant=0.5, tikzit draw=black]
\tikzstyle{slant label red}=[font={\tiny}, xslant=0.5, text={rgb,255: red,191; green,0; blue,64}, tikzit draw={rgb,255: red,191; green,0; blue,64}]
\tikzstyle{slant label orange}=[font={\tiny}, xslant=0.5, text={rgb,255: red,255; green,128; blue,0}, tikzit draw={rgb,255: red,255; green,128; blue,0}]
\tikzstyle{slanted point}=[fill={rgb,255: red,191; green,0; blue,64}, draw={rgb,255: red,191; green,0; blue,64}, shape=triangle, regular polygon, regular polygon sides=3, scale=0.75, inner sep=-0.5pt, minimum width=5mm, regular polygon rotate=90, xslant=0.5, fill opacity=0.5, font={\tiny}, line width=1pt, text={rgb,255: red,191; green,0; blue,64}]
\tikzstyle{slanted point black}=[draw=black, shape=triangle, regular polygon, regular polygon sides=3, scale=0.75, inner sep=-0.5pt, minimum width=5mm, regular polygon rotate=90, xslant=0.5, font={\tiny}, line width=0.2pt, text=black, fill=white, tikzit fill=white]
\tikzstyle{red dot}=[fill={rgb,255: red,191; green,0; blue,64}, draw={rgb,255: red,191; green,0; blue,64}, shape=circle, inner sep=0, minimum width=1.5mm, minimum height=1.5mm]
\tikzstyle{black dot}=[fill=black, draw=black, shape=circle, inner sep=0, minimum width=1.5mm, minimum height=1.5mm]
\tikzstyle{orange dot}=[fill={rgb,255: red,255; green,128; blue,0}, draw={rgb,255: red,255; green,128; blue,0}, shape=circle, inner sep=0, minimum width=1.5mm, minimum height=1.5mm]
\tikzstyle{blue dot}=[fill={rgb,255: red,0; green,0; blue,228}, draw={rgb,255: red,0; green,0; blue,228}, shape=circle, inner sep=0, minimum width=1.5mm, minimum height=1.5mm]
\tikzstyle{slant white}=[fill=white, draw=black, shape=rectangle, xslant=0.5, font={\tiny}, line width=1pt]
\tikzstyle{slant small map}=[fill=white, draw=black, xslant=0.5, shape=rectangle, font={\tiny}, line width=1pt, inner sep=0.6mm]
\tikzstyle{slanted copoint black}=[draw=black, shape=triangle, regular polygon, regular polygon sides=3, scale=0.75, inner sep=-0.5pt, minimum width=5mm, regular polygon rotate=-90, xslant=0.5, font={\tiny}, line width=0.2pt, text=black, fill=white, tikzit fill=white]
\tikzstyle{purple dot}=[fill={rgb,255: red,128; green,0; blue,128}, draw={rgb,255: red,128; green,0; blue,128}, shape=circle, inner sep=0, minimum width=1.5mm, minimum height=1.5mm]
\tikzstyle{white dot 2}=[fill=white, draw=black, shape=circle]
\tikzstyle{horizontal point}=[style=point, rotate=-90, tikzit shape=rectangle, tikzit fill={rgb,255: red,191; green,128; blue,64}]
\tikzstyle{pslant orange}=[style=slanted point black, fill={rgb,255: red,255; green,128; blue,0}, draw={rgb,255: red,255; green,128; blue,0}, tikzit fill={rgb,255: red,255; green,128; blue,0}, tikzit draw={rgb,255: red,255; green,128; blue,0}]
\tikzstyle{upground horizontal}=[style=upground, rotate=-90]
\tikzstyle{double horizontal point}=[style=horizontal point, line width=1pt]
\tikzstyle{double point}=[style=point, line width=1pt]
\tikzstyle{double copoint}=[style=copoint, line width=1pt]
\tikzstyle{horizontal copoint}=[style=double copoint, rotate=-90]
\tikzstyle{slant label purple}=[style=slant label black, tikzit draw={rgb,255: red,128; green,0; blue,128}, text={rgb,255: red,128; green,0; blue,128}]
\tikzstyle{orange copoint}=[style=pslant orange, rotate=-180, tikzit fill={rgb,255: red,255; green,128; blue,0}]
\tikzstyle{new style 0}=[style=slant white, draw={rgb,255: red,0; green,0; blue,228}, fill={rgb,255: red,0; green,0; blue,228}, fill opacity=0.5, shape=rectangle]
\tikzstyle{wide slanted point}=[style=wide point, xslant=0.5, fill=white, rotate=-90, minimum width=0.8cm, fill={rgb,255: red,128; green,128; blue,128}, fill opacity=0.5, line width=1pt]
\tikzstyle{black dot white}=[style=black dot, text=white, draw=none, tikzit draw={rgb,255: red,191; green,255; blue,0}, shape=circle]

\tikzstyle{new edge style 1}=[-, line width=1pt, shorten <=-0.17mm, shorten >=-0.17mm, tikzit draw={rgb,255: red,204; green,0; blue,3}]
\tikzstyle{diredge}=[-, postaction=decorate, decoration={markings, mark=at position 0.55 with \edgearrow}]
\tikzstyle{bold diredge}=[-, diredge, line width=1pt, tikzit draw={rgb,255: red,128; green,0; blue,128}]
\tikzstyle{grey}=[-, draw={rgb,255: red,188; green,188; blue,188}]
\tikzstyle{classical}=[-, dashed, tikzit draw={rgb,255: red,255; green,128; blue,0}]
\tikzstyle{reddashed}=[-, dashed, draw={rgb,255: red,0; green,128; blue,128}, postaction=decorate, decoration={markings, mark=at position 0.55 with \edgearrow}]
\tikzstyle{reddahednoarrow}=[-, dashed, draw={rgb,255: red,179; green,40; blue,40}]
\tikzstyle{arrow edge}=[-, ->, draw={rgb,255: red,191; green,191; blue,191}, tikzit draw={rgb,255: red,191; green,191; blue,191}, ultra thick]
\tikzstyle{tarrow edge}=[-, ->, draw={rgb,255: red,191; green,191; blue,191}, tikzit draw={rgb,255: red,191; green,191; blue,191}]
\tikzstyle{gray edge}=[-, draw={rgb,255: red,191; green,191; blue,191}, tikzit draw={rgb,255: red,191; green,191; blue,191}, ultra thick]
\tikzstyle{lightgrayedge}=[-, draw={rgb,255: red,207; green,207; blue,207}]
\tikzstyle{green edge}=[-, tikzit draw={rgb,255: red,128; green,128; blue,0}, draw={rgb,255: red,128; green,128; blue,0}]
\tikzstyle{red edge}=[-, draw={rgb,255: red,191; green,0; blue,64}, tikzit draw={rgb,255: red,191; green,0; blue,64}]
\tikzstyle{arrow edge black}=[-, ->]
\tikzstyle{solid blue}=[-, draw={rgb,255: red,0; green,96; blue,167}, tikzit draw={rgb,255: red,0; green,96; blue,167}]
\tikzstyle{classical blue}=[-, draw={rgb,255: red,0; green,96; blue,167}, tikzit draw={rgb,255: red,0; green,96; blue,167}, dashed]
\tikzstyle{fill gray}=[-, fill=gray]
\tikzstyle{bold gray}=[-, line width=1pt, tikzit draw={rgb,255: red,128; green,128; blue,128}]
\tikzstyle{fill pink}=[-, fill={rgb,255: red,193; green,100; blue,94}, fill opacity=0.5, draw={rgb,255: red,134; green,68; blue,65}, line width=1pt, tikzit draw={rgb,255: red,134; green,68; blue,65}, tikzit fill={rgb,255: red,193; green,100; blue,94}]
\tikzstyle{fill carta da zucchero}=[-, fill={rgb,255: red,129; green,158; blue,219}, fill opacity=0.5, line width=0.4mm]
\tikzstyle{fill white}=[-, fill=white]
\tikzstyle{fill purple}=[-, fill={rgb,255: red,113; green,69; blue,128}, fill opacity=0.5, draw={rgb,255: red,79; green,48; blue,90}, tikzit fill={rgb,255: red,113; green,69; blue,128}, tikzit draw={rgb,255: red,79; green,48; blue,90}, line width=1pt]
\tikzstyle{fill green}=[-, fill={rgb,255: red,62; green,128; blue,120}, fill opacity=0.5, draw={rgb,255: red,33; green,68; blue,63}, tikzit fill={rgb,255: red,62; green,128; blue,120}, tikzit draw={rgb,255: red,33; green,68; blue,63}, line width=1pt]
\tikzstyle{bold orange}=[-, draw={rgb,255: red,255; green,128; blue,0}, fill=none, line width=1pt]
\tikzstyle{bold black}=[-, line width=1pt, draw=black, fill=none, tikzit draw=black]
\tikzstyle{bold red}=[-, draw={rgb,255: red,191; green,0; blue,64}, fill=none, line width=1pt]
\tikzstyle{fill light green}=[-, fill={rgb,255: red,166; green,166; blue,112}, fill opacity=0.5, draw={rgb,255: red,121; green,121; blue,81}, line width=1pt]
\tikzstyle{new edge style 0}=[-, fill=yellow, fill opacity=0.5, draw={rgb,255: red,146; green,146; blue,0}, tikzit fill=yellow, tikzit draw={rgb,255: red,146; green,146; blue,0}]
\tikzstyle{bold dashed red}=[-, draw={rgb,255: red,191; green,0; blue,64}, fill=none, line width=1pt, dashed]
\tikzstyle{bold dashed orange}=[-, draw={rgb,255: red,255; green,128; blue,0}, dashed, line width=1pt]
\tikzstyle{bold blue}=[-, draw={rgb,255: red,0; green,0; blue,228}, line width=1pt]
\tikzstyle{arrow red}=[draw={rgb,255: red,191; green,0; blue,64}, ->, line width=1pt]
\tikzstyle{new edge style 2}=[-, draw={rgb,255: red,191; green,0; blue,64}, line width=1pt]
\tikzstyle{boldish}=[-, line width=0.6mm, fill=cyan]
\tikzstyle{white edge}=[-, draw=white]
\tikzstyle{purple edge}=[-, draw={rgb,255: red,128; green,0; blue,128}, line width=1pt]
\tikzstyle{light gray}=[-, fill={rgb,255: red,191; green,191; blue,191}, draw={rgb,255: red,191; green,191; blue,191}, tikzit fill={rgb,255: red,191; green,191; blue,191}, tikzit draw={rgb,255: red,191; green,191; blue,191}, fill opacity=0.3]
\tikzstyle{invisible edge}=[-, fill opacity=0, fill=none]
\tikzstyle{carta da zucchero thin}=[-, style=fill carta da zucchero, line width=0.1pt, fill={rgb,255: red,129; green,158; blue,219}, tikzit fill={rgb,255: red,129; green,158; blue,219}]
\tikzstyle{pink thin}=[-, style=fill pink, line width=0.1pt, fill={rgb,255: red,193; green,100; blue,94}]
\tikzstyle{fill green thin edge}=[-, style=fill green, tikzit fill={rgb,255: red,62; green,128; blue,120}, line width=0.1pt]

\definecolor{evred}{rgb}{0.996, 0.403, 0.537}
\definecolor{evgreen}{rgb}{0.501, 1.0, 0.505}
\definecolor{evblue}{rgb}{0.2, 0.588, 1.0}

\usepackage[frozencache]{minted} 

\setminted[python]{
bgcolor=lightgray!10,
xleftmargin=2mm,
fontsize=\footnotesize,
frame=lines,
linenos,
python3=true}
\setminted[text]{
bgcolor=lightgray!10,
xleftmargin=2mm,
fontsize=\footnotesize}

\theoremstyle{definition}
\newtheorem{definition}{Definition}[section]
\newtheorem{remark}{Remark}[section]
\newtheorem{theorem}{Theorem}[section]
\newtheorem{proposition}[theorem]{Proposition}

\newtheorem{observation}[theorem]{Observation}
\newtheorem{corollary}[theorem]{Corollary}

\usepackage{xifthen}

\newcommand{\opapp}[2]{\ensuremath{#1\left(#2\right)}} 
\newcommand{\opapptxt}[2]{\ensuremath{\text{#1}\left(#2\right)}} 

\newcommand{\tsuchthat}[2]{\ensuremath{\left\{#1\middle|#2\right\}}} 
\newcommand{\suchthat}[2]{\tsuchthat{\,#1\,}{\,#2\,}} 

\newcommand{\downset}[1]{\ensuremath{#1\!\downarrow}}
\newcommand{\upset}[1]{\ensuremath{#1\!\uparrow}}
\newcommand{\domSym}{\text{dom}}
\newcommand{\dom}[1]{\opapp{\domSym}{#1}}

\newcommand{\restrict}[2]{#1|_{#2}}
\newcommand{\Subsets}[1]{\opapp{\mathcal{P}\!}{#1}} 
\newcommand{\PFun}[1]{\opapptxt{PFun}{#1}} 

\newcommand{\ev}[1]{\text{#1}} 
\newcommand{\discrete}[1]{\opapptxt{discrete}{#1}} 
\newcommand{\indiscrete}[1]{\opapptxt{indiscrete}{#1}} 
\newcommand{\total}[1]{\opapptxt{total}{#1}} 
\newcommand{\seqcomposeSym}{\rightsquigarrow}
\newcommand{\causeqcls}[1]{\ensuremath{\left[#1\right]_{\simeq}}}
\newcommand{\LsetsSym}{\Lambda} 
\newcommand{\Lsets}[1]{\opapp{\LsetsSym}{#1}} 
\newcommand{\Hist}[1]{\opapptxt{Hist}{#1}} 
\newcommand{\ExtHist}[1]{\opapptxt{ExtHist}{#1}} 
\newcommand{\Ext}[1]{\opapptxt{Ext}{#1}} 
\newcommand{\Prime}[1]{\opapptxt{Prime}{#1}} 
\newcommand{\allJoinsSym}{\;\dot{\vee}\;}
\newcommand{\tips}[2]{\opapp{\text{tips}_{#1}}{#2}} 
\newcommand{\tip}[2]{\opapp{\text{tip}_{#1}}{#2}} 
\newcommand{\CausCompl}[1]{\opapptxt{CausCompl}{#1}} 
\newcommand{\Events}[1]{{E}^{#1}} 
\newcommand{\Inputs}[1]{{I}^{#1}} 
\newcommand{\AllSpaces}{\ensuremath{\text{Spaces}}} 
\newcommand{\Spaces}[1]{\opapptxt{Spaces}{#1}} 
\newcommand{\SpacesFC}[1]{\opapp{\text{Spaces}_{\text{FC}}}{#1}} 
\newcommand{\CCSpaces}[1]{\opapptxt{CCSpaces}{#1}} 
\newcommand{\CSwitchSpaces}[1]{\opapptxt{CSwitchSpaces}{#1}} 

\newcommand{\hist}[1]{
    \ensuremath{
        \left\{
            \foreach \i\j [count=\idx] in {#1}{%
                \ifnum\idx=1%
                    \ev{\i}\!:\!\j%
                \else%
                    ,\,\ev{\i}\!:\!\j%
                \fi%
            }
        \right\}
    }
}
\newcommand{\evset}[1]{
    \ensuremath{
        \left\{
            \foreach \i [count=\idx] in {#1}{%
                \ifnum\idx=1%
                    \ev{\i}%
                \else%
                    ,\ev{\i}%
                \fi%
            }
        \right\}
    }
}

\usetikzlibrary{external}
\begin{document}

\title{The Combinatorics of Causality}

\author{Stefano Gogioso$^{1,2}$ and Nicola Pinzani$^{1,3}$}

\address{$^1$Hashberg Ltd, London, UK}
\address{$^2$Department of Computer Science, University of Oxford, Oxford, UK}
\address{$^3$QuIC, Universit\'{e} Libre de Bruxelles, Brussels, BE}
\ead{$^1$stefano.gogioso@cs.ox.ac.uk, $^2$nicola.pinzani@ulb.be}
\vspace{10pt}

\begin{abstract}
    We introduce and explore the notion of ``spaces of input histories'', a broad family of combinatorial objects which can be used to model input-dependent, dynamical causal order.
    We motivate our definition with reference to traditional partial order- and preorder-based notions of causal order---adopted by the majority of previous literature on the subject---and we proceed to explore the novel landscape of combinatorial complexity made available by our generalisation of those notions.

    In the process, we discover that the fine-grained structure of causality is significantly more complex than we might have previously believed: in the simplest case of binary inputs, the number of available ``causally complete'' spaces grows from 7 on 2 events, to 2644 on 3 events, to an unknown number---likely around a billion---on 4 events.
    For perspective, previous literature on non-locality and contextuality used a single one of the 2644 available spaces on 3 events, work on definite causality used 19 spaces---derived from partial orders---and work on indefinite causality used only 6 more, for a grand total of 25.

    This paper is the first instalment in a trilogy: sthe sheaf-theoretic treatment of causal distributions will be detailed in ``The Topology of Causality'', while the polytopes formed by the associated empirical models will be studied in ``The Geometry of Causality''.
    An exhaustive classification of the 2644 causally complete spaces on 3 events with binary inputs is provided in the supplementary work ``Classification of causally complete spaces on 3 events with binary inputs'', together with the algorithm used for the classification and partial results from the ongoing search on 4 events.
\end{abstract}

\maketitle







\section{Introduction}
\label{section:introduction}

The attempt of providing a description of spacetime involving discrete atomic components is as old as special relativity itself: the first appearance of a structure resembling a causal order dates back at least as far as 1914, in a book by Alfred Robb titled ``A Theory of Time and Space'' \cite{robb1914theory}.
Robb attempts an axiomatic derivation of special relativity from temporal succession of individual events: the relativity of simultaneity is captured by the introduction of a different type of ordering, taking into account the possibility that individual events may be incomparable with respect to the succession of time.
This desideratum is embodied in the description of a partial order, leading to the question: What is the physical meaning of this new kind of connections between events?
Robb's response is that such connections represent the ``possibility'' of causal influence:
\begin{quotation}
\noindent If an instant \ev{B} be distinct from an instant \ev{A}, then \ev{B} will be said to be \textit{after} \ev{A}, if, and only if, it be abstractly possible for a person, at the instant \ev{A}, to produce an effect at the instant \ev{B}.\\
\phantom{.}\hfill--- Alfred Robb, ``A Theory of Time and Space'' \cite{robb1914theory}, p.7
\end{quotation}
It is a celebrated result by David Malament \cite{Malament1977} that the reconstruction of Robb can be extended, at least in spirit, from special relativity to general relativity.
Specifically, Malament showed that an isomorphism of causal structures between two sufficiently well-behaved spacetimes $(\mathcal{M},g_{ab})$ and $(\mathcal{M'},g'_{ab})$ can be extended to a smooth conformal isometry.
In other words, the underlying combinatorial structure of causality fully constrains the topological, differential and conformal structures built on top of it.

Norbert Weiner, reviewing the work by Robb in \cite{wiener1916thoery}, found Robb's description of causal order in terms of agency and the ``abstract possibility'' of influence to be ``utterly pointless'', feeling that the notion of causality itself was at least as obscure to that of time succession.
In this work, instead, we take Robb's perspective as the core for a broad combinatorial generalisation of causal orders: instead of talking about causal relation between events as the presence---or even the necessity---of causal influence, we take the \textit{absence} of causal relation to impose the \textit{impossibility} of causal influence.

For example, consider the following causal relationships between three events \ev{A}, \ev{B} and \ev{C}, where arrows indicate that the head event is ``after'' the tail event in the causal order:
\begin{center}
    \includegraphics[height=2cm]{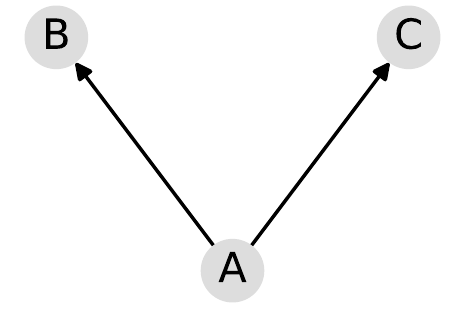}
\end{center}
The arrows $\ev{A} \rightarrow \ev{B}$ and $\ev{A} \rightarrow \ev{C}$ indicate the \textit{possibility} for an action at event \ev{A} to causally influence outcomes at events $\ev{B}$ and $\ev{C}$, but leaves open the possibility for such influence not to be exerted.
On the other hand, the absence of an arrow $\ev{B} \rightarrow \ev{C}$ \textit{mandates} the \textit{impossibility} for an action at event \ev{B} to causally influence outcomes at event \ev{C}, and analogously the absence of an arrow $\ev{C} \rightarrow \ev{B}$ mandates the impossibility for an action at event \ev{C} to causally influence outcomes at event \ev{B}.
In other words, events \ev{B} and \ev{C} are causally unrelated.

Both Robb's original perspective and Malament's result fall under the remit of \textit{static}---or \textit{definite}---causality: the causal order between events is fixed and cannot be influenced by the processes which it constrains.
On the other hand, proposals for generalised notions of computation consistent with a \textit{dynamical}---or \textit{indefinite}---causal background date back at least as far as a 2009 work on quantum information by Lucien Hardy \cite{Hardy2009}.
This was followed by the introduction of the \textit{quantum switch} \cite{chiribella2013quantum,goswami2018indefinite} and subsequent attempts to justify the experimental realisability of indefinite causal order \cite{procopio2015experimental,rubino2017experimental,goswami2018indefinite,rubino2019experimental,bavaresco2019semidevice,rubino2021experimental,dourdent2021semidevice}.
It was shown that superposition of the causal orders for quantum instruments is mathematically sound \cite{pinzani2020giving,wechs2021quantum}, overcoming known issues with more general notions of quantum control \cite{oi2003interference,chiribella2019quantum,abbott2020communication}.
The need for a generalisation of the notions of sequential and parallel composition of processes prompted the definition of \textit{quantum supermaps}, where quantum theory is enhanced by the possibility of indefinite causal order.
The associated notion of \textit{process matrices} was introduced in 2012 by Oreshkov, Costa and Bruckner \cite{oreshkov2012quantum} and it quickly developed into an active area of research \cite{branciard2015simplest,abbott2016multipartite,araujo2017purification,kissinger2017categorical,wechs2021quantum}.

In this work, we introduce and investigate a combinatorial generalisation of causal orders which explicitly allows for the possibility of indefinite and dynamical order (as well as a variety of exotic input-dependent causality constraints).
Events are replaced by \textit{input histories}, which are defined to be all possible combinations of inputs upon which the outputs at events are allowed to depend: ordering the input histories by extension yields combinatorial objects known to us as \textit{spaces of input histories}.
For example, below is the space of input histories for our previous causal order on 3 events, in the case where each event admits a binary input: 
\begin{center}
    \includegraphics[height=3cm]{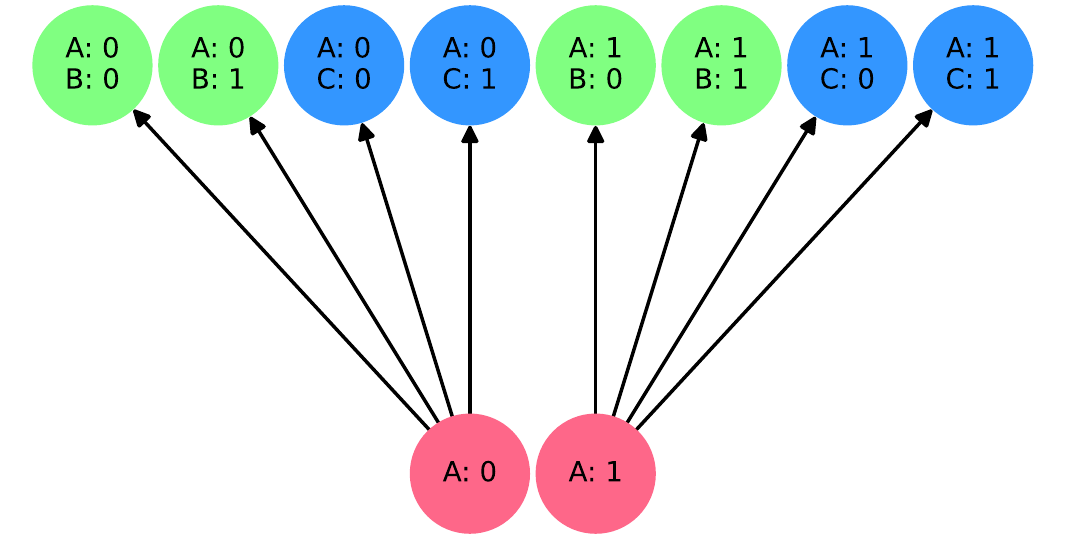}
\end{center}
The output at event \ev{A} is allowed to depend on input histories $\hist{A/i_A}$ for $i_A \in \{0,1\}$ (highlighted in red), the output at event \ev{B} is allowed to depend in input histories $\hist{A/i_A,B/i_B}$ for $i_A, i_B \in \{0,1\}$ and the output at event \ev{C} is allowed to depend in input histories $\hist{A/i_A,C/i_C}$ for $i_A, i_C \in \{0,1\}$.
As a second, more interesting example, the space of input histories below captures a form of \textit{dynamical} causal order, on 3 events with binary inputs:
\begin{center}
    \includegraphics[height=3cm]{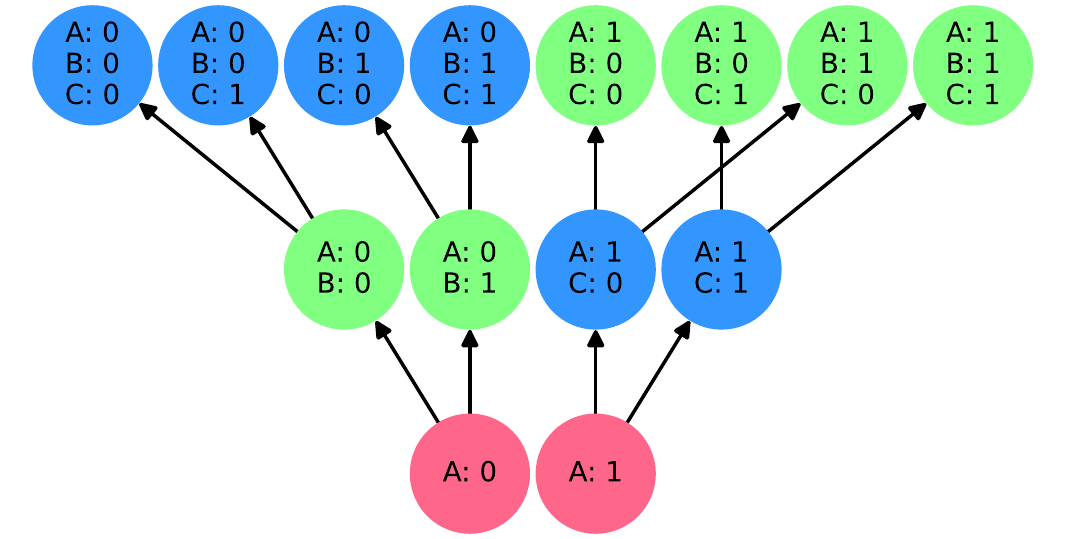}
\end{center}
The input histories indicate that the input $i_A \in \{0,1\}$ at event \ev{A} determines the causal order between events \ev{B,C}: input histories $\hist{A/0,B/i_B,C/i_C}$ (highlighted blue) show that the output at \ev{C} is allowed to depend on the input $i_B$ at event \ev{B} when $i_A = 0$, while input histories $\hist{A/1,C/i_C}$ (also highlighted blue) show that the same output is not allowed to depend on the input $i_B$ at event \ev{B} when $i_A=1$.














\section{Causal orders}
\label{section:causal-orders}

In this work, a \emph{causal order} models a discrete subset of events in some (suitably well-behaved) spacetime, where all details other than event identity and their mutual causal relationship have been abstracted away.
The possibility of abstracting from geometric to order-theoretic structures is consequence of a much celebrated result by Malament \cite{Malament1977}, based on previous work by Kronheimer, Penrose, Hawking, King and McCarthy \cite{Kronheimer1967,Hawking1976}: subject to a mild requirement of causal well-behaviour, Lorentzian manifolds are identified (up to conformal equivalence) by the causal order between their events.
Malament's result is the motivation behind many past and current lines of enquiry in causality: examples include the "causal sets" research programme \cite{Bombelli1987}, the domain-theoretic investigations of Martin and Panangaden \cite{Martin2010,Martin2012}, and the functorial approach to quantum field dynamics \cite{gogioso2021fields}.

The works mentioned above are all concerned with recovering relativistic structure or understanding quantum fields in an approximation of the spacetime continuum.
Here, instead, we will limit our efforts to the needs of quantum information protocols and experiments, where operations are performed locally at a finite set of spacetime events.
Furthermore, our approach is independent of both the theory underpinning the experiments and the concrete realisation of the local operations involved.
As a consequence, we will work directly with finite order structures, without the need for Lorentzian geometry to ever be involved.

\subsection{Causal Orders and Hasse Diagrams}
\label{subsection:causal-orders-intro}

\begin{definition}
A \emph{causal order} $\Omega$ is a preorder: a set $|\Omega|$ of events---finite, in this work---equipped with a reflexive transitive relation $\leq$, which we refer to as the \emph{causal relation}.
In cases where multiple cause orders are involved, we might also use the more explicit notation $\leq_{\Omega}$, to indicate that the relation is order-dependent.
\end{definition}

\begin{definition}
There are four possible ways in which two distinct events $\omega, \xi \in \Omega$ can relate to each other causally:
\begin{itemize}
\item $\omega$ \emph{causally precedes} $\xi$ if $\omega \leq \xi$ and $\xi \not \leq \omega$, which we write succinctly as $\omega \prec \xi$ (to distinguish it from $\omega < \xi$, meaning instead that $\omega \leq \xi$ and $\omega \neq \xi$)
\item $\omega$ \emph{causally succeeds} $\xi$ if $\xi \leq \omega$ and $\omega \not \leq \xi$, which we write succinctly as $\omega \succ \xi$ (to distinguish it from $\omega > \xi$, meaning instead that $\omega \geq \xi$ and $\omega \neq \xi$)
\item $\omega$ and $\xi$ are \emph{causally unrelated} if $\omega \not\leq \xi$ and $\xi \not \leq \omega$
\item $\omega$ and $\xi$ are in \emph{indefinite causal order} if $\omega \neq \xi$, $\omega \leq \xi$ and $\xi \leq \omega$, which we write succinctly as $\omega \simeq \xi$
\end{itemize}
We say that a causal order is \emph{definite} when the last case cannot occur, i.e. when $\leq$ is anti-symmetric ($\omega \leq \xi$ and $\omega \geq \xi$ together imply $\omega = \xi$); otherwise, we say that it is \emph{indefinite}.
A definite causal order is thus a \emph{partial order}, or \emph{poset}: in this case, $\omega\prec\xi$ is the same as $\omega<\xi$, and $\omega\succ\xi$ is the same as $\omega > \xi$.
\end{definition}

\begin{definition}
We say that two events $\omega, \xi$ are \emph{causally related} if they are not causally unrelated, i.e. if at least one of $\omega \leq \xi$ or $\omega \geq \xi$ holds.
We also define the \emph{causal past} $\downset{\omega}$ and \emph{causal future} $\upset{\omega}$ of an event $\omega \in \Omega$, as well as its \emph{causal equivalence class} $\causeqcls{\omega}$:
\begin{eqnarray}
    \downset{\omega} & := \suchthat{\xi \in \Omega}{\xi \leq \omega} \\
    \upset{\omega} & := \suchthat{\xi \in \Omega}{\xi \geq \omega} \\
    \causeqcls{\omega} & :=  \suchthat{\xi \in \Omega}{\xi \simeq \omega} = \downset{\omega} \cap\; \upset{\omega}
\end{eqnarray}
Note that the $\omega$ always lies in both its own causal future and its own causal past, but also that their intersection can comprise more events (if the order is indefinite).
\end{definition}

Our interpretation of causality is a "negative" one, as "no-signalling from the future": when $\omega$ causally precedes $\xi$, for example, we are not so much interested in the "possibility" of causal influence from $\omega$ to $\xi$ (because $\omega \leq \xi$) as we are in the "impossibility" of causal influence from $\xi$ to $\omega$ (because $\xi \not \leq \omega$).
This generalises the "spatial" no-signalling case, where one is interested in the statements $\omega \not\leq \xi$ and $\xi \not \leq \omega$.
Far from being merely an interpretation, such no-signalling approach to causality permeates the entirety of this work.
From a topological perspective, it points to the lattice of lowersets $\Lambda(\Omega)$ of a causal order $\Omega$ as the correct combinatorial object to consider. Indeed, the inclusion order $U \subseteq V$ of lowersets is defined by the follow condition.
\[
    U \subseteq V
    \;\Leftrightarrow\;
    \forall \xi \in V \backslash U.\;
    \forall \omega \in U.\;
    \xi \not \leq \omega
\]
From a geometric perspective, it leads to an alternative way to define causal polytopes, by restriction of a simpler higher-dimensional polytope rather than using causal inequalities directly \cite{gogioso2022geometry}.
Specifically, causal polytopes are defined by slicing the polytope of conditional probability distributions---a product of simplices---with hyperplanes determined by no-signalling equations.

Definite causal orders have an equivalent presentation as directed acyclic graphs (DAGs), known as \emph{Hasse diagrams}: vertices in the graph correspond to events $\omega \in \Omega$, while edges $x \rightarrow y$ correspond to those causally related pairs $\omega \leq \xi$ with no intermediate event (i.e. where there is no $\zeta \in \Omega$ such that $\omega < \zeta < \xi$).
For example, below are the Hasse diagrams for three definite causal orders on three events \ev{A}, \ev{B} and \ev{C}.
\begin{center}
    \includegraphics[height=2.5cm]{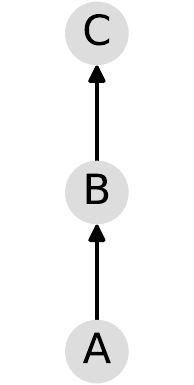}
    \hspace{1.5cm}
    \includegraphics[height=2cm]{svg-inkscape/vee-A-BC_svg-tex.pdf}
    \hspace{1.5cm}
    \includegraphics[height=2cm]{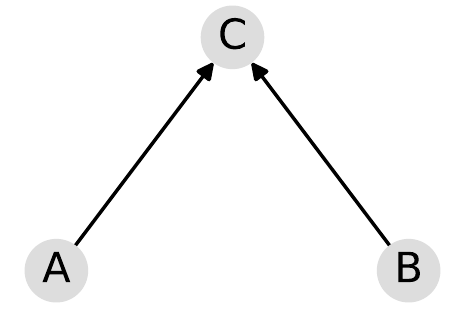}
\end{center}
On the left, \ev{A} causally precedes \ev{B}, which in turn causally precedes \ev{C}: this is an example of a \emph{total} order, one corresponding to a line Hasse diagram.
In the middle, \ev{A} causally precedes both \ev{B} and \ev{C}, which are causally unrelated to one another.
On the right, \ev{C} causally succeeds both \ev{A} and \ev{B}, which are causally unrelated to one another.

\begin{remark}
More precisely, there is a bijective correspondence between finite partial orders and finite \textit{intransitive} DAGs---loosely speaking, those without unnecessary edges \cite{pinzani2019categorical}.
The correspondence further generalises to \textit{locally finite} partial orders---where any two elements have finitely many elements in between---and arbitrary intransitive DAGs \cite{gogioso2021fields}.
\end{remark}

The Hasse diagram representation extends to arbitrary causal orders, by making vertices in the graph correspond to causal equivalence classes instead of individual events; in the case of definite orders, the equivalence classes are all singletons, and can be safely identified with the unique event they contain.
For example, below are the Hasse diagrams for three indefinite causal orders on three events \ev{A}, \ev{B} and \ev{C}.
\begin{center}
    \includegraphics[height=2.5cm]{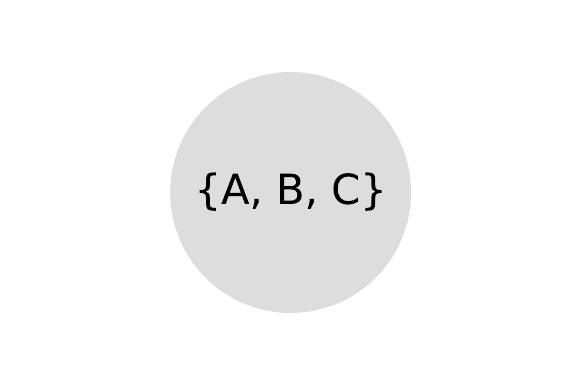}
    \hspace{1.5cm}
    \includegraphics[height=2.5cm]{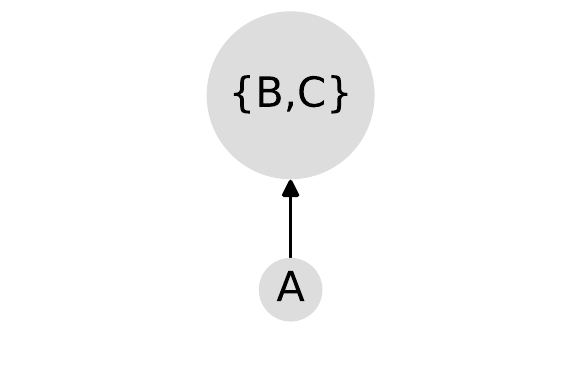}
    \hspace{1.5cm}
    \includegraphics[height=2.5cm]{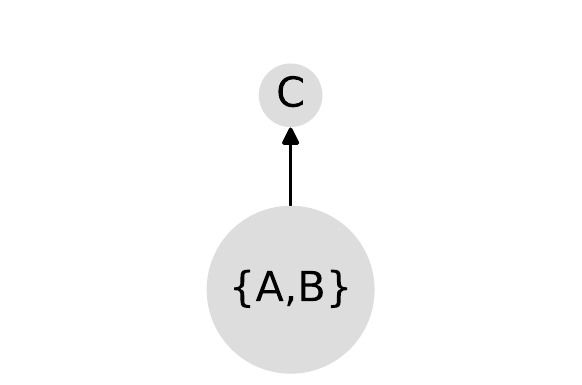}
\end{center}

\subsection{Operations on Causal Orders}
\label{subsection:causal-orders-ops}

Rather than defining the graphs themselves, it is often convenient to construct complex causal orders from simpler components, by combining smaller orders together, refining larger orders, or a combination of both.
On a given set of event, all orders lie between two extremes: the \emph{discrete order}, where all elements are causally unrelated, and the \emph{indiscrete order}, where all elements lie in a single causal equivalence class.
Below are the discrete and indiscrete orders on three events \ev{A}, \ev{B} and \ev{C}.
\begin{center}
    \includegraphics[height=2.5cm]{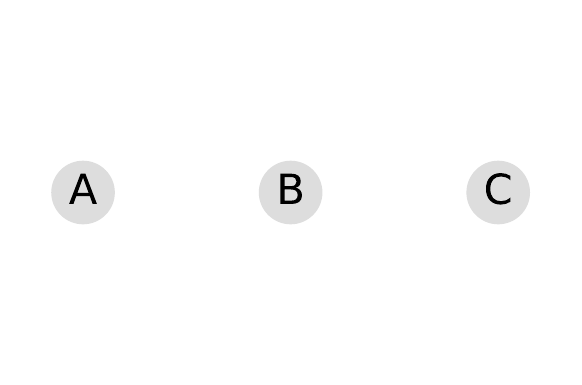}
    \hspace{3cm}
    \includegraphics[height=2.5cm]{svg-inkscape/indiscrete-ABC_svg-tex.pdf}
\end{center}
Additionally, the $n!$ possible \emph{total orders} on $n$ events are often of interest.
For example, below are the 6 total orders on three events \ev{A}, \ev{B} and \ev{C}.
\begin{center}
    \includegraphics[height=2.5cm]{svg-inkscape/total-ABC_svg-tex.pdf}
    \hspace{0.75cm}
    \includegraphics[height=2.5cm]{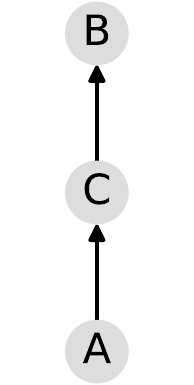}
    \hspace{0.75cm}
    \includegraphics[height=2.5cm]{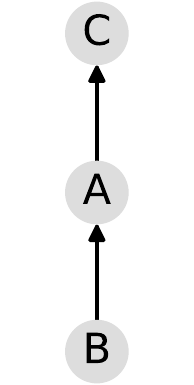}
    \hspace{0.75cm}
    \includegraphics[height=2.5cm]{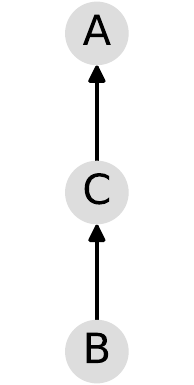}
    \hspace{0.75cm}
    \includegraphics[height=2.5cm]{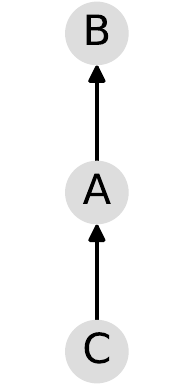}
    \hspace{0.75cm}
    \includegraphics[height=2.5cm]{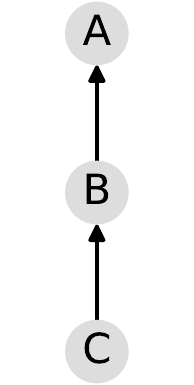}
\end{center}
\begin{definition}
For any finite set $X$ of events, we write $\discrete{X}$ for the discrete order on the events and $\indiscrete{X}$ for the indiscrete order.
For any finite sequence $\omega_1,...,\omega_n$ of events, we write $\total{\omega_1,...,\omega_n}$ for the total order on the events which matches the sequence order.
\end{definition}

\subsubsection{Join/union of causal orders}

\begin{definition}
The \emph{join} of a family $(\Omega_j)_{j=1}^n$ of causal orders, denoted by $\bigvee_{j=1}^{n} \Omega_j$, is the union of their events equipped with the transitive closure of the union of the respective causal relations.
Explicitly, two events $\omega$ and $\xi$ are related by $\omega \leq \xi$ in the join $\bigvee_{j=1}^{n} \Omega_j$ iff there is a sequence of events $(\omega_{k})_{k=0}^{m}$ and a sequence of causal orders $(\Omega_{j_k})_{k=1}^{m}$ such that $\omega_0 = \omega$, $\omega_m = \xi$ and $\omega_{k-1} \leq_{\Omega_{j_k}} \omega_{k}$ for all $k=1,...,m$.
\end{definition}

The join operation is commutative (order doesn't matter), associative (bracketing doesn't matter) and idempotent (repetition doesn't mater).
When two causal orders are \emph{disjoint}, i.e. when they share no common events, their join represents a scenario where all events from one order are causally unrelated to all events from the other, as in the following example.
\begin{center}
    \includegraphics[height=2.5cm]{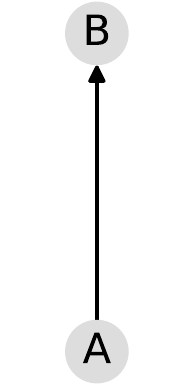}
    \hspace{0.75cm}
    \raisebox{1.15cm}{$\bigvee$}
    \hspace{0.75cm}
    \includegraphics[height=2.5cm]{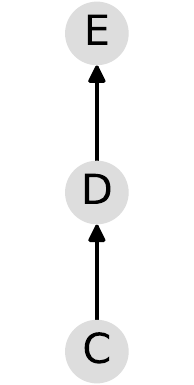}
    \hspace{0.75cm}
    \raisebox{1.15cm}{$=$}
    \hspace{0.75cm}
    \includegraphics[height=2.5cm]{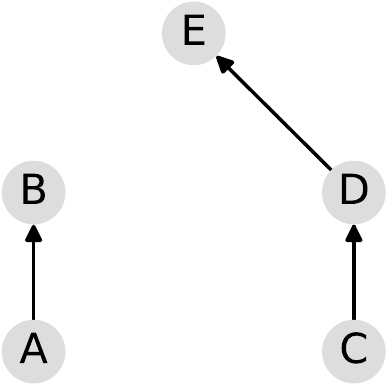}
\end{center}
Because of this behaviour, we will refer to disjoint joins as \emph{parallel composition}.
When two causal orders have events in common, their join ``glues'' them along the common events.
Below is an example of two orders sharing an initial totally-ordered segment $\ev{A} \rightarrow \ev{B}$, followed by two distinct events.
Their join then has the same initial totally ordered segment, with a fork at \ev{B} that leads to two causally unrelated events.
\begin{center}
    \includegraphics[height=2.5cm]{svg-inkscape/total-ABC_svg-tex.pdf}
    \hspace{0.75cm}
    \raisebox{1.15cm}{$\bigvee$}
    \hspace{0.75cm}
    \includegraphics[height=2.5cm]{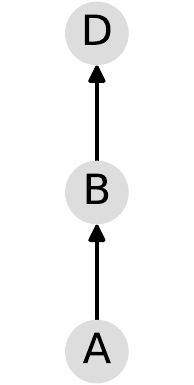}
    \hspace{0.75cm}
    \raisebox{1.15cm}{$=$}
    \hspace{0.75cm}
    \includegraphics[height=2.5cm]{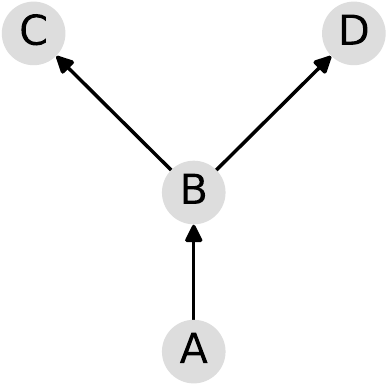}
\end{center}
Below is a second example, where the two order share a pair of causally unrelated events \ev{B} and \ev{C}.
Taking their join glues the two orders into a diamond, with causally unrelated events \ev{B} and \ev{C} separating the bottom event \ev{A} from the top event \ev{D}.
\begin{center}
    \includegraphics[height=2cm]{svg-inkscape/vee-A-BC_svg-tex.pdf}
    \hspace{0.75cm}
    \raisebox{1.15cm}{$\bigvee$}
    \hspace{0.75cm}
    \includegraphics[height=2cm]{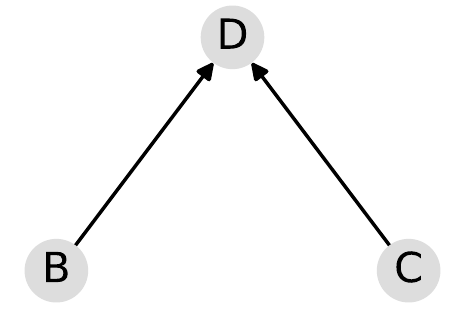}
    \hspace{0.75cm}
    \raisebox{1.15cm}{$=$}
    \hspace{0.75cm}
    \includegraphics[height=2.5cm]{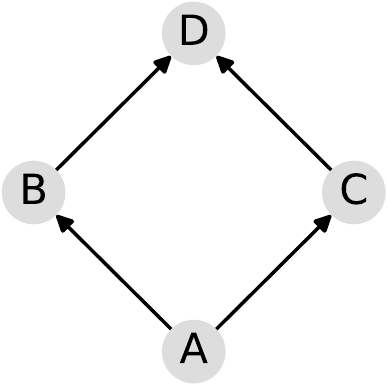}
\end{center}
In the two examples above, the common events had an identical mutual causal relation in both others.
However, it is generally the case for causal orders involved in a join to impose different causal relations on their common events.
In particular, if two events are causally related in different ways in two causal orders, then the same two events will be in indefinite causal order in the join.
For example, event \ev{B} causally precedes \ev{C} in the first causal order below, while the same event \ev{B} causally succeeds \ev{C} in the second causal order: in join (on the right), events \ev{B} and \ev{C} are therefore in indefinite causal order.
\begin{center}
    \includegraphics[height=3cm]{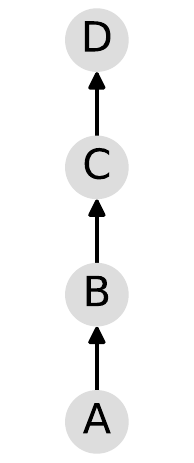}
    \hspace{0.75cm}
    \raisebox{1.4cm}{$\bigvee$}
    \hspace{0.75cm}
    \includegraphics[height=3cm]{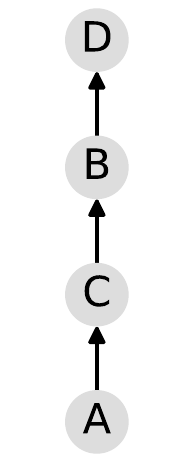}
    \hspace{0.75cm}
    \raisebox{1.4cm}{$=$}
    \hspace{0.75cm}
    \includegraphics[height=3cm]{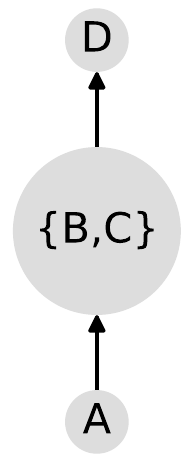}
\end{center}
The space above is not a total order, but its Hasse diagram takes the same shape, so we extend our notation slightly and write $\total{\ev{A}, \{\ev{B}, \ev{C}\}, \ev{D}}$ to denote it.

\subsubsection{Sequential composition of causal orders}

Discrete orders with two or more events can be decomposed into the join of their individual events.
Analogously, total orders with two or more events can be decomposed into the sequential composition of their individual events.
\begin{center}
    \includegraphics[height=2.5cm]{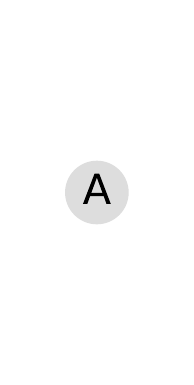}
    \hspace{0.5cm}
    \raisebox{1.15cm}{$\seqcomposeSym$}
    \hspace{0.5cm}
    \includegraphics[height=2.5cm]{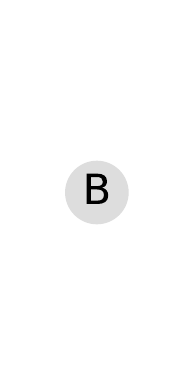}
    \hspace{0.5cm}
    \raisebox{1.15cm}{$\seqcomposeSym$}
    \hspace{0.5cm}
    \includegraphics[height=2.5cm]{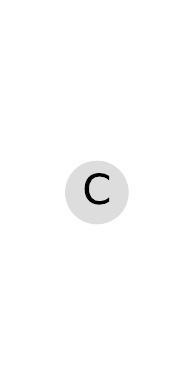}
    \hspace{0.75cm}
    \raisebox{1.15cm}{$=$}
    \hspace{0.75cm}
    \includegraphics[height=2.5cm]{svg-inkscape/total-ABC_svg-tex.pdf}
\end{center}

\begin{definition}
The \emph{sequential composition} of a family $(\Omega_j)_{j=1}^n$ of disjoint causal orders, denoted by $\Omega_1 \seqcomposeSym ... \seqcomposeSym \Omega_n$, is the (disjoint) union of their events and causal relations, with the additional stipulation that all events in $\Omega_i$ causally precede all events in $\Omega_{i+1}$ for all $i=1,...,n-1$.
Explicitly, two events $\omega \in \Omega_i$ and $\xi \in \Omega_j$ are related by $\omega \leq \xi$ in the sequential composition $\Omega_1 \seqcomposeSym ... \seqcomposeSym \Omega_n$ iff either $i < j$, or $i = j$ and $\omega \leq_{\Omega_i} \xi$. 
\end{definition}

Below is an example where a total order $\Omega$ on two events \ev{A} and \ev{B} is sequentially composed with a discrete order $\Omega'$ on two events \ev{C} and \ev{D}: this creates a fork connecting \ev{B} (the maximum of $\Omega$) to \ev{C} and \ev{D} (the minima of $\Omega'$).
\begin{center}
    \includegraphics[height=2.5cm]{svg-inkscape/total-AB_svg-tex.pdf}
    \hspace{0.75cm}
    \raisebox{1.15cm}{$\seqcomposeSym$}
    \hspace{0.75cm}
    \includegraphics[height=2.5cm]{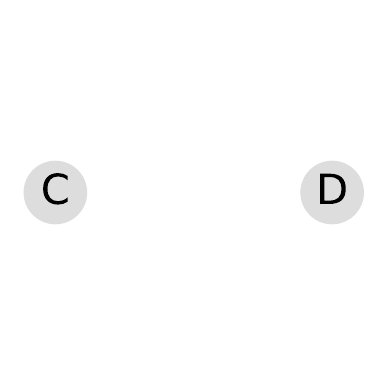}
    \hspace{0.75cm}
    \raisebox{1.15cm}{$=$}
    \hspace{0.75cm}
    \includegraphics[height=2.5cm]{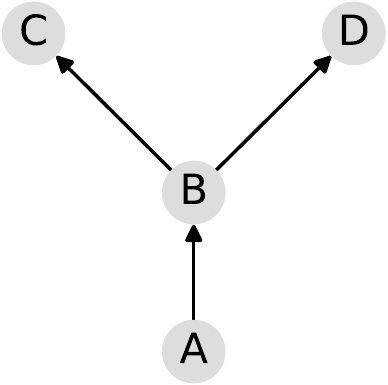}
\end{center}
Below is a sequential composition of the same total order $\Omega$ and discrete order $\Omega'$ used in the example above, but this time with $\Omega'$ preceding $\Omega$: this creates a wedge connecting \ev{C} and \ev{D} (the maxima of $\Omega'$) to \ev{A} (the minimum of $\Omega$).
\begin{center}
    \includegraphics[height=2.5cm]{svg-inkscape/discrete-CD_svg-tex.pdf}
    \hspace{0.75cm}
    \raisebox{1.15cm}{$\seqcomposeSym$}
    \hspace{0.75cm}
    \includegraphics[height=2.5cm]{svg-inkscape/total-AB_svg-tex.pdf}
    \hspace{0.75cm}
    \raisebox{1.15cm}{$=$}
    \hspace{0.75cm}
    \includegraphics[height=2.5cm]{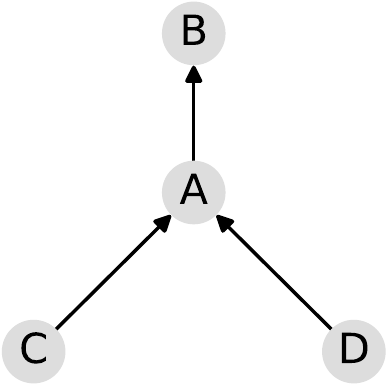}
\end{center}

\subsubsection{Meet/intersection of causal orders}

\begin{definition}
The \emph{meet} of a family $(\Omega_j)_{j=1}^n$ of causal orders, denoted by $\bigwedge_{j=1}^{n} \Omega_j$, is the intersection of the events from the individual orders, equipped with the intersection of the respective causal relations.
Explicitly, two events $\omega$ and $\xi$ are related by $\omega \leq \xi$ in the meet $\bigwedge_{j=1}^n \Omega_j$ iff they are related in all orders, i.e. if $\omega \leq_{\Omega_j} \xi$ for all $j=1,...,n$.
\end{definition}

For example, below is the intersection of two total orders on the same 4 events $\{\ev{A}, \ev{B}, \ev{C}, \ev{D}\}$: in both orders, we have that event \ev{A} causally precedes events \ev{B} and \ev{C}, which in turn causally precede event \ev{D}.
However, \ev{B} precedes \ev{C} in the first order, while it succeeds it in the second, resulting in \ev{B} and \ev{C} being causally unrelated in the meet.
We will sometimes refer to this as the \emph{diamond order}.
\begin{center}
    \includegraphics[height=3cm]{svg-inkscape/total-ABCD_svg-tex.pdf}
    \hspace{0.75cm}
    \raisebox{1.4cm}{$\bigwedge$}
    \hspace{0.75cm}
    \includegraphics[height=3cm]{svg-inkscape/total-ACBD_svg-tex.pdf}
    \hspace{0.75cm}
    \raisebox{1.4cm}{$=$}
    \hspace{0.75cm}
    \includegraphics[height=2.5cm]{svg-inkscape/diamond-ABCD_svg-tex.pdf}
\end{center}
Below is a more complicated example, involving events in indefinite causal order.
Events \ev{B} and \ev{C} are in indefinite causal order on the left, but \ev{B} causally precedes \ev{C} on the right, so \ev{B} causally precedes \ev{C} in the meet.
Similarly, events \ev{C} and \ev{D} are in indefinite causal order on the right, but \ev{C} causally precedes \ev{D} on the left, so \ev{C} causally precedes \ev{D} in the meet.
The situation leading to events \ev{A} and \ev{B} being causally unrelated in the meet is analogous to the one from the previous example.
\begin{center}
    \includegraphics[height=3cm]{svg-inkscape/total-AZBCZD_svg-tex.pdf}
    \hspace{0.75cm}
    \raisebox{1.4cm}{$\bigwedge$}
    \hspace{0.75cm}
    \includegraphics[height=4cm]{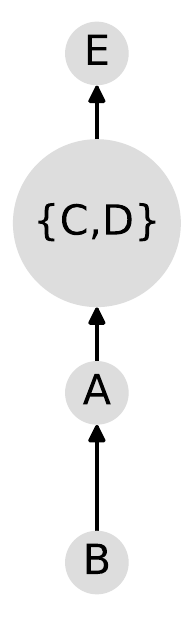}
    \hspace{0.75cm}
    \raisebox{1.4cm}{$=$}
    \hspace{0.75cm}
    \includegraphics[height=3cm]{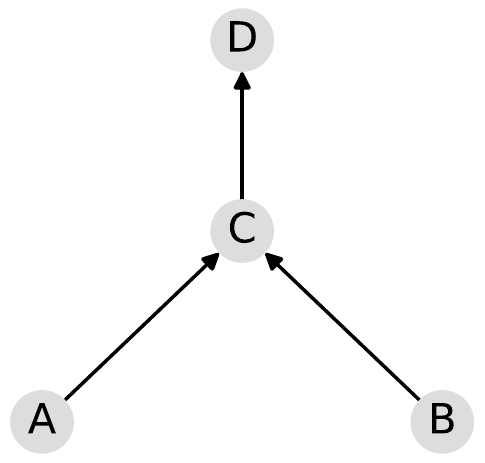}
\end{center}

\subsubsection{Replacement in causal orders}

Events in causal orders are an abstraction: they model spacetime regions which are sufficiently localised to be considered point-like, at least as far as causality is concerned.
However, events which looked ``causally atomic'' at a certain scale might turn out to comprise several sub-events when we look closer.
In terms of causal orders, this operation is known as ``replacement'', where one or more events are refined into causal orders, inheriting their causal relationships wholesale.

\begin{definition}
A \emph{replacement} in a causal order $\Omega$ is dictated by a family $(\Xi_{\omega})_{\omega \in \Omega}$ of pairwise disjoint causal orders $\Xi_{\omega}$ for each event $\omega \in \Omega$, where the trivial discrete case $\Xi_{\omega} := \{\omega\}$ is used to indicate that an event $\omega$ was not refined.
The result of the replacement is the causal order on $\bigcup_{\omega \in \Omega}|\Xi_{\omega}|$ where two events $\xi \in \Xi_{\omega}$ and $\xi' \in \Xi_{\omega'}$ are related by $\xi \leq \xi'$ if one of the following conditions holds:
\begin{itemize}
    \item $\omega < \omega'$
    \item $\omega = \omega'$ and $\xi \leq \xi'$ in $\Xi_\omega$
\end{itemize}
\end{definition}

Below is a simple example of replacement, where event \ev{B} in the total order $\ev{A} \rightarrow \ev{B} \rightarrow \ev{C}$ on the left is refined into the 5-event order $\Xi_{\ev{B}}$ on its right.
Event \ev{A} causally preceded \ev{B} in the total order, so it causally precedes all events of $\Xi_{\ev{B}}$ in the resulting order. Analogously, event \ev{C} causally succeeded \ev{B} in the total order, so it causally succeeds all events of $\Xi_{\ev{B}}$ in the resulting order.
\begin{center}
    \includegraphics[height=2.5cm]{svg-inkscape/total-ABC_svg-tex.pdf}
    \hspace{0.25cm}
    \raisebox{1.15cm}{\text{ where \textsf{B} $\mapsto$ }}
    \raisebox{1.15cm}{$\left(\rule{0cm}{1.35cm}\right.$}
    \hspace{0.1cm}
    \includegraphics[height=2.5cm]{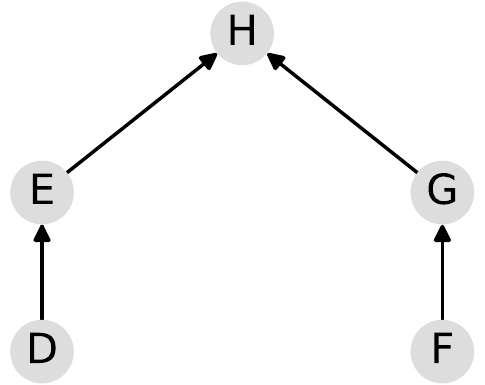}
    \hspace{0.1cm}
    \raisebox{1.15cm}{$\left.\rule{0cm}{1.35cm}\right)$}
    \hspace{0.75cm}
    \raisebox{1.15cm}{$=$}
    \hspace{0.75cm}
    \includegraphics[height=3.5cm]{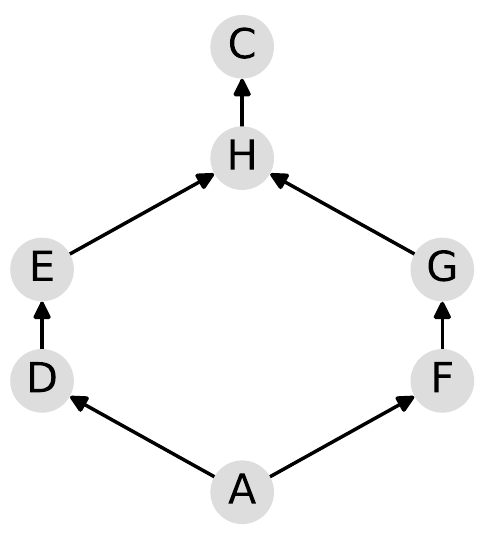}
\end{center}
A special case of replacement arises when all events in an order $\Omega$ are replaced with copies of the same order $\Xi$.
In this case, the elements of $\Xi_{\omega}$ are naturally labelled by pairs of events, as $\Xi_{\omega}=\suchthat{(\omega, \xi)}{\xi \in \Xi}$.
Below is an example where each event of the wedge order $\Omega$ on the left is replaced by a copy of the fork order $\Xi$ on the right, where events have been re-labelled to make copies disjoint (e.g. event \ev{A E} is the event \ev{E} of the copy of $\Xi$ that replaced event \ev{A} in $\Omega$).
\begin{center}
    \raisebox{1cm}{
        \includegraphics[height=2cm]{svg-inkscape/wedge-AB-C_svg-tex.pdf}
    }
    \hspace{-2mm}
    \raisebox{19mm}{$\times_{lex}$}
    \hspace{-2mm}
    \raisebox{1cm}{
        \includegraphics[height=2cm]{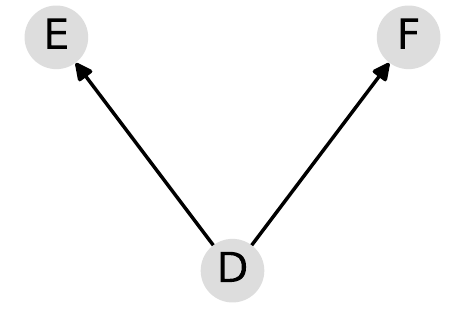}
    }
    \raisebox{19mm}{$=$}
    \hspace{0.5cm}
    \includegraphics[height=4cm]{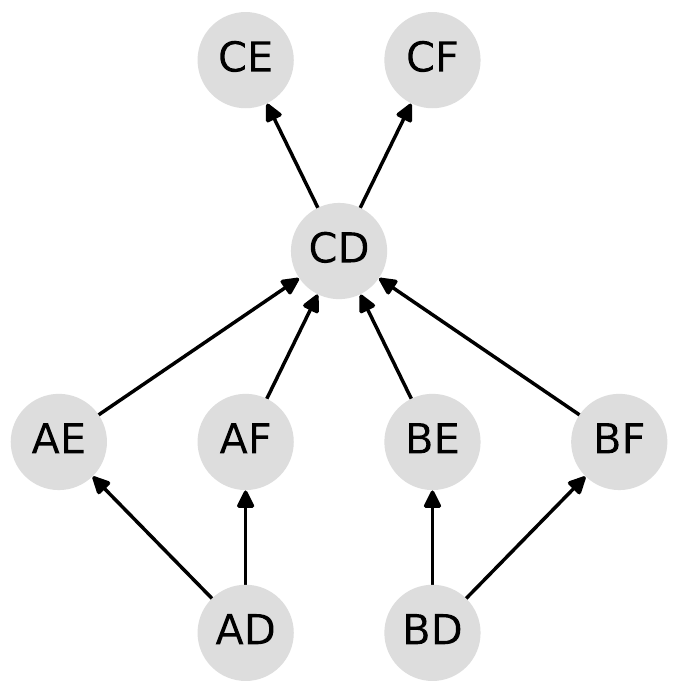}
\end{center}

\begin{definition}
The \emph{lexicographic product} $\Omega \times_{lex} \Xi$ is the replacement in $\Omega$ defined by the family $(\suchthat{(\omega, \xi)}{\xi \in \Xi})_{\omega \in \Omega}$.
\end{definition}

\subsubsection{Cartesian product of causal orders}

There is a certain algebraic significance to the operations described above.
Join $\bigvee$ and meet $\bigwedge$ are the lattice operations on the hierarchy of causal orders under inclusion, the main topic of the next subsection.
Sequential composition $\rightarrow$ is the natural sum operation on partial orders, from an algebraic and category-theoretic perspective.
There are also two distinct notions of product on causal orders: the lexicographic product $\times_{lex}$, arising as a special case of replacement, and the Cartesian product.

\begin{definition}
The \emph{Cartesian product} $\Omega \times \Xi$ is the order on the set $|\Omega| \times |\Xi|$ defined by setting $(\omega, \xi) \leq (\omega', \xi')$ exactly when $\omega \leq_\Omega \omega'$ and $\xi \leq_\Xi \xi'$.
\end{definition}

Cartesian products can be used to build causal lattices: for example, below is a (truncated) diamond lattice in 1+1 Minkowski spacetime, built by Cartesian product of two total orders on 3 events. 
\begin{center}
    \raisebox{0.75cm}{
        \includegraphics[height=2.5cm]{svg-inkscape/total-ABC_svg-tex.pdf}
    }
    \raisebox{19mm}{$\times$}
    \raisebox{0.75cm}{
        \includegraphics[height=2.5cm]{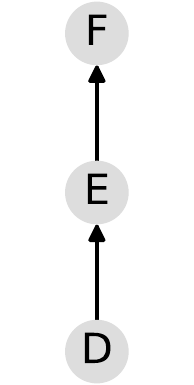}
    }
    \raisebox{19mm}{$=$}
    \hspace{0.5cm}
    \includegraphics[height=4cm]{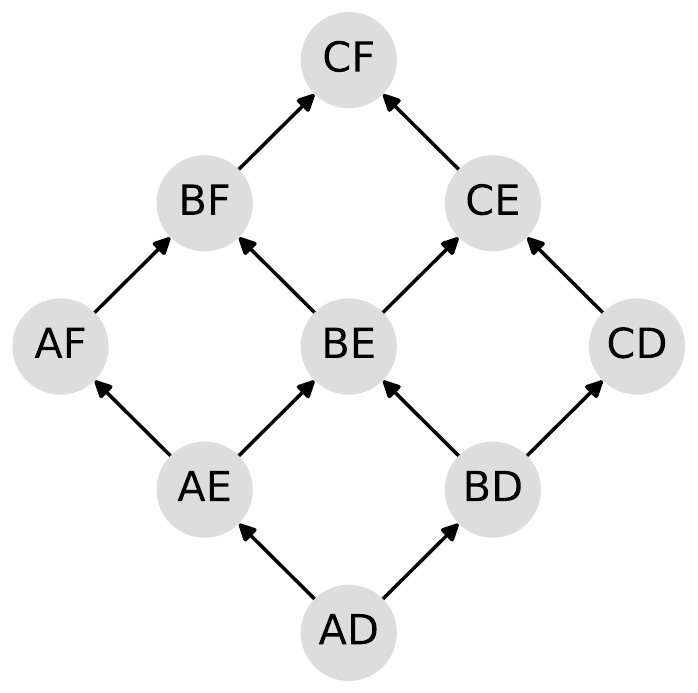}
\end{center}
The diamond lattice above arises naturally from the causal structure of regularly spaced events in $(1+1)$-dimensional Minkowski spacetime.
In general, however, products of causal orders require additional spatial dimensions: for example, below is the Cartesian product of a total order and a fork, which arises naturally in $(1+2)$ dimensional Minkowski spacetime.
\begin{center}
    \raisebox{0.75cm}{
        \includegraphics[height=2.5cm]{svg-inkscape/total-ABC_svg-tex.pdf}
    }
    \raisebox{1.9cm}{$\times$}
    \raisebox{1cm}{
        \includegraphics[height=2cm]{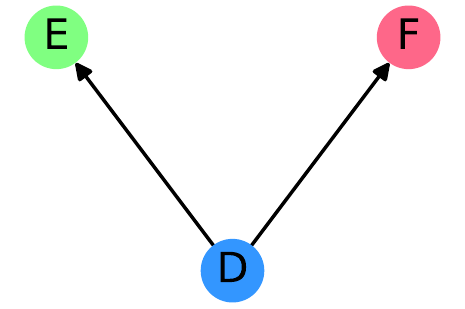}
    }
    \raisebox{1.9cm}{$=$}
    \hspace{2mm}
    \includegraphics[height=4cm]{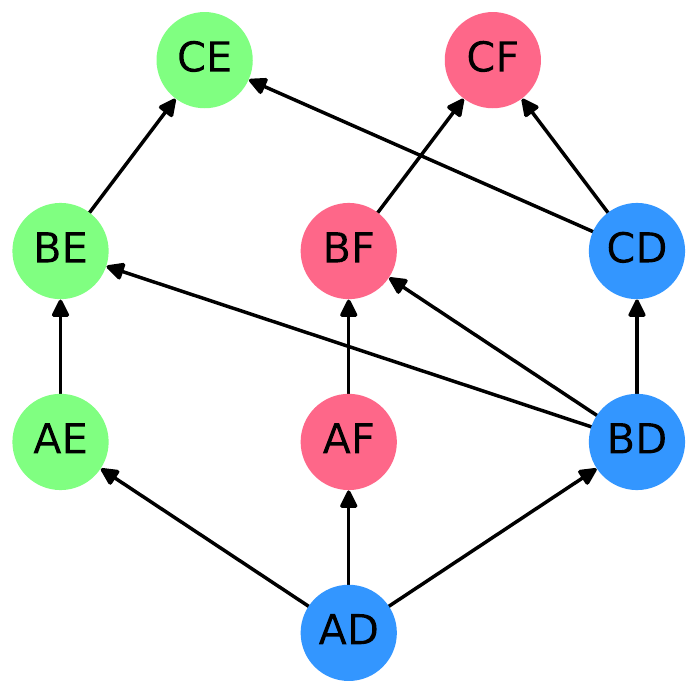}
    \hspace{2mm}
    \raisebox{1.9cm}{$=$}
    \includegraphics[height=4cm]{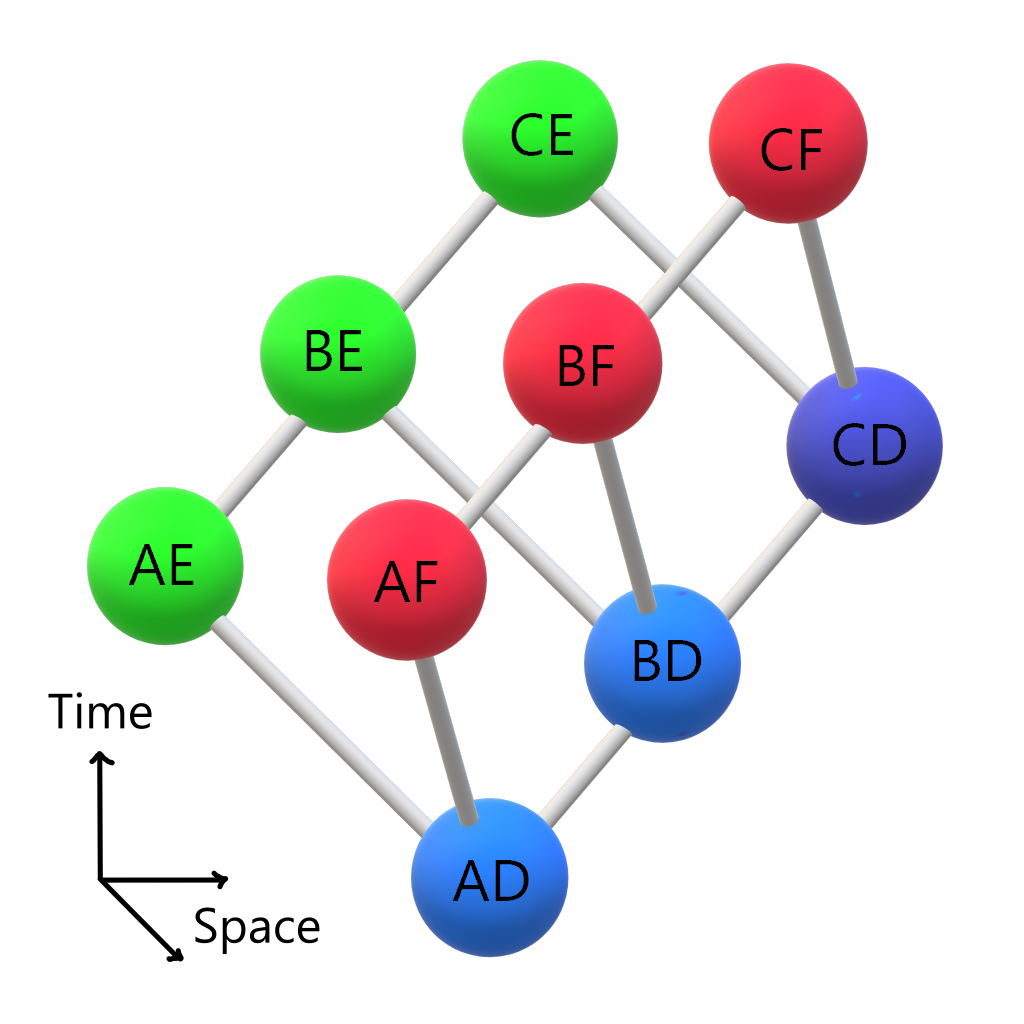}
\end{center}
Other common causal lattices can be obtained from Cartesian products by refining each event into a copy of some common causal order, using the lexicographic product.
For example, Figure \ref{fig:diamond-honeycomb-lattices} (p.\pageref{fig:diamond-honeycomb-lattices}) shows how a honeycomb lattice can be obtained from a diamond lattice by lexicographic product, expanding each event of the latter into a copy of the total order on two events.
\begin{figure}[h]
    \centering
    \raisebox{1.25cm}{
        \includegraphics[height=5cm]{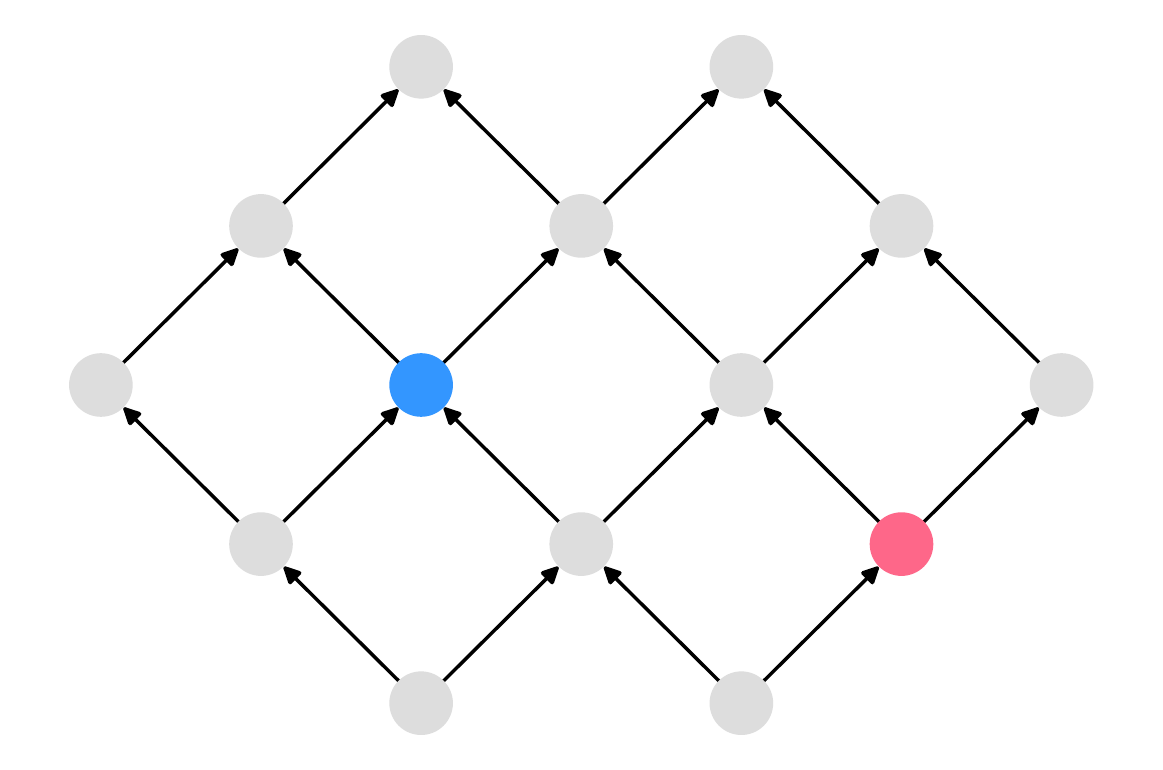}
    }
    \includegraphics[height=7.5cm]{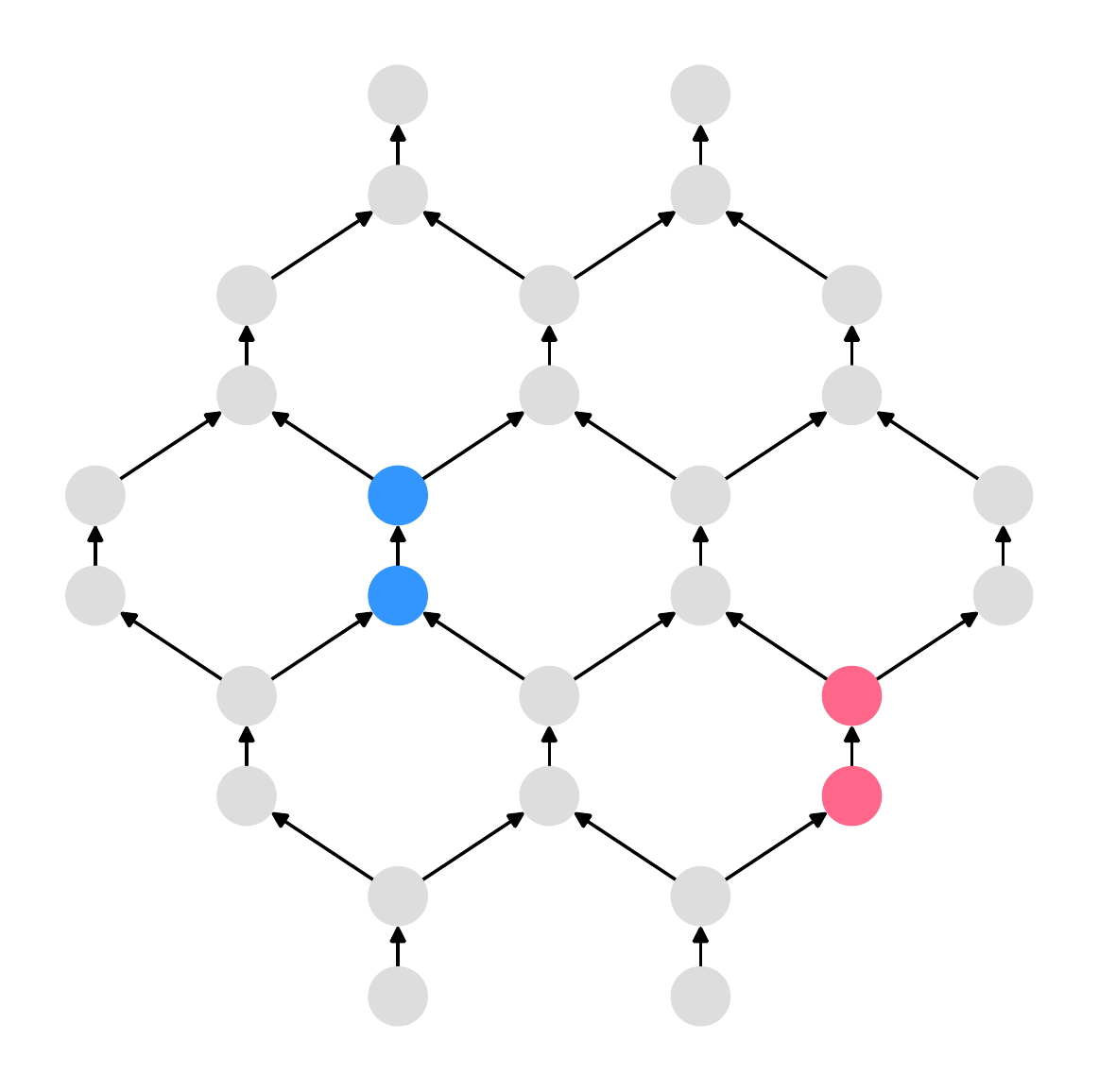}
    \caption{A fragment of diamond lattice (left) expanded into a corresponding fragment of honeycomb lattice (right) by lexicographic product. Events highlighted in blue and red on the left are replaced with correspondingly coloured sub-orders on the right.}
\label{fig:diamond-honeycomb-lattices}
\end{figure}
\begin{remark}
The sequential composition operation is also called ``series composition'' or ``ordinal sum''. It is sometimes denoted $\oplus$, but we shall not do so here.
Partial orders which can be built from single events using only sequential composition and parallel composition (i.e. disjoint joins) are known as ``series-parallel partial orders'': they are exactly the (non-empty) finite partial orders which are ``N-free'', that is, those which don't contain any 4 distinct elements $\omega_1, \omega_2, \omega_3, \omega_4$ forming an ``N shape'' $\omega_1 < \omega_2 > \omega_3 < \omega_4$.
If arbitrary joins are allowed, all finite partial orders can be built. 
\end{remark}

\subsection{Hierarchy of Causal Orders}
\label{subsection:causal-orders-hierarchy}

Causal orders are naturally ordered by inclusion: $\Omega \leq \Xi$ if $|\Omega| \subseteq |\Xi|$ as sets and $\leq_{\Omega} \subseteq \leq_{\Xi}$ as relations (i.e. as subsets $\suchthat{(\omega, \omega')}{\omega \leq_\Omega \omega'} \subseteq |\Omega|^2$ and $\suchthat{(\xi, \xi')}{\xi \leq_\Xi \xi'} \subseteq |\Xi|^2$).
The requirement that $\leq_{\Omega} \subseteq \leq_{\Xi}$ explicitly means that for all $\omega, \omega' \in \Omega$ the constraint $\omega \not \leq_\Xi \omega'$ in $\Xi$ implies the constraint $\omega' \not \leq_\Omega \omega'$.
Put in different words:
\begin{itemize}
    \item If $\omega$ and $\omega'$ are causally unrelated in $\Xi$ , then they are causally unrelated in $\Omega$.
    \item If $\omega$ causally precedes $\omega'$ in $\Xi$, then it can either causally precede $\omega'$ in $\Omega$ or it can be causally unrelated to $\omega'$ in $\Omega$; it cannot causally succeed $\omega'$ or be in indefinite causal order with it.
    \item If $\omega$ and $\omega'$ are in indefinite causal order in $\Xi$, then their causal relationship in $\Omega$ is unconstrained: $\omega$ can causally precede $\omega'$, causally succeed it, be causally unrelated to it or be in indefinite causal order with it.
\end{itemize}
From a causal standpoint, $\Omega \leq \Xi$ means that $\Omega$ imposes on its own events at least the same causal constraints as $\Xi$, and possibly more.
In particular, if $\Xi$ is definite (no two events in indefinite causal order) then so is $\Omega$; conversely, if $\Omega$ is indefinite, then so is $\Xi$.

\begin{observation}
Causal orders on a given set of events form a finite lattice, which we refer to as the \emph{hierarchy of causal orders}.
The join and meet operations on this lattice are those described in the previous subsection, the indiscrete order is the unique maximum (all events in indefinite causal order, i.e. no causal constraints), while the discrete order is the unique minimum (all events are causally unrelated).
\end{observation}

The hierarchy of causal orders on three events $\{\ev{A},\ev{B},\ev{C}\}$ is displayed by Figure \ref{fig:hierarchy-orders-3} (p.\pageref{fig:hierarchy-orders-3}), with definite causal order coloured red and indefinite ones coloured blue:
\begin{enumerate}
    \item The discrete order \textsf{A B C}, the minimum of the hierarchy, is at the left of the diagram.
    \item Immediately after the discrete order, we find the 6 orders in the shape of \textsf{A→B C}.
    \item The 6 orders above are included into the 3 fork orders (e.g. \textsf{B←A→C}), the 3 wedge orders (e.g. \textsf{B→A←C}), and the 3 indefinite causal orders in the shape of \textsf{\{AB\} C}.
    \item The fork and wedge orders are included into the 6 total orders (e.g. \textsf{A→B→C}).
    \item The 6 total orders and the 3 indefinite causal orders in the shape of \textsf{\{AB\} C} are included into the the 3 indefinite causal orders in the shape of \textsf{\{AB\}→C} and the 3 ones in the shape of \textsf{A→\{BC\}}. 
    \item The indiscrete order \textsf{\{ABC\}}, the maximum of the hierarchy, is at the right of the diagram.
\end{enumerate}
\begin{figure}[h]
    \centering
    \includegraphics[height=10cm]{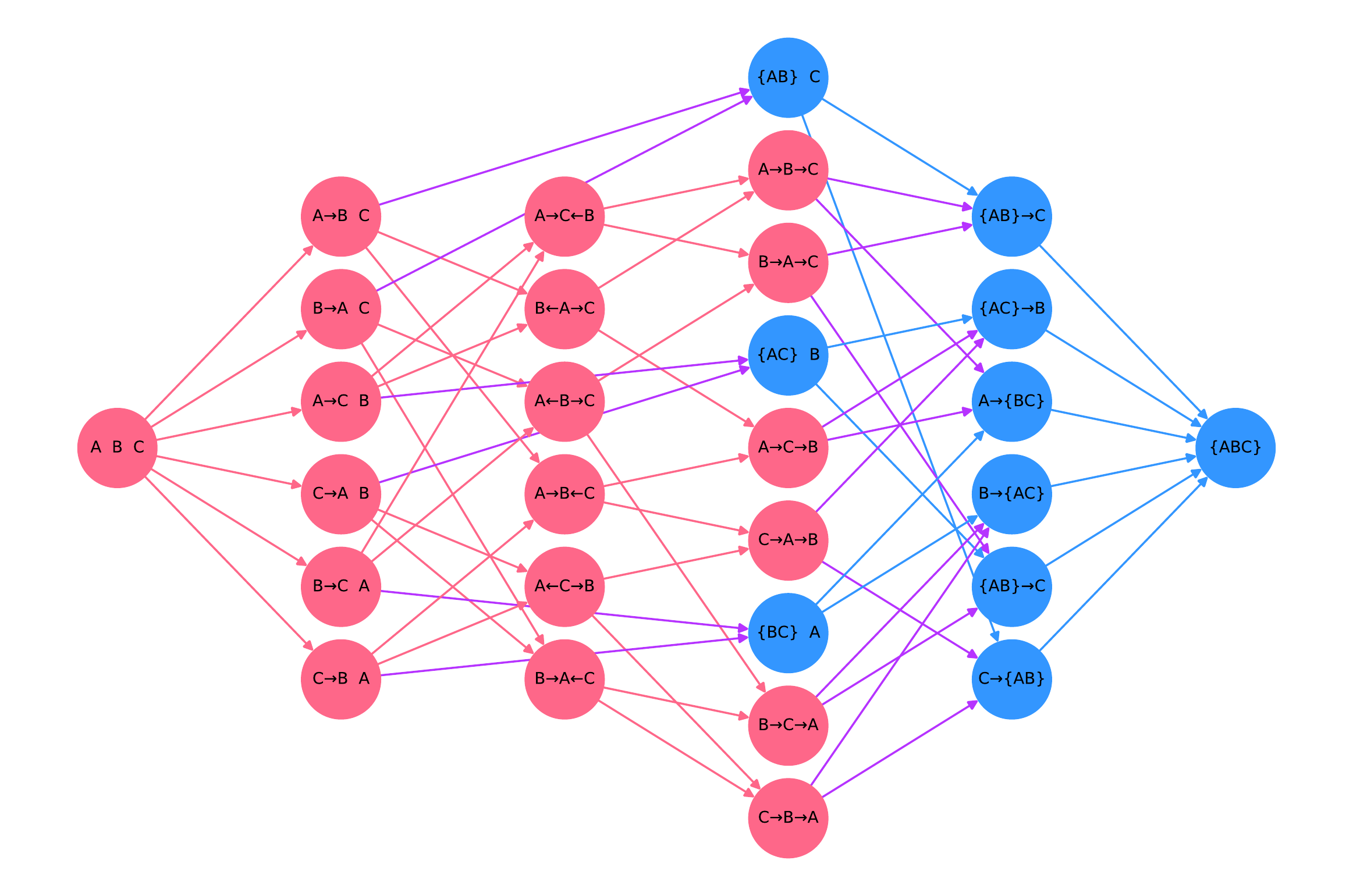}
    \caption{
    Hasse diagram for the hierarchy of causal orders on three events $\{\ev{A},\ev{B},\ev{C}\}$, left-to-right in inclusion order.
    Definite causal orders (left and middle) are coloured red, while indefinite causal orders (middle and right) are coloured blue.
    Inclusions between definite orders are coloured red, inclusions between indefinite orders are coloured blue, inclusions of a definite order into an indefinite one are coloured violet.
    In the order labels, a space is used to indicate causal unrelatedness, arrows are used to indicate that the event at the tail causally precedes the event at the head, and braces are used to indicate that the events contained are in indefinite causal order.
    }
\label{fig:hierarchy-orders-3}
\end{figure}
The definite causal order always form a lowerset in the hierarchy---if $\Xi$ is definite then all $\Omega \leq \Xi$ are also definite---and the maxima for this lowerset are exactly the total orders.
However, total orders don't form a separating set: there are inclusions of definite orders into indefinite ones that don't factor through a total order (cf. Figure \ref{fig:hierarchy-orders-3}, arrows from red notes in the second layer to blue nodes in the fourth layer).

The number of causal orders on $n$ events grows extremely rapidly, following OEIS sequence A000798: there are 4 orders on 2 events, 29 orders on 3 events, 355 orders on 4 events, 6942 on 5 events, and so on, until we reach 261492535743634374805066126901117203 orders on 18 events; the number of causal orders on 19 events is currently unknown.
As fast as this number grows, however, it pales in comparison to the number of possible causal spaces (2644 on 3 events with binary inputs \cite{gogioso2022classification}) and causal polytopes \cite{gogioso2022geometry}.

The super-exponential growth of causal orders is extremely problematic for unstructured causal decomposition tasks.
However, the issue can be significantly mitigated if some causal constraints are known about the scenario at hand.
For example, we might know that the actions of 4 parties \textsf{\{A,B,C,D\}} are constrained by a diamond causal order, and wish to divine how much of their experience can be explained without certain parties being able to signal others.
In such circumstances, we only need to investigate (certain) sub-orders of the diamond order: instead of all 355 orders on 4 events, we have limited our search to at most 25 of them.
Figure \ref{fig:hierarchy-suborders-diamond-ABCD} (p.\pageref{fig:hierarchy-suborders-diamond-ABCD}) shows the hierarchy of sub-orders of the diamond order: the maximum is now the diamond order itself, and all causal orders in the hierarchy are definite.
Some new interesting orders also make an appearance, such as the N-shaped \textsf{C←A→D←B}.
\begin{figure}[h]
    \centering
    \includegraphics[height=10cm]{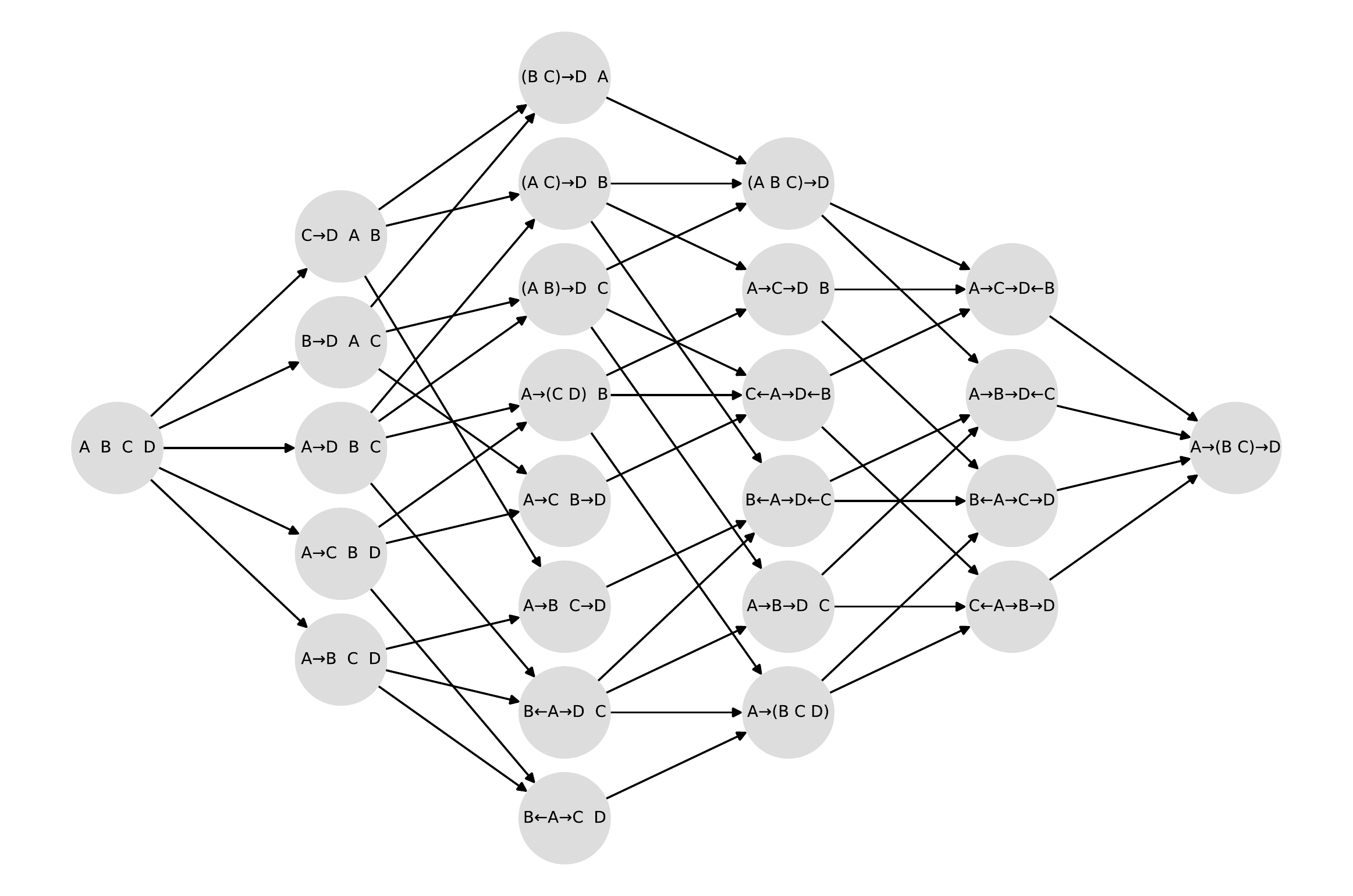}
    \caption{
    Hasse diagram for the hierarchy of sub-orders of the diamond order, left-to-right in inclusion order.
    All orders are definite, so no colour-coding of nodes and edges is necessary.
    In the order labels, a space is used to indicate causal unrelatedness, arrows are used to indicate that the event(s) at the tail causally precedes the event(s) at the head, and brackets are used to group multiple causally unrelated events together (for ease of notation).
    For example, \textsf{A→(B C)→D} on the right indicates that event \ev{A} precedes events \ev{B} and \ev{C}, which are causally unrelated to each other and both precede event \ev{D}.
    }
\label{fig:hierarchy-suborders-diamond-ABCD}
\end{figure}

\subsection{Lattice of Lowersets}
\label{subsection:causal-orders-lowersets}


As discussed in detail by Section \ref{section:spaces}, this work is concerned with a certain class of operational scenarios: blackbox devices are operated locally at events in spacetime, determining a probability distribution on their joint outputs conditional to their (freely chosen) joint inputs.
In such scenarios, causality constraints essentially state that the output at any subset of events cannot depend on inputs at events which causally succeed them or are causally unrelated to them.
Furthermore, the output at any event is only well-defined conditional to inputs for all events in its past: we are not interested in all sub-sets of events of a causal order, but rather in its lowersets.

The discussion above indicates that the object we seek to understand is not the causal order $\Omega$ itself, but rather its \emph{lattice of lowersets} $\Lsets{\Omega}$.
This is the subsets of events closed in the past, ordered by inclusion:
\begin{equation}
    \Lsets{\Omega}
    :=
    \suchthat{U \subseteq \Omega}{\forall \omega \in U.\,\downset{\omega} \subseteq U}
\end{equation}
In this case, being a lattice means that lowersets are closed under both intersection and union; we always omit the empty set from our Hasse diagrams, for clarity.

Inclusions between lowersets determine the causality constraints for the causal order: if $U, V \in \Lsets{\Omega}$ are such that $U \subseteq V$, then the output at events in $U$ cannot depend on the inputs at events in $V \backslash U$.
Consider the total order $\ev{A}\rightarrow\ev{B}\rightarrow\ev{C}$, and its associated lattice of lowersets: the inclusion $\{\ev{A},\ev{B}\} \subseteq \{\ev{A},\ev{B},\ev{C}\}$, for example, tells us that the outputs at events \ev{A} and \ev{B} cannot depend on the input at event \ev{C}; the inclusion $\{\ev{A}\} \subseteq \{\ev{A},\ev{B}\}$, additionally, tells us that the outputs at event \ev{A} cannot depend on the input at event \ev{B}.
\begin{center}
    \raisebox{1.40cm}{$\LsetsSym$}
    \raisebox{1.40cm}{$\left(\rule{0cm}{1.35cm}\right.$}
    \hspace{0.0cm}
    \raisebox{0.25cm}{
        \includegraphics[height=2.5cm]{svg-inkscape/total-ABC_svg-tex.pdf}
    }
    \hspace{0.0cm}
    \raisebox{1.40cm}{$\left.\rule{0cm}{1.35cm}\right)$}
    \hspace{0.75cm}
    \raisebox{1.40cm}{$=$}
    \hspace{0.5cm}
    \includegraphics[height=3cm]{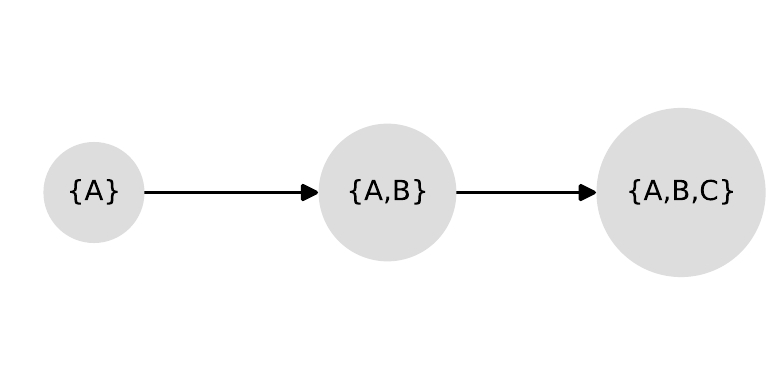}
\end{center}
Below is a more complicated example, for the diamond order: the inclusion $\{\ev{A},\ev{B}\} \subseteq \{\ev{A},\ev{B},\ev{C},\ev{D}\}$, for example, tells us that the outputs at events \ev{A} and \ev{B} cannot depend on the input at events \ev{C} and \ev{D}.
\begin{center}
    \raisebox{2.4cm}{$\LsetsSym$}
    \raisebox{2.4cm}{$\left(\rule{0cm}{1.35cm}\right.$}
    \hspace{0.0cm}
    \raisebox{1.25cm}{
        \includegraphics[height=2.5cm]{svg-inkscape/diamond-ABCD_svg-tex.pdf}
    }
    \hspace{0.0cm}
    \raisebox{2.4cm}{$\left.\rule{0cm}{1.35cm}\right)$}
    \hspace{0.75cm}
    \raisebox{2.4cm}{$=$}
    \hspace{0.5cm}
    \includegraphics[height=5cm]{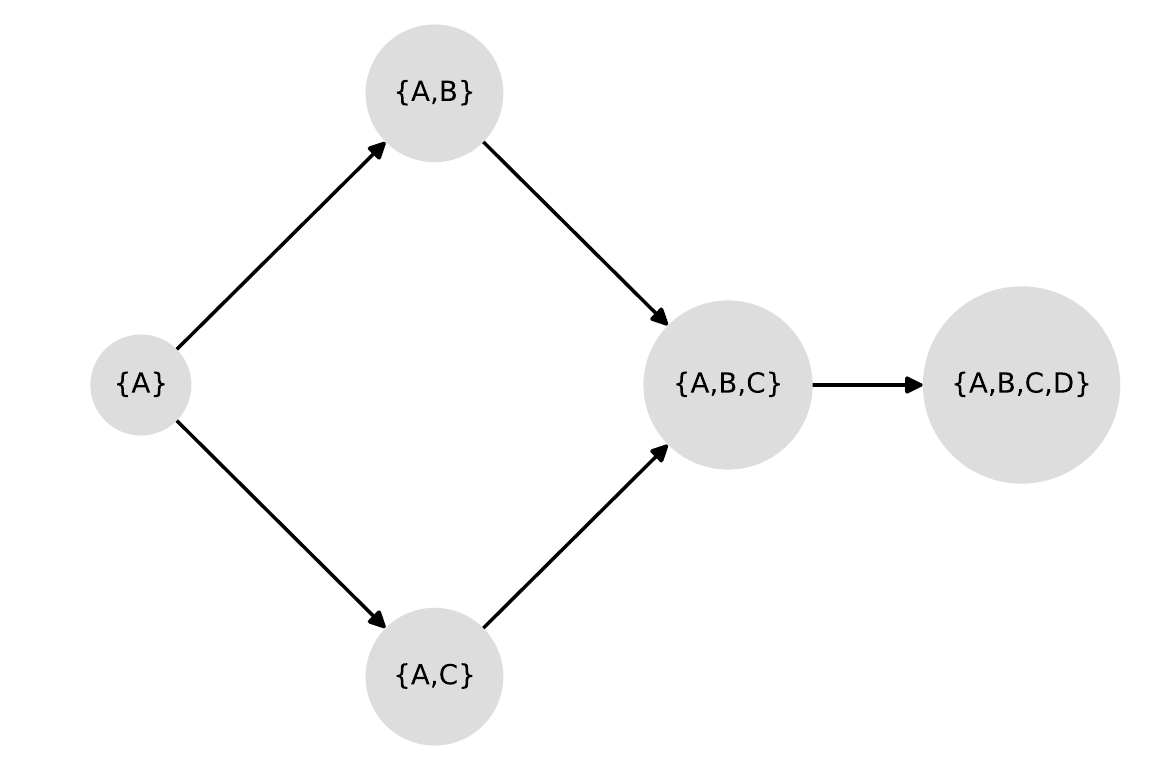}
\end{center}
Here, we note for the first time how lowersets are more general than downsets: we have $\downset{\ev{A}}=\{\ev{A}\}$, $\downset{\ev{B}}=\{\ev{A},\ev{B}\}$, $\downset{\ev{C}}=\{\ev{A},\ev{C}\}$ and $\downset{\ev{D}}=\{\ev{A},\ev{B},\ev{C},\ev{D}\}$, but lowerset $\{\ev{A},\ev{B},\ev{C}\}$ does not originate from any individual event.
Hence, lowersets strictly generalise the notion of causal past from individual events to arbitrary subsets of events:
\[
\{\ev{A},\ev{B},\ev{C}\}
=
\downset{\ev{B}} \cup \downset{\ev{C}}
=
\downset{\{\ev{B},\ev{C}\}}
\]
When the causal order is indefinite, lowersets cannot split causal equivalence classes: either no event from the class is in the lowerset, or all events are.
We can see this in the lattice of lowersets for the indefinite causal order $\ev{A}\rightarrow\{\ev{B},\ev{C}\}\rightarrow\ev{D}$, where events $\{\ev{B},\ev{C}\}$ form a causal equivalence class.
\begin{center}
    \raisebox{1.40cm}{$\LsetsSym$}
    \raisebox{1.40cm}{$\left(\rule{0cm}{1.35cm}\right.$}
    \hspace{0.0cm}
    \raisebox{0.25cm}{
        \includegraphics[height=2.5cm]{svg-inkscape/total-AZBCZD_svg-tex.pdf}
    }
    \hspace{0.0cm}
    \raisebox{1.40cm}{$\left.\rule{0cm}{1.35cm}\right)$}
    \hspace{0.75cm}
    \raisebox{1.40cm}{$=$}
    \hspace{0.5cm}
    \includegraphics[height=3cm]{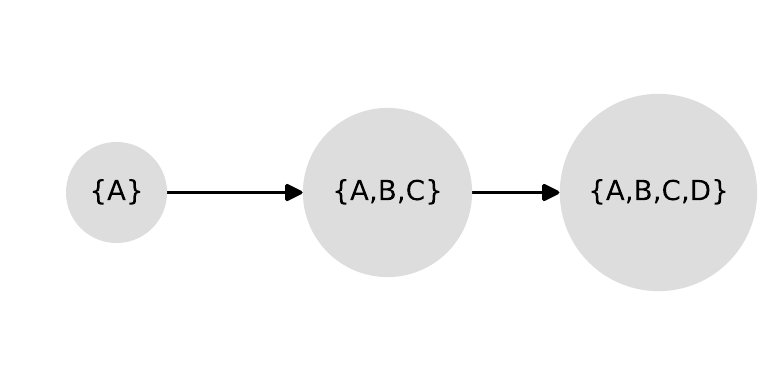}
\end{center}

An interesting question arises when we consider the interaction of causal constraints for multiple causal orders.
A scenario is explainable by two causal orders $\Omega$ and $\Omega'$ if it satisfies the causal constrains of both: in terms of lowersets, such constraints correspond to the union $\Lsets{\Omega}\cup\Lsets{\Omega'}$ of the lowersets for the individual orders.
Consider, for example, the total orders $\Omega = \ev{A}\rightarrow\ev{B}\rightarrow\ev{C}\rightarrow\ev{D}$ and $\Omega' = \ev{A}\rightarrow\ev{C}\rightarrow\ev{B}\rightarrow\ev{D}$, together with the associated lowersets.
\begin{center}
    \raisebox{1.40cm}{$\LsetsSym$}
    \raisebox{1.40cm}{$\left(\rule{0cm}{1.60cm}\right.$}
    \hspace{0.0cm}
    \raisebox{0.0cm}{
        \includegraphics[height=3cm]{svg-inkscape/total-ABCD_svg-tex.pdf}
    }
    \hspace{0.0cm}
    \raisebox{1.40cm}{$\left.\rule{0cm}{1.60cm}\right)$}
    \hspace{0.75cm}
    \raisebox{1.40cm}{$=$}
    \hspace{0.5cm}
    \includegraphics[height=3cm]{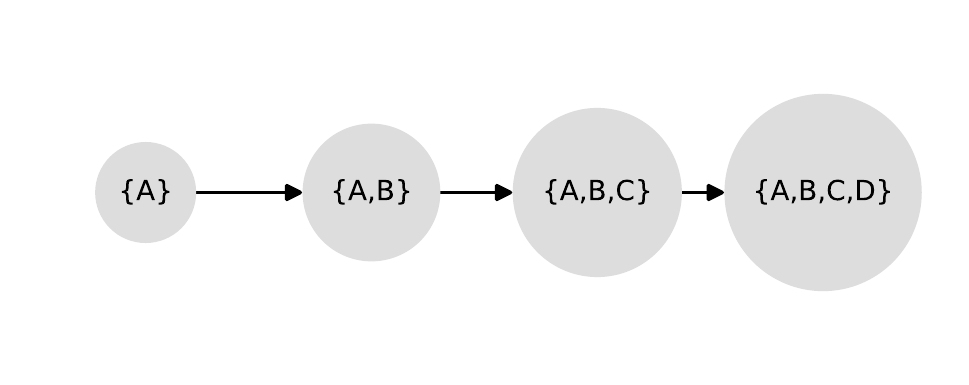}
\end{center}
\begin{center}
    \raisebox{1.40cm}{$\LsetsSym$}
    \raisebox{1.40cm}{$\left(\rule{0cm}{1.60cm}\right.$}
    \hspace{0.0cm}
    \raisebox{0.0cm}{
        \includegraphics[height=3cm]{svg-inkscape/total-ACBD_svg-tex.pdf}
    }
    \hspace{0.0cm}
    \raisebox{1.40cm}{$\left.\rule{0cm}{1.60cm}\right)$}
    \hspace{0.75cm}
    \raisebox{1.40cm}{$=$}
    \hspace{0.5cm}
    \includegraphics[height=3cm]{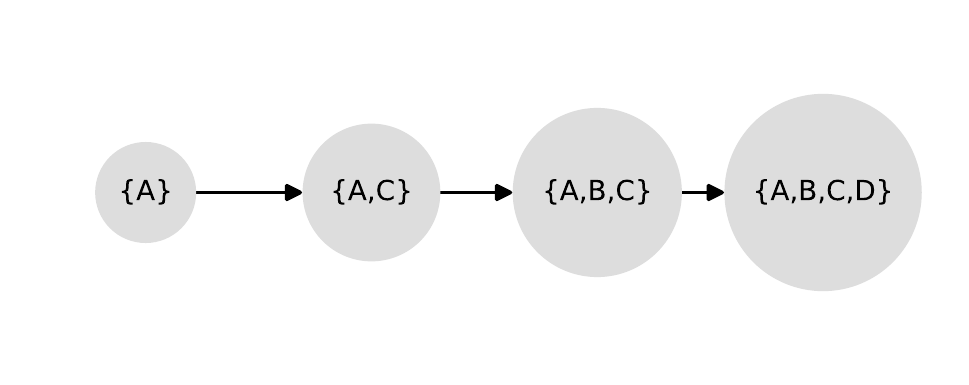}
\end{center}
For a scenario to satisfy both orders, it has to satisfy the constraints derived from $\Lsets{\ev{A}\rightarrow\ev{B}\rightarrow\ev{C}\rightarrow\ev{D}}\cup\Lsets{\ev{A}\rightarrow\ev{C}\rightarrow\ev{B}\rightarrow\ev{D}}$, depicted below.
\begin{center}
    \raisebox{1.90cm}{$\LsetsSym$}
    \raisebox{1.90cm}{$\left(\rule{0cm}{1.60cm}\right.$}
    \hspace{-0.4cm}
    \raisebox{0.5cm}{
        \includegraphics[height=3cm]{svg-inkscape/total-ABCD_svg-tex.pdf}
    }
    \hspace{-0.5cm}
    \raisebox{1.90cm}{$\left.\rule{0cm}{1.60cm}\right)$}
    \hspace{0.25cm}
    \raisebox{1.90cm}{$\bigcup$}
    \hspace{0.25cm}
    \raisebox{1.90cm}{$\LsetsSym$}
    \raisebox{1.90cm}{$\left(\rule{0cm}{1.60cm}\right.$}
    \hspace{-0.4cm}
    \raisebox{0.5cm}{
        \includegraphics[height=3cm]{svg-inkscape/total-ACBD_svg-tex.pdf}
    }
    \hspace{-0.5cm}
    \raisebox{1.90cm}{$\left.\rule{0cm}{1.60cm}\right)$}
    \hspace{0.25cm}
    \raisebox{1.90cm}{$=$}
    \hspace{0.0cm}
    \includegraphics[height=4cm]{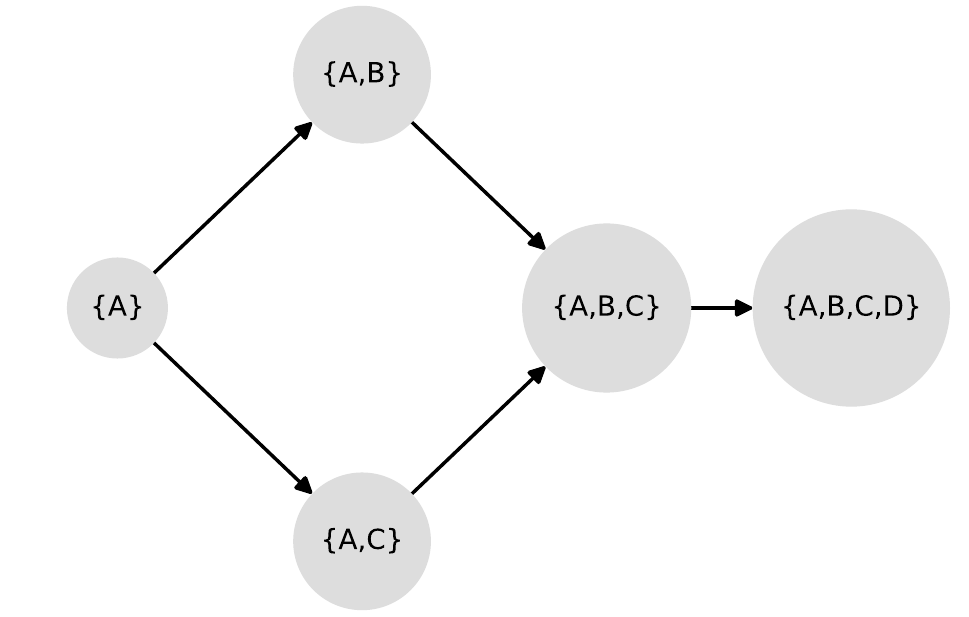}
\end{center}
We immediately recognise the lowersets as those of the diamond order, which we also know to take the form:
\[
\text{diamond}_{\textsf{ABCD}}
=
(\ev{A}\rightarrow\ev{B}\rightarrow\ev{C}\rightarrow\ev{D})
\wedge
(\ev{A}\rightarrow\ev{C}\rightarrow\ev{B}\rightarrow\ev{D})
\]
So the question arises: is simultaneously satisfying the causal constraints for two (or more) causal orders always the same as satisfying the causal constraints for their meet?
To answer it we first note that the hierarchy of causal orders is (contravariantly) related to the hierarchy formed by the corresponding lattices of lowersets under inclusion.

\begin{proposition}
\label{proposition:caus-ord-lsets-inclusion}
For any two causal orders $\Omega$ and $\Omega'$ we have:
\begin{equation}
    \Omega \leq \Omega'
    \;\Leftrightarrow\;
    \Lsets{\Omega}\supseteq\Lsets{\Omega'}
\end{equation}
\end{proposition}
\begin{proof}
See \ref{proof:proposition:caus-ord-lsets-inclusion}
\end{proof}

\begin{corollary}
For any two causal orders $\Omega$ and $\Omega'$ we have:
\begin{equation}
\Lsets{\Omega}\cup\Lsets{\Omega'} \subseteq \Lsets{\Omega \wedge \Omega'}
\end{equation}
\end{corollary}

Unfortunately, the above inclusion cannot be strengthened to an equality: in general, $\Lsets{\Omega}\cup\Lsets{\Omega'}$ is not even a lattice!
For a counterexample, we consider the following orders on four events and their associated lattices of lowersets.
\begin{center}
    \raisebox{1.40cm}{$\LsetsSym$}
    \raisebox{1.40cm}{$\left(\rule{0cm}{1.35cm}\right.$}
    \hspace{0.0cm}
    \raisebox{0.25cm}{
        \includegraphics[height=2.5cm]{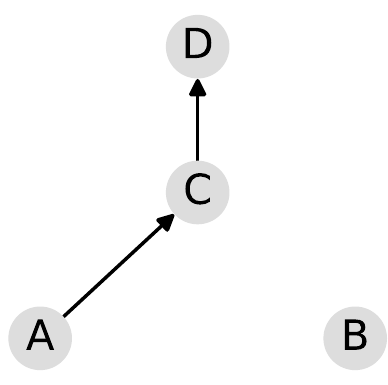}
    }
    \hspace{0.0cm}
    \raisebox{1.40cm}{$\left.\rule{0cm}{1.35cm}\right)$}
    \hspace{0.75cm}
    \raisebox{1.40cm}{$=$}
    \hspace{0.5cm}
    \includegraphics[height=3cm]{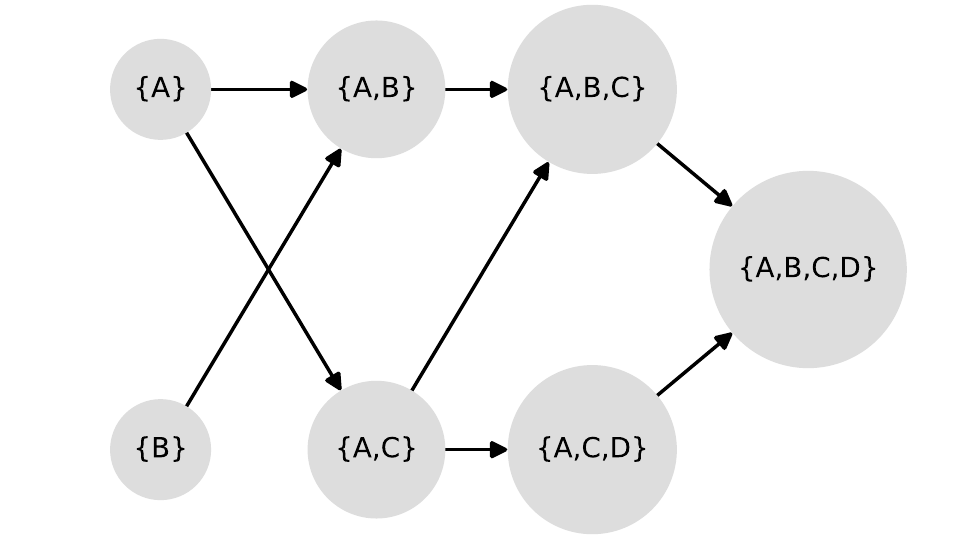}
\end{center}
\begin{center}
    \raisebox{1.40cm}{$\LsetsSym$}
    \raisebox{1.40cm}{$\left(\rule{0cm}{1.35cm}\right.$}
    \hspace{0.0cm}
    \raisebox{0.25cm}{
        \includegraphics[height=2.5cm]{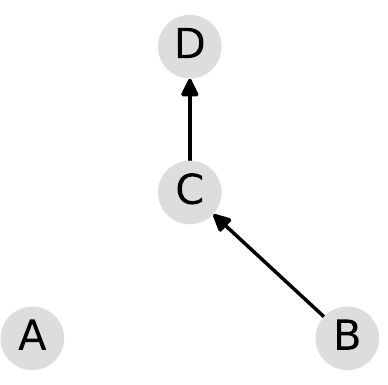}
    }
    \hspace{0.0cm}
    \raisebox{1.40cm}{$\left.\rule{0cm}{1.35cm}\right)$}
    \hspace{0.75cm}
    \raisebox{1.40cm}{$=$}
    \hspace{0.5cm}
    \includegraphics[height=3cm]{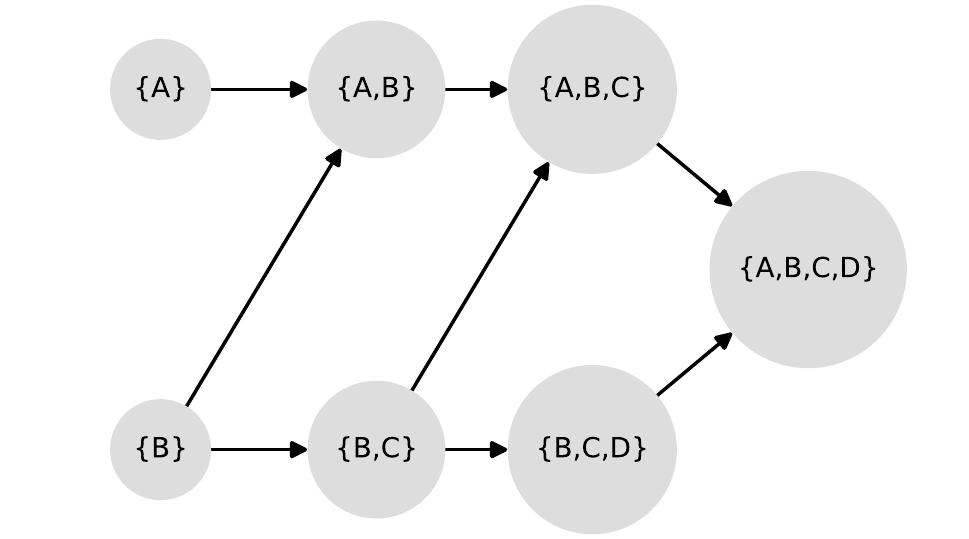}
\end{center}
The union of the corresponding lattices of lowersets is the following set, which is evidently not closed under intersection: the intersections $\{\ev{C}\}=\{\ev{A},\ev{C}\}\cap\{\ev{B},\ev{C}\}$ and $\{\ev{C},\ev{D}\}=\{\ev{A},\ev{C},\ev{D}\}\cap\{\ev{B},\ev{C},\ev{D}\}$ are both conspicuously missing.
\begin{center}
    \raisebox{1.90cm}{$\LsetsSym$}
    \raisebox{1.90cm}{$\left(\rule{0cm}{1.60cm}\right.$}
    \hspace{-0.4cm}
    \raisebox{0.75cm}{
        \includegraphics[height=2.5cm]{svg-inkscape/total-ACD-join-B_svg-tex.pdf}
    }
    \hspace{-0.5cm}
    \raisebox{1.90cm}{$\left.\rule{0cm}{1.60cm}\right)$}
    \hspace{0.0cm}
    \raisebox{1.90cm}{$\bigcup$}
    \hspace{0.0cm}
    \raisebox{1.90cm}{$\LsetsSym$}
    \raisebox{1.90cm}{$\left(\rule{0cm}{1.60cm}\right.$}
    \hspace{-0.4cm}
    \raisebox{0.75cm}{
        \includegraphics[height=2.5cm]{svg-inkscape/total-BCD-join-A_svg-tex.pdf}
    }
    \hspace{-0.5cm}
    \raisebox{1.90cm}{$\left.\rule{0cm}{1.60cm}\right)$}
    \hspace{0.25cm}
    \raisebox{1.90cm}{$=$}
    \hspace{0.0cm}
    \includegraphics[height=4cm]{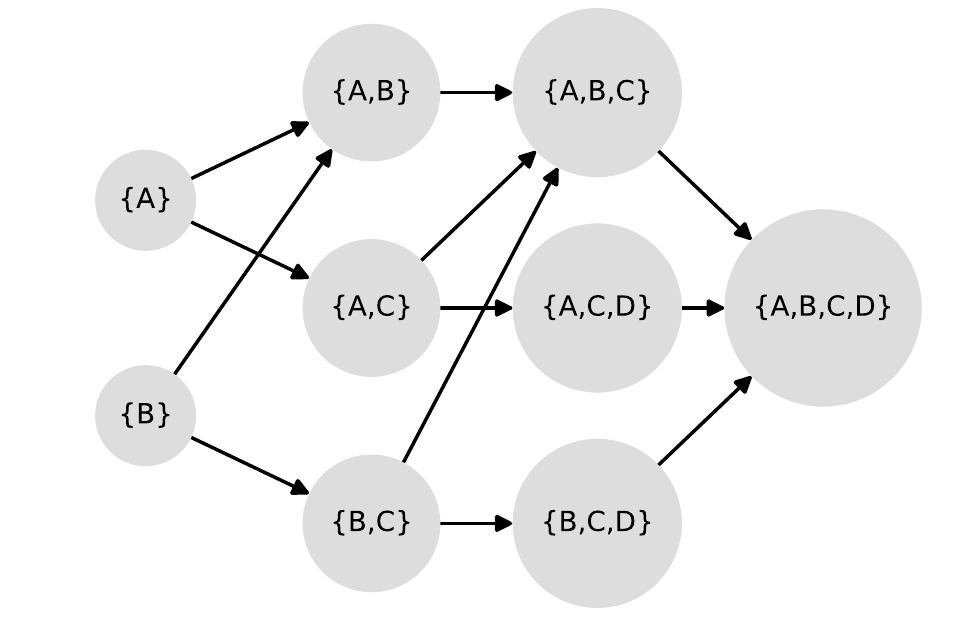}
\end{center}
In particular, the collection of lowersets displayed above is not the lattice of lowersets for the meet of the two orders, which (in this case) is obtained by including the two missing lowerset intersections.
\begin{center}
    \raisebox{2.40cm}{$\LsetsSym$}
    \raisebox{2.40cm}{$\left(\rule{0cm}{1.35cm}\right.$}
    \hspace{0.0cm}
    \raisebox{1.25cm}{
        \includegraphics[height=2.5cm]{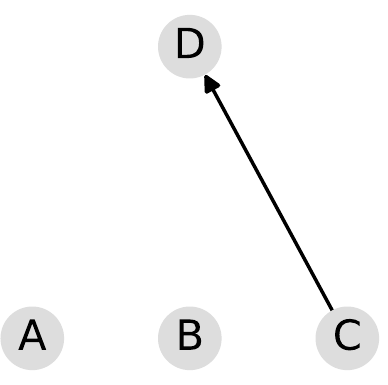}
    }
    \hspace{0.0cm}
    \raisebox{2.40cm}{$\left.\rule{0cm}{1.35cm}\right)$}
    \hspace{0.75cm}
    \raisebox{2.40cm}{$=$}
    \hspace{0.5cm}
    \includegraphics[height=5cm]{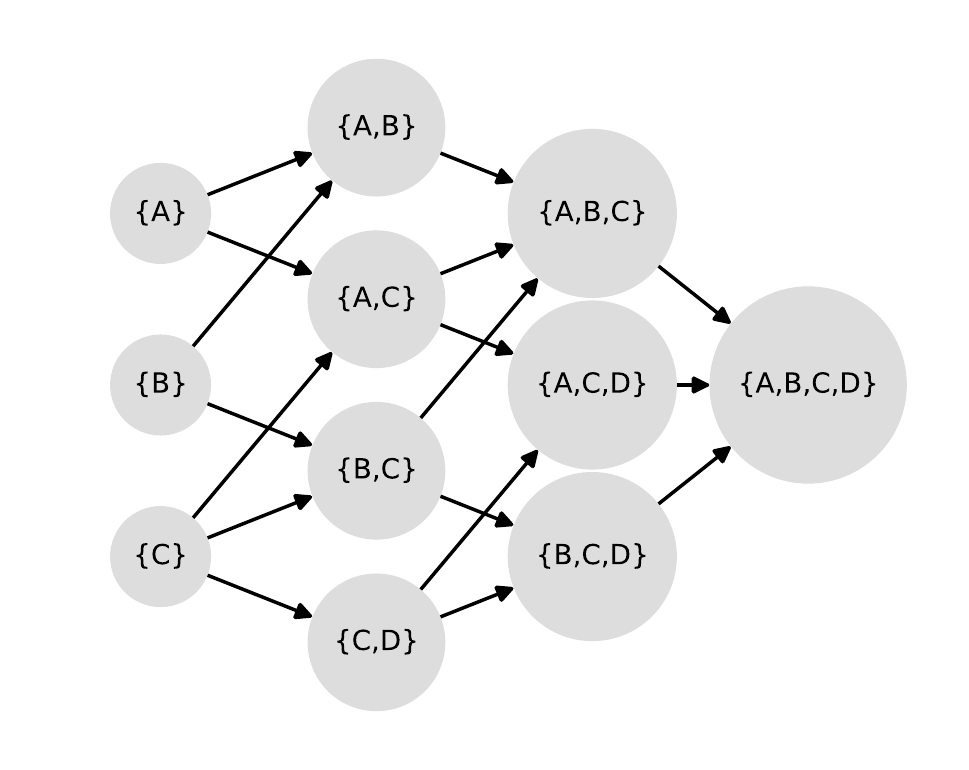}
\end{center}

Dually to the above, we can ask whether the intersection of the lattices of lowersets for two (or more) causal events is the lattice of lowersets for their join, i.e. whether it holds that $\Lsets{\Omega}\cap\Lsets{\Omega'} = \Lsets{\Omega \vee \Omega'}$.
This is more promising: at the very least, the intersection $\Lsets{\Omega}\cap\Lsets{\Omega'}$ is always a lattice!
As a motivating example, we go back to the total orders $\Omega = \ev{A}\rightarrow\ev{B}\rightarrow\ev{C}\rightarrow\ev{D}$ and $\Omega' = \ev{A}\rightarrow\ev{C}\rightarrow\ev{B}\rightarrow\ev{D}$.
\begin{center}
    \raisebox{1.40cm}{$\LsetsSym$}
    \raisebox{1.40cm}{$\left(\rule{0cm}{1.60cm}\right.$}
    \hspace{-0.4cm}
    \raisebox{0.0cm}{
        \includegraphics[height=3cm]{svg-inkscape/total-ABCD_svg-tex.pdf}
    }
    \hspace{-0.5cm}
    \raisebox{1.40cm}{$\left.\rule{0cm}{1.60cm}\right)$}
    \hspace{0.25cm}
    \raisebox{1.40cm}{$\bigcap$}
    \hspace{0.25cm}
    \raisebox{1.40cm}{$\LsetsSym$}
    \raisebox{1.40cm}{$\left(\rule{0cm}{1.60cm}\right.$}
    \hspace{-0.4cm}
    \raisebox{0.0cm}{
        \includegraphics[height=3cm]{svg-inkscape/total-ACBD_svg-tex.pdf}
    }
    \hspace{-0.5cm}
    \raisebox{1.40cm}{$\left.\rule{0cm}{1.60cm}\right)$}
    \hspace{0.25cm}
    \raisebox{1.40cm}{$=$}
    \hspace{0.0cm}
    \includegraphics[height=3cm]{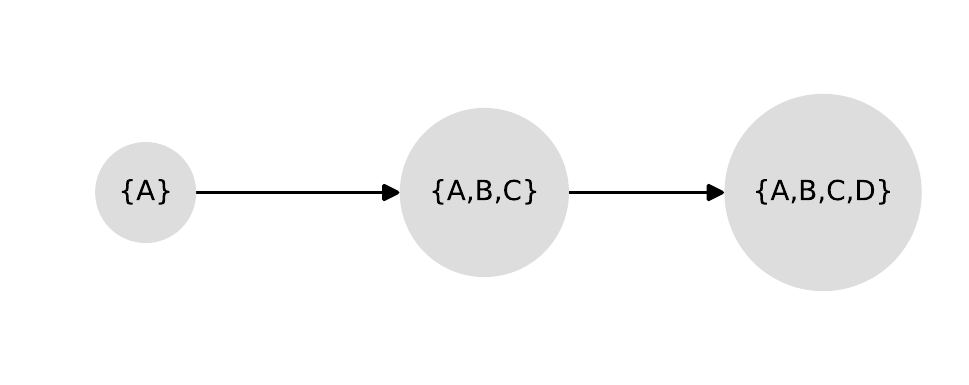}
\end{center}
Indeed, we immediately recognise the intersection as the lattice of lowersets for the join of the two total orders, where the events \ev{B} and \ev{C} are in indefinite causal order:
\begin{center}
    \raisebox{1.40cm}{$\LsetsSym$}
    \raisebox{1.40cm}{$\left(\rule{0cm}{1.60cm}\right.$}
    \hspace{-0.4cm}
    \raisebox{0.0cm}{
        \includegraphics[height=3cm]{svg-inkscape/total-ABCD_svg-tex.pdf}
        \hspace{-0.1cm}
        \raisebox{1.40cm}{$\bigvee$}
        \hspace{-0.1cm}
        \includegraphics[height=3cm]{svg-inkscape/total-ABCD_svg-tex.pdf}
    }
    \hspace{-0.5cm}
    \raisebox{1.40cm}{$\left.\rule{0cm}{1.60cm}\right)$}
    \hspace{0.2cm}
    \raisebox{1.40cm}{$=$}
    \hspace{0.2cm}
    \raisebox{1.40cm}{$\LsetsSym$}
    \raisebox{1.40cm}{$\left(\rule{0cm}{1.35cm}\right.$}
    \hspace{-0.4cm}
    \raisebox{0.25cm}{
        \includegraphics[height=2.5cm]{svg-inkscape/total-AZBCZD_svg-tex.pdf}
    }
    \hspace{-0.5cm}
    \raisebox{1.40cm}{$\left.\rule{0cm}{1.35cm}\right)$}
    \hspace{0.25cm}
    \raisebox{1.40cm}{$=$}
    \hspace{0.0cm}
    \includegraphics[height=3cm]{svg-inkscape/total-AZBCZD-lsets_svg-tex.pdf}
\end{center}
In this dual scenario, the contravariant relation between the hierarchy of causal orders and the associated lattices of lowersets implies that $\Lsets{\Omega}\cap\Lsets{\Omega'} \supseteq \Lsets{\Omega \vee \Omega'}$.
This time, the inclusion can be strengthened to an equality.

\begin{proposition}
\label{proposition:caus-ord-lsets-intersection}
For any two causal orders $\Omega$ and $\Omega'$ we have:
\begin{equation}
    \Lsets{\Omega}\cap\Lsets{\Omega'} = \Lsets{\Omega \vee \Omega'}
\end{equation}
\end{proposition}
\begin{proof}
See \ref{proof:proposition:caus-ord-lsets-intersection}
\end{proof}
The above result provides operational meaning to joins of causal orders: the causal constraints imposed by the join order are exactly the constraints common to all causal orders involved.

\newpage
\subsection{Proofs for Section \ref{section:causal-orders}}

\subsubsection{Proof of Proposition \ref{proposition:caus-ord-lsets-inclusion}}
\label{proof:proposition:caus-ord-lsets-inclusion}
\begin{proof}
Follows immediately from the definition of $\Omega \leq \Omega'$ ($\omega \leq_{\Omega} \xi$ implies $\omega \leq_{\Omega'} \xi$), the definition of lowersets (for all $\omega \in U$ we must have $\downset{\omega} \subseteq U$), and the observation that $\omega \leq \xi$ is equivalent to $\downset{\omega} \subseteq \downset{\xi}$.
\end{proof}

\subsubsection{Proof of Proposition \ref{proposition:caus-ord-lsets-intersection}}
\label{proof:proposition:caus-ord-lsets-intersection}
\begin{proof}
Any $U \in \Lsets{\Omega}\cap\Lsets{\Omega'}$ is a lowerset for both $\Omega$ and $\Omega'$: if $\xi \in U$, then either one of $\omega \leq_{\Omega} \xi$ or $\omega \leq_{\Omega'} \xi$ is enough to imply $\omega \in U$.
It is enough to show that $U \in \Lsets{\Omega \vee \Omega'}$.
Take a generic $\xi \in U$, and a generic $\omega$ such that $\omega \leq_{\Omega \vee \Omega'} \xi$: by definition of the join order, there is a sequence of events $(\omega_{k})_{k=0}^{m}$ such that $\omega_0 = \omega$, $\omega_m = \xi$ and for all $k=1,...,m$ either $\omega_{k-1} \leq_{\Omega} \omega_{k}$ or $\omega_{k-1} \leq_{\Omega'} \omega_{k}$.
We follow the sequence backwards, starting from $\omega_m = \xi$: for each $k=m,...,1$, $\omega_{k} \in U$ implies $\omega_{k-1} \in U$.
Hence $\omega = \omega_0 \in U$, so that $U$ is a lowerset for $\Omega \vee \Omega'$.
\end{proof}

\section{Spaces of input histories}
\label{section:spaces}

This work is concerned with the causal structure of a certain class of experiments or protocols, where events correspond to the local operation of black-box devices.
At each event, an input to the device is freely chosen from a finite input set, in response to which the device produces (probabilistically) an output in a finite output set.
The ensuing probability distribution on joint outputs for all devices, conditional on joint inputs for all devices, forms the basis of our causal analysis.
In the most general case, no causal constraints are given on the events.

When we say that the devices are operated locally at each event, we mean that no information about the other events is explicitly used in the operation: every dependence on the inputs and outputs at other events must be entirely mediated by the causal structure.
If event \ev{A} causally succeeds event \ev{B}, for example, then the output at \ev{A} is allowed depend on the input and output at \ev{B}: the devices being operated at the two events are black-box, and it is causally possible for one of them to signal the other.
However, the input at \ev{A} is still freely chosen, regardless of what happened at \ev{B}, and the input/output sets for the device at \ev{A} are fixed beforehand.
In the absence of causal constraints, it is therefore possible for the output at each event to arbitrarily depend on inputs at all events, and for the outputs at any set of events to be correlated.
As a consequence, the only conditional probability distribution that is well-defined in general is one on joint outputs for all events, conditional on joint inputs for all events.

We like to imagine that, in practice, such probability distribution could be inferred by performing experiments multiple times, subject to a guarantee that the (otherwise black-box) devices behave identically at each iteration---without memory, if iterations are causally sequential---and that the (otherwise unknown, possibly indefinite) causal structure on events is the same at each iteration.
Alternatively---as is the case for all examples in this work---the probability distribution could be derived from a theoretical model, with limited or no specification of causal structure.
Whatever the case may be, this work is not concerned with how such probability distributions are obtained: they are a purely mathematical given, decoupled from any practicalities.
We partition our causal study of such protocols and experiments into three distinct concerns:
\begin{itemize}
    \item The ``spaces of input histories'', presented in this work, describe the combinatorial structure of the sequences of local inputs that can determine the local output at each event.
    \item The ``empirical models'', presented in the companion work ``The Topology of Causality'' \cite{gogioso2022topology}, assign probability distributions to joint outputs over the events, conditional to joint inputs, in a way compatible with causal constraints.
    \item The ``causal polytopes'', presented in the companion work ``The Geometry of Causality'' \cite{gogioso2022geometry}, provide a geometric perspective which allows for the numerical calculation of interesting quantities, such as the fraction of an empirical jointly explainable by any given family of causal orders.
\end{itemize}
Spaces of input histories provide the ``causal canvas'' upon which empirical models are specified: they have a simple definition and a rich taxonomy, the exploration of which will fill the next three dozen pages.

\subsection{Operational Scenarios}
\label{subsection:spaces-scenarios}

\begin{definition}
An \emph{operational scenario} $(E, \underline{I}, \underline{O})$ specifies a finite non-empty set $E$ of \emph{events}, a finite non-empty set $I_\omega$ of \emph{inputs} for each event $\omega \in E$, and a finite non-empty set $O_\omega$ of \emph{outputs} for each event $\omega \in E$; we write $\underline{I} = (I_\omega)_{\omega \in E}$ and $\underline{O} = (O_\omega)_{\omega \in E}$.
The set of \emph{joint inputs} for all events is defined by:
\[
    \prod_{\omega \in E}I_\omega
    =
    \suchthat{(i_\omega)_{\omega \in E}}{i_\omega \in I_\omega}
\]
Similarly, the set of \emph{joint outputs} for all events is defined by:
\[
    \prod_{\omega \in E}O_\omega
    =
    \suchthat{(o_\omega)_{\omega \in E}}{i_\omega \in O_\omega}
\]
\end{definition}

It is sometimes convenient for the input sets $I_\omega = I$ and outputs sets $O_\omega = O$ to be the same for all events, in which case we have $\prod_{\omega \in E}I_\omega = I^E$ and $\prod_{\omega \in E}O_\omega = O^E$.
When $E = |\Omega|$ is the set of events underlying a given causal order $\Omega$, we will slightly abuse notation and write $(\Omega, \underline{I}, \underline{O})$ for the operational scenario, $\underline{I} = (I_\omega)_{\omega \in \Omega}$ for the family of input sets, and $\underline{O} = (O_\omega)_{\omega \in \Omega}$ for the family of output sets.

\begin{remark}
It is also possible to consider a more general definition in which the output set at each event is allowed to be \emph{dependent} on the input at the same event, so that $ \underline{O} = (O_{\omega, i})_{\omega \in E, i \in I_\omega}$.
Aside from providing additional type-theoretic sophistication to our framework, dependent output sets have an important practical effect in reducing the dimensionality of the embedding space for causal polytopes.
For sake of simplicity, we will not be using dependent output sets in our main text, but we will at times remark on how our definitions can be extended to include them.
\end{remark}

The full definition of operational scenarios---including both inputs and outputs---is useful to contextualise this work within the context of its companion works \cite{gogioso2022topology,gogioso2022geometry}, this Section is only be concerned with the input side of things.
Specifically, given an operational scenario and a causal order on its events, we will investigate the structure of ``input histories'', the possible assignments of joint inputs to subsets of events upon which individual outputs are allowed to depend.
Input histories are partial functions from events to inputs, so we start our investigation of the former by recapping some fundamental features of the latter.

\subsection{Partial Functions}
\label{subsection:spaces-pfuns}

\begin{definition}
Given a family $\underline{Y} = (Y_x)_{x \in X}$ of sets, the \emph{partial functions} $\PFun{\underline{Y}}$ on $\underline{Y}$ are defined to be all possible functions $f$ having subsets $D \subseteq X$ as their domain $\dom{f} := D$ and such that $f(x) \in Y_x$ for all $x \in D$.
\begin{equation}
    \PFun{\underline{Y}}
    :=
    \bigcup_{D \subseteq X}
    \prod_{x \in D}
    Y_x
\end{equation}
Partial functions are partially ordered by restriction:
\begin{equation}
    f \leq g
    \hspace{2mm}\stackrel{def}{\Leftrightarrow}\hspace{2mm}
    \dom{f} \subseteq \dom{g}
    \text{ and }
    \restrict{g}{\dom{f}} = f
\end{equation}
\end{definition}

In the special case where the set $Y_x = Y$ is the same for all $x \in X$, partial functions are exactly the functions $D \rightarrow Y$ for all $D \subset X$.
We can also interpret $\underline{Y} = (Y_x)_{x \in X}$ as a set-valued function $x \mapsto Y_x$, so that $\dom{\underline{Y}} = X$ and $\Subsets{\dom{\underline{Y}}}$ is the set of subsets of $X$, partially ordered by inclusion. With this notation, the domain function can be written as $\domSym:\PFun{\underline{Y}} \rightarrow \Subsets{\dom{\underline{Y}}}$, and it is always order-preserving:
\[
    f \leq g \Rightarrow \dom{f} \subseteq \dom{g}
\]

\begin{observation}
Under their restriction order, partial functions form a lower semilattice, with the empty function $\emptyset$ as its minimum and meets given by:
\begin{equation}
\begin{array}{rcl}
    \dom{f \wedge g}
    &=&
    \suchthat{x \in \dom{f}\cap\dom{g}}{f(x) = g(x)}
    \\
    f \wedge g
    &=&
    \restrict{f}{\dom{f \wedge g}}
    =
    \restrict{g}{\dom{f \wedge g}}
\end{array}
\end{equation}
\end{observation}

The domain function $\domSym:\PFun{\underline{Y}} \rightarrow \Subsets{\dom{\underline{Y}}}$ is not meet-preserving, but the following inclusion holds by definition of the meet:
\[
    \dom{f \wedge g}
    \subseteq
    \dom{f} \cap \dom{g}
\]

\begin{definition}
\label{definition:compatible-pfun}
We say that $f$ and $g$ are \emph{compatible} when the inclusion above is an equality:
\begin{equation}
    \text{$f$ and $g$ compatible}
    \hspace{2mm} \Leftrightarrow \hspace{2mm}
    \dom{f \wedge g}
    =
    \dom{f} \cap \dom{g}
\end{equation}
More generally, we say that a set $\mathcal{F} \subseteq \PFun{\underline{Y}}$ of partial functions is \emph{compatible} if $f$ and $g$ are compatible for all $f, g \in \mathcal{F}$.
\end{definition}

\begin{definition}
\label{definition:compatible-join}
The \emph{join} of a set $\mathcal{F}$ of partial functions exists exactly when the set is compatible, in which case it is given by:
\begin{equation}
\label{eq:definition:join}
\begin{array}{rcl}
    \dom{\bigvee \mathcal{F}}
    &=&
    \bigcup\limits_{f \in \mathcal{F}} \dom{f}
    \\
    \bigvee \mathcal{F}
    &=&
    x \mapsto f(x) \text{ for any $f$ such that } x \in \dom{f}
\end{array}
\end{equation}
The \emph{compatible joins} in a set $\mathcal{F}'$ of partial functions are all possible joins $\bigvee\mathcal{F}$ of compatible subsets $\mathcal{F} \subseteq \mathcal{F}'$.
\end{definition}

The definition of $\bigvee \mathcal{F}$ is sound because compatibility of $\mathcal{F}$ means that $f(x) = g(x)$ for all $f, g \in \mathcal{F}$ such that $x \in \dom{f}$ and $x \in \dom{g}$.
For two functions $f,g$, the join $f\vee g$ takes the following explicit form:
\[
\begin{array}{rcll}
    (f \vee g)(x)
    &=&
    f(x) &\hspace{2mm} \text{ if } x \in \dom{f}\backslash\dom{g}
    \\
    (f \vee g)(x)
    &=&
    g(x) &\hspace{2mm} \text{ if } x \in \dom{g}\backslash\dom{f}
    \\
    (f \vee g)(x)
    &=&
    f(x) = g(x) &\hspace{2mm} \text{ if } x \in \dom{f}\cap\dom{g}
\end{array}
\]
In a join, two or more partial functions are effectively ``stitched together'' around their common values, yielding a partial function which extends them all.
The maxima of $\PFun{\underline{Y}}$ are those where this stitching process cannot go any further; they are exactly the total functions, defined on the entirety of $\dom{\underline{Y}}$:
\[
\max\left(\PFun{\underline{Y}}\right)
=
\prod_{x \in \dom{Y}} Y_x
\]

\subsection{Input Histories for Causal Orders}
\label{subsection:spaces-order-induced}

Consider a causal order $\Omega$ and an associated operational scenario $(\Omega, \underline{I}, \underline{O})$; we only need the family of input sets $\underline{I}$ in what follows, but we mention the whole scenario for context.
Because of causality, the output at an event $\xi$ can only depend on choices of inputs for the events $\omega$ in $\downset{\xi}$, the causal past of $\xi$. This observation motivates the following definition.

\begin{definition}
The \emph{input histories} for a given choice of order $\Omega$ and inputs $\underline{I} = (I_\omega)_{\omega \in \Omega}$ are defined to be the partial functions in the following set:
\begin{equation}
    \Hist{\Omega, \underline{I}}
    \;:=\;
    \bigcup_{\xi \in \Omega}
    \prod_{\omega \in \downset{\xi}}
    I_\omega
    \;\;\subseteq\;\;
    \PFun{\underline{I}}    
\end{equation}
Input histories inherit the restriction order of partial functions, and we refer to the partially ordered set $\Hist{\Omega, \underline{I}}$ as a \emph{space of input histories}.
\end{definition}

As a simple concrete example, let $\Omega$ be the total order on 3 events and consider its associated lattice of lowersets $\Lsets{\Omega}$.
In the lattice, the causal pasts of individual events have been colour-coded.
\begin{center}
    \raisebox{1.15cm}{$\LsetsSym$}
    \raisebox{1.15cm}{$\left(\rule{0cm}{1.35cm}\right.$}
    \hspace{-0.3cm}
    \raisebox{0cm}{
        \includegraphics[height=2.5cm]{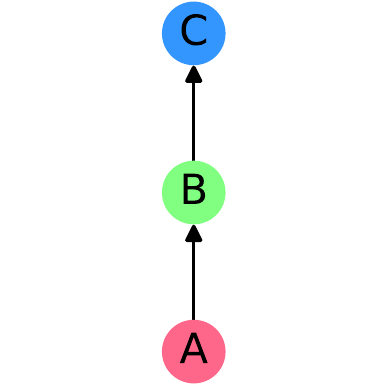}
    }
    \hspace{-0.3cm}
    \raisebox{1.15cm}{$\left.\rule{0cm}{1.35cm}\right)$}
    \hspace{0.5cm}
    \raisebox{1.15cm}{$=$}
    \hspace{0.25cm}
    \includegraphics[height=2.5cm]{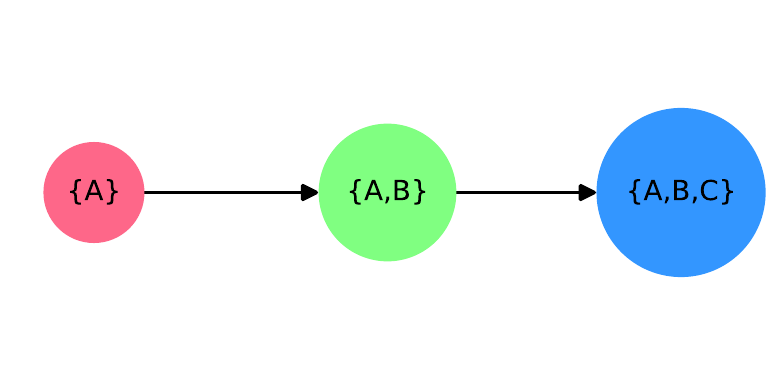}
\end{center}
Below is the Hasse diagram for the space of input histories, consisting of all binary functions on $\evset{A}$, $\evset{A,B}$ and $\evset{A,B,C}$.
\begin{center}
    \includegraphics[height=3cm]{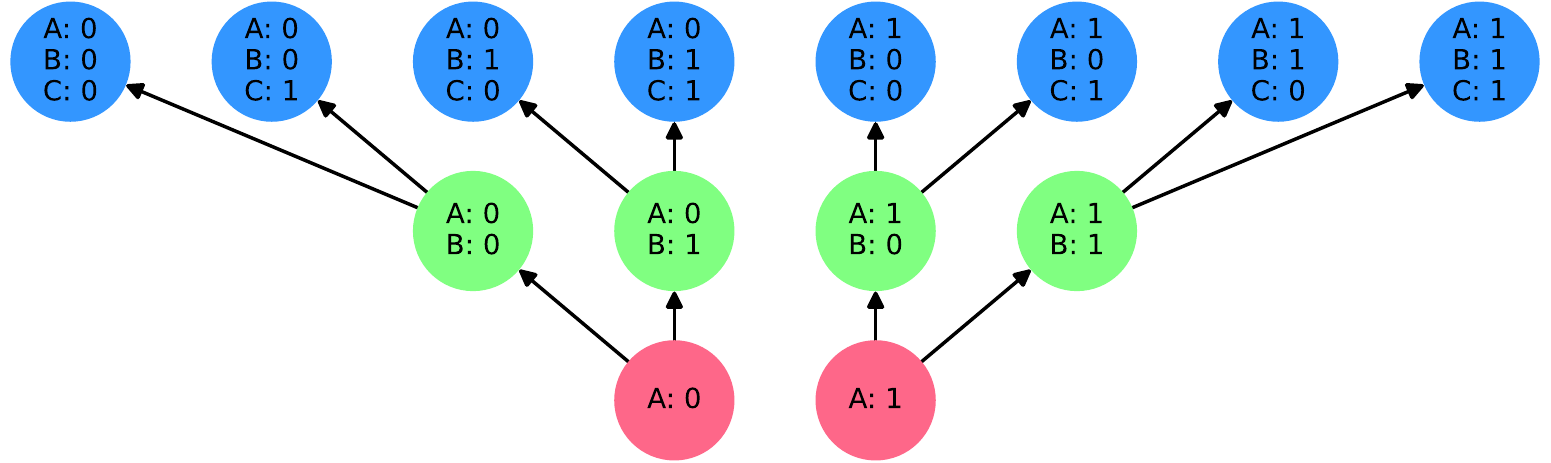}
\end{center}

Input histories are not generally closed under meets, and the only subsets closed under joins are chains (in which case the join is the maximum).
When talking about the meet of two (or more) input histories, we will always mean the meet in $\PFun{\underline{I}}$; similarly, when saying that a subset $\mathcal{F} \subseteq \Hist{\Omega, \underline{I}}$ of input histories is compatible, we will always mean that it is compatible in $\PFun{\underline{I}}$.
As a concrete example of the lack of closure under joins and meets, we consider the following ``M''-shaped causal order on 4 events.
In the lattice, the causal pasts of individual events have been colour-coded: we observe that both the intersection $\evset{A,B}$ and the union $\{\ev{A},\ev{B},\ev{C},\ev{D}\}$ of the causal pasts $\downset{\ev{C}}=\evset{A,B,C}$ and $\downset{\ev{D}}=\evset{A,B,D}$ are not causal pasts of events themselves.
\begin{center}
    \raisebox{2.4cm}{$\LsetsSym$}
    \raisebox{2.4cm}{$\left(\rule{0cm}{1.35cm}\right.$}
    \hspace{-0.3cm}
    \raisebox{1.25cm}{
        \includegraphics[height=2.5cm]{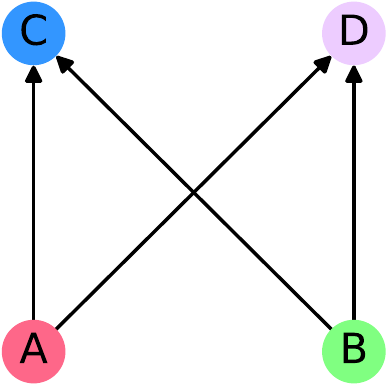}
    }
    \hspace{-0.3cm}
    \raisebox{2.4cm}{$\left.\rule{0cm}{1.35cm}\right)$}
    \hspace{0.25cm}
    \raisebox{2.4cm}{$=$}
    \hspace{-0.25cm}
    \includegraphics[height=5cm]{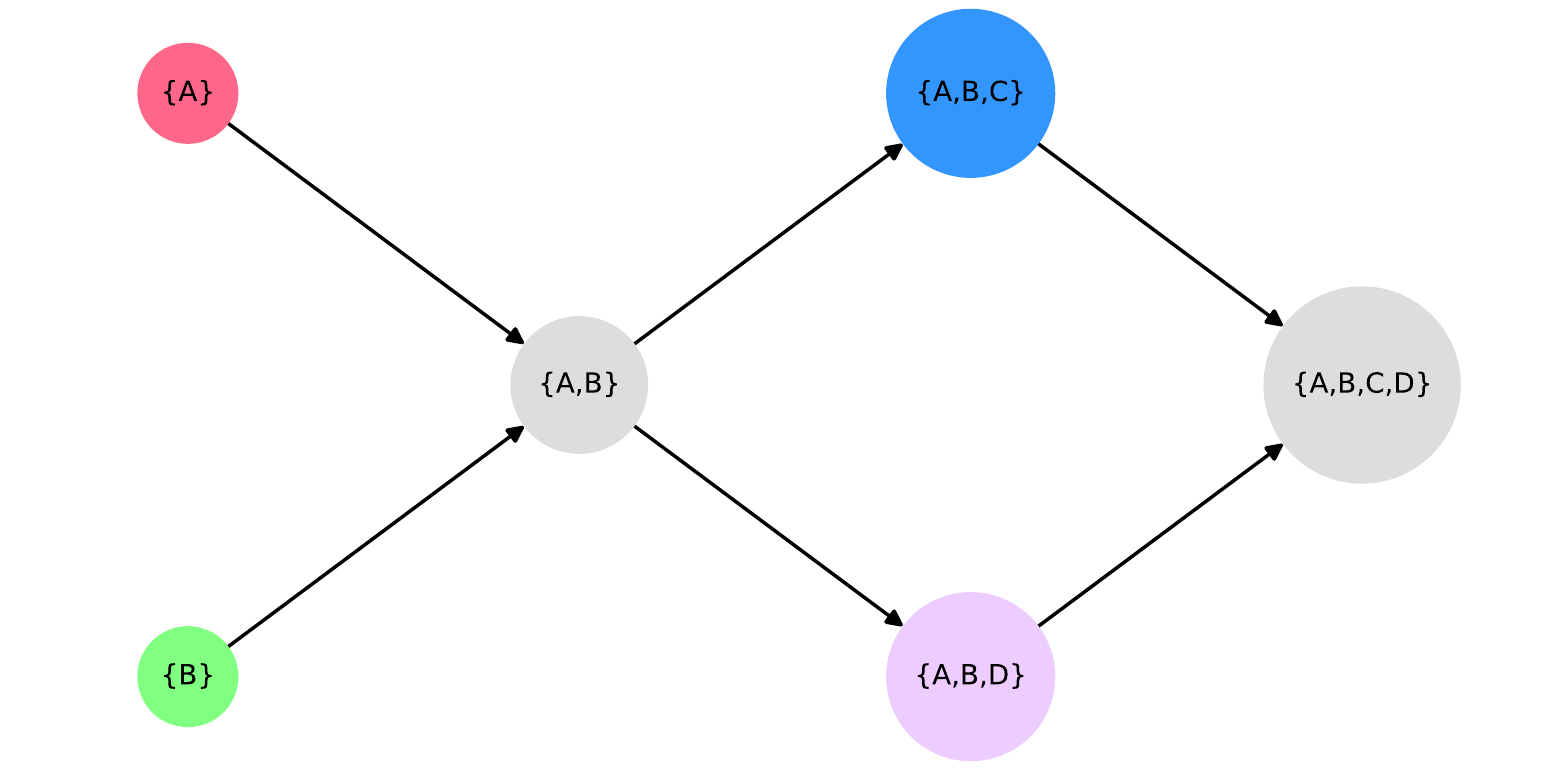}
\end{center}
The associated space of input histories doesn't feature any meets $f \wedge g$ or joins $f \vee g$ for compatible histories $f,g$ with domain $\downset{\ev{C}}$ and $\downset{\ev{D}}$ respectively (we remind the reader that the meets and joins being referred to are those in $\PFun{\underline{I}}$).
\begin{center}
    \includegraphics[height=4cm]{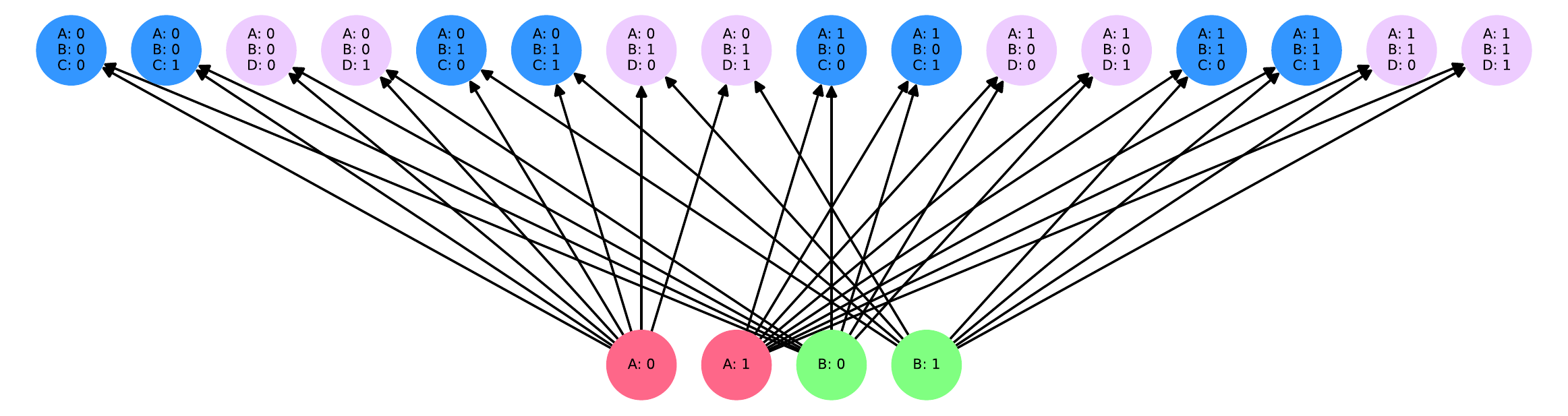}
\end{center}
The space above is also an example where the maxima of the space of input histories differ from those of $\PFun{\underline{I}}$: $\{\ev{A},\ev{B},\ev{C},\ev{D}\}$ is not the causal past of an event, so the total functions in $\PFun{\underline{I}}$ are not input histories.

To overcome the limitations of input histories in terms of meets and compatible joins, we also introduce a notion of ``extended'' input histories, defined on all causal lowersets.

\begin{definition}
The \emph{extended input histories} for a given choice of order $\Omega$ and inputs $\underline{I} = (I_\omega)_{\omega \in \Omega}$ are defined to be the partial functions in the following set:
\begin{equation}
    \ExtHist{\Omega, \underline{I}}
    \;:=\;
    \bigcup_{U \in \Lsets{\Omega}}
    \prod_{\omega \in U}
    I_\omega
    \;\;\subseteq\;\;
    \PFun{\underline{I}}    
\end{equation}
Extended input histories inherit the restriction order of partial functions, and we refer to the partially ordered set $\ExtHist{\Omega, \underline{I}}$ as a \emph{space of extended input histories}.
\end{definition}

The space of extended input histories $\ExtHist{\Omega, \underline{I}}$ contains the space of histories $\Hist{\Omega, \underline{I}}$, so that input histories are a special case of extended input histories.
Because $\Lsets{\Omega}$ is a lattice, all meets and compatible joins of input histories are also extended input histories.
The total functions $\prod_{\omega \in \Omega} I_\omega$ are also all in $\ExtHist{\Omega, \underline{I}}$; in certain circumstances, we will refer to them as the ``maximal extended input histories''.
In the case of total orders, where all lowersets are causal pasts of events, the spaces of input histories and extended input histories always coincide.

The spaces of input histories derived from causal orders work quite well when the orders are definite, but they do not quite capture the full desired gamut of possibilities for indefinite causal orders.
Indeed, consider the following indefinite causal order $\Omega$ on 3 events, and its associated lattice of lowersets $\Lsets{\Omega}$.
\begin{center}
    \raisebox{1.15cm}{$\LsetsSym$}
    \raisebox{1.15cm}{$\left(\rule{0cm}{1.35cm}\right.$}
    \hspace{-0.4cm}
    \raisebox{0cm}{
        \includegraphics[height=2.5cm]{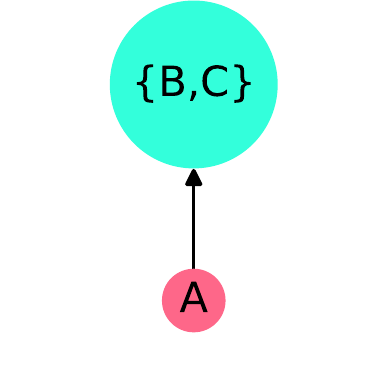}
    }
    \hspace{-0.5cm}
    \raisebox{1.15cm}{$\left.\rule{0cm}{1.35cm}\right)$}
    \hspace{0.5cm}
    \raisebox{1.15cm}{$=$}
    \hspace{0.25cm}
    \raisebox{0.0cm}{
        \includegraphics[height=2.5cm]{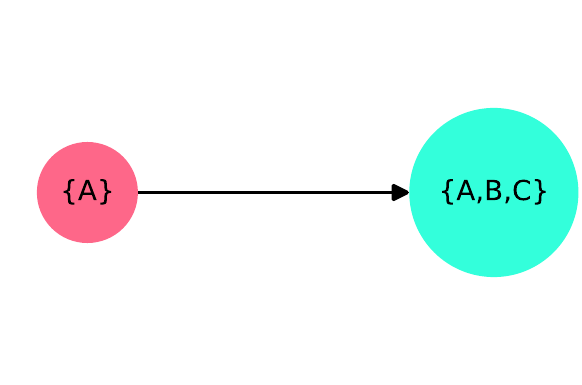}
    }
\end{center}
Because events \ev{B} and \ev{C} are in indefinite causal order, they have the same causal past, and hence they are never separated by input histories.
\begin{center}
    \includegraphics[height=2.5cm]{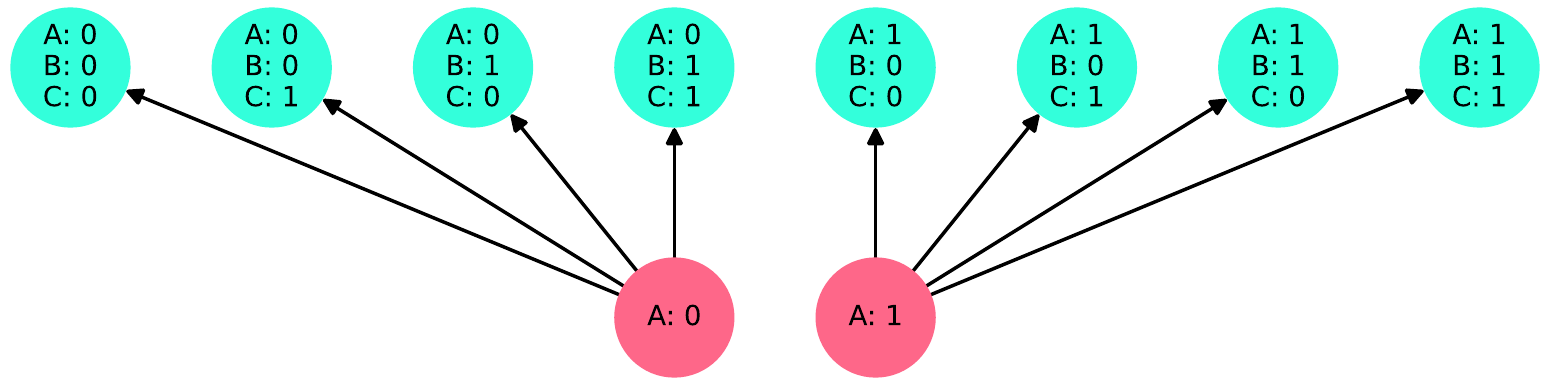}
\end{center}
We will revisit this specific issue later, when talking about ``causal completeness'', but it already prompts the question: can we extend our spaces of input histories to capture additional information about causal indefiniteness?

For example, we might want to consider a ``3-party causal switch'' space, in which event \ev{A} controls the order of events \ev{B} and \ev{C} with its input, e.g. by setting $\ev{B} < \ev{C}$ when the input is 0 and $\ev{C} < \ev{B}$ when the input is 1.
In such a setting, the output at \ev{B} is fully determined by the inputs at events \ev{A} and \ev{B} when the input at \ev{A} is 0, but the input at event \ev{C} is also needed---in the general case---when the input at \ev{A} is 1.
Taking this observation---and the analogous one about the output at \ev{C}---we obtain our desired space of input histories.
\begin{center}
    \includegraphics[height=3cm]{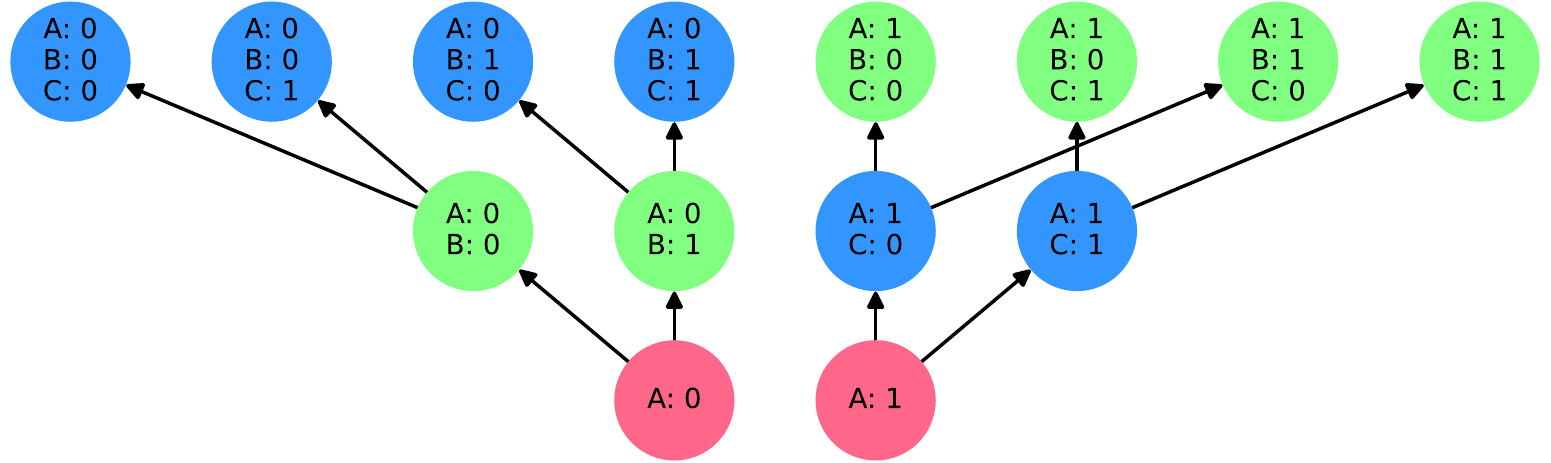}
\end{center}
The space above is a subset of $\PFun{\underline{I}}$, but does not arise from $\Hist{\Omega, \underline{I}}$ for any causal order $\Omega$: the order between \ev{B} and \ev{C} is indefinite overall, but the input histories are able to discriminate between the two events based on the input at \ev{A} (colour coding of input histories reflects this fact).
Clearly, we need a more general definition of the possible spaces of input histories for an operational scenario.

\subsection{Spaces of Input Histories}
\label{subsection:spaces-definition}

It might be tempting to define spaces of input histories as generic subsets of $\PFun{\underline{I}}$, but this definition turns out to be too broad: for one thing, it blurs the distinction between input histories (the data upon which outputs at events depend) and extended input histories (obtained by stitching together compatible input histories).
Instead, we observe that the spaces of input histories $\Hist{\Omega, \underline{I}}$ we defined for causal orders have a couple of special properties: they are $\vee$-prime (read ``join-prime'') and they satisfy the ``free-choice condition''.

\begin{definition}
A subset $\Theta \subseteq \PFun{\underline{I}}$ is said to be \emph{$\vee$-prime} (read ``join-prime'') if no $h \in \Theta$ can be written as the compatible join $h = \bigvee \mathcal{F}$ of a subset $\mathcal{F} \subseteq \Theta$ such that $h \notin \mathcal{F}$:
\[
\left(
    \mathcal{F} \subseteq \Theta \text{ compatible and }
    \bigvee \mathcal{F} \in \Theta
\right)
\Rightarrow \bigvee \mathcal{F} \in \mathcal{F}
\]
Dually, a subset $W \subseteq \PFun{\underline{I}}$ is said to be \emph{$\vee$-closed} (read ``join-closed'') if for every pair of compatible $h, k \in W$ the join $h \vee k$ is itself in $W$.
This implies that, more generally:
\[
\mathcal{F} \subseteq \Theta \text{ compatible }
\Rightarrow
\bigvee \mathcal{F} \in \Theta
\]
\end{definition}

\begin{proposition}
\label{proposition:hist-construction-is-join-prime-subset}
For any causal order $\Omega$, $\Hist{\Omega, \underline{I}} \subseteq \PFun{\underline{I}}$ is always a $\vee$-prime subset of $\PFun{\underline{I}}$.
\end{proposition}
\begin{proof}
See \ref{proof:proposition:hist-construction-is-join-prime-subset}
\end{proof}

The $\vee$-primality condition forms the basis for our generalisation of the notion of space of input histories: it gives a ``normal form'' to spaces, making their correspondence with causal constraints exact, and it removes ``redundant extended input histories'', so that causal functions can be defined freely.

\begin{definition}
A \emph{space of input histories} is a finite set $\Theta$ of partial functions which is $\vee$-prime.
The associated \emph{event set} $\Events{\Theta}$ and family of \emph{input sets} $\underline{\Inputs{\Theta}} = (\Inputs{\Theta}_\omega)_{\omega \in \Events{\Theta}}$ are defined as follows:
\begin{equation}
    \begin{array}{lcl}
        \Events{\Theta}
        &:=&
        \bigcup_{h \in \Theta} \dom{h}
        \\
        \Inputs{\Theta}_\omega
        &:=&
        \suchthat{h_\omega}{h \in \Theta, \omega \in \dom{h}}
    \end{array}
\end{equation}
We have $\Theta \subseteq \PFun{\underline{\Inputs{\Theta}}}$ and the space $\Theta$ is equipped with the partial order inherited from $\PFun{\underline{\Inputs{\Theta}}}$.
The \emph{space of extended input histories} $\Ext{\Theta}$ associated to $\Theta$ is defined to be its $\vee$-closure:
\begin{equation}
    \Ext{\Theta}
    :=
    \suchthat{\bigvee \mathcal{F}}{\mathcal{F} \subseteq \Theta \text{ compatible}}
\end{equation}
We have $\Ext{\Theta} \subseteq \PFun{\underline{\Inputs{\Theta}}}$ and the space $\Ext{\Theta}$ is equipped with the partial order inherited from $\PFun{\underline{\Inputs{\Theta}}}$.
\end{definition}

\begin{observation}
Given any subset $W \subseteq \PFun{\underline{I}}$, we can obtain a space of input histories by taking its $\vee$-prime elements:
\begin{equation}
    \Prime{W}
    :=
    \suchthat{w \in W}{
    \forall \mathcal{F} \subseteq W \text{ compatible}.\,
    w = \bigvee\mathcal{F} \Rightarrow w \in \mathcal{F}
    }
\end{equation}
In particular, we can recover a space of input histories from its corresponding space of extended input histories:
\[
\Prime{\Ext{\Theta}}
=
\Theta
\]
Conversely, any $\vee$-closed subset $W \subseteq \PFun{\underline{I}}$ is a space of extended input histories:
\[
W \text{ $\vee$-closed}
\Rightarrow
\Ext{\Prime{W}} = W
\]
\end{observation}

\begin{observation}
For a space of input histories $\Theta = \Hist{\Omega, \underline{I}}$ induced by given causal order $\Omega$ and input sets $\underline{I}$, we have $\Events{\Theta}=|\Omega|$, $\underline{\Inputs{\Theta}}=\underline{I}$ and $\Ext{\Hist{\Omega, \underline{I}}} = \ExtHist{\Omega, \underline{I}}$
\end{observation}

\begin{observation}
Because the extended input histories are all obtained as compatible joins of input histories, the event set and input set could have been equivalently defined on all extended input histories:
\begin{equation}
    \begin{array}{lcl}
        \Events{\Theta}
        &=&
        \bigcup_{h \in \Ext{\Theta}} \dom{h}
        \\
        \Inputs{\Theta}_\omega
        &=&
        \suchthat{h_\omega}{h \in \Ext{\Theta}, \omega \in \dom{h}}
    \end{array}
\end{equation}
\end{observation}

Below is an example of a space of input histories $\Theta$ (on the left) together with its corresponding space of extended input histories $\Ext{\Theta}$ (on the right).
Input histories have been colour-coded by the events whose output they refer to (more on this later), in both diagrams: grey coloured extended input histories on the right are those which are not input histories (i.e. they arise by join).
This space is a variation on the total order $\total{\ev{A}, \ev{B}, \ev{C}}$, where input 0 at event $\ev{B}$ causally disconnects $\ev{B}$ from $\ev{A}$ and input 0 at event $\ev{C}$ causally disconnects $\ev{C}$ from both $\ev{B}$ and $\ev{A}$.
\begin{center}
    \begin{tabular}{cc}
    \includegraphics[height=3.5cm]{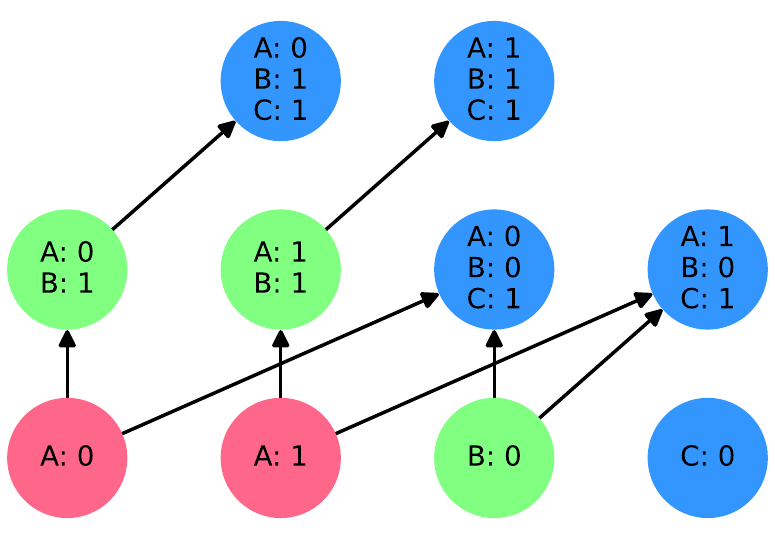}
    &
    \includegraphics[height=3.5cm]{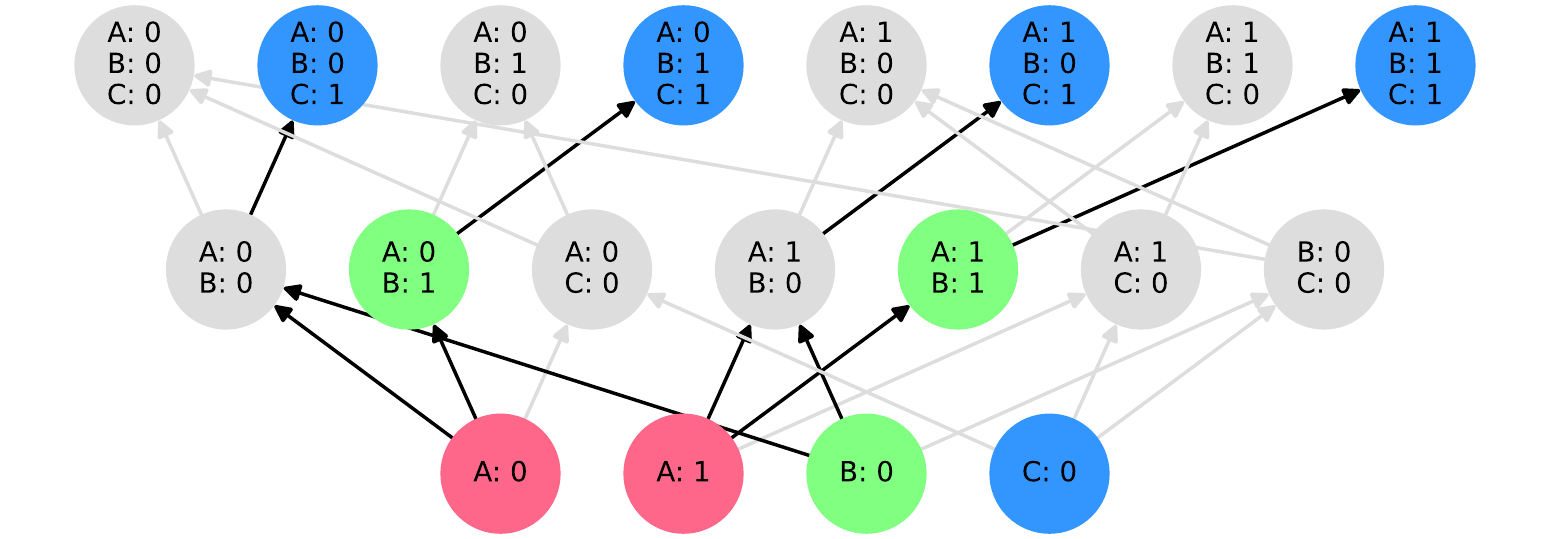}
    \\
    $\Theta$
    &
    $\Ext{\Theta}$
    \end{tabular}
\end{center}
To understand the $\vee$-primality and $\vee$-closure conditions, we study the Hasse diagram on the right, depicting the space of extended input histories $\Ext{\Theta}$.
With regards to the $\vee$-closure condition, note how all compatible input histories have some common successor in the graph: for some of them, such as \hist{A/1} and \hist{B/0}, this is an immediate common successor, namely \hist{A/1,B/0}; for others, such as \hist{B/0} and \hist{C/0}, this is a common successor further up the graph, e.g. \hist{A/1,B/0,C/0}.
Extended input histories without a join are always incompatible ones, such as \hist{A/0,C/0} and \hist{A/1,C/0} (differing in value on a common event $\ev{A}$).
With regards to the $\vee$-primality condition, note how all input histories (coloured nodes) have at most one predecessor in the graph: they cannot arise as joins of extended input histories below them.
This is a general fact: extended input histories arise as (possibly trivial) joins of input histories and they are $\vee$-closed, so the input histories can never have more than one predecessor in $\Ext{\Theta}$.

In our description of operational scenarios we stated that inputs at events are ``freely chosen'', i.e. without any local or global constraint.
By stitching together input histories, it must therefore be possible to obtain all possible combinations of joint input values over all events: in other words, the ``maximal extended input histories''---the total functions in $\PFun{\underline{\Inputs{\Theta}}}$---must all arise by joins of compatible subsets of the input histories in a space.

\begin{definition}
A space of input histories $\Theta$ is said to satisfy the \emph{free-choice condition} if:
\[
    \max\Ext{\Theta} = \prod_{\omega \in \Events{\Theta}} \Inputs{\Theta}_\omega
\]
In spaces satisfying the free-choice condition, we refer to the histories in $\prod_{\omega \in \Events{\Theta}} \Inputs{\Theta}_\omega$ as the \emph{maximal extended input histories}.
\end{definition}

\begin{proposition}
\label{proposition:order-induced-space-free-choice}
The spaces of input histories $\Hist{\Omega, \underline{I}}$ constructed from causal orders always satisfy the free-choice condition.
\end{proposition}
\begin{proof}
See \ref{proof:proposition:order-induced-space-free-choice}
\end{proof}

\begin{observation}
The ``minimal'' extended input histories $k \in \Ext{\Theta}$ are those without sub-histories, i.e. those such:
\[
\forall k' \in \Ext{\Theta}. k' \leq k \Rightarrow k' = k
\]
Such $k$ are necessarily $\vee$-prime, so we refer to them as the \emph{minimal input histories}.
\end{observation}

Recall now that causal orders form a hierarchy (a lattice) when ordered by inclusion.
We would like this hierarchy to generalise from causal orders $\Omega$ to their spaces of input histories $\Hist{\Omega, \underline{I}}$, and then to all spaces of input histories, including ones that don't arise from orders.
Unfortunately, this is not as simple as ordering the spaces themselves by inclusion: $\Omega \leq \Xi$ does not in general imply an inclusion relationship between $\Hist{\Omega, \underline{I}}$ and $\Hist{\Xi, \underline{I}}$.
Indeed, consider the following fork $\Omega$ and total order $\Xi$ on 3 events.
\begin{center}
\begin{tabular}{ccc}
    \raisebox{0.25cm}{
        \includegraphics[height=2cm]{svg-inkscape/vee-A-BC_svg-tex.pdf}
    }
    &
    \hspace{1cm}
    \raisebox{1.15cm}{$\leq$}
    \hspace{1cm}
    &
    \raisebox{0cm}{
        \includegraphics[height=2.5cm]{svg-inkscape/total-ABC_svg-tex.pdf}
    }
    \\
    $\Theta$
    &&
    $\Xi$
\end{tabular}
\end{center}
The corresponding spaces of input histories $\Hist{\Omega, \underline{I}}$ and $\Hist{\Xi, \underline{I}}$ are not related by inclusion in either direction.
To make this evident, no colour-coding is used for the input histories in these diagrams: instead, the common input histories have been highlighted with a darker colour.
\begin{center}
\begin{tabular}{ccc}
    \raisebox{0cm}{
        \includegraphics[height=2cm]{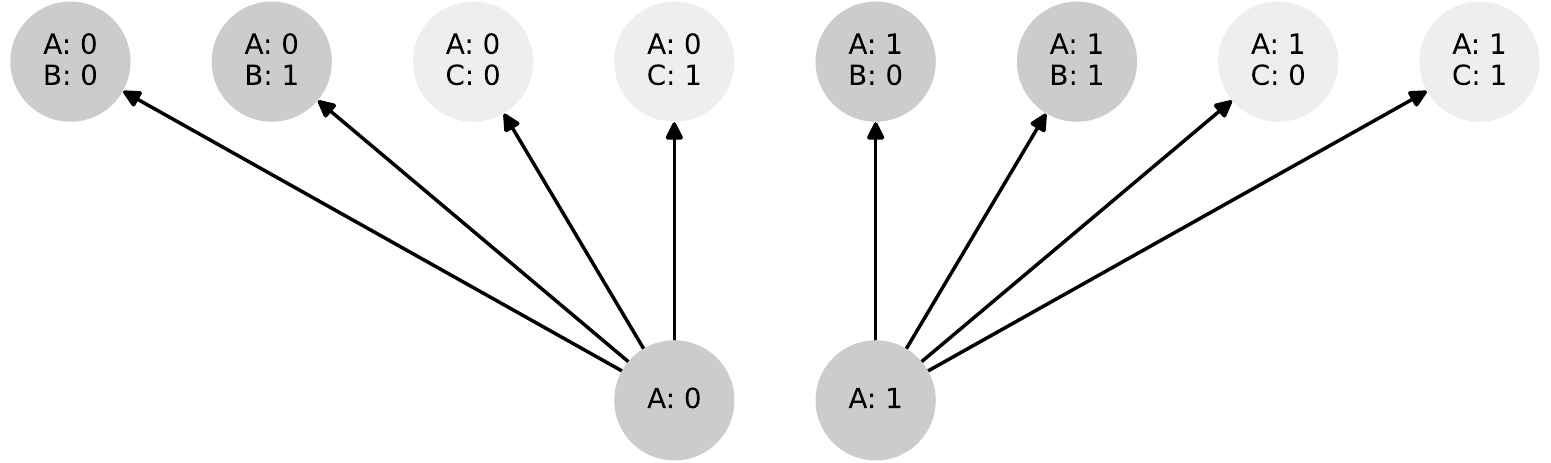}
    }
    &
    \hspace{0.1cm}
    \raisebox{0.9cm}{$\left/\rule{0cm}{0.35cm}\right.\hspace{-4.22mm}\supsub$}
    \hspace{0.1cm}
    &
    \raisebox{0cm}{
        \includegraphics[height=2cm]{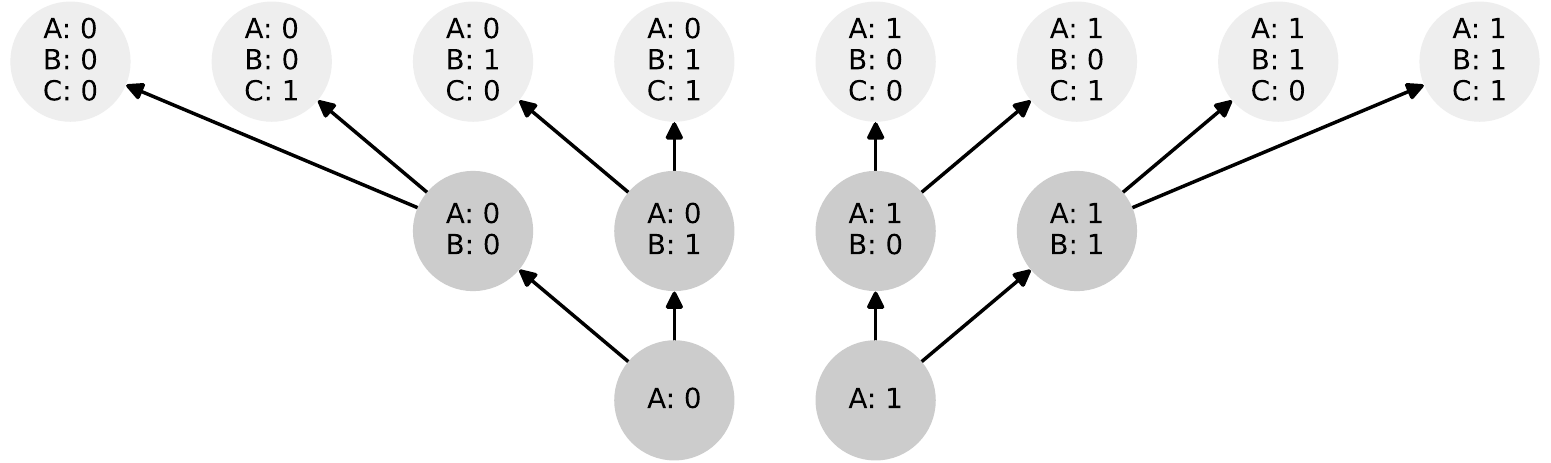}
    }
    \\
    $\Hist{\Omega, \underline{I}}$
    &&
    $\Hist{\Xi, \underline{I}}$
\end{tabular}
\end{center}
However, a suitable statement of inclusion holds for the corresponding spaces of extended input histories.
\begin{proposition}
\label{proposition:induced-spaces-order}
For any two causal orders $\Omega$ and $\Xi$ we have:
\begin{equation}
    \Omega \leq \Xi
    \;\;\Leftrightarrow\;\;
    \ExtHist{\Omega, \underline{I}} \supseteq \ExtHist{\Xi, \underline{I}}    
\end{equation}
\end{proposition}
\begin{proof}
See \ref{proof:proposition:induced-spaces-order}
\end{proof}

We take the result of Proposition \ref{proposition:induced-spaces-order} above as the basis to define a partial order on spaces of input histories, via the corresponding spaces of extended input histories.
We have a choice between two opposite order conventions for $\Theta \leq \Theta'$: we could define it to be $\Ext{\Theta} \subseteq \Ext{\Theta'}$ or we could define it to be $\Ext{\Theta} \supseteq \Ext{\Theta'}$.
The former choice is more straightforward to remember and understand, while the latter matches the convention for causal orders, where $\Omega \leq \Omega'$ means that $\Omega$ has more causal constraints than $\Omega'$.
As we discuss in \cite{gogioso2022geometry}, the latter choice is also aligned with the inclusion order for the corresponding causal polytopes, so we take it as our definition.

\begin{definition}
\label{definition:spaces-order}
We define the following partial order on spaces of input histories:
\begin{equation}
    \Theta' \leq \Theta
    \stackrel{def}{\Leftrightarrow}
    \Ext{\Theta'} \supseteq \Ext{\Theta}
\end{equation}
We say that $\Theta'$ is a \emph{causal refinement} of $\Theta$ (more causal constraints), or that $\Theta$ is a \emph{causal coarsening} of $\Theta'$ (fewer causal constraints).
Equivalently, sometimes we say that $\Theta'$ is a \emph{sub-space} of $\Theta$ or that $\Theta'$ is a \emph{super-space} of $\Theta$.
\end{definition}

\begin{proposition}
\label{proposition:subspace-input-histories-events-inputs}
If $\Theta$ is a space of input histories and $\Theta' \leq \Theta$ is a sub-space, then we have $\Events{\Theta'} \supseteq \Events{\Theta}$ and $\Inputs{\Theta'}_\omega \supseteq \Inputs{\Theta}_\omega$ for all $\omega \in \Events{\Theta}$.
\end{proposition}
\begin{proof}
See \ref{proof:proposition:subspace-input-histories-events-inputs}
\end{proof}

To exemplify the ordering on spaces of input histories that we just defined, we consider once again the causal fork and total order on three events.
\begin{center}
\begin{tabular}{ccc}
    \raisebox{0.25cm}{
        \includegraphics[height=2cm]{svg-inkscape/vee-A-BC_svg-tex.pdf}
    }
    &
    \hspace{1cm}
    \raisebox{1.15cm}{$\leq$}
    \hspace{1cm}
    &
    \raisebox{0cm}{
        \includegraphics[height=2.5cm]{svg-inkscape/total-ABC_svg-tex.pdf}
    }
    \\
    $\Theta$
    &&
    $\Xi$
\end{tabular}
\end{center}
We now show that the space of input histories for the causal fork lies below the space of input histories for the total order.
\begin{center}
\begin{tabular}{ccc}
    \raisebox{0cm}{
        \includegraphics[height=2cm]{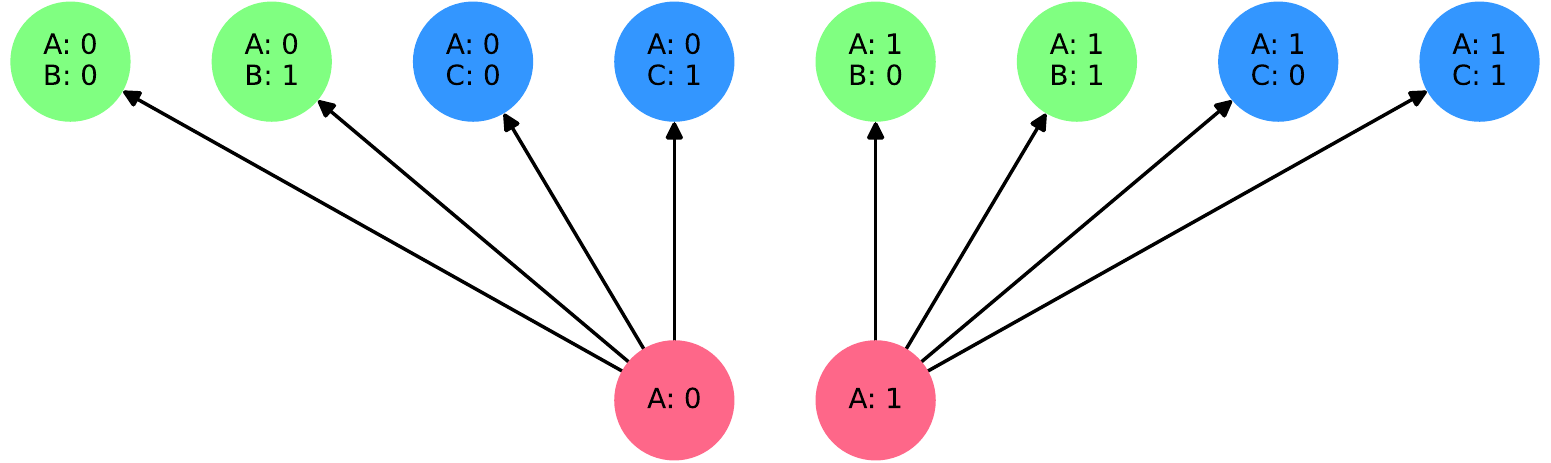}
    }
    &
    \hspace{0.1cm}
    \raisebox{0.9cm}{$\leq$}
    \hspace{0.1cm}
    &
    \raisebox{0.08cm}{
        \includegraphics[height=1.84cm]{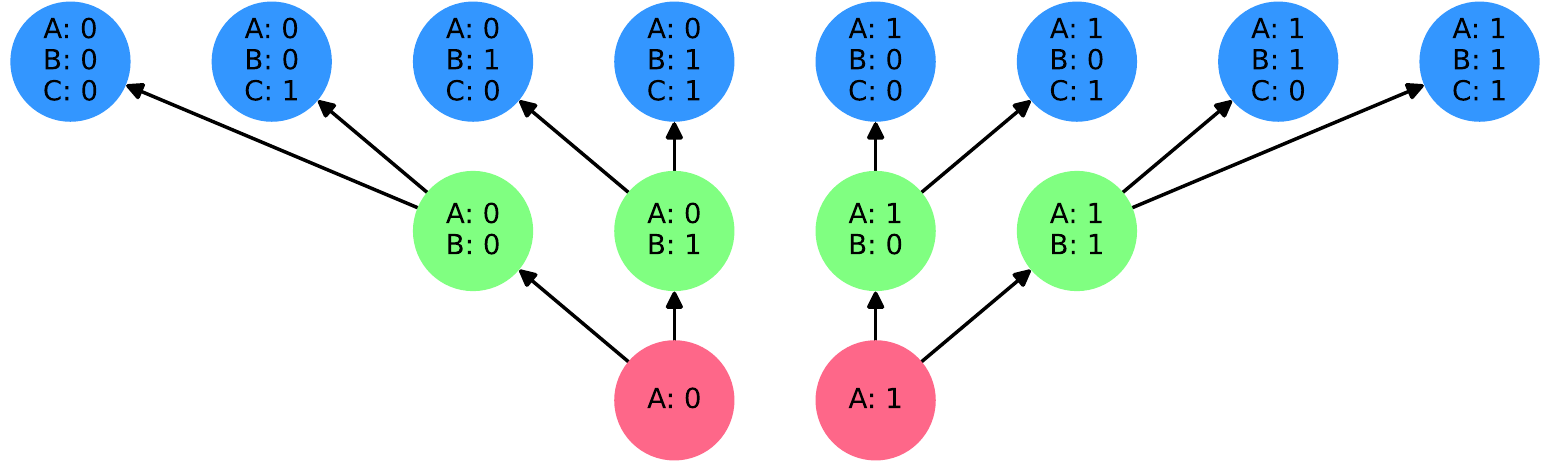}
    }
    \\
    $\Hist{\Omega, \underline{I}}$
    &&
    $\Hist{\Xi, \underline{I}}$
\end{tabular}
\end{center}
Indeed, all we have to check is that the reverse inclusion holds for the corresponding spaces of extended input histories, shown below.
To make it easier to spot the inclusion, no colour-coding is used for input histories in these diagrams: instead, the extended input histories from the space on the right have been highlighted in both spaces with a darker colour.
\begin{center}
\begin{tabular}{ccc}
    \raisebox{0cm}{
        \includegraphics[height=2cm]{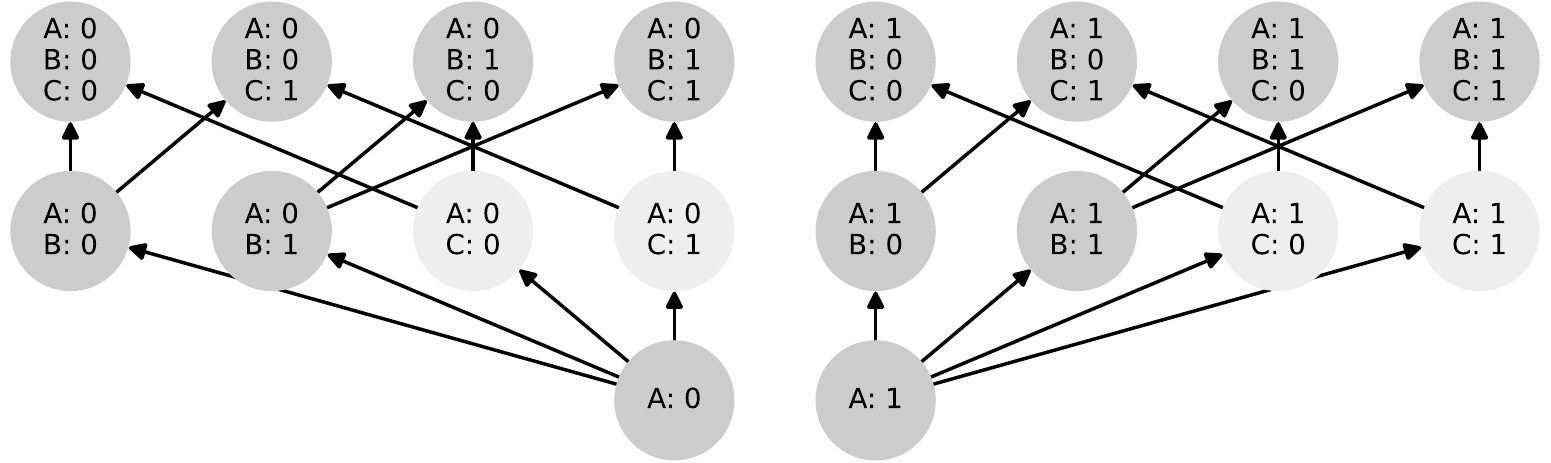}
    }
    &
    \hspace{0.1cm}
    \raisebox{0.9cm}{$\supseteq$}
    \hspace{0.1cm}
    &
    \raisebox{0.08cm}{
        \includegraphics[height=1.84cm]{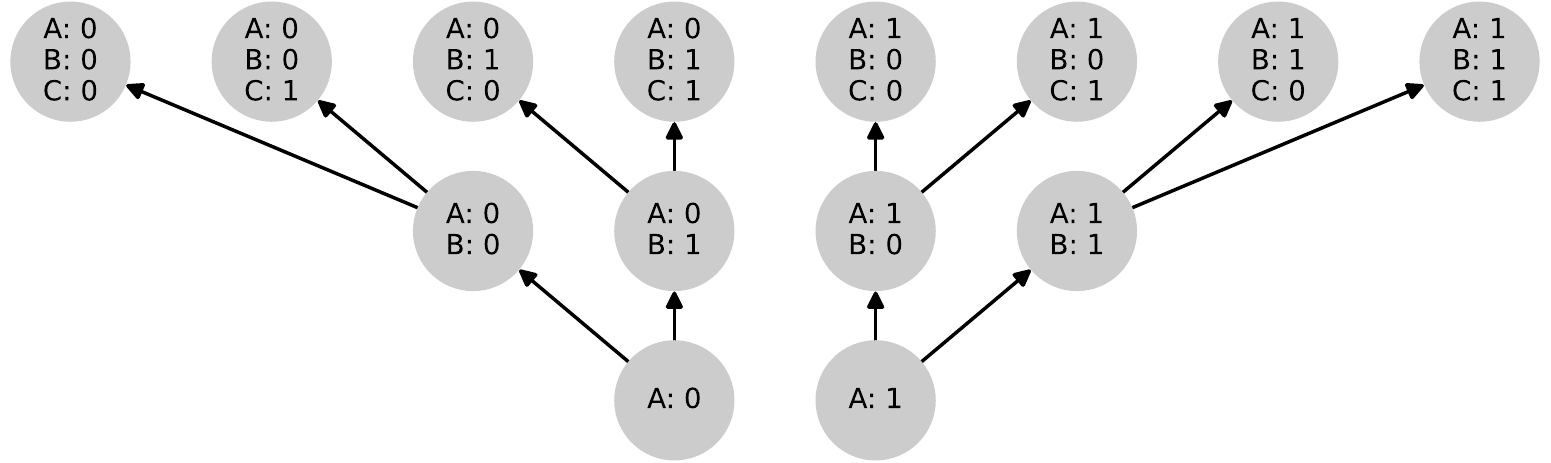}
    }
    \\
    $\ExtHist{\Omega, \underline{I}}$
    &&
    $\ExtHist{\Xi, \underline{I}}$
\end{tabular}
\end{center}
As another example, we look again at the space of input histories for the 3-party causal switch, and we compare it to the space of input histories for the indefinite causal order on 3 events which originally inspired it.
\begin{center}
\begin{tabular}{ccc}
    \raisebox{0cm}{
        \includegraphics[height=2cm]{svg-inkscape/switch-ABC-hists_svg-tex.pdf}
    }
    &
    \raisebox{0.9cm}{$\leq$}
    &
    \raisebox{0.08cm}{
        \includegraphics[height=1.84cm]{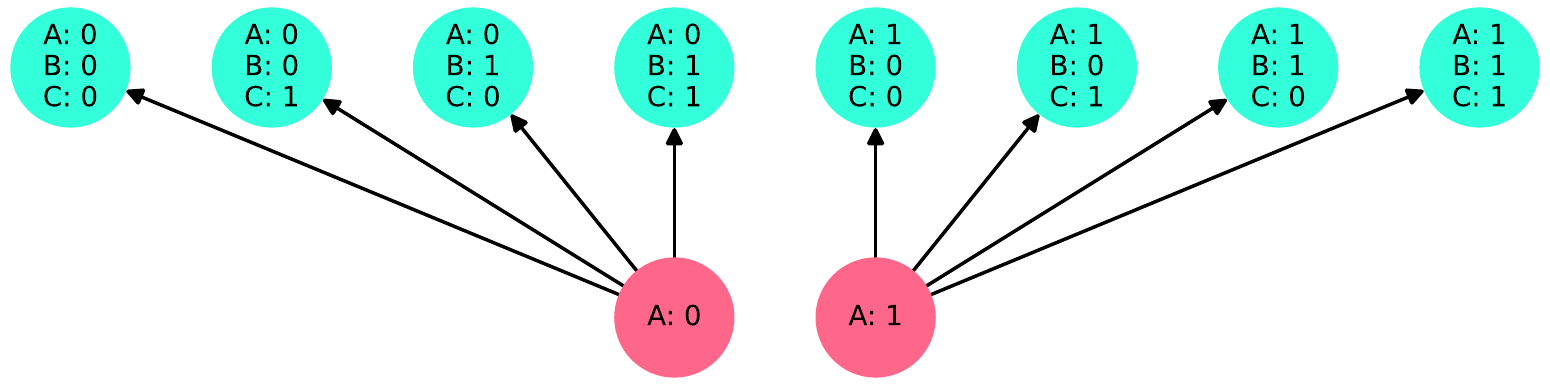}
    }
    \\
    $\Theta$
    &&
    $\Hist{\Omega, \underline{I}}$
\end{tabular}
\end{center}
Indeed, both spaces coincide with their own spaces of extended input histories, and it is easy to check that the reverse inclusion holds.
To make it easier to spot the inclusion, no colour-coding is used for input histories in these diagrams: instead, the extended input histories from the space on the right have been highlighted in both spaces with a darker colour.
\begin{center}
\begin{tabular}{ccc}
    \raisebox{0cm}{
        \includegraphics[height=2cm]{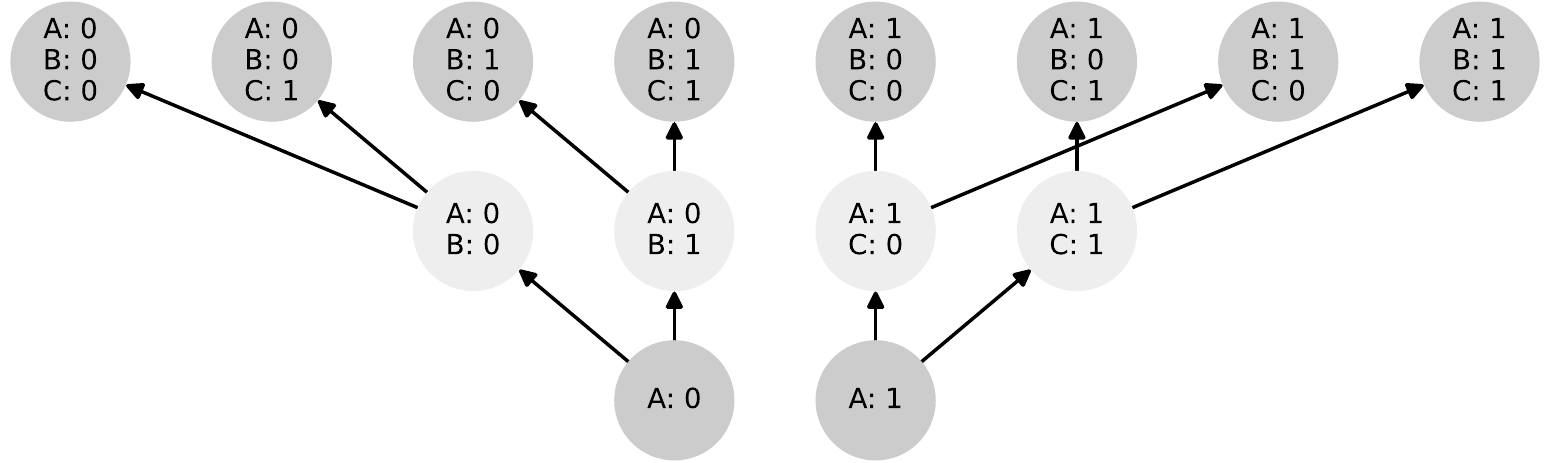}
    }
    &
    \raisebox{0.9cm}{$\supseteq$}
    &
    \raisebox{0.08cm}{
        \includegraphics[height=1.84cm]{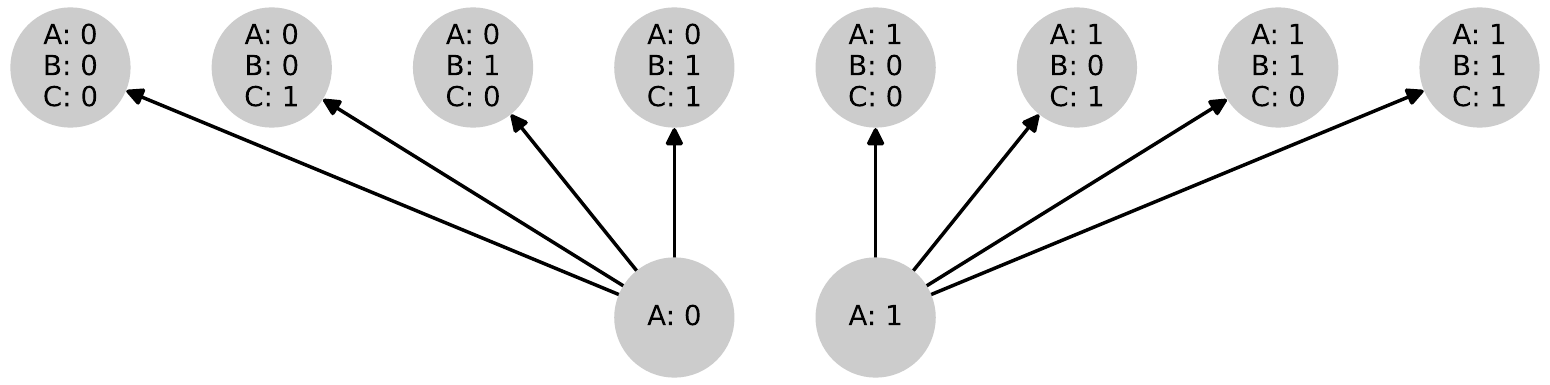}
    }
    \\
    $\Ext{\Theta}$
    &&
    $\ExtHist{\Omega, \underline{I}}$
\end{tabular}
\end{center}

Note that Definition \ref{definition:spaces-order} allows us to compare spaces of input histories with different underlying sets of events and inputs.
We refer to the partial order, or ``hierarchy'', of all spaces of input histories simply as $\AllSpaces$. Spaces $\Theta$ with $\underline{\Inputs{\Theta}} = \underline{I}$ for a specific choice $\underline{I} = (I_\omega)_{\omega \in E}$ form a sub-hierarchy $\Spaces{\underline{I}}$, and spaces satisfying the free-choice condition form a further sub-hierarchy $\SpacesFC{\underline{I}}$. All three hierarchies are lattices, sharing the same notion of join and meet.

\begin{proposition}
\label{proposition:spaces-input-histories-hierarchies}
All spaces of input histories---without restrictions on given on events or input values---form an infinite lattice $\AllSpaces$ under the partial order of Definition \ref{definition:spaces-order}, having the empty space $\emptyset$ as its maximum.
The join of two spaces is their \emph{closest common coarsening}, while their meet is their \emph{closest common refinement}.
Joins and meets take the following form:
\begin{equation}
\begin{array}{rcl}
    \Theta \vee \Theta' &=& \Prime{\Ext{\Theta} \cap \Ext{\Theta'}}\\
    \Theta \wedge \Theta' &=& \Prime{\Ext{\Theta} \cup \Ext{\Theta'}}
\end{array}
\end{equation}
When talking about the \underline{join} and \underline{meet} of spaces of input histories, we shall henceforth refer to the operations above.
Spaces $\Theta$ such that $\underline{\Inputs{\Theta}} = \underline{I}$ for some choice of $\underline{I}$ form a finite upperset of $\AllSpaces$:
\[
    \Spaces{\underline{I}}
    \hookrightarrow
    \AllSpaces
\]
The upperset $\Spaces{\underline{I}}$ has the empty space as its maximum and the \emph{discrete space} $\Hist{\discrete{E},\underline{I}}$ as its minimum, therefore it is a lattice.
The spaces $\Theta \in \Spaces{\underline{I}}$ which satisfy the free-choice condition form a lowerset $\SpacesFC{\underline{I}}$:
\[
    \SpacesFC{\underline{I}}
    \hookrightarrow\Spaces{\underline{I}}
    \hookrightarrow \AllSpaces
\]
The lowerset $\SpacesFC{\underline{I}}$ has the discrete space as its minimum and the \emph{indiscrete space} $\Hist{\indiscrete{E}, \underline{I}}$ as its maximum, therefore it is a lattice.
\end{proposition}
\begin{proof}
See \ref{proof:proposition:spaces-input-histories-hierarchies}
\end{proof}

We might now wonder how spaces induced by causal orders---those in the form $\Theta = \Hist{\Omega, \underline{I}}$---fit within these hierarchies.
We already know from Proposition \ref{proposition:order-induced-space-free-choice} that $\Hist{\Omega, \underline{I}} \in \SpacesFC{\underline{I}}$: we now complete the picture by showing that these spaces are closed under join, but not under meet.

\begin{proposition}
\label{proposition:order-induced-spaces-closed-under-join}
For any given $\underline{I} = (I_\omega)_{\omega \in E}$, the $\Omega \mapsto \Hist{\Omega, \underline{I}}$ function---sending the causal orders on $E$ to the associated spaces of input histories---commutes with the join operation:
\[
    \begin{array}{rcl}
    \ExtHist{\Omega,\underline{I}}
    \cap
    \ExtHist{\Omega',\underline{I}}
    &=&
    \ExtHist{\Omega\vee\Omega',\underline{I}}
    \\
    \Hist{\Omega,\underline{I}}
    \vee
    \Hist{\Omega',\underline{I}}
    &=&
    \Hist{\Omega\vee\Omega',\underline{I}}
    \end{array}
\]
\end{proposition}
\begin{proof}
See \ref{proof:proposition:order-induced-spaces-closed-under-join}
\end{proof}

\begin{proposition}
\label{proposition:order-induced-spaces-not-closed-under-meet}
Spaces of input histories induced by causal orders are not closed under meet.
\end{proposition}
\begin{proof}
See \ref{proof:proposition:order-induced-spaces-not-closed-under-meet}
\end{proof}

Further to the join and meet, we can also define the sequential and parallel composition of spaces of input histories on disjoint event sets, generalising the definition for causal orders previously given in Section \ref{section:causal-orders}.

\begin{definition}
Let $\Theta, \Theta'$ be spaces of input histories with $\Events{\Theta} \cap \Events{\Theta'} = \emptyset$.
The \emph{parallel composition} of $\Theta$ and $\Theta'$ is defined to be their union as sets:
\begin{equation}
    \Theta \cup \Theta'
\end{equation}
The \emph{sequential composition} of $\Theta$ before $\Theta'$ is defined as follows:
\begin{equation}
    \Theta \seqcomposeSym \Theta'
    :=
    \Theta \cup \left(\max{\Ext{\Theta}}\allJoinsSym\Theta'\right)
\end{equation}
where we adopted the symbol $\allJoinsSym$ to indicate all possible compatible joins between two sets (or families) of partial functions:
\begin{equation}
    \max{\Ext{\Theta}}\allJoinsSym\Theta'
    :=
    \suchthat{k\vee h'}{k \in \max{\Ext{\Theta}}, h' \in \Theta'}
\end{equation}
Note: because $\Events{\Theta} \cap \Events{\Theta'} = \emptyset$, all joins above are necessarily compatible.
\end{definition}

\begin{proposition}
\label{proposition:parallel-sequential-composition-spaces}
Let $\Theta, \Theta'$ be spaces of input histories with $\Events{\Theta} \cap \Events{\Theta'} = \emptyset$.
The parallel composition $\Theta \cup \Theta'$ and sequential composition $\Theta \seqcomposeSym \Theta'$ are well-defined spaces of input histories.
\end{proposition}
\begin{proof}
See \ref{proof:proposition:parallel-sequential-composition-spaces}
\end{proof}

As simple examples of sequential and parallel composition of spaces, we consider the spaces $\Theta$ and $\Theta'$ induced by the discrete order $\discrete{A,B}$ and total order $\total{C,D}$ respectively.
The two spaces are depicted below.
\begin{center}
    \begin{tabular}{cc}
    \includegraphics[height=2.5cm]{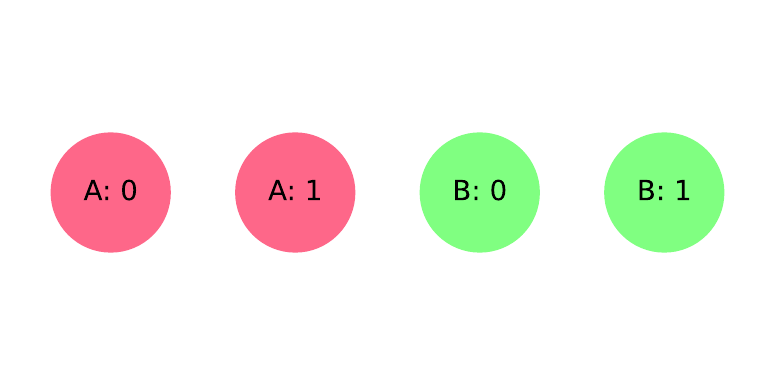}
    &
    \includegraphics[height=2.5cm]{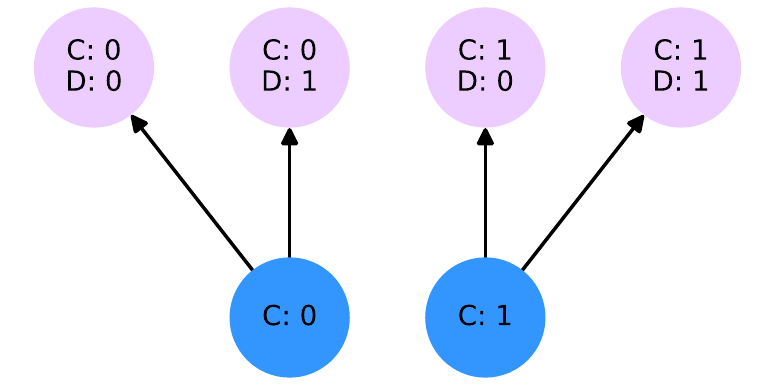}
    \\
    $\Theta$
    &
    $\Theta'$
    \end{tabular}
\end{center}
The parallel composition $\Theta \cup \Theta'$ of the two spaces is simply the disjoint union of their histories, with no additional causal relationship between them. 
\begin{center}
    \begin{tabular}{c}
    \includegraphics[height=2.5cm]{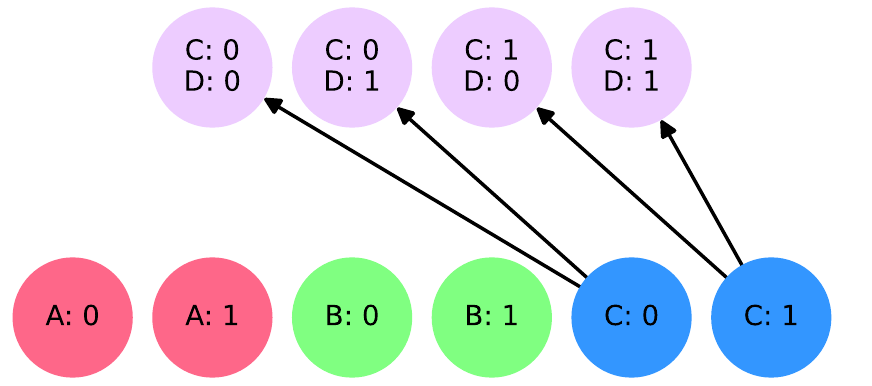}
    \\
    $\Theta \cup \Theta'$
    \end{tabular}
\end{center}
The sequential composition $\Theta \seqcomposeSym \Theta'$ of the two spaces consists of a copy of $\Theta'$ appearing after each maximal extended input history of $\Ext{\Theta}$, for a total of four copies.
\begin{center}
    \begin{tabular}{c}
    \includegraphics[height=5cm]{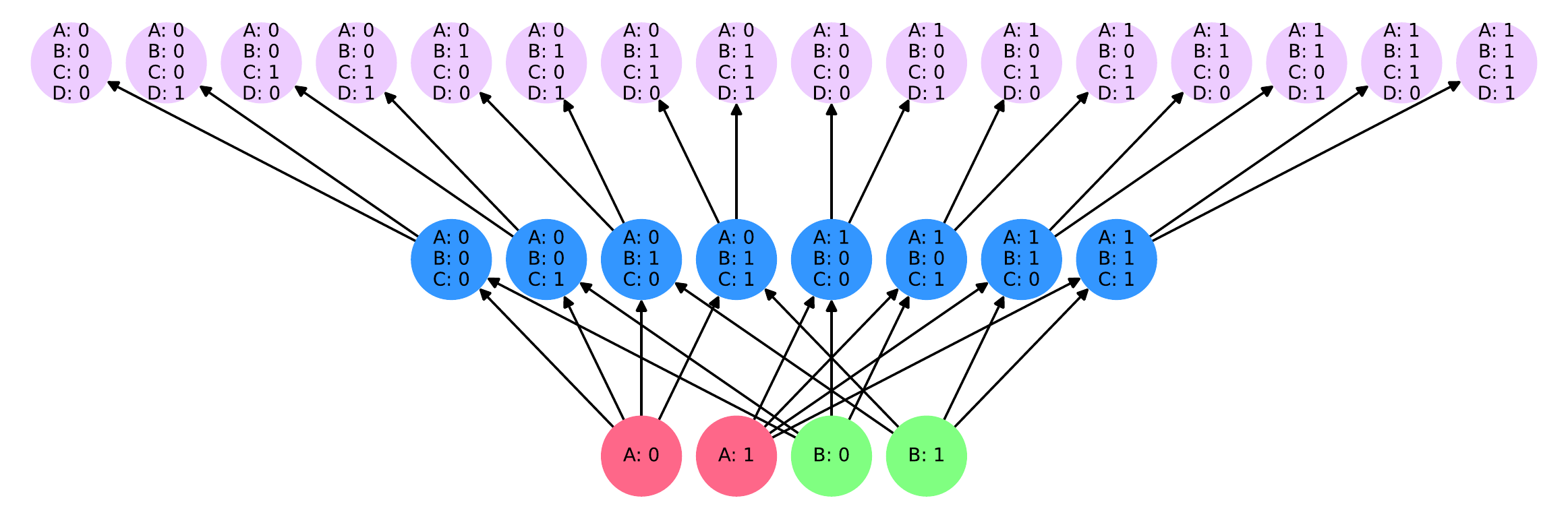}
    \\
    $\Theta \seqcomposeSym \Theta'$
    \end{tabular}
\end{center}

\begin{observation}
Let $\Theta, \Theta'$ be spaces of input histories with $\Events{\Theta} \cap \Events{\Theta'} = \emptyset$.
The parallel composition $\Theta \cup \Theta'$ and sequential composition $\Theta \seqcomposeSym \Theta'$ have the same input sets:
\[
\underline{\Inputs{\Theta\cup\Theta'}}
= 
\underline{\Inputs{\Theta\seqcomposeSym\Theta'}}
=
\underline{\Inputs{\Theta}}\allJoinsSym\underline{\Inputs{\Theta'}}
:=
\left(
\prod_{\omega \in \Events{\Theta}} \Inputs{\Theta}_\omega
\right)
\times
\left(
\prod_{\omega \in \Events{\Theta'}} \Inputs{\Theta'}_\omega
\right)
\]
\end{observation}

\begin{proposition}
\label{proposition:parallel-sequential-composition-free-choice}
Let $\Theta, \Theta'$ be spaces of input histories with $\Events{\Theta} \cap \Events{\Theta'} = \emptyset$.
If $\Theta$ and $\Theta'$ both satisfy the free-choice condition, then so do their parallel and sequential compositions.
\end{proposition}
\begin{proof}
See \ref{proof:proposition:parallel-sequential-composition-free-choice}
\end{proof}

In the previous examples, the spaces $\Theta$ and $\Theta'$ were induced by the causal orders $\discrete{A,B}$ and $\total{C,D}$ respectively: 
\begin{center}
    \begin{tabular}{cc}
    \includegraphics[height=2.5cm]{svg-inkscape/space-AB-0-highlighted_svg-tex.pdf}
    &
    \hspace{15mm}
    \includegraphics[height=2.5cm]{svg-inkscape/space-CD-5-highlighted_svg-tex.pdf}
    \\
    $\Theta = \Hist{\discrete{A,B}, \{0,1\}}$
    &
    \hspace{15mm}
    $\Theta' = \Hist{\total{C,D}, \{0,1\}}$
    \end{tabular}
\end{center}
The parallel composition $\Theta \cup \Theta'$ and sequential composition $\Theta \seqcomposeSym \Theta'$ of the two spaces $\Theta$ and $\Theta'$ above are the spaces of input histories induced by the parallel composition $\discrete{A,B} \vee \total{C,D}$ and sequential composition $\discrete{A,B} \seqcomposeSym \total{C,D}$, respectively, of the underlying causal orders:
\begin{center}
    \begin{tabular}{cc}
    \includegraphics[height=2cm]{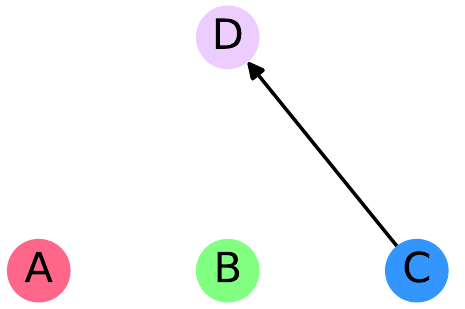}
    &
    \hspace{20mm}
    \includegraphics[height=2cm]{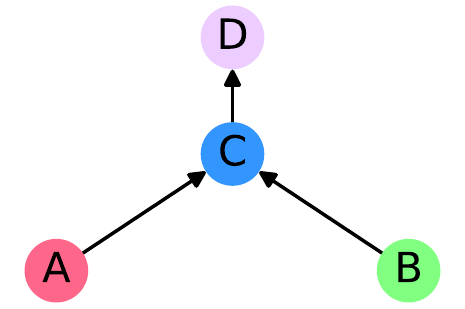}
    \\
    $\discrete{A,B} \vee \total{C,D}$
    &
    \hspace{20mm}
    $\discrete{A,B} \seqcomposeSym \total{C,D}$
    \end{tabular}
\end{center}
The following result proves that, for spaces of input histories induced by causal orders, the definition of parallel and sequential composition is compatible with that previously given for the underlying orders.

\begin{proposition}
\label{proposition:parallel-sequential-composition-spaces-and-orders}
Let $\Omega, \Omega'$ be disjoint causal orders.
Let $\Theta := \Hist{\Omega, \underline{I}}$ for some family $\underline{I}$ of non-empty input sets.
Let $\Theta' := \Hist{\Omega', \underline{I}'}$ for some family $\underline{I}'$ of non-empty input sets.
The sequential and parallel composition of spaces of input histories mirror the sequential and parallel composition of the orders that induced them:
\[
\begin{array}{rcl}
\Theta \cup \Theta'
&=&
\Hist{\Omega\vee\Omega', \underline{I}\vee\underline{I}'}
\\
\Theta \seqcomposeSym \Theta'
&=&
\Hist{\Omega\seqcomposeSym\Omega', \underline{I}\vee\underline{I}'}
\end{array}
\]
\end{proposition}
\begin{proof}
See \ref{proof:proposition:parallel-sequential-composition-spaces-and-orders}
\end{proof}

\begin{definition}
Let $\Theta$ be a space of input histories and let $\underline{\Theta'} := (\Theta'_k)_{k \in \max{\Ext{\Theta}}}$ be a family of spaces of input histories such that $\Events{\Theta} \cap \Events{\Theta'_k} = \emptyset$ for all $k \in \max{\Ext{\Theta}}$.
The \emph{conditional sequential composition} of $\Theta$ and $\underline{\Theta'}$ is defined as follows:
\begin{equation}
    \Theta \seqcomposeSym \underline{\Theta'}
    :=
    \Theta \cup \suchthat{k\vee h'}{k \in \max{\Ext{\Theta}}, h' \in \Theta'_k}
\end{equation}
Sequential composition $\Theta \seqcomposeSym \Theta'$ arises as the special case of conditional sequential composition where $\Theta'_k = \Theta'$ for all $k \in \max\Ext{\Theta}$.
\end{definition}

\begin{proposition}
\label{proposition:conditional-sequential-composition-spaces}
Let $\Theta$ be a space of input histories and let $\underline{\Theta'} := (\Theta'_k)_{k \in \max{\Ext{\Theta}}}$ be a family of spaces of input histories such that $\Events{\Theta} \cap \Events{\Theta'_k} = \emptyset$ for all $k \in \max{\Ext{\Theta}}$.
The conditional sequential composition $\Theta \seqcomposeSym \underline{\Theta'}$ is a well-defined space of input histories.
\end{proposition}
\begin{proof}
See \ref{proof:proposition:conditional-sequential-composition-spaces}
\end{proof}

As a simple example of conditional sequential composition, we compose the space induced by the discrete order on one event $\ev{A}$ (having \hist{A/0} and \hist{A/1} as its maximal extended input histories) with the spaces induced by the two total orders on two events $\ev{B}$ and $\ev{C}$:
\begin{center}
    \begin{tabular}{ccc}
    \includegraphics[height=2.5cm]{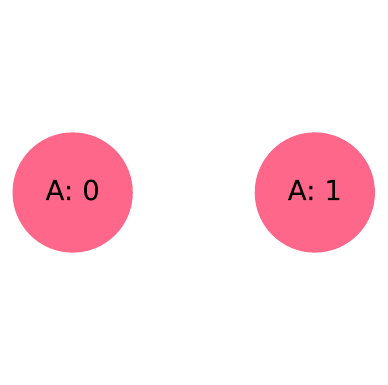}
    &
    \includegraphics[height=2.5cm]{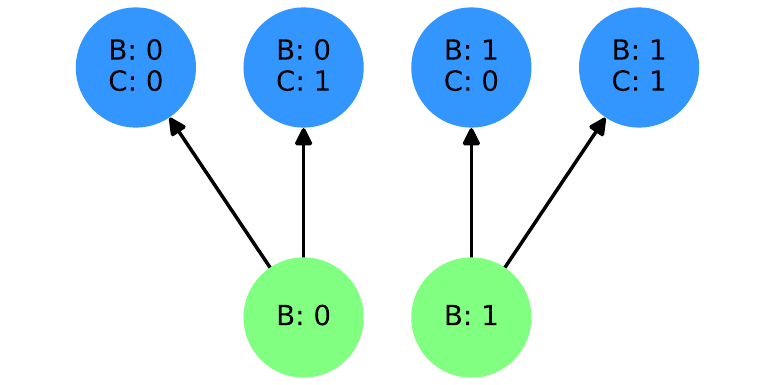}
    &
    \includegraphics[height=2.5cm]{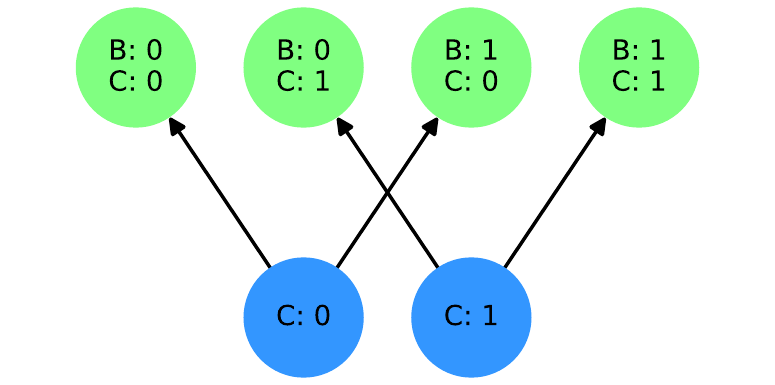}
    \\
    $\Theta$
    &
    $\Theta'_{\hist{A/0}}$
    &
    $\Theta'_{\hist{A/1}}$
    \end{tabular}
\end{center}
The result of this conditional sequential composition is the space of input histories for the 3-party causal switch which we discussed earlier on:
\begin{center}
    \begin{tabular}{c}
    \includegraphics[height=3.5cm]{svg-inkscape/space-ABC-unique-tight-101-highlighted_svg-tex.pdf}
    \\
    $\Theta \seqcomposeSym \underline{\Theta'}$
    \end{tabular}
\end{center}
In this case, $\Theta$ satisfies the free-choice condition, as do all the $\Theta'_k$, which in addition feature the exact same event set and input sets: as the following result shows, such circumstances are both necessary and sufficient for the conditional sequential composition to satisfy the free-choice condition.

\begin{proposition}
\label{proposition:conditional-sequential-composition-spaces-free-will}
Let $\Theta$ be a space of input histories and let $\underline{\Theta'} := (\Theta'_k)_{k \in \max{\Ext{\Theta}}}$ be a family of spaces of input histories such that $\Events{\Theta} \cap \Events{\Theta'_k} = \emptyset$ for all $k \in \max{\Ext{\Theta}}$.
The conditional sequential composition $\Theta \seqcomposeSym \underline{\Theta'}$ satisfies the free-choice condition if and only if the following conditions all hold:
\begin{itemize}
    \item the space $\Theta$ satisfies the free-choice condition;
    \item the space $\Theta'_k$ satisfies the free-choice condition for all $k \in \max\Ext{\Theta}$;
    \item the families of input sets $\underline{\Inputs{\Theta'_k}}$ are identical for all $k \in \max\Ext{\Theta}$.
\end{itemize}
\end{proposition}
\begin{proof}
See \ref{proof:proposition:conditional-sequential-composition-spaces-free-will}
\end{proof}

\subsection{Causally Complete Spaces}
\label{subsection:spaces-cc}

In our operational interpretation, input histories are the data upon which the output values at individual events are allowed to depend.
When the causal order is given, it is always clear which histories refer to which outputs: the output at event $\omega$ is determined by the input histories $h$ with domain $\dom{h} = \downset{\omega}$.
In the more general setting of spaces of input histories, however, a causal order might not be given: in the absence of a $\downset{\omega}$, how do we associate events to the input histories that determine their outputs?
To get ourselves started, we consider the example of the causal diamond $\Omega$.
\begin{center}
    \includegraphics[height=2.5cm]{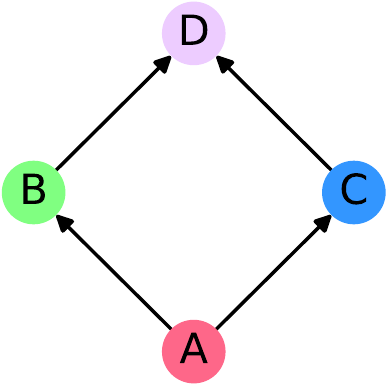}
\end{center}
Looking at the space of input histories $\Hist{\Omega,\{0,1\}}$---a shorthand by which we mean $\Hist{\Omega,(\{0,1\})_{\omega \in \Omega}}$---we observe that an association between input histories and events can be made from the order of histories alone.
Indeed, if $h$ is a history with $\dom{h} = \downset{\omega}$, then we can look at all input histories $k < h$ strictly below it and recover $\omega$ as the only event in $\dom{h}\backslash\bigcup_{k < h}\dom{k}$: this is the only event not covered by the domains of the histories strictly below $h$, which we will refer to as a ``tip event''.
In the Hasse diagram below, we have colour-coded input histories according to the tip event associated to them by this procedure.
\begin{center}
    \includegraphics[height=4cm]{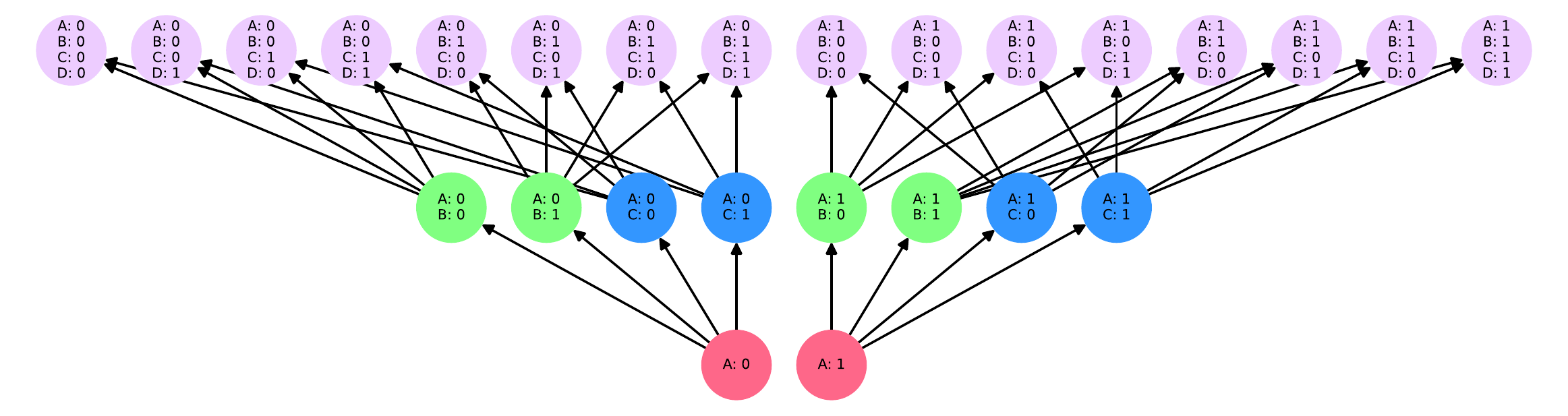}
\end{center}
The procedure works well for definite causal orders, but something goes wrong for indefinite ones: if two or more events are in indefinite causal order, they will appear together at the tip of histories.
Indeed, consider the following indefinite version of the diamond order above: the space $\total{\ev{A}, \evset{B,C}, \ev{D}}$, where the events \ev{B} and \ev{C} are in indefinite causal order rather than causally unrelated.
\begin{center}
    \includegraphics[height=3cm]{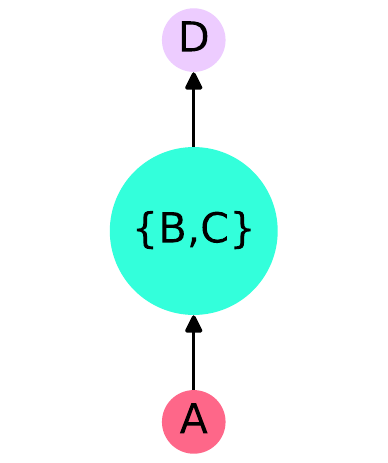}
\end{center}
Because \ev{B} and \ev{C} cannot be distinguished by input histories in the space, the histories in the middle layer now have two ``tip events'' instead of one.
\begin{center}
    \includegraphics[height=4cm]{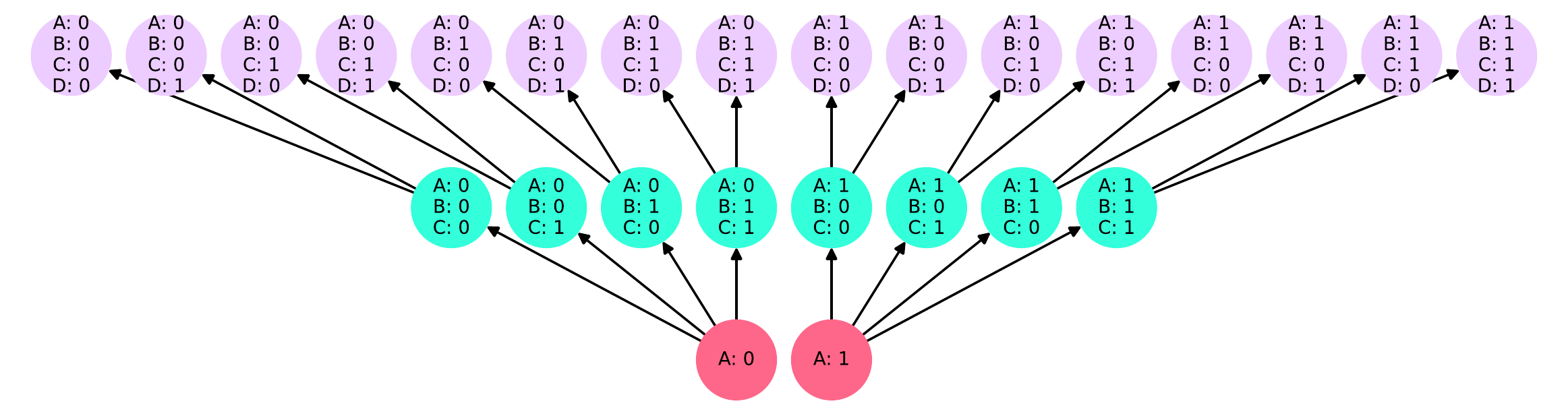}
\end{center}

The operational interpretation of multiple tip events is challenging: in a naive sense, it means that the output value at two events in indefinite causal order is to be produced ``simultaneously'', using the input values at both events.
This is problematic, because indefinite causal order should not trivialise to causal collapse: under our operational interpretation, distinct events should retain their independent local nature.
It should not, for example, be possible to perform the ``swap'' function $(b, c) \mapsto (c, b)$ on two events \ev{B} and \ev{C} in indefinite causal order: the devices would have to wait for both inputs to be given before producing their outputs, with the effect of delocalising the events.

However, there is an alternative way to look at the presence of multiple tip events, as a form of ``causal incompleteness''.
Rather than interpreting such spaces as allowing event delocalisation, we think of such spaces as not providing sufficient information for causal inference to be performed.
As such, we will focus our efforts on ``causally complete'' spaces, studying the incomplete spaces through the lens of their ``causal completions''.

\begin{definition}
Let $\Theta$ be a space of input histories.
Given an extended input history $h \in \Ext{\Theta}$, we define the \emph{tip events} of $h$ in $\Theta$ as the events which are in the domain of $h$ but not in the domain of any history strictly below it:
\begin{equation}
    \begin{array}{rcl}
    \tips{\Theta}{h}
    &:=&
    \dom{h}\backslash\bigcup_{k < h}\dom{k}
    \\
    &=&
    \suchthat{\omega \in \dom{h}}{\forall k < h.\, \omega \notin \dom{k}}
    \end{array}
\end{equation}
\end{definition}

\begin{observation}
The definition of tip events could have been equivalently formulated in terms of extended input histories, because the latter arise as joins of the former:
\[
    \begin{array}{rl}
     &\bigcup\suchthat{\dom{k}}{k \in \Ext{\Theta} \text{ s.t. } k < h}\\
    =&\bigcup\suchthat{\dom{k'}}{k \in \Ext{\Theta} \text{ s.t. } k < h, k'\in \Theta \text{ s.t. } k' \leq k}\\
    =&\bigcup\suchthat{\dom{k'}}{k' \in \Theta \text{ s.t. } k' < h}
    \end{array}
\]
\end{observation}

\begin{proposition}
\label{proposition:input-history-at-least-one-tip}
Every input history $h \in \Theta$ has at least one tip event. Every extended input history $h \in \Ext{\Theta}$ which is not an input history---i.e. one such that $h \not\in \Theta$---has no tip events.
\end{proposition}
\begin{proof}
See \ref{proof:proposition:input-history-at-least-one-tip}
\end{proof}

\begin{definition}
Let $\Theta$ be a space of input histories satisfying the free-choice condition.
We say that $\Theta$ is \emph{causally complete} if all input histories $h \in \Theta$ have exactly one tip event, and that it is \emph{causally incomplete} otherwise.
If $\Theta$ is causally complete and $h \in \Theta$, we define the \emph{tip event} of $h$ in $\Theta$ to be the unique event in $\tips{\Theta}{h}$:
\begin{equation}
    \Theta \text{ causally complete }
    \Leftrightarrow
    \forall h \in \Theta.\,
    \tips{\Theta}{h} = \{\tip{\Theta}{h}\}
\end{equation}
\end{definition}

\begin{proposition}
\label{proposition:order-induced-space-causal-completeness}
A space of input histories $\Theta = \Hist{\Omega,\underline{I}}$ induced by a causal order $\Omega$ is causally complete if and only if the causal order $\Omega$ is causally definite.
\end{proposition}
\begin{proof}
See \ref{proof:proposition:order-induced-space-causal-completeness}
\end{proof}

\begin{observation}
\label{observation:minimal-input-history-tips}
For a minimal input history $h \in \Theta$, we always have $\tips{\Theta}{h} = \dom{h}$.
If $\Theta$ is causally complete, this forces any minimal input history $h$ to have $|\dom{h}| = 1$.
\end{observation}

\begin{remark}
From a purely mathematical standpoint, the definition of causal completeness in terms of number of tip events does not involve the free-choice condition.
However, spaces not satisfying the free-choice condition are challenging in terms of operational interpretation, so in this work we made a choice to include the condition as part of the definition of causal completeness.
This allows us to talk simply about ``causally complete spaces'', avoiding the somewhat verbose alternative ``causally complete spaces satisfying the free-choice condition''.
\end{remark}

\begin{definition}
Let $\Theta$ be a space of input histories satisfying the free-choice condition.
The \emph{causal completions} of $\Theta$ are the closest refinements of $\Theta$ which are causally complete, i.e. the maxima of the set of causally complete spaces which are causal refinements of $\Theta$:
\begin{equation}
    \CausCompl{\Theta}
    :=
    \max
    \suchthat{\Theta' \leq \Theta}{\Theta' \text{ causally complete}}
\end{equation}
Since the discrete space $\Hist{\discrete{E^\Theta}, \underline{I}^\Theta}$ is always causally complete, the set of causal completions of $\Theta$ is never empty. If $\Theta$ is itself causally complete, then $\CausCompl{\Theta} = \{\Theta\}$.
\end{definition}

As an example of causal completion, we refer back to the indefinite causal order $\total{\ev{A}, \evset{B,C}}$.
The associated space of input histories is causally incomplete, because \ev{B} and \ev{C} always appear together as tip events (coloured aquamarine, at the top).
\begin{center}
    \includegraphics[height=2.5cm]{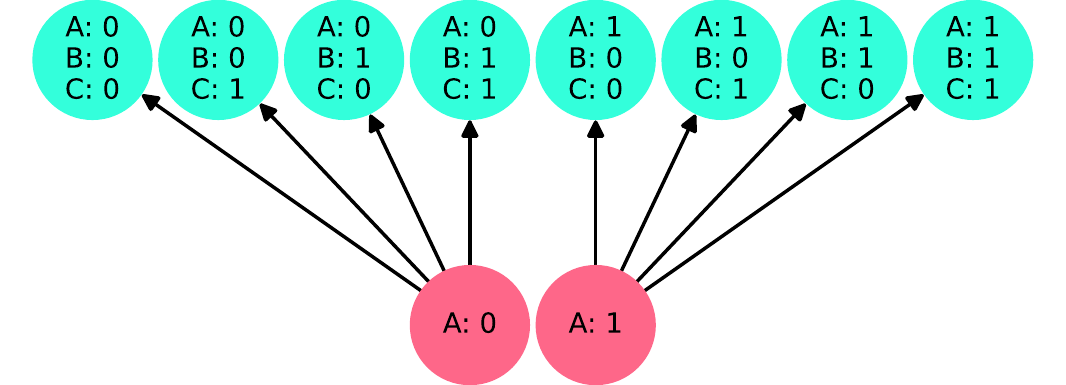}
\end{center}
There are four possible causal completions for this space.
Two of the causal completions are obtained by imposing a fixed order on events \ev{B} and \ev{C}: either \ev{B} causally precedes \ev{C} (left below) or \ev{B} causally succeeds \ev{C} (right below).
\begin{center}
\begin{tabular}{cc}
    \includegraphics[height=3cm]{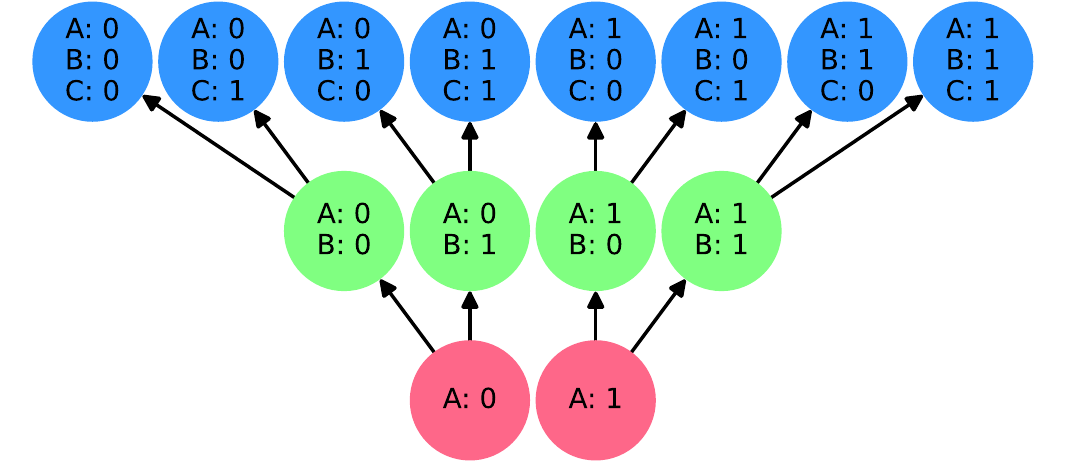}
    &
    \includegraphics[height=3cm]{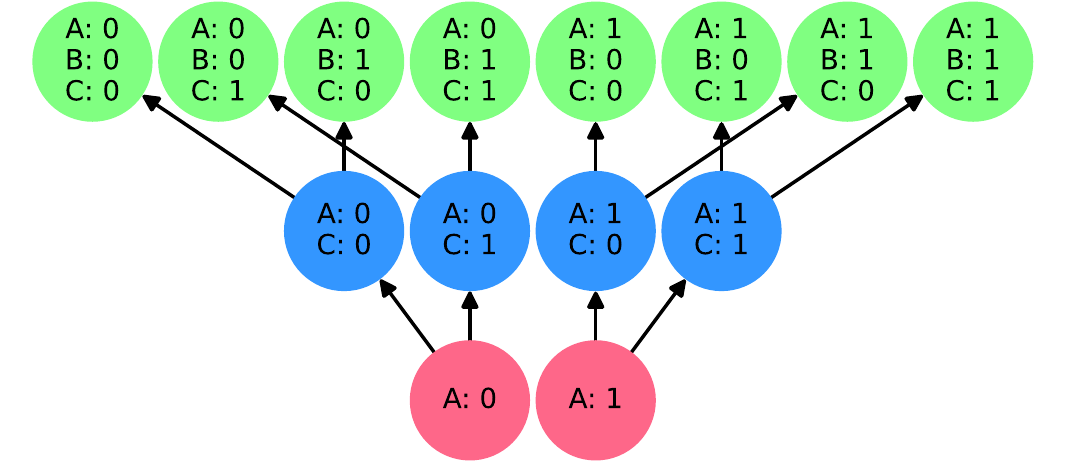}
\end{tabular}
\end{center}
The remaining two causal completions are obtained by imposing an order on events \ev{B} and \ev{C} that depends on the input at event \ev{A}: either \ev{B} causally precedes \ev{C} when the input at \ev{A} is 0 and causally succeeds \ev{C} when the input at \ev{A} is 1 (left below), or \ev{B} causally succeeds \ev{C} when the input at \ev{A} is 0 and causally precedes \ev{C} when the input at \ev{A} is 1 (right below). 
\begin{center}
\begin{tabular}{cc}
    \includegraphics[height=3cm]{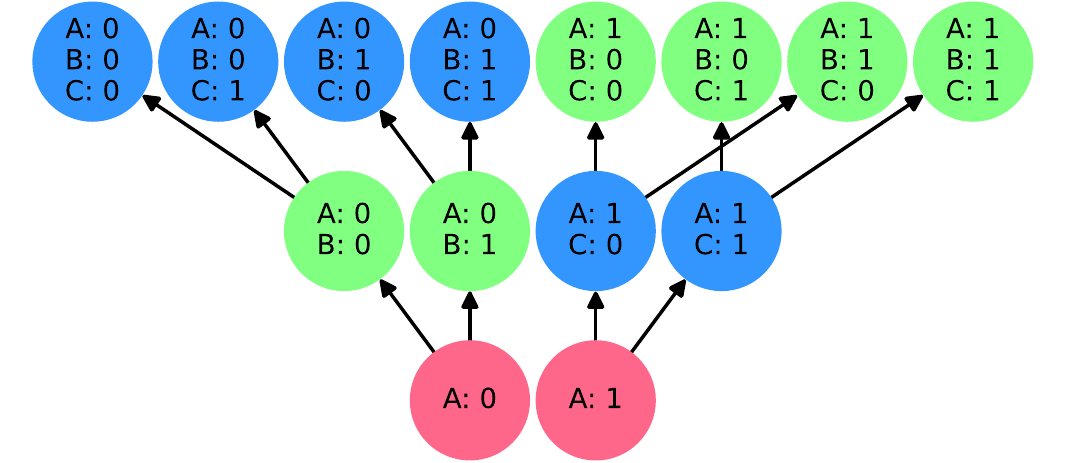}
    &
    \includegraphics[height=3cm]{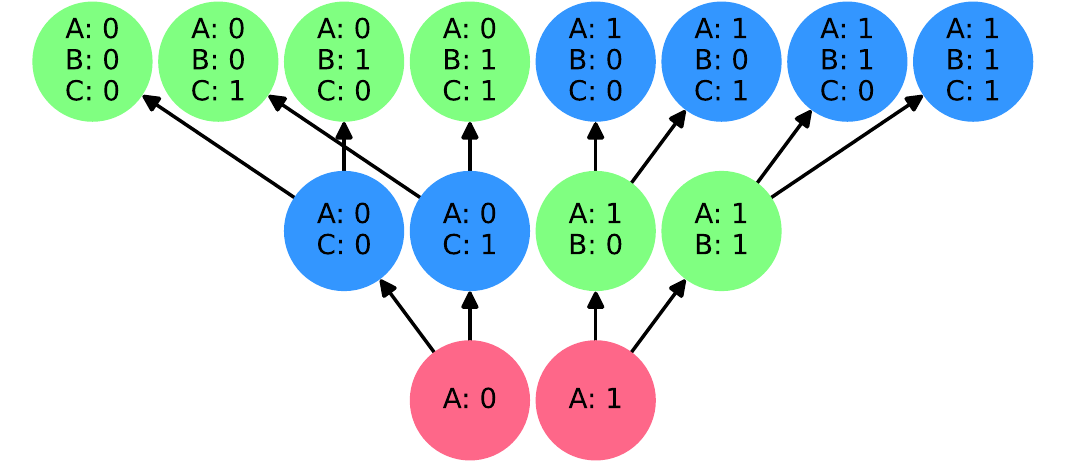}
\end{tabular}
\end{center}

To conclude this subsection, we show that both causal composition and sequential composition respect causal completeness, as does conditional sequential composition (in the circumstances under which it respects the free-choice condition).

\begin{proposition}
\label{proposition:parallel-sequential-composition-causally-complete}
Let $\Theta$ and $\Theta'$ be causally complete spaces of input histories such that $\Events{\Theta} \cap \Events{\Theta'} = \emptyset$.
The parallel composition $\Theta \cup \Theta'$ and sequential composition $\Theta \seqcomposeSym \Theta'$ are causally complete.
\end{proposition}
\begin{proof}
See \ref{proof:proposition:parallel-sequential-composition-causally-complete}
\end{proof}

\begin{proposition}
\label{proposition:conditional-sequential-composition-causally-complete}
Let $\Theta$ be a causally complete space of input histories.
Let $(\Theta'_k)_{k \in \max\Ext{\Theta}}$ be a family of causally complete spaces of input histories, with $\Events{\Theta} \cap \Events{\Theta'_k} = \emptyset$ for all $k \in \max\Ext{\Theta}$.
Assume that the families of input sets $\underline{\Inputs{\Theta'_k}}$ are identical for all $k \in \max\Ext{\Theta}$.
Then the conditional sequential composition $\Theta \seqcomposeSym \underline{\Theta'}$ is causally complete.
\end{proposition}
\begin{proof}
See \ref{proof:proposition:conditional-sequential-composition-causally-complete}
\end{proof}

\begin{proposition}
\label{proposition:causal-completion-events-inputs}
Let $\Theta$ be a space of input histories satisfying the free-choice condition and let $\hat{\Theta} \in \CausCompl{\Theta}$ be a causal completion of $\Theta$.
Then $\Events{\hat{\Theta}} = \Events{\Theta}$ and $\underline{\Inputs{\hat{\Theta}}} = \underline{\Inputs{\Theta}}$.
\end{proposition}
\begin{proof}
See \ref{proof:proposition:causal-completion-events-inputs}
\end{proof}

\begin{theorem}
\label{theorem:parallel-composition-causal-completions}
Let $\Theta$ and $\Theta'$ be spaces of input histories satisfying the free-choice condition, such that $\Events{\Theta} \cap \Events{\Theta'} = \emptyset$.
The parallel composition $\Theta \cup \Theta'$ has the following causal completions:
\begin{equation}
    \CausCompl{\Theta \cup \Theta'}
    =
    \suchthat{\hat{\Theta} \cup \hat{\Theta}'\;\;}{
        \begin{array}{l}
        \hat{\Theta} \in \CausCompl{\Theta},\\
        \hat{\Theta}' \in \CausCompl{\Theta'}
        \end{array}
    }
\end{equation}
\end{theorem}
\begin{proof}
See \ref{proof:theorem:parallel-composition-causal-completions}
\end{proof}

\begin{theorem}
\label{theorem:conditional-sequential-composition-causal-completions}
Let $\Theta$ be a space of input histories satisfying the free-choice condition.
Let $(\Theta'_k)_{k \in \max\Ext{\Theta}}$ be a family of spaces of input histories, with $\Events{\Theta} \cap \Events{\Theta'_k} = \emptyset$ for all $k \in \max\Ext{\Theta}$.
Assume that the families of input sets $\underline{\Inputs{\Theta'_k}}$ are identical for all $k \in \max\Ext{\Theta}$ and that they all satisfy the free-choice condition.
Then the conditional sequential composition $\Theta \seqcomposeSym \underline{\Theta'}$ has the following causal completions:
\begin{equation}
    \CausCompl{\Theta \seqcomposeSym \underline{\Theta'}}
    =
    \suchthat{\hat{\Theta} \seqcomposeSym \underline{\hat{\Theta}'}\;\;}{
        \begin{array}{l}
        \hat{\Theta} \in \CausCompl{\Theta},\\
        \forall k.\; \hat{\Theta}'_k \in \CausCompl{\Theta'_k}
        \end{array}
    }
\end{equation}
\end{theorem}
\begin{proof}
See \ref{proof:theorem:conditional-sequential-composition-causal-completions}
\end{proof}

\begin{corollary}
\label{corollary:sequential-composition-causal-completions}
Let $\Theta$ and $\Theta'$ be causally complete spaces of input histories such that $\Events{\Theta} \cap \Events{\Theta'} = \emptyset$ and satisfying the free-choice condition.
The sequential composition $\Theta \seqcomposeSym \Theta'$ has the following causal completions:
\begin{equation}
    \CausCompl{\Theta \seqcomposeSym \Theta'}
    =
    \suchthat{\hat{\Theta} \seqcomposeSym \hat{\Theta}'\;\;}{
        \begin{array}{l}
        \hat{\Theta} \in \CausCompl{\Theta},\\
        \hat{\Theta}' \in \CausCompl{\Theta'}
        \end{array}
    }
\end{equation}
\end{corollary}

\subsection{The Hierarchy of Causally Complete Spaces}
\label{subsection:spaces-hierarchy}

Causally complete spaces---and the associated causal polytopes---are the main focus of this work. The result below clarifies their standing within the hierarchy of spaces satisfying the free-choice condition.

\begin{proposition}
\label{proposition:hierarchy-causally-complete-spaces}
Causally complete spaces satisfying $\underline{\Inputs{\Theta}} = \underline{I}$ form a subset $\CCSpaces{\underline{I}} \subseteq \SpacesFC{\underline{I}}$ which is closed under meet but not under join (for two or more events). We refer to the $\wedge$-semilattice $\CCSpaces{\underline{I}}$ as the \emph{hierarchy of causally complete spaces} for $\underline{I}$.
\end{proposition}
\begin{proof}
See \ref{proof:proposition:hierarchy-causally-complete-spaces}
\end{proof}

As our simplest non-trivial example, we look at the hierarchy of causally complete spaces $\CCSpaces{\left(\{0,1\}\right)_{\omega \in \evset{A,B}}}$ on 2 events \ev{A} and \ev{B} with binary inputs $\{0,1\}$.
This hierarchy contains 7 causally complete spaces of input histories, ordered in 3 layers.
For additional ease of understanding, each space of input histories we examine is displayed together with the associated space of extended input histories: this way, it is easy to check whether a given space refines another.

At the bottom of the hierarchy $\CCSpaces{\left(\{0,1\}\right)_{\omega \in \evset{A,B}}}$ is the discrete space, induced by the discrete order on two events.
This space has 4 histories: because the two events are causally unrelated, the input histories \hist{A/0} and \hist{A/1} determine the output on event \ev{A}, while the input histories \hist{B/0} and \hist{B/1} determine the output on event \ev{B}.
\begin{center}
    \begin{tabular}{cc}
    \includegraphics[height=2.5cm]{svg-inkscape/space-AB-0-highlighted_svg-tex.pdf}
    &
    \includegraphics[height=2.5cm]{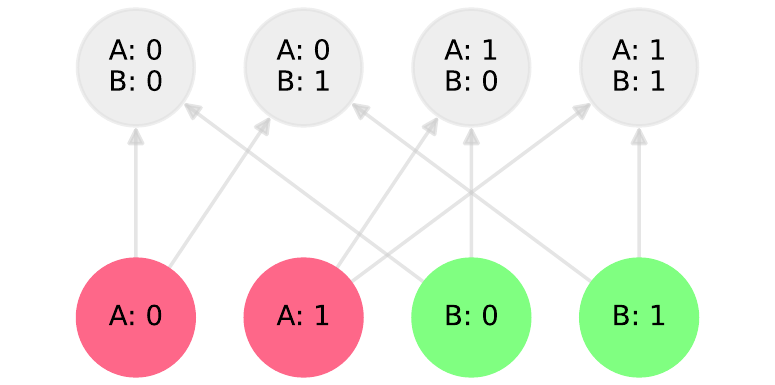}
    \\
    $\Theta$
    &
    $\Ext{\Theta}$
    \end{tabular}
\end{center}
At the top of the hierarchy are the 2 spaces, induced by the two possible total orders on two events.
Below is the space corresponding to the total order $\ev{A}\rightarrow \ev{B}$.
This space has 6 histories: the input histories \hist{A/0} and \hist{A/1} determine the output on event \ev{A}, while the remaining four histories are needed to determine the output on event \ev{B}, because the latter causally succeeds \ev{A}.
\begin{center}
    \begin{tabular}{cc}
    \includegraphics[height=2.5cm]{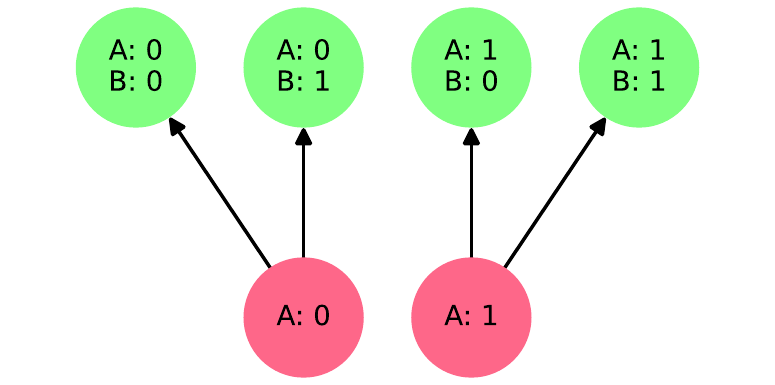}
    &
    \includegraphics[height=2.5cm]{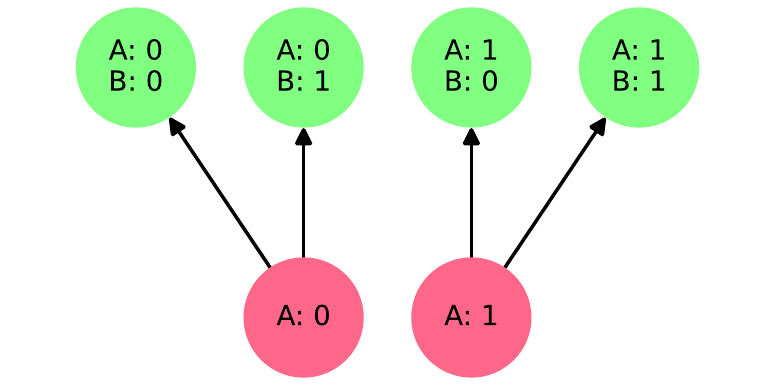}
    \\
    $\Theta$
    &
    $\Ext{\Theta}$
    \end{tabular}
\end{center}
The two spaces induced by total orders are related by event permutation symmetry $S(\evset{A,B})$.
\begin{center}
    \begin{tabular}{cc}
    \includegraphics[height=2.5cm]{svg-inkscape/space-AB-5-highlighted_svg-tex.pdf}
    &
    \includegraphics[height=2.5cm]{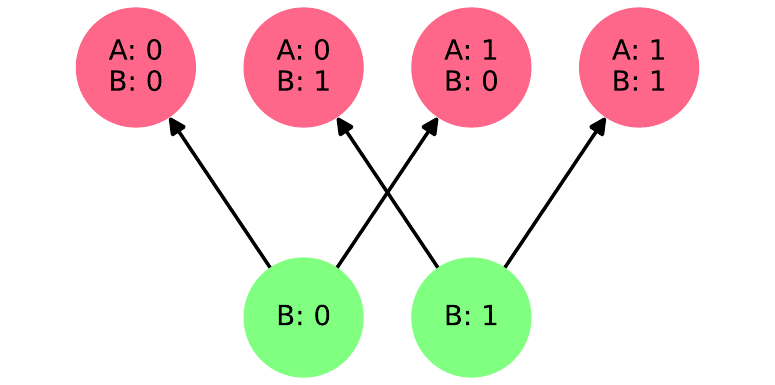}
    \end{tabular}
\end{center}
The middle layer of the hierarchy contains 4 spaces, each of them a coarsening of the discrete space and a refinement of one of the two total order spaces.
Below is one of the four spaces.
This space is a refinement of the space for the total order $\ev{A}\rightarrow \ev{B}$: by looking at the space of extended input histories, we note that the input history \hist{B/0} has been added, with tip event \ev{B}.
This means that the output at \ev{B} does not depend on the input at \ev{A} when the input at \ev{B} is 0: choosing 0 at \ev{B} causally disconnects \ev{B} from \ev{A}.
When the input at \ev{B} is 1, the output at \ev{B} can still depend on the input at \ev{A}, as demonstrated by the two input histories \hist{A/0,B/1} and \hist{A/1,B/1} with tip event \ev{B}.
\begin{center}
    \begin{tabular}{cc}
    \includegraphics[height=2.5cm]{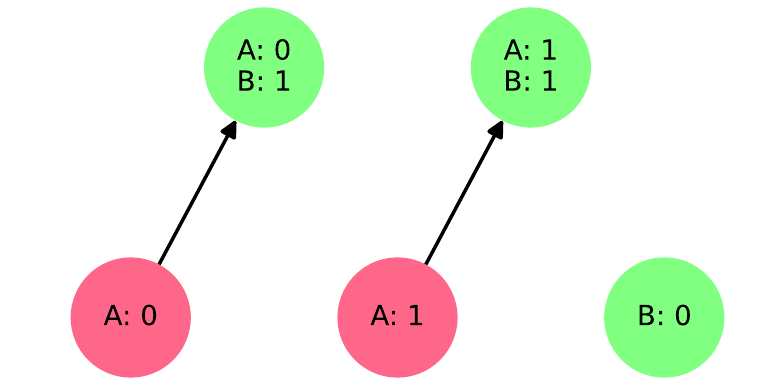}
    &
    \includegraphics[height=2.5cm]{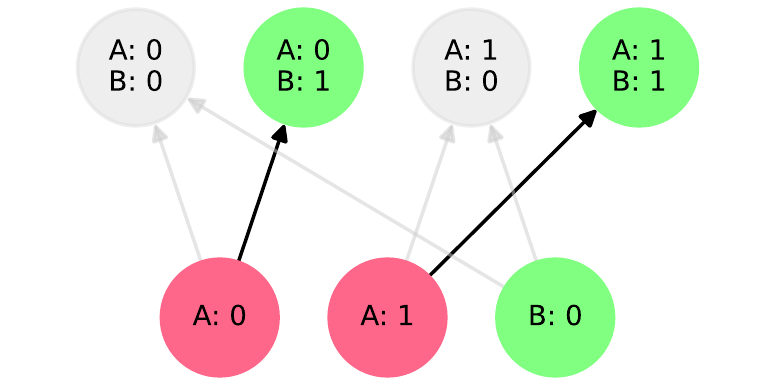}
    \\
    $\Theta$
    &
    $\Ext{\Theta}$
    \end{tabular}
\end{center}
The four spaces in the middle layer are related by event-input permutation symmetry $S(\evset{A,B})\times S(I_{\ev{A}})\times S(I_{\ev{B}})$: that is, by independently permuting the event set $\evset{A,B}$ and each of the input value sets $I_\omega$ (in fact, permuting one of the input sets is enough in this case).
\begin{center}
    \begin{tabular}{cc}
    \includegraphics[height=2.5cm]{svg-inkscape/space-AB-1-highlighted_svg-tex.pdf}
    &
    \includegraphics[height=2.5cm]{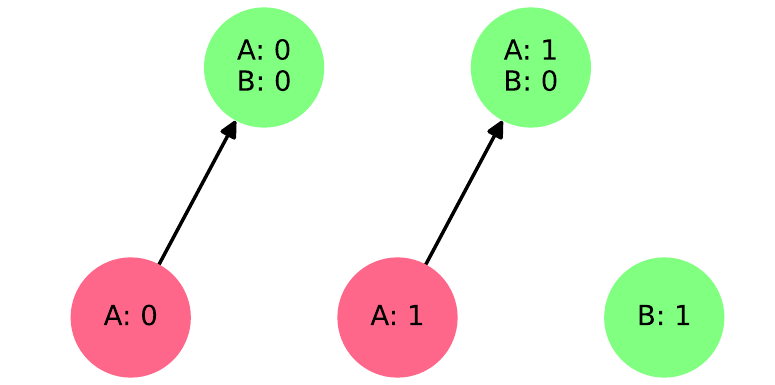}
    \\
    \includegraphics[height=2.5cm]{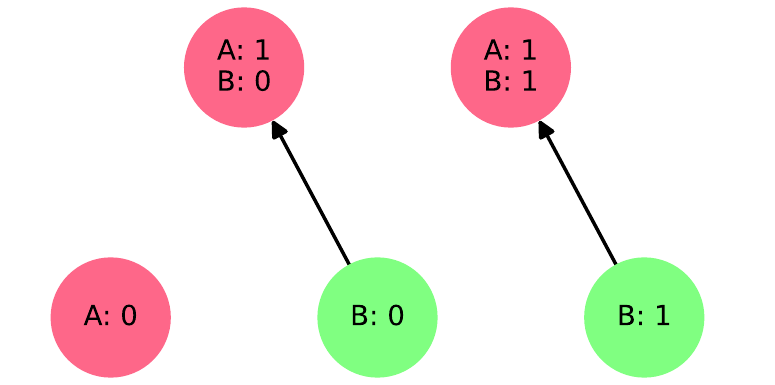}
    &
    \includegraphics[height=2.5cm]{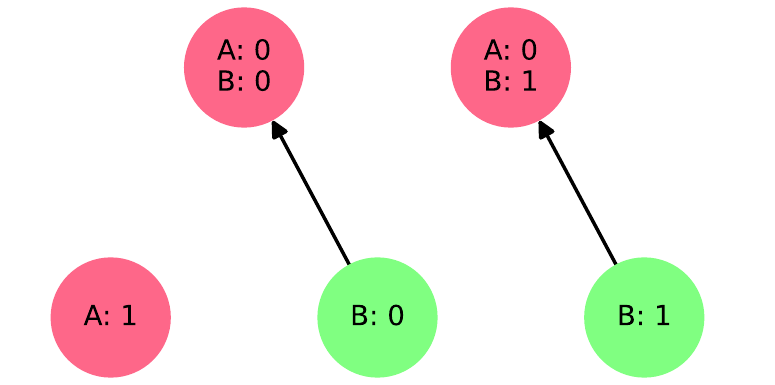}
    \end{tabular}
\end{center}

Event-input permutation symmetry is extremely helpful when classifying spaces: because the event and input labels are arbitrary, permutations don't contribute to our general understanding of causality.
For a general $\underline{I} = (I_\omega)_{\omega \in E}$, event-input permutation symmetry corresponds to the following group, where $S(X)$ is the group of permutations on a set $X$:
\[
S(E) \times \prod_{\omega \in E} S(I_\omega)
\]
Permutation symmetry is broken once an empirical model is specified, because conditional probability distributions are not, in general, invariant under its action.
In those cases where empirical models retain some symmetry, the latter can be used to reduce the computational burden for causal decomposition.
Figure \ref{fig:space-ABC-unique-tight-28-perm} (p.\pageref{fig:space-ABC-unique-tight-28-perm}) shows the action of permutation symmetry on a causally complete space on 3 events with binary inputs: the symmetry group does not act freely on this particular equivalence class (which features 24 spaces), but it does on other equivalence classes (27 such equivalence classes in total, e.g. equivalence class 30).

\begin{figure}[h]
    \centering
    \begin{tabular}{cccccc}
    \includegraphics[height=1.5cm]{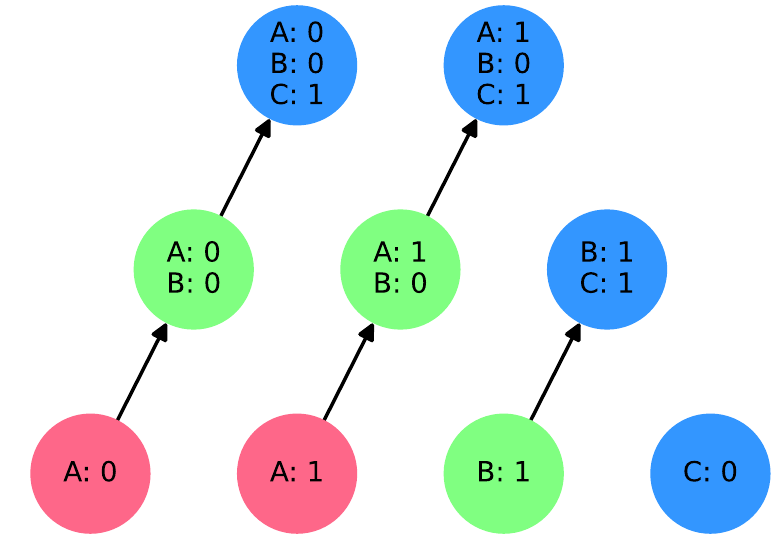}
    &
    \includegraphics[height=1.5cm]{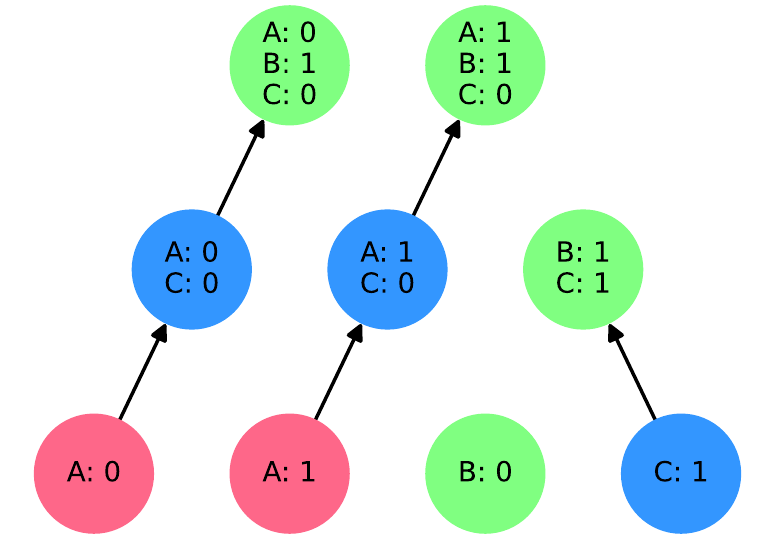}
    &
    \includegraphics[height=1.5cm]{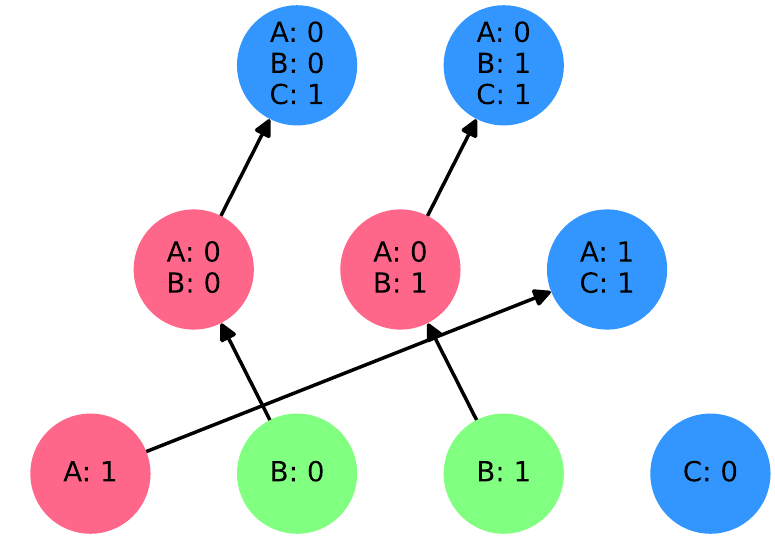}
    &
    \includegraphics[height=1.5cm]{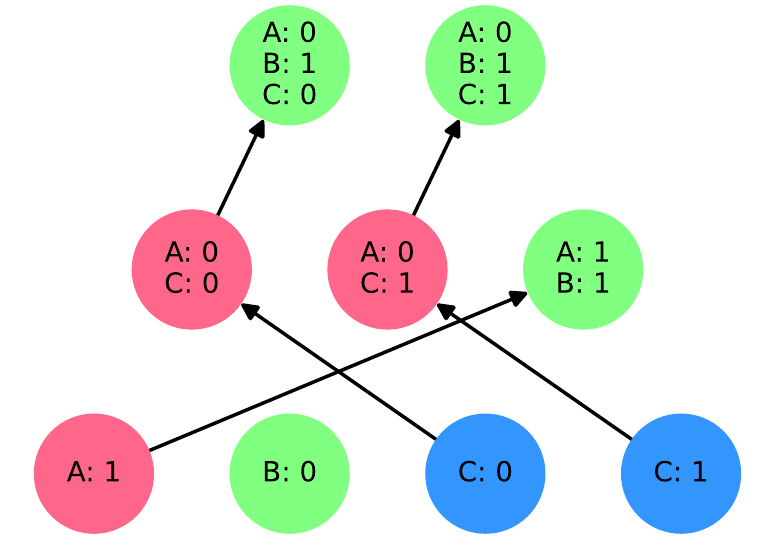}
    &
    \includegraphics[height=1.5cm]{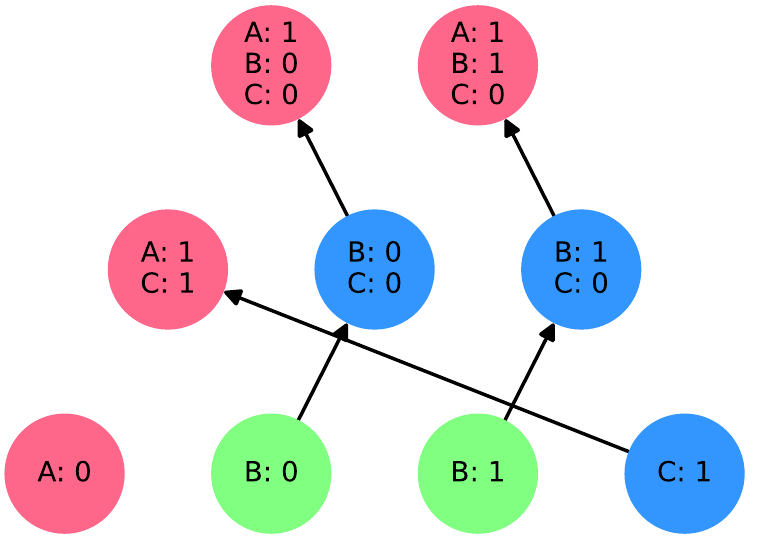}
    &
    \includegraphics[height=1.5cm]{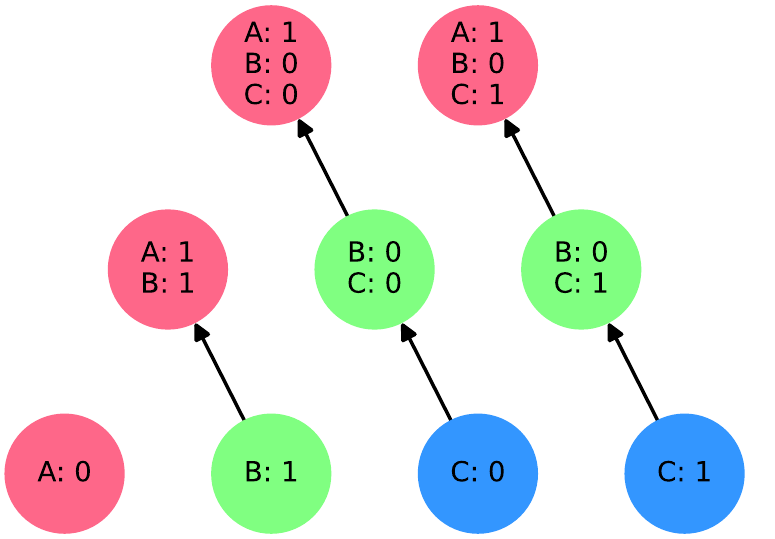}
    \\
    \includegraphics[height=1.5cm]{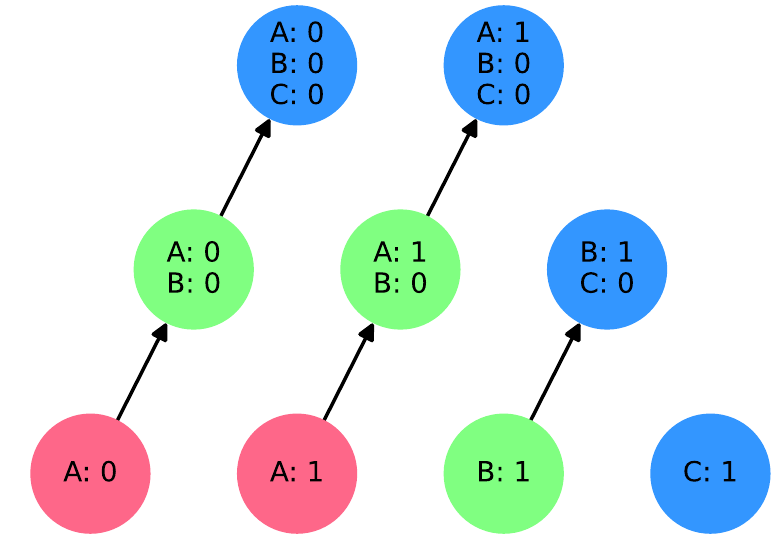}
    &
    \includegraphics[height=1.5cm]{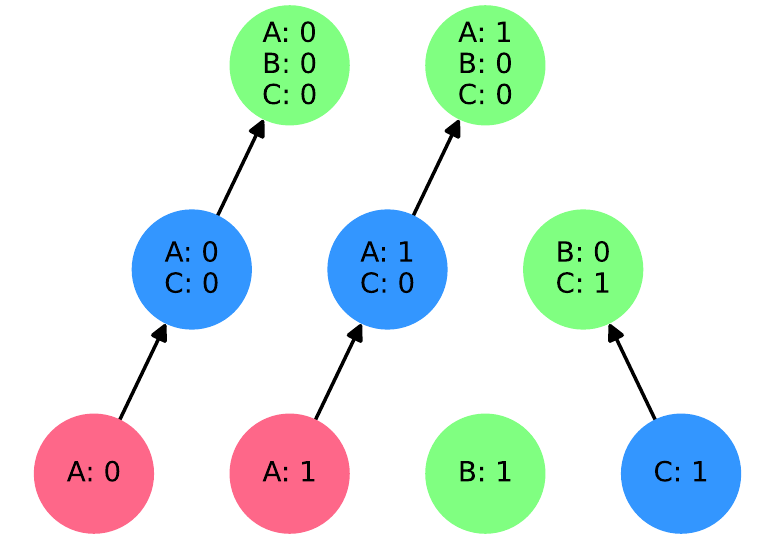}
    &
    \includegraphics[height=1.5cm]{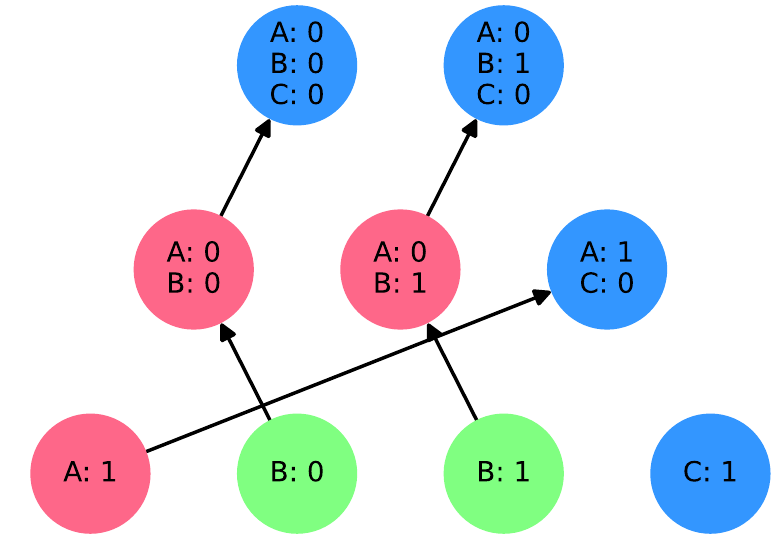}
    &
    \includegraphics[height=1.5cm]{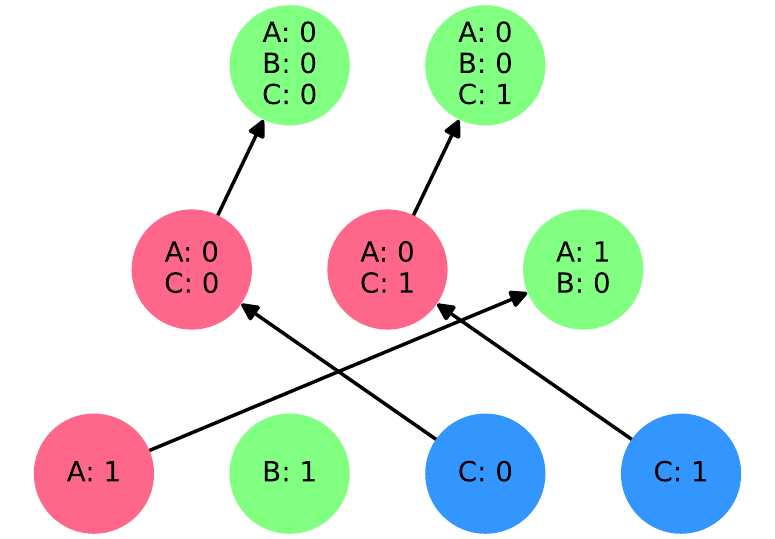}
    &
    \includegraphics[height=1.5cm]{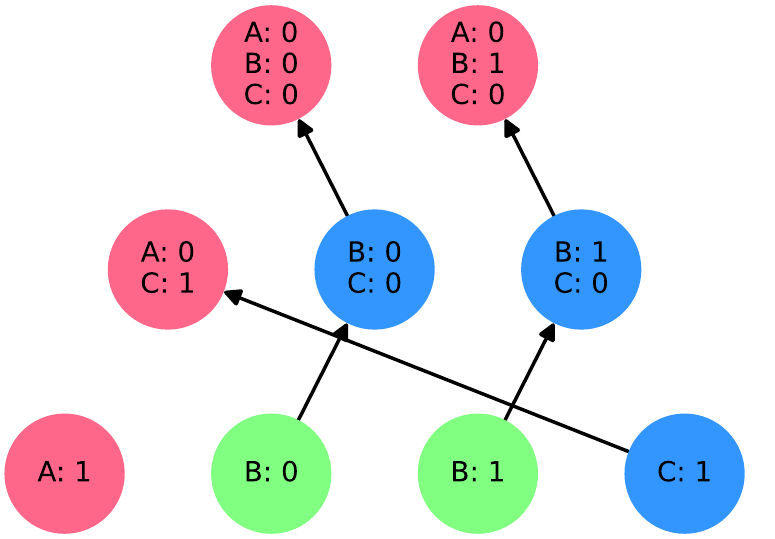}
    &
    \includegraphics[height=1.5cm]{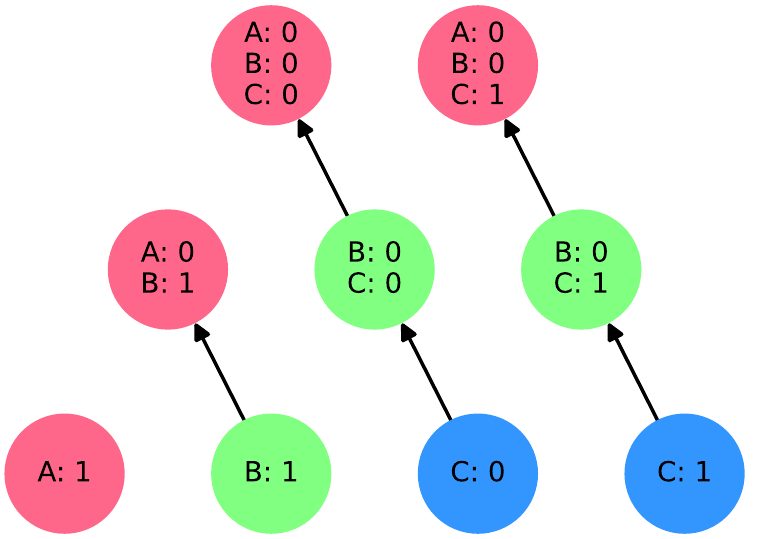}
    \\
    \includegraphics[height=1.5cm]{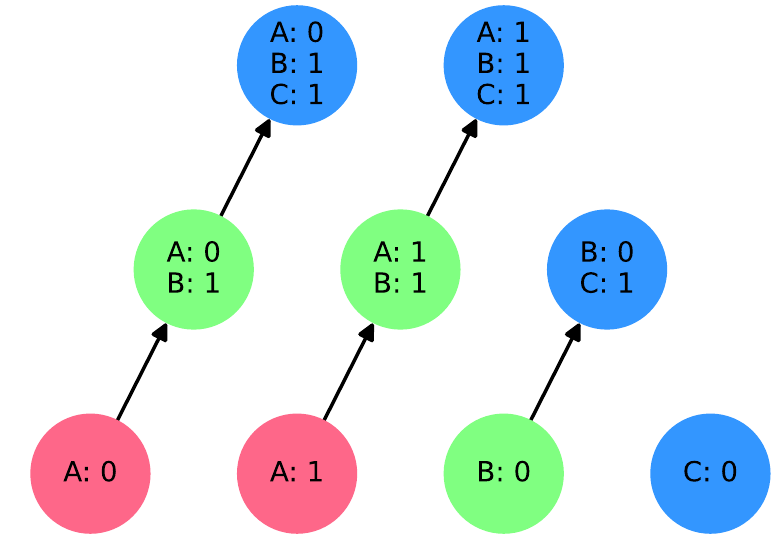}
    &
    \includegraphics[height=1.5cm]{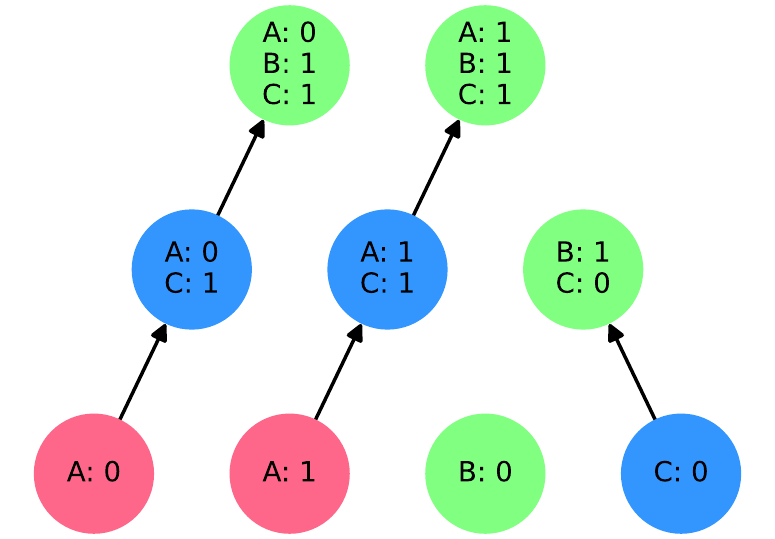}
    &
    \includegraphics[height=1.5cm]{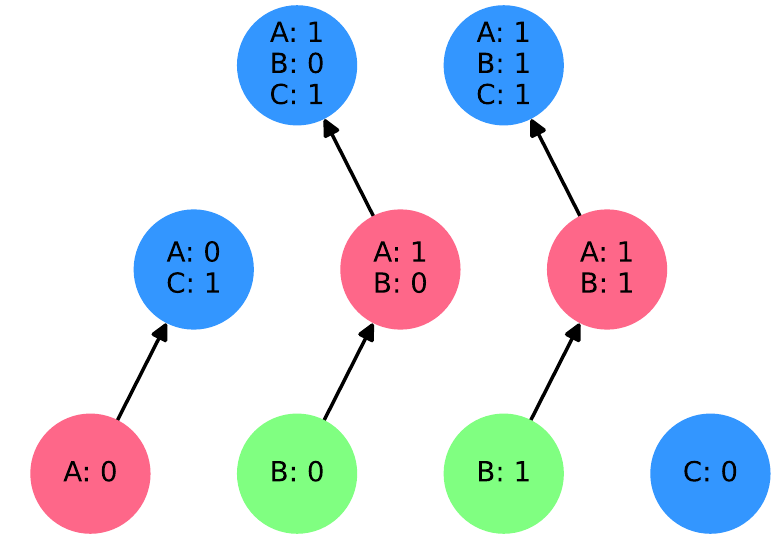}
    &
    \includegraphics[height=1.5cm]{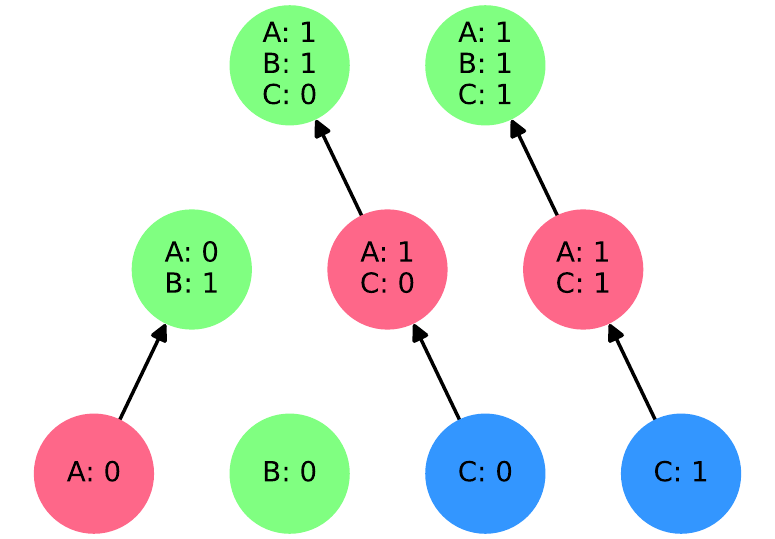}
    &
    \includegraphics[height=1.5cm]{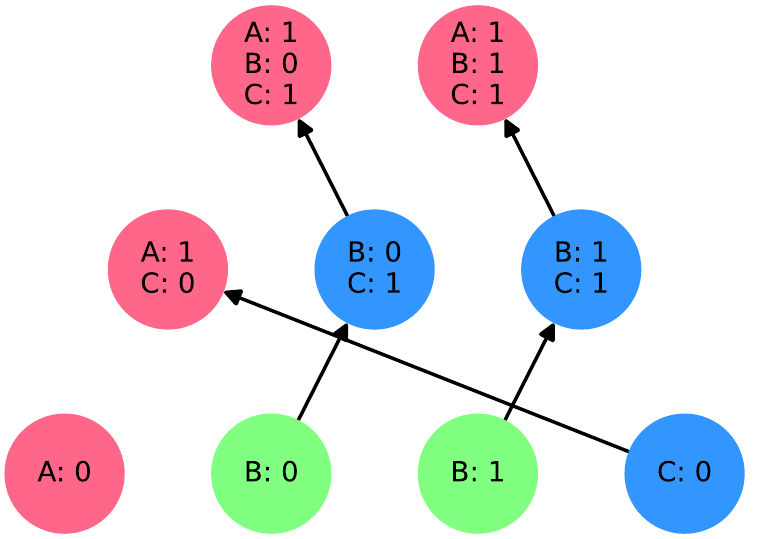}
    &
    \includegraphics[height=1.5cm]{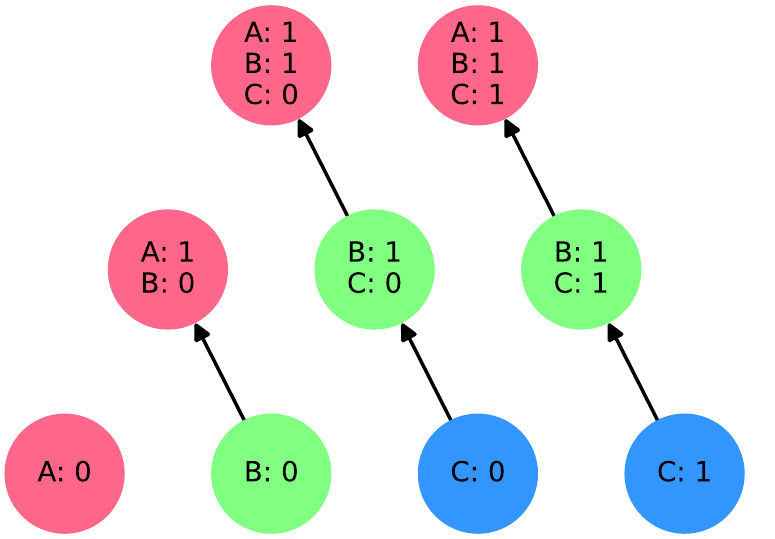}
    \\
    \includegraphics[height=1.5cm]{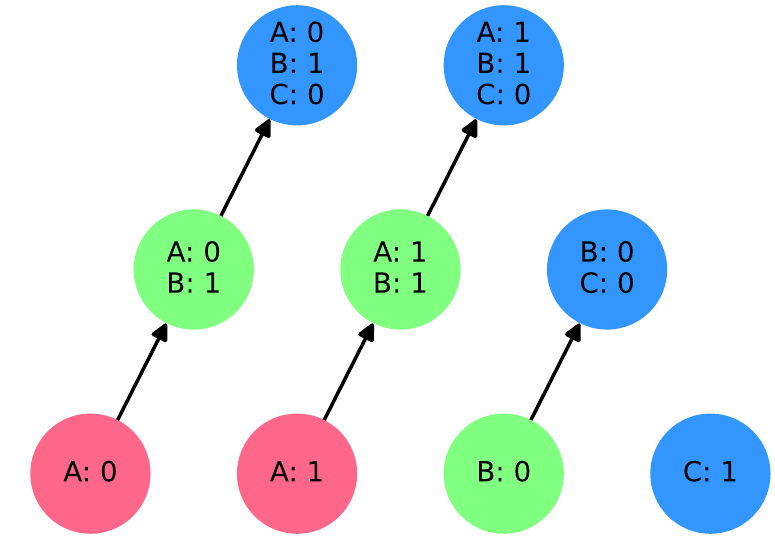}
    &
    \includegraphics[height=1.5cm]{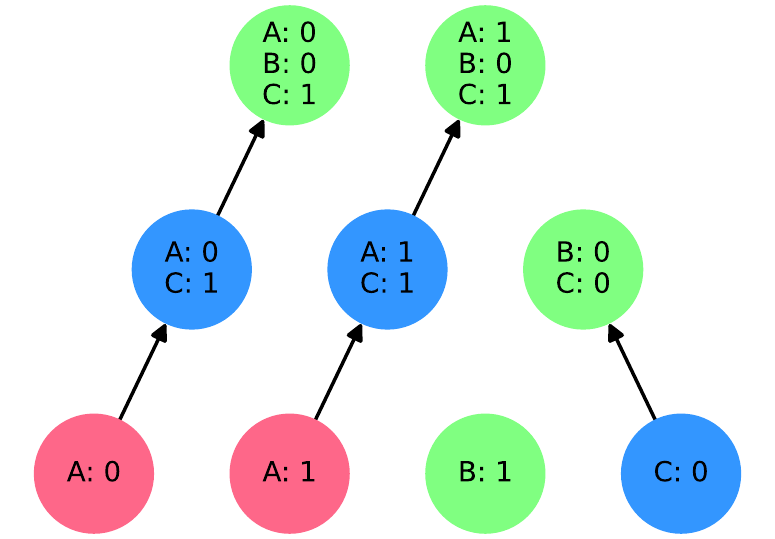}
    &
    \includegraphics[height=1.5cm]{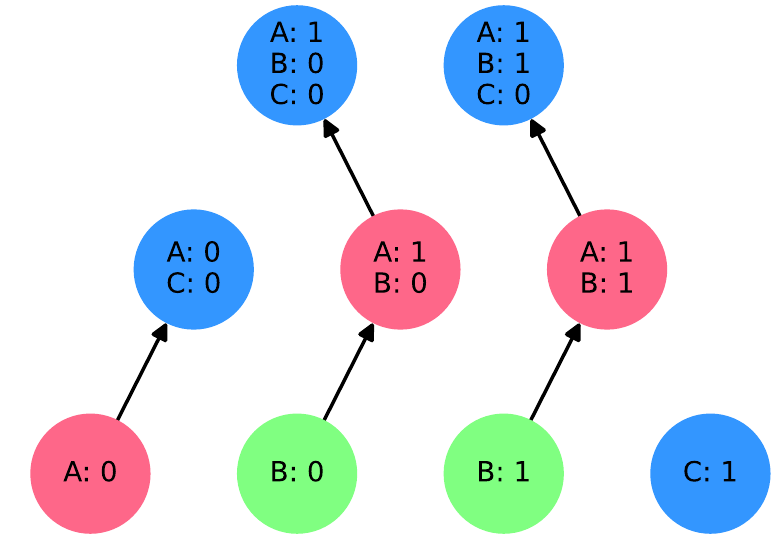}
    &
    \includegraphics[height=1.5cm]{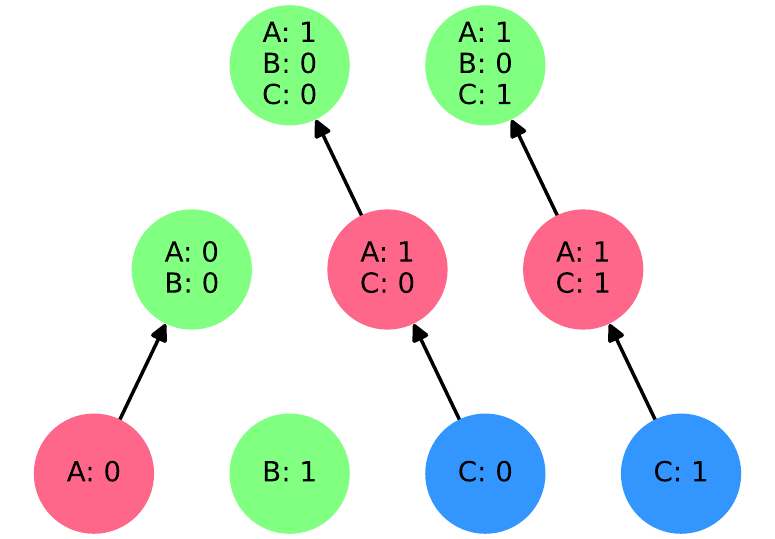}
    &
    \includegraphics[height=1.5cm]{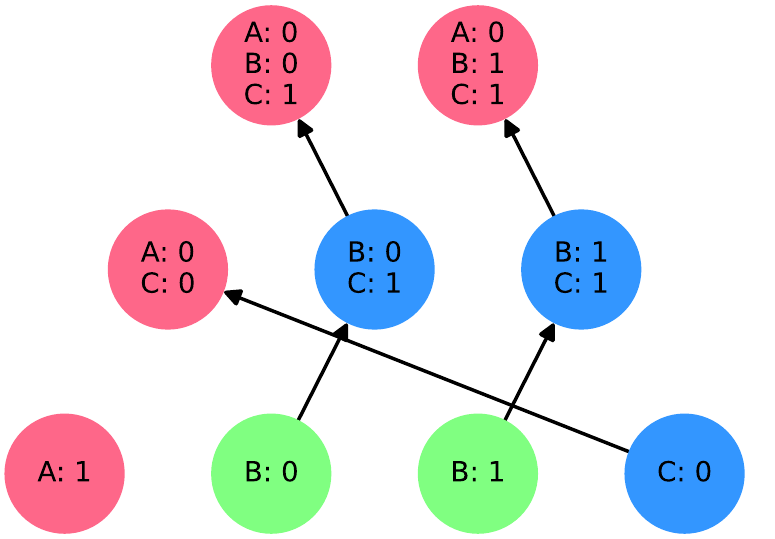}
    &
    \includegraphics[height=1.5cm]{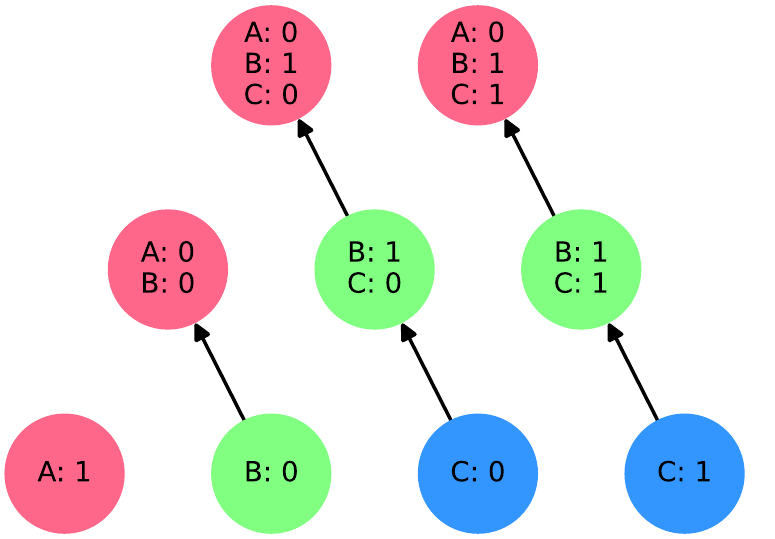}
    \end{tabular}
    \caption{
    All 24 permutations of a causally complete space on 3 events with binary inputs.
    Specifically, these are the contents of equivalence class 28 in the hierarchy, as depicted by Figure \ref{fig:hierarchy-spaces-ABC} (p.\pageref{fig:hierarchy-spaces-ABC}).
    Each row is a coset for the action of event permutation symmetry $S(\evset{A,B,C})$, which acts freely on this equivalence class.
    Each column is a coset for the action of input permutation symmetry $\prod_{\omega \in \evset{A,B,C}} S(I_\omega)$, which doesn't act freely on this equivalence class.
    }
\label{fig:space-ABC-unique-tight-28-perm}
\end{figure}

Having completed our exposition of the hierarchy of spaces on two events with binary inputs, we now move to the hierarchy $\CCSpaces{\left(\{0,1\}\right)_{\omega \in \evset{A,B,C}}}$ on three events.
This hierarchy has 2644 spaces, forming 102 equivalence classes under event-input permutation symmetry \cite{gogioso2022classification}.
While the full hierarchy is too complex to display, Figure \ref{fig:hierarchy-spaces-ABC} (p.\pageref{fig:hierarchy-spaces-ABC}) depicts the corresponding hierarchy of 102 equivalence classes: in this condensed graph, an edge $i \rightarrow j$ indicates that some space---and hence every space---in equivalence class $i$ is a closest refinement of some space of equivalence class $j$.
To get a reasonably orderly 3D view of the full hierarchy, one could imagine stacking all spaces in each equivalence class vertically: edges between spaces in equivalence classes $i$ and $j$ would line up, and their 2D vertical projections would form the edges seen in Figure \ref{fig:hierarchy-spaces-ABC}.

\begin{figure}[h]
    \centering
    \includegraphics[width=\textwidth]{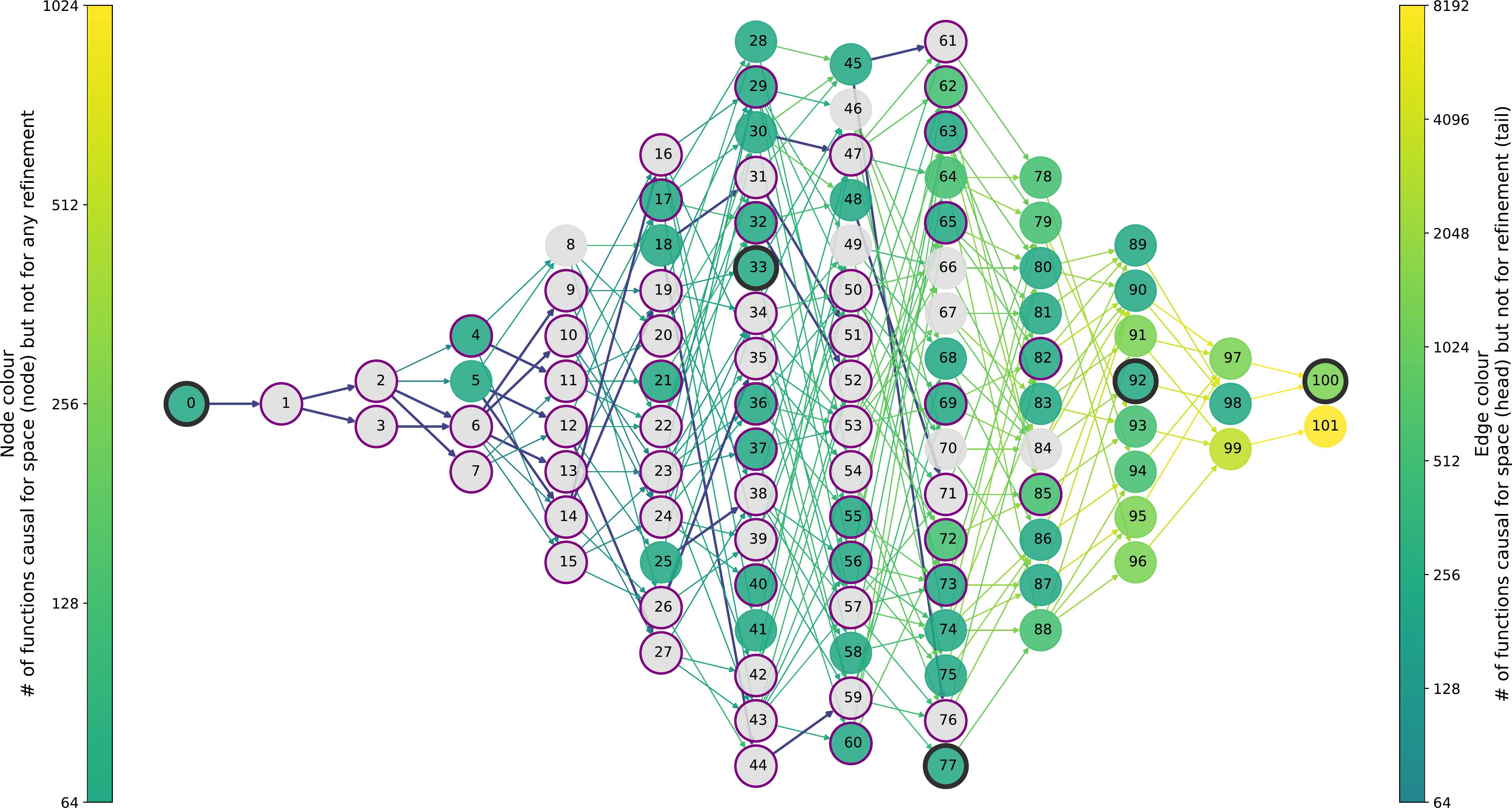}
    \caption{
    The hierarchy of causally complete spaces on 3 events with binary inputs, grouped into 102 equivalence classes under event-input permutation symmetry.
    An edge $i \rightarrow j$ indicates that some space---and hence every space---in equivalence class $i$ is a closest refinement for some space in equivalence class $j$.
    See \cite{gogioso2022classification} for description of all spaces in the hierarchy and \cite{gogioso2022topology} for the definition of causal functions.

    Node colour indicates the number of causal functions for a space which are not causal for any of its subspaces, while edge colour indicates the number of causal functions for the head space that are not causal for the tail space.
    Grey nodes (e.g. eq. class 1) indicate spaces where every causal function is also causal for some subspace, while thicker dark blue edges (e.g. edge $0 \rightarrow 1$) indicate that all causal functions for the head space are also causal a single tail space.

    Thin purple borders for nodes indicate eq. classes of non-tight spaces (e.g. eq. class 1).
    Thick black borders for nodes indicate the eq. classes of spaces induced by causal orders.
    }
\label{fig:hierarchy-spaces-ABC}
\end{figure}

At the bottom of the hierarchy we find the discrete space, induced by the discrete order $\discrete{\ev{A},\ev{B},\ev{C}}$, sitting alone in equivalence class 0.
This space has 6 histories---one for each event and input choice at that event---all unrelated: this the \emph{no-signalling space}, where the output at each event depends only on the input at that event.
The corresponding space of extended input histories contains all 26 binary-valued partial functions on the 3 events: histories supported by more than one event are not $\vee$-prime in this space.
\begin{center}
    \begin{tabular}{cc}
    \includegraphics[height=3.5cm]{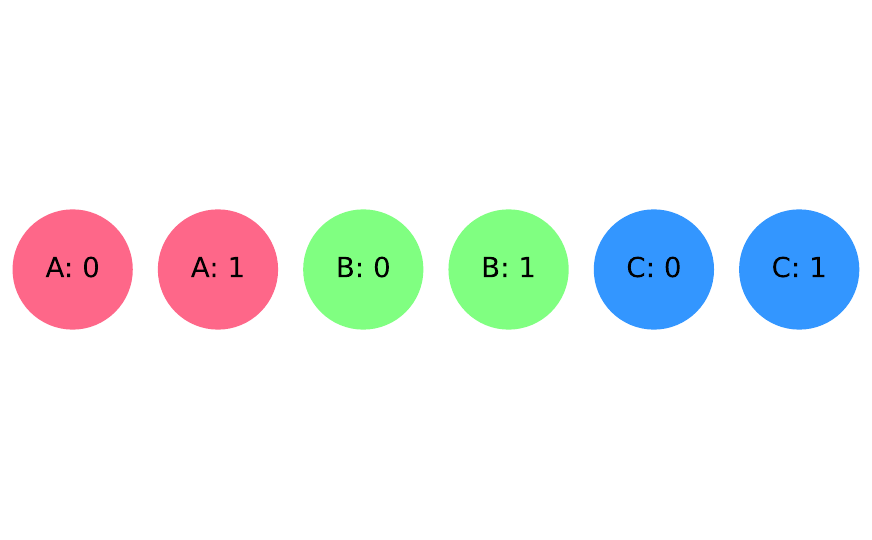}
    &
    \includegraphics[height=3.5cm]{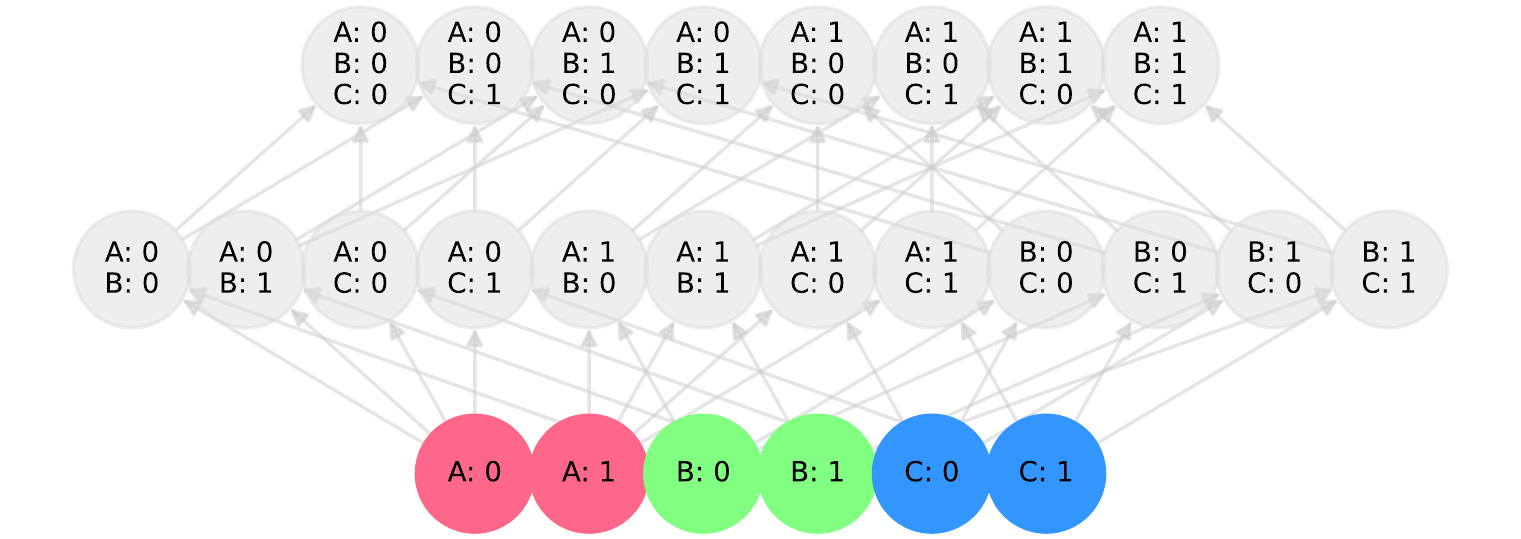}
    \\
    $\Theta_{0}$
    &
    $\Ext{\Theta_{0}}$
    \end{tabular}
\end{center}

At the top of the hierarchy we find two equivalence classes of spaces, labelled 100 and 101.
Equivalence class 100 contains the 6 spaces induced by total order: below is the space induced by $\total{\ev{A},\ev{B},\ev{C}}$.
This space has 14 histories, covering all possible combinations of inputs for event $\ev{A}$ (determining the output at event $\ev{A}$), for events $\evset{A,B}$ (determining the output at event $\ev{B}$) and for events $\evset{A,B,C}$ (determining the output at event $\ev{C}$).
This space coincides with its own space of extended input histories.
\begin{center}
    \begin{tabular}{cc}
    \includegraphics[height=3.5cm]{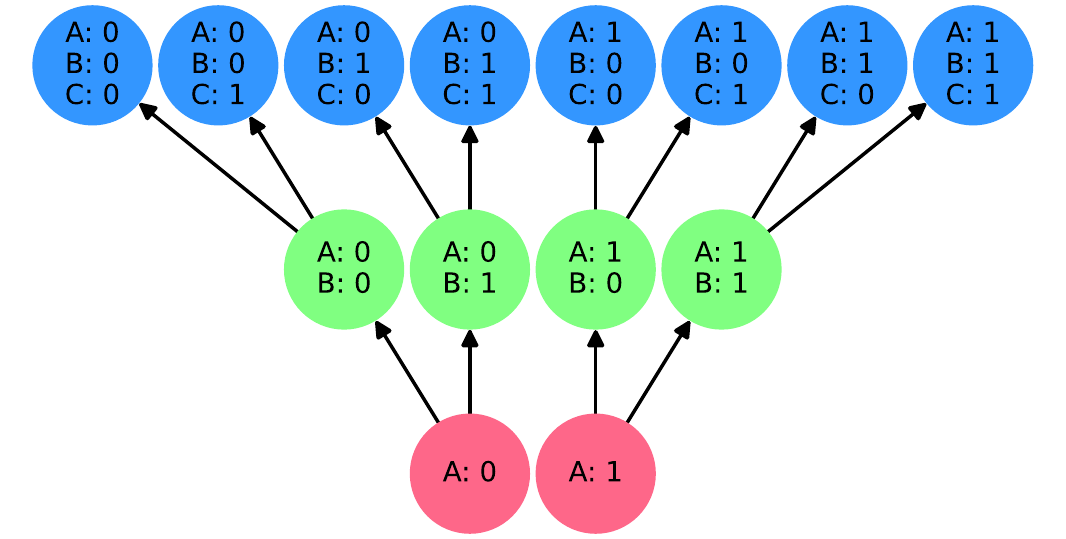}
    &
    \includegraphics[height=3.5cm]{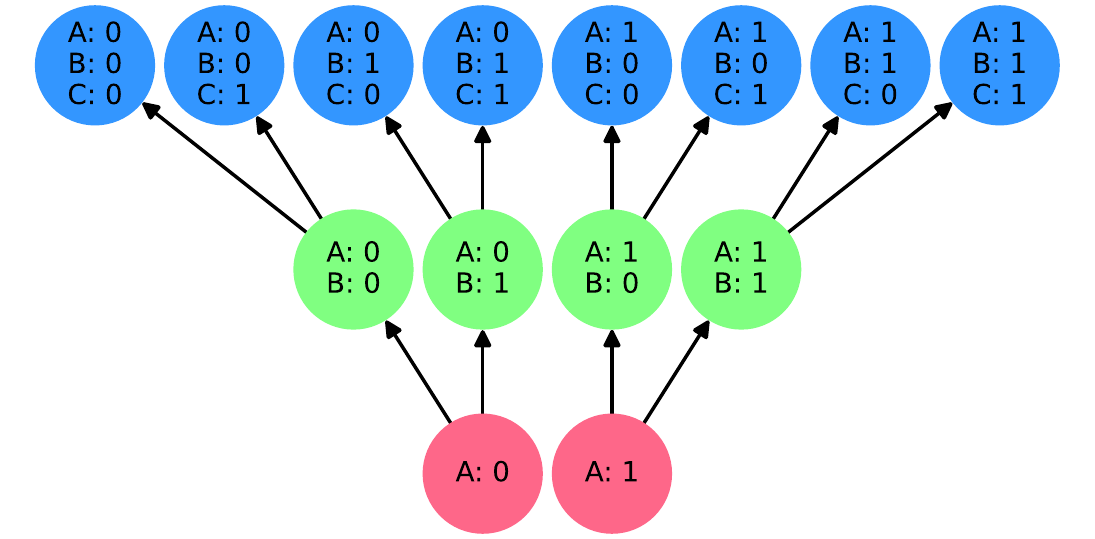}
    \\
    $\Theta_{100}$
    &
    $\Ext{\Theta_{100}}$
    \end{tabular}
\end{center}
Equivalence class 101 contains the 6 spaces for a 3-party causal switch: below is the space where the input of \ev{A} determines the total order between \ev{B} and \ev{C}, with input 0 at \ev{A} setting $\ev{B} < \ev{C}$ and input 1 at \ev{A} setting $\ev{C} < \ev{B}$.
This space has 14 histories, covering:
\begin{itemize}
    \item all inputs for event \ev{A}, determining the output at $\ev{A}$ and the total order between \ev{B} and \ev{C}
    \item all inputs for event \ev{B} when \ev{A} has input 0, determining the output at \ev{B}
    \item all inputs for events $\evset{B,C}$ when \ev{A} has input 0, determining the output at $\ev{C}$
    \item all inputs for event \ev{C} when \ev{A} has input 1, determining the output at \ev{C}
    \item all inputs for events $\evset{C,B}$ when \ev{A} has input 1, determining the output at $\ev{B}$
\end{itemize}
This space coincides with its own space of extended input histories.
\begin{center}
    \begin{tabular}{cc}
    \includegraphics[height=3.5cm]{svg-inkscape/space-ABC-unique-tight-101-highlighted_svg-tex.pdf}
    &
    \includegraphics[height=3.5cm]{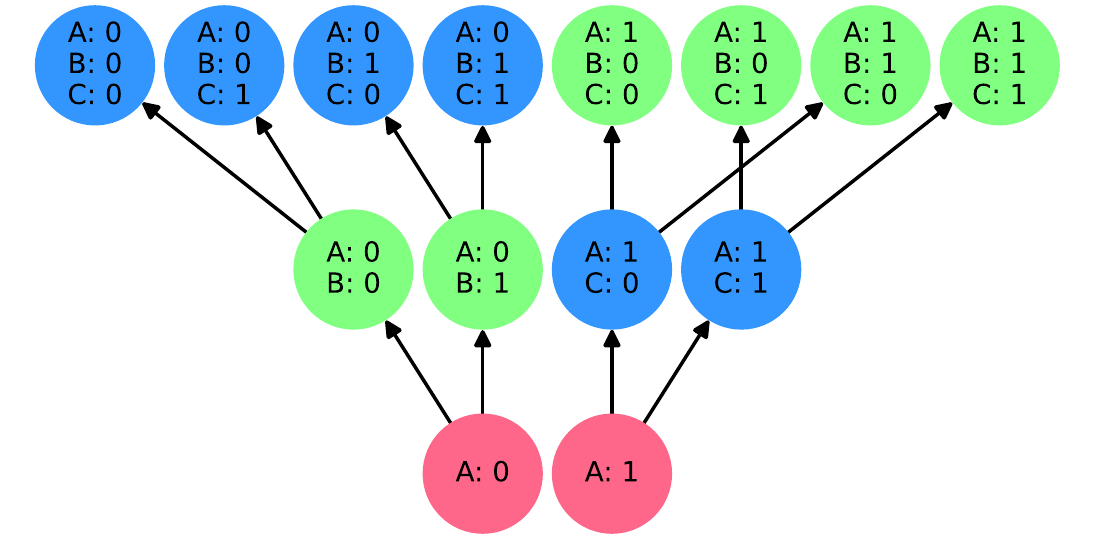}
    \\
    $\Theta_{101}$
    &
    $\Ext{\Theta_{101}}$
    \end{tabular}
\end{center}
The spaces in equivalence class 101 are examples of causally complete spaces not admitting a fixed definite causal order: they are not refinements of $\Hist{\Omega, \{0,1\}}$ for any definite causal order $\Omega$ on \ev{A}, \ev{B} and \ev{C}.
There are 13 equivalence classes consisting of spaces that don't admit a fixed definite causal order, highlighted in Figure \ref{fig:hierarchy-spaces-ABC-indefinite} (p.\pageref{fig:hierarchy-spaces-ABC-indefinite}).
\begin{figure}[h]
    \centering
    \includegraphics[width=\textwidth]{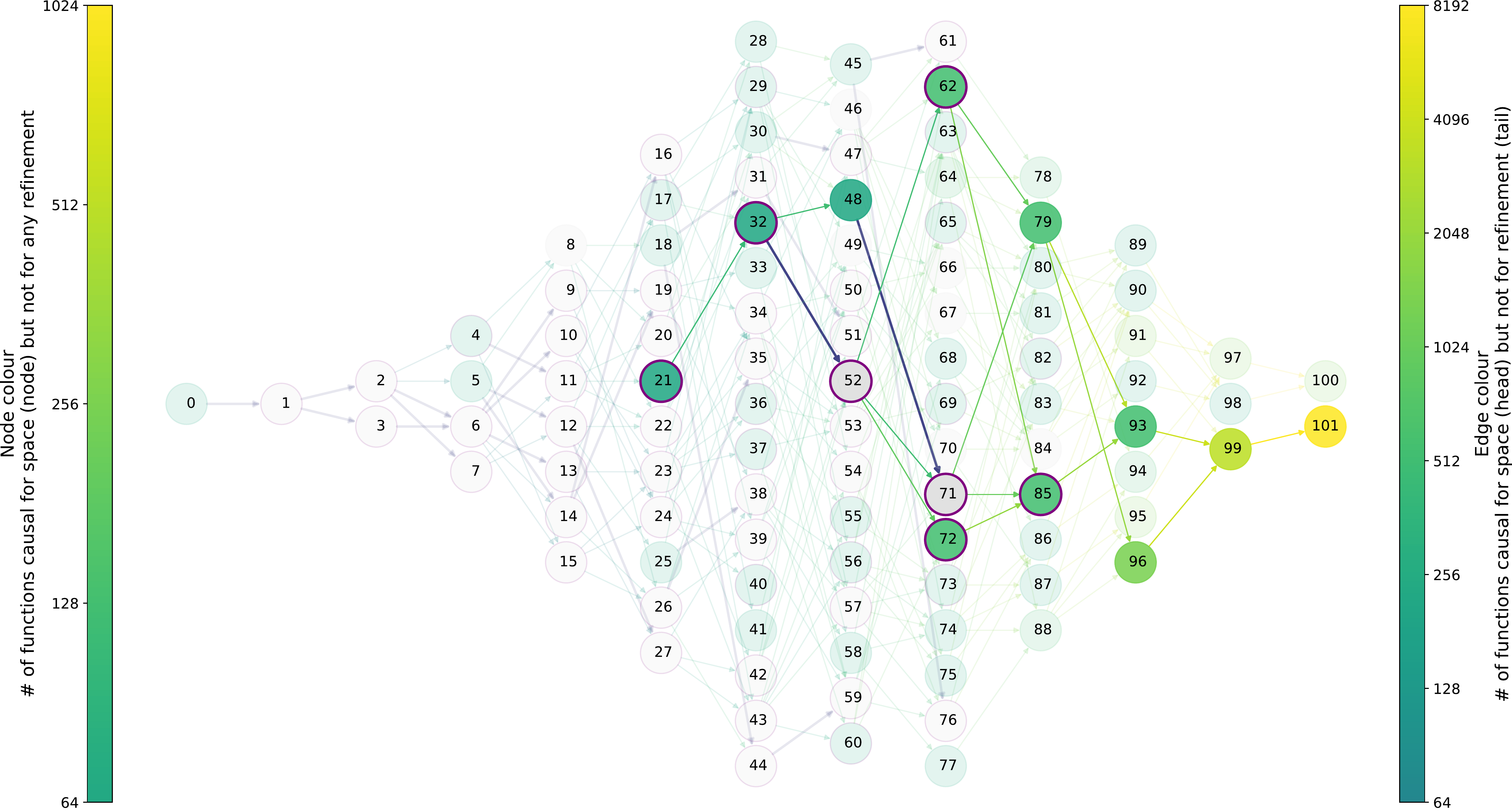}
    \caption{
    The 13 equivalence classes not admitting a fixed definite causal order, highlighted within the hierarchy of causally complete spaces on 3 events with binary inputs.
    See Figure \ref{fig:hierarchy-spaces-ABC} (p.\pageref{fig:hierarchy-spaces-ABC}) for discussion of colours and markings.
    }
\label{fig:hierarchy-spaces-ABC-indefinite}
\end{figure}

The 5 equivalence classes of spaces induced by total orders are marked by a thick black border in Figure \ref{fig:hierarchy-spaces-ABC} (p.\pageref{fig:hierarchy-spaces-ABC}).
We have already seen equivalence class 0 (for the discrete order) and equivalence class 100 (for total orders): we now look at the remaining three.
Equivalence class 92 contains the 3 spaces induced by wedge orders: below is the space induced by order $\total{\ev{A},\ev{C}}\vee\total{\ev{B},\ev{C}}$.
This space has 12 histories, covering all possible combinations of inputs for event $\ev{A}$ (determining the output at event $\ev{A}$), for event $\ev{B}$ (determining the output at event $\ev{B}$), and for events $\evset{A,B,C}$ (determining the output at event $\ev{C}$).
The extended input histories supported by $\evset{A,B}$ are not $\vee$-prime in this space.
\begin{center}
    \begin{tabular}{cc}
    \includegraphics[height=3.5cm]{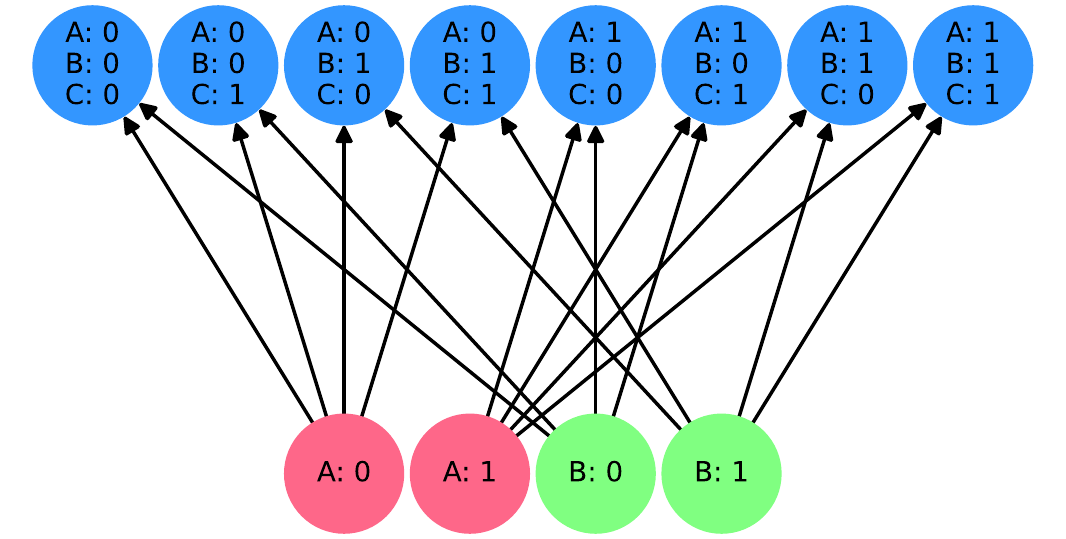}
    &
    \includegraphics[height=3.5cm]{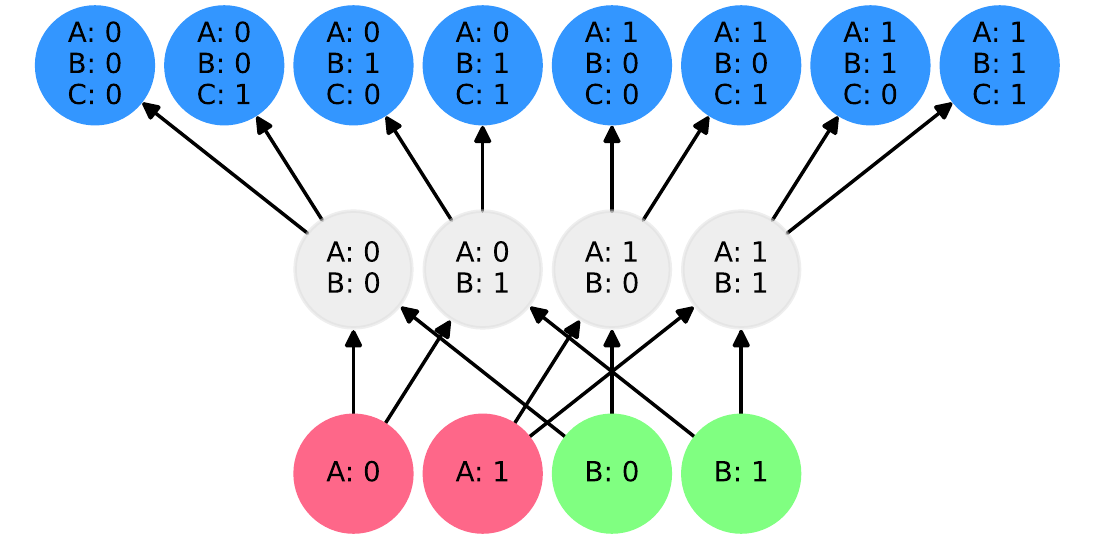}
    \\
    $\Theta_{92}$
    &
    $\Ext{\Theta_{92}}$
    \end{tabular}
\end{center}
Equivalence class 77 contains the 3 spaces induced by fork orders: below is the space induced by order $\total{\ev{A},\ev{B}}\vee\total{\ev{A},\ev{C}}$.
This space has 10 histories, covering all possible combinations of inputs for event $\ev{A}$ (determining the output at event $\ev{A}$), for events $\evset{A,B}$ (determining the output at event $\ev{B}$), and for events $\evset{A,C}$ (determining the output at event $\ev{C}$).
The extended input histories supported by all three events are not $\vee$-prime in this space.
\begin{center}
    \begin{tabular}{cc}
    \includegraphics[height=3.5cm]{svg-inkscape/space-ABC-unique-tight-77-highlighted_svg-tex.pdf}
    &
    \includegraphics[height=3.5cm]{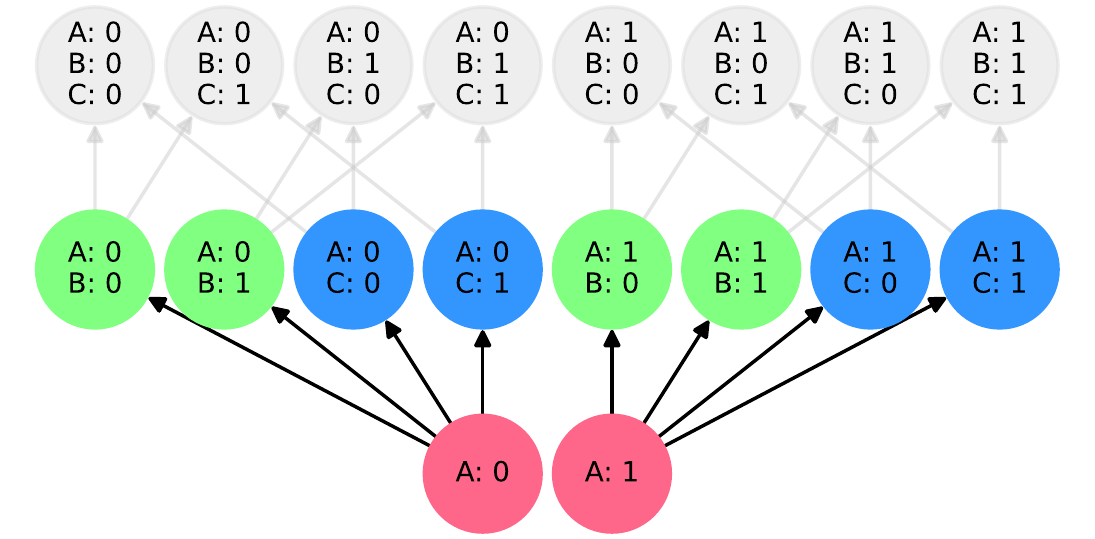}
    \\
    $\Theta_{77}$
    &
    $\Ext{\Theta_{77}}$
    \end{tabular}
\end{center}
Equivalence class 33 contains the 6 spaces induced by the disjoint join of a total order on two events with a discrete third event: below is the space induced by order $\total{\ev{A},\ev{B}}\vee\discrete{\ev{C}}$.
This space has 8 histories, covering all possible combinations of inputs for event $\ev{A}$ (determining the output at event $\ev{A}$), for events $\evset{A,B}$ (determining the output at event $\ev{B}$), and for event \ev{C} (determining the output at event $\ev{C}$).
The extended input histories supported by either $\evset{A,C}$ or by all three events are not $\vee$-prime in this space.
\begin{center}
    \begin{tabular}{cc}
    \includegraphics[height=3.5cm]{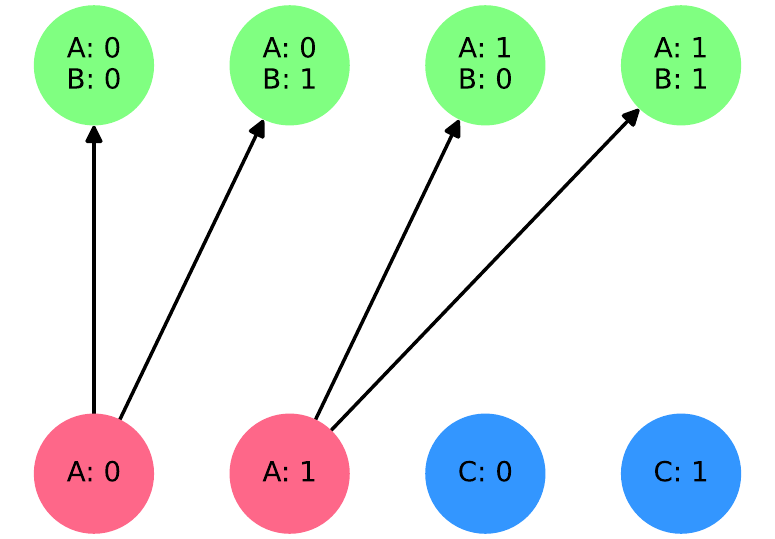}
    &
    \includegraphics[height=3.5cm]{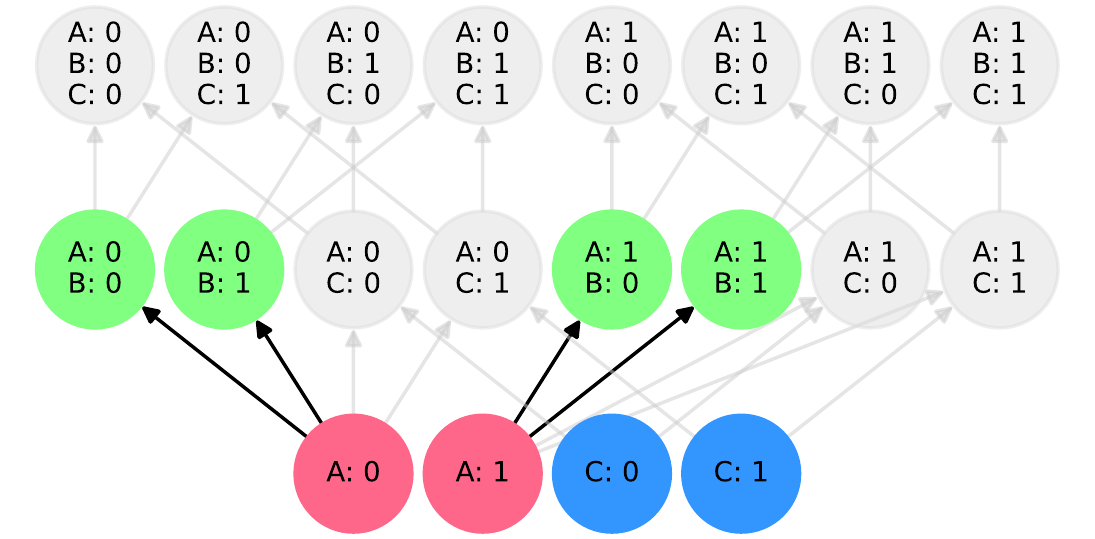}
    \\
    $\Theta_{33}$
    &
    $\Ext{\Theta_{33}}$
    \end{tabular}
\end{center}

Spaces not induced by causal orders can all be understood as introducing input-dependent causal constraints.
We already saw this in the 3-party causal switch space $\Theta_{101}$: it refines the (non causally complete) order-induced space $\Hist{\total{\ev{A},\evset{B,C}}, \{0,1\}}$, by introducing causal constrains on $\evset{B,C}$ which depend on the input at event $\ev{A}$.
The spaces in equivalence class 101 might be iconic example of this mechanism, but all 97 equivalence classes of non-order-induced spaces can be understood this way: we take some order-induced coarsening and study the additional input-dependent causal constraints.

In the most general case of this procedure, we consider a space $\Theta \in \Spaces{\underline{I}}$ and a causal order $\Omega$ such that $\Theta \leq \Hist{\Omega, \underline{I}}$, i.e. such that:
\[
    \Ext{\Theta} \supseteq \Hist{\Omega, \underline{I}}
\]
In particular, the above implies that $\Theta \in \SpacesFC{\underline{I}}$.
The extended input histories in $\Ext{\Theta} \backslash \ExtHist{\Omega, \underline{I}}$ correspond to causal constraints that $\Theta$ imposes additionally to $\Hist{\Omega, \underline{I}}$: if there is a unique minimal choice for $\Omega$ (e.g. for equivalence class 98, discussed below), then the additional constraints are truly input-dependent; if there are multiple minimal choices for $\Omega$ (e.g. equivalence class 3, discussed below), then the additional constraints might instead be those of a different causal order, independent of any input values.
For each extended input history $h \in \Ext{\Theta} \backslash \ExtHist{\Omega, \underline{I}}$, we consider the set $K_h$ of minimal extended input histories from the order-induced space which lie above $h$:
\[
    K_h 
    :=
    \min\left(\upset{h} \cap \ExtHist{\Omega, \underline{I}}\right)
    \subseteq \Ext{\Theta}
\]
We then consider the set $E_h$ of all events which are in the domain of some $k \in K_h$ but not in the domain of $h$:
\[
E_h := \bigcup_{k \in K_h}\dom{k} \backslash \dom{h}
\]
The additional constraint imposed by $h$ can then be understood as follows: when the events in $\dom{h}$ have inputs specified by $h$, the outputs at the events $\dom{h}$ are independent of the inputs at events in $E_h$.

As the simplest example of input-dependent causal constraints, we consider space $\Theta_{98}$ below, a representative from equivalence class 98 which is a closest refinement of $\Hist{\total{\ev{A},\ev{B},\ev{C}}, \{0,1\}}$.
The only additional history in this case is $\hist{B/1}$, imposing the following constraint: when the input at \ev{B} is 1, the output at \ev{B} is independent of the input at event \ev{A}.
\begin{center}
    \begin{tabular}{cc}
    \includegraphics[height=3.5cm]{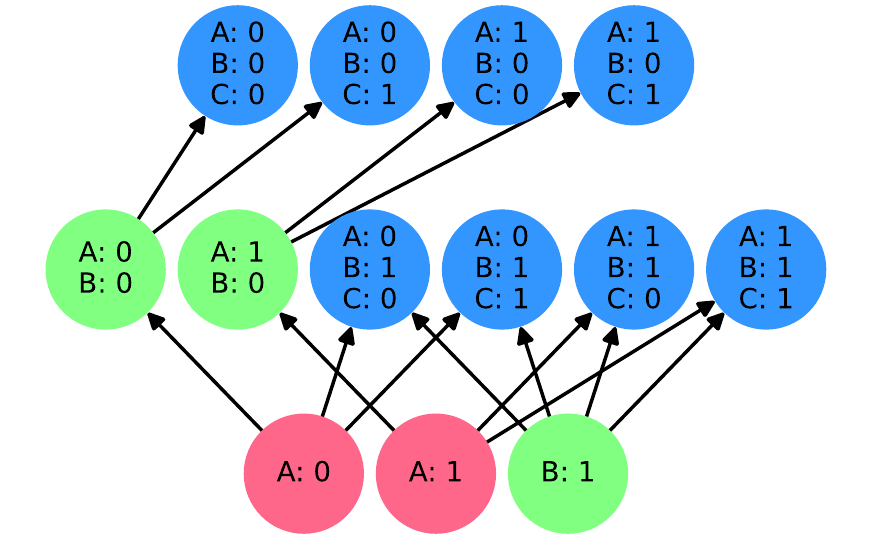}
    &
    \includegraphics[height=3.5cm]{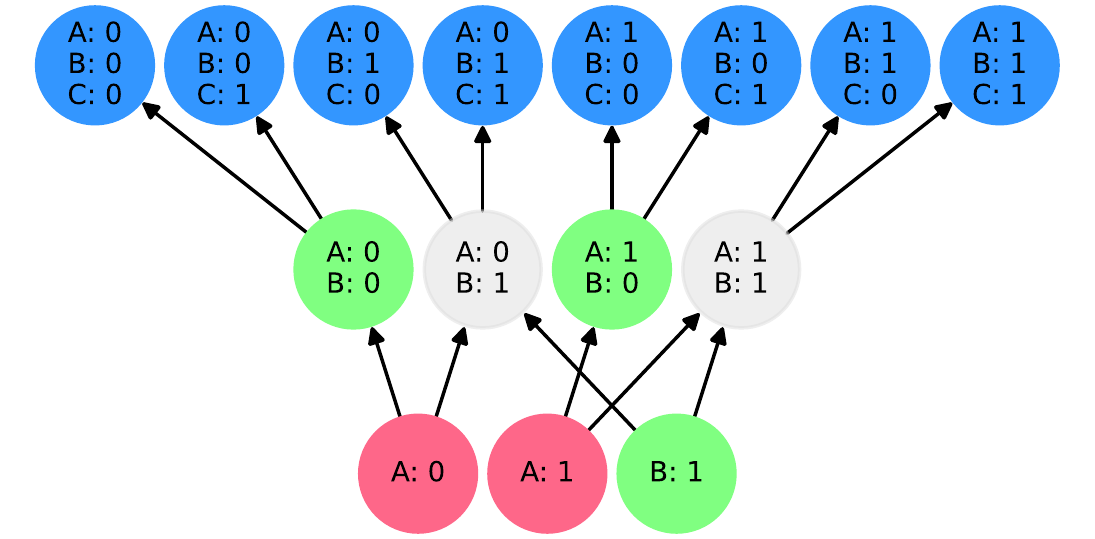}
    \\
    $\Theta_{98}$
    &
    $\Ext{\Theta_{98}}$
    \end{tabular}
\end{center}
Another simple example is given by $\Theta_{97}$ below, a representative from equivalence class 97 which is also a closest refinement of $\Hist{\total{\ev{A},\ev{B},\ev{C}}, \{0,1\}}$.
The only additional history in this case is $\hist{C/1}$, imposing the following constraint: when the inputs at events \evset{A,C} are given by $\hist{A/1, C/1}$, the outputs at \evset{A,C} are independent of the input at event \ev{B}.
\begin{center}
    \begin{tabular}{cc}
    \includegraphics[height=3.5cm]{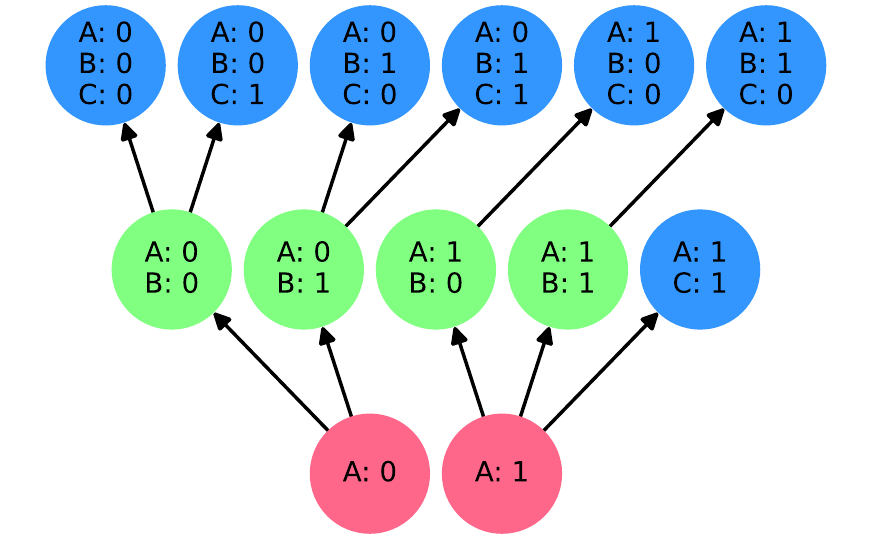}
    &
    \includegraphics[height=3.5cm]{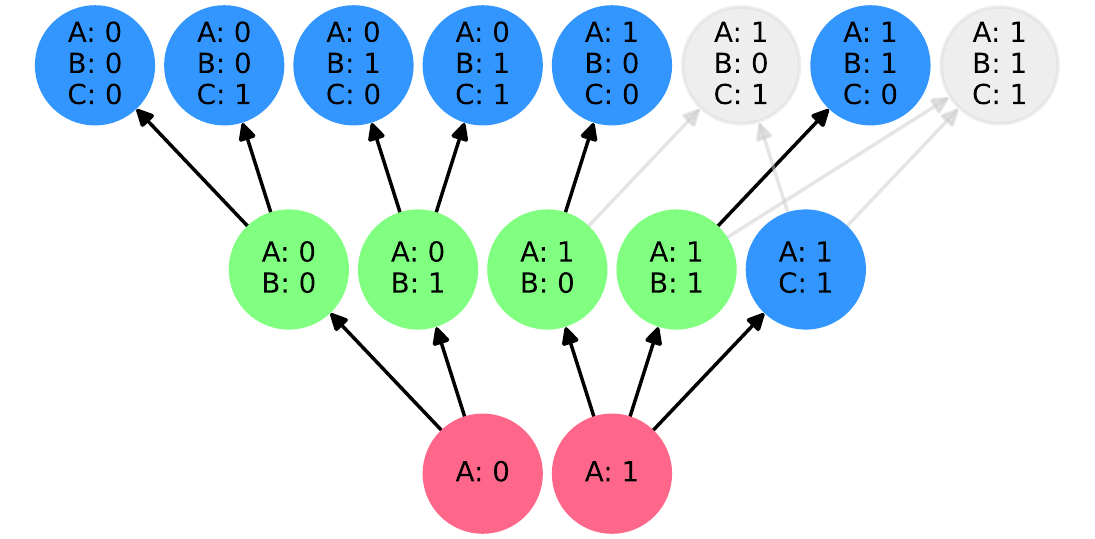}
    \\
    $\Theta_{97}$
    &
    $\Ext{\Theta_{97}}$
    \end{tabular}
\end{center}
Both examples above are clear cases of input-dependent causal constraints.
However, we mentioned that additional causal constraints need not be truly input dependent, as witnessed by our previous example on the meet of order-induced spaces for causal orders $\Omega = \total{\ev{A},\ev{B}}\vee\discrete{\ev{C}}$ and $\Omega' = \discrete{A}\vee\total{\ev{C},\ev{B}}$.
\begin{center}
    \begin{tabular}{ccc}
    \includegraphics[height=2.5cm]{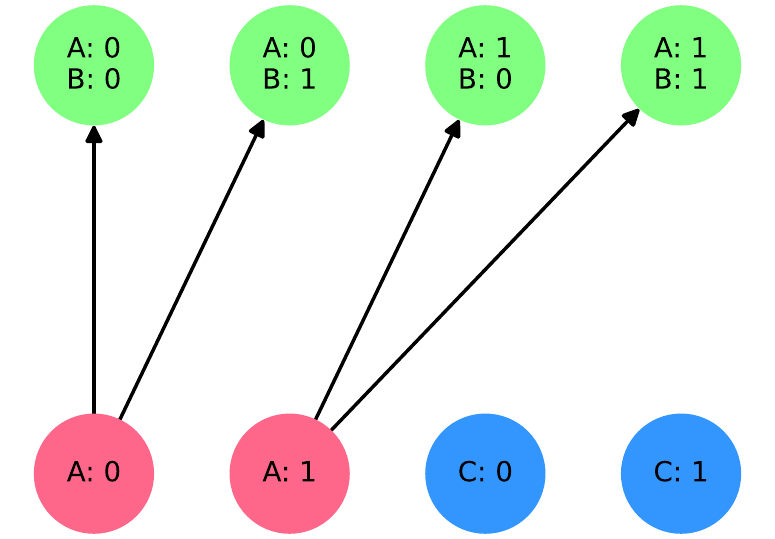}
    &
    \includegraphics[height=2.5cm]{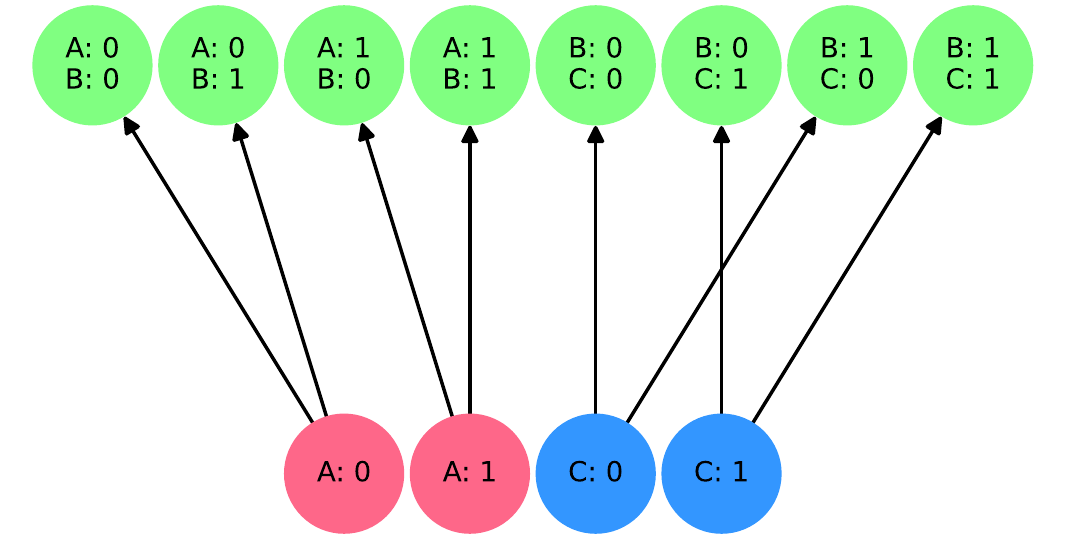}
    &
    \includegraphics[height=2.5cm]{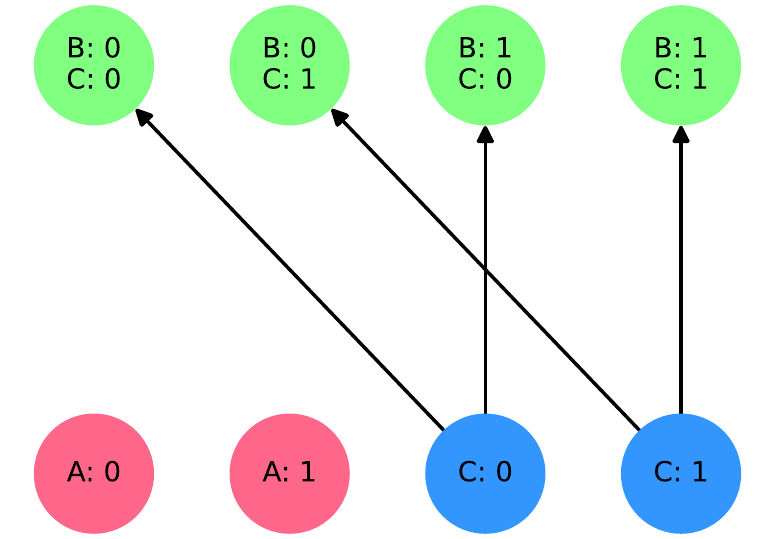}
    \\
    $\Theta_{33} = \Hist{\Omega,\{0,1\}}$
    &
    \hspace{2mm}
    $\Theta_3 = \Hist{\Omega',\{0,1\}}\wedge\Hist{\Omega',\{0,1\}}$
    \hspace{2mm}
    &
    $\Hist{\Omega',\{0,1\}}$
    \end{tabular}
\end{center}
Indeed, the spaces in equivalence class 3 are exactly the meets of 3 pairs of spaces from equivalence class 33 (the other 15 non-trivial meets of pairs in equivalence class 33 all yield the discrete space, in equivalence class 0).
For space $\Theta_{3}$, specifically, we get the following additional constraints:
\begin{itemize}
    \item as a coarsening of order-induced space $\Theta_{33} = \Hist{\total{\ev{A},\ev{B}}\vee\discrete{\ev{C}},\{0,1\}}$, the additional constraints come from the 4 histories with domain $\evset{B,C}$: they state that the outputs on \evset{B,C} are independent of the input on \ev{A} for all possible choices of inputs on $\evset{B,C}$.
    \item as a coarsening of order-induced space $\Hist{\discrete{\ev{A}}\vee\total{\ev{C},\ev{B}},\{0,1\}}$, the additional constraints come from the 4 histories with domain $\evset{A,B}$: they state that the outputs on \evset{A,B} are independent of the input on \ev{C} for all possible choices of inputs on $\evset{A,B}$.
\end{itemize}
Because the additional constraints appear for all possible choices of inputs on their common support, they are not truly input-dependent in this case.
\begin{center}
    \begin{tabular}{cc}
    \includegraphics[height=3cm]{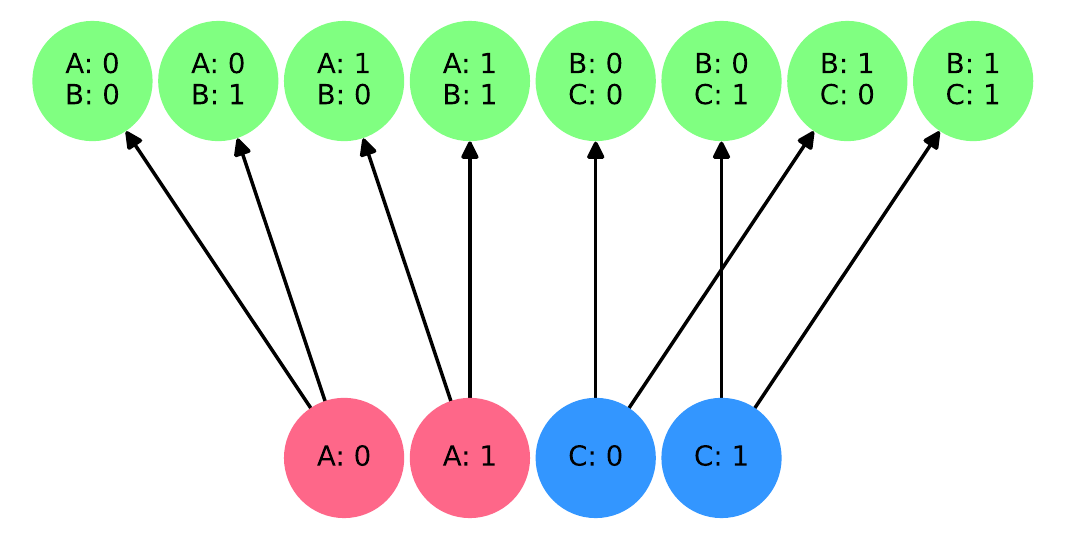}
    &
    \includegraphics[height=3cm]{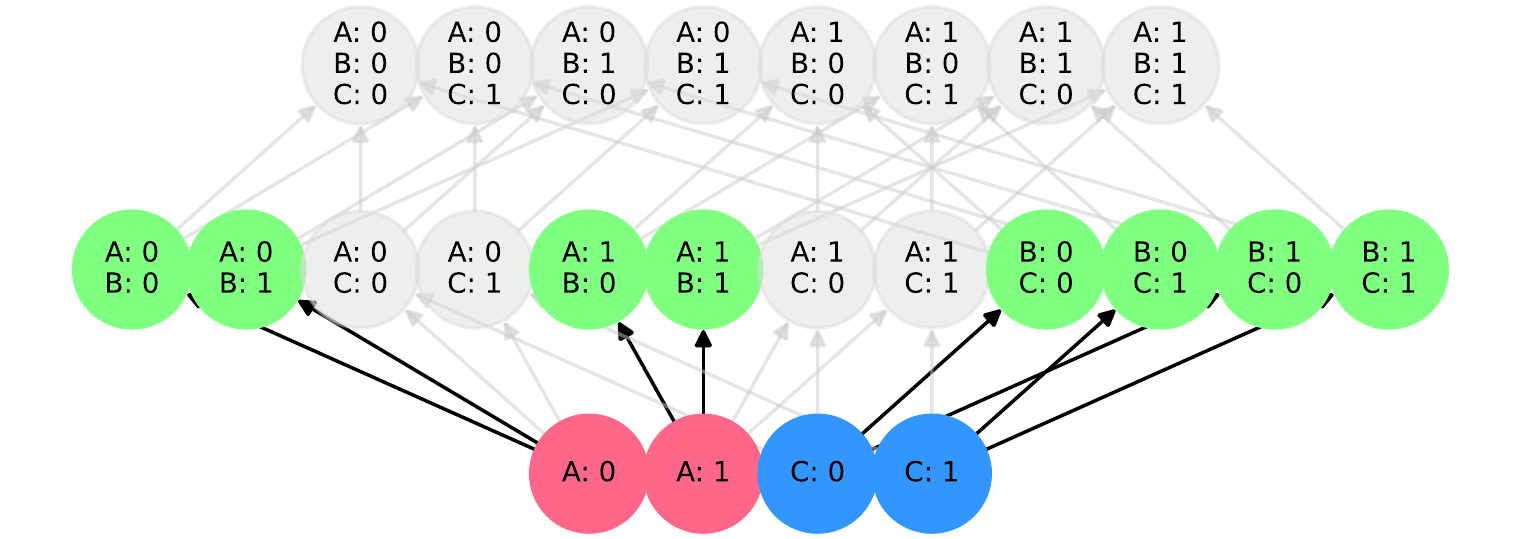}
    \\
    $\Theta_{3}$
    &
    $\Ext{\Theta_{3}}$
    \end{tabular}
\end{center}

The description of the constraints for space $\Theta_3$ is a bit confusing: one would certainly be forgiven for thinking that these constraints should be equivalent to the no-signalling ones, generated by the discrete space.
And, in a sense, they are: as discussed in the next subsection, the spaces in equivalence class 3 have exactly the same causal functions as the discrete space (as do the spaces in equivalence classes 1, 2, 6, 7, 9, 10 and 13).
Furthermore, we show in \cite{gogioso2022geometry} that the causal polytope for space $\Theta_3$ (as well as $\Theta_1$) coincides with the no-signalling polytope---the causal polytope of the discrete space $\Theta_0$---when the ``standard cover'' is considered.
This means that spaces $\Theta_3$ and $\Theta_0$ are causally equivalent when studying non-locality, which is based on the standard cover; however, the former admits strictly more contextual empirical models than the latter for other choices of cover, modelling various notions of contextuality.

Space $\Theta_3$ is also an example of a ``non-tight'' space, one where the events in some histories are constrained by multiple causal orders.
Lack of tightness is a peculiar pathology: in some cases, it implies a form of contextuality where deterministic functions defined compatibly on certain subsets of input histories cannot always be glued together into functions defined on all histories. Put it in more technical terms, we will see later on that the presheaf of causal functions on a non-tight space of input histories is not necessarily a sheaf.

\begin{definition}
\label{definition:tight-space}
Let $\Theta$ be a space of input histories.
We say that $\Theta$ is \emph{tight} if for every (maximal) extended input history $k \in \Ext{\Theta}$ and every event $\omega \in \dom{k}$ there is a unique input history $h \in \Theta$ such that $h \leq k$ and $\omega \in \tips{\Theta}{h}$.
We say that $\Theta$ is \emph{non-tight} otherwise.
\end{definition}

Non-tight spaces are indicated in Figure \ref{fig:hierarchy-spaces-ABC} (p.\pageref{fig:hierarchy-spaces-ABC}) by a thin violet border, and they constitute the majority of examples: out of 102 equivalence classes, 58 consist of non-tight spaces and 44 consist of tight spaces.
To understand what lack of tightness means concretely, let's consider space $\Theta_{17}$ below.
In the input histories below extended input history $\hist{A/1,B/1,C/2}$ (circled in blue), the event $\ev{C}$ appears as a tip event in two separate histories, namely $\hist{A/1, C/1}$ and $\hist{B/1, C/1}$; edges from the latter input histories to the former extended input histories have also been highlighted blue, for clarity.
\begin{center}
    \begin{tabular}{cc}
    \includegraphics[height=3.5cm]{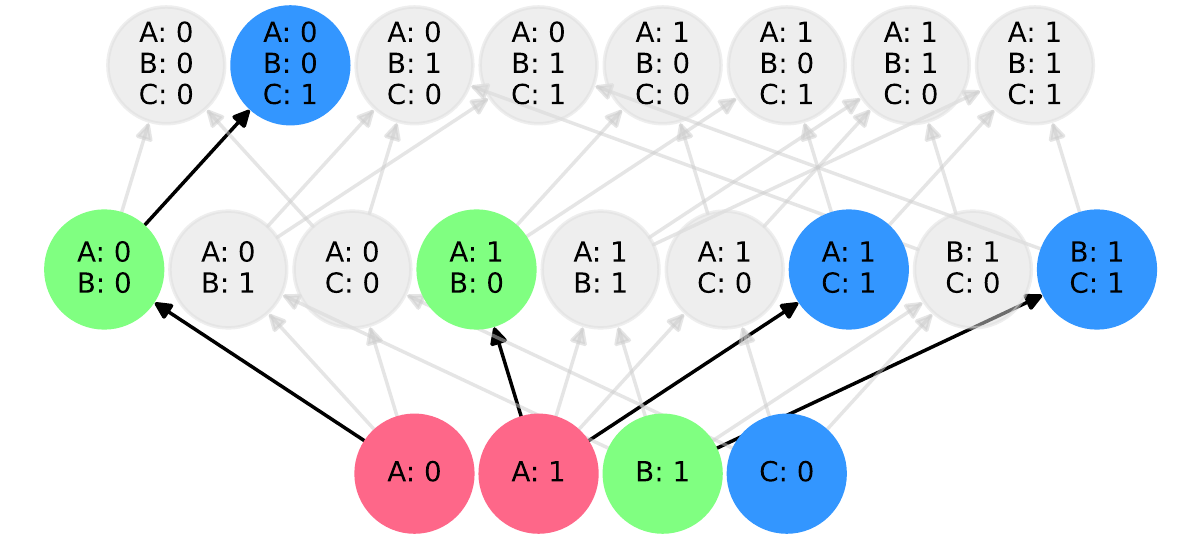}
    &
    \includegraphics[height=3.5cm]{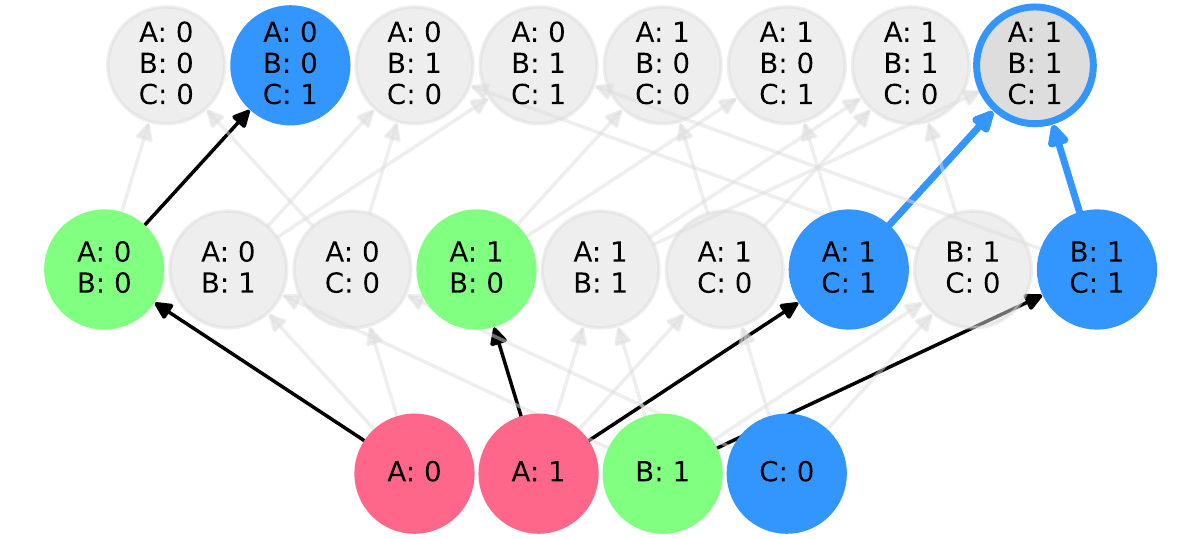}
    \\
    $\Ext{\Theta_{17}}$
    &
    $\Ext{\Theta_{17}}$ with highlights
    \end{tabular}
\end{center}
The effect of this multiple appearance of $\ev{C}$ as a tip event is that causal functions on space $\Theta_{17}$ must yield identical output values at event $\ev{C}$ for both input histories $\hist{A/1, C/1}$ and $\hist{B/1, C/1}$, which would have otherwise been unrelated.
Put in other words, in history $\hist{A/1,B/1,C/2}$ the output at event $\ev{C}$ must satisfy the constraints of two different causal orders: \total{\ev{A}, \ev{C}, \ev{B}} (from $\hist{A/1} \rightarrow \hist{A/1, C/1} \rightarrow \hist{A/1, B/1, C/1}$) and \total{\ev{B}, \ev{C}, \ev{A}} (from $\hist{B/1} \rightarrow \hist{B/1, C/1} \rightarrow \hist{A/1, B/1, C/1}$).

A further example of non-tight space is given by space $\Theta_{21}$, which doesn't admit a fixed definite causal order: \ev{B} causally precedes \ev{C} when the input at \ev{B} is 0 or the input at \ev{A} is 1, while it causally succeeds \ev{C} when the input at \ev{C} is 0 or the input at \ev{A} is 1.
In this space, there are two extended input histories with ``tip event conflicts'' below them: the extended input history $\hist{A/0,B/1,C/0}$ (circled in green) sees $\ev{B}$ appearing as tip event in the two input histories $\hist{A/0,B/1}$ and $\hist{B/1,C/0}$ below it, while the extended input history $\hist{A/1,B/0,C/1}$ (circled in blue) sees $\ev{C}$ appearing as tip event in the two input histories $\hist{A/1,C/1}$ and $\hist{B/0,C/1}$ below it.
\begin{center}
    \begin{tabular}{cc}
    \includegraphics[height=3cm]{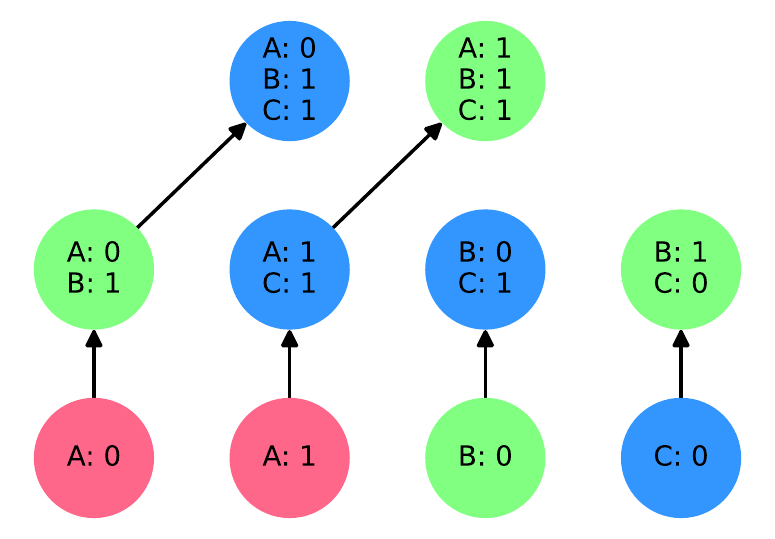}
    &
    \includegraphics[height=3.5cm]{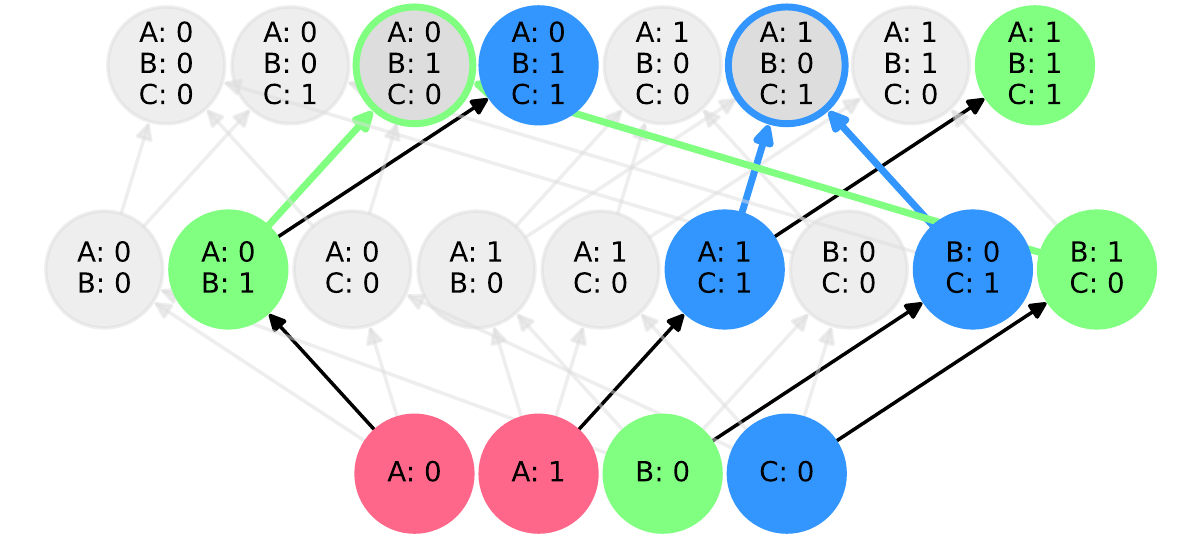}
    \\
    $\Theta_{21}$
    &
    $\Ext{\Theta_{21}}$ with highlights
    \end{tabular}
\end{center}
The (unique) closest causal coarsening of $\Theta_{21}$ which is tight is the space in equivalence class 48 obtained by removing the ``conflicting'' input histories $\hist{B/1,C/0}$ (for event \ev{B}) and $\hist{B/0,C/1}$ (for event \ev{C}).
The space is displayed below as $\Theta_{48,1}$---to differentiate it from the representative $\Theta_{48}$ used in the hierarchy from \cite{gogioso2022classification}---and it also doesn't admit a fixed definite causal order.
\begin{center}
    \begin{tabular}{cc}
    \includegraphics[height=3cm]{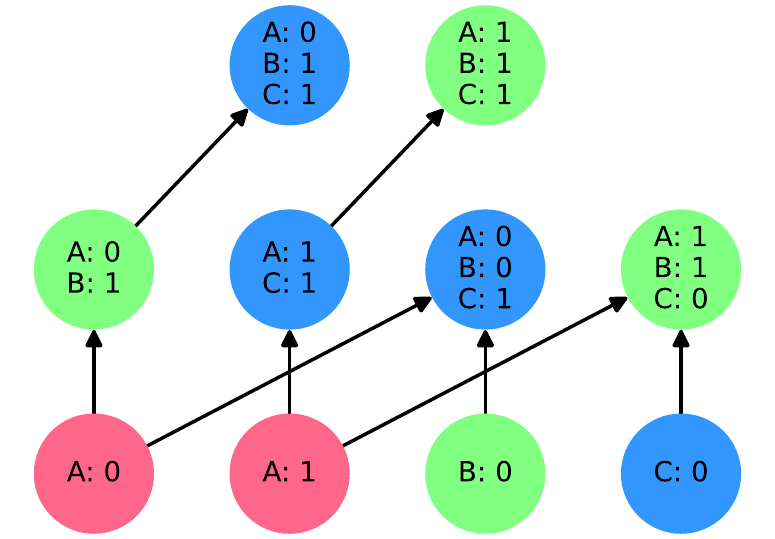}
    &
    \includegraphics[height=3cm]{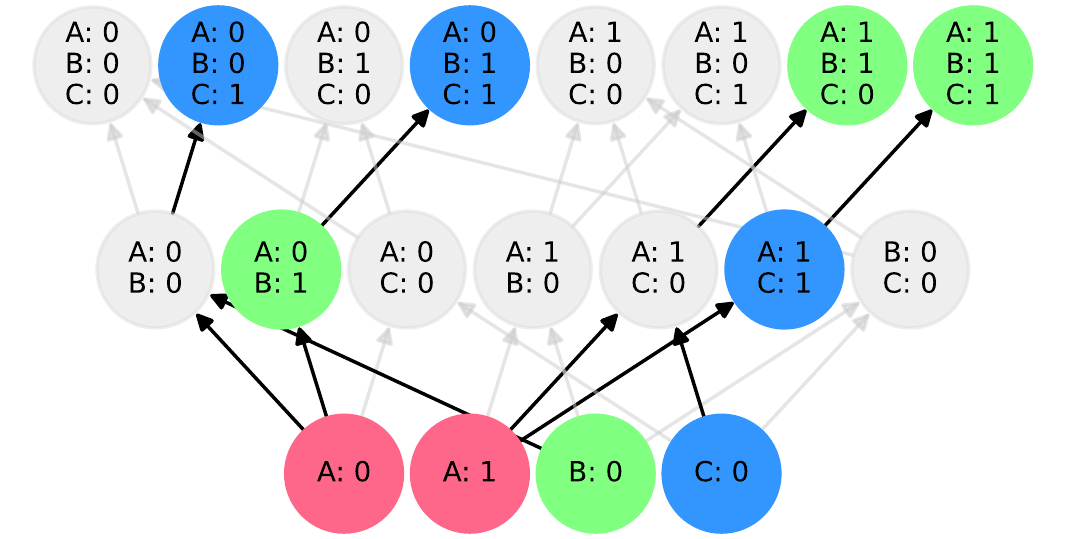}
    \\
    $\Theta_{48,1}$
    &
    $\Ext{\Theta_{48,1}}$
    \end{tabular}
\end{center}

A more complicated example of tight space---imposing multiple input-dependent causal constraints---is given $\Theta_{80}$, a representative of equivalence class 80 and causal refinement of $\Hist{\total{\ev{A},\ev{B},\ev{C}}, \{0,1\}}$.
\begin{center}
    \begin{tabular}{cc}
    \includegraphics[height=3cm]{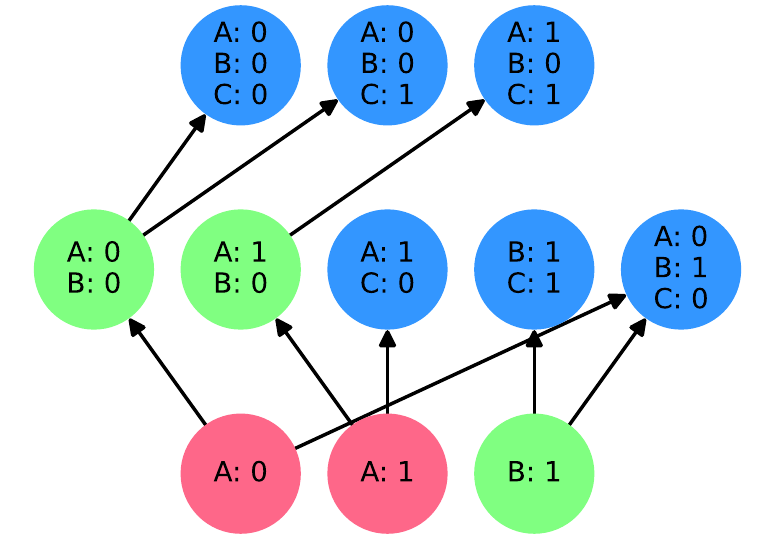}
    &
    \includegraphics[height=3cm]{svg-inkscape/space-ABC-unique-tight-100-highlighted_svg-tex.pdf}
    \\
    $\Theta_{80}$
    &
    $\Hist{\total{\ev{A},\ev{B},\ev{C}}, \{0,1\}}$
    \end{tabular}
\end{center}
In addition to the causal constraints associated with the total order $\total{\ev{A},\ev{B},\ev{C}}$, space $\Theta_{80}$ imposes the following input-dependent causal constraints:
\begin{itemize}
    \item From the additional history \hist{B/1} we get that the input at \ev{B} is independent of the input at \ev{A} when the input at \ev{B} is 1.
    \item From the additional history \hist{A/1,C/0} we get that the outputs at \evset{A,C} are independent of the input at \ev{B} when the input at \ev{A} is 1 and the input at \ev{C} is 0.
    \item From the additional history \hist{B/1,C/1} we get that the outputs at \ev{B,C} are independent of the input at \ev{A} when the input at \ev{B} is 1 and the input at \ev{C} is 1.
\end{itemize}

To conclude our exploration of tight spaces, we prove results about tightness of parallel composition, sequential composition and conditional sequential composition.

\begin{proposition}
\label{proposition:parallel-sequential-composition-tight}
Let $\Theta, \Theta'$ be tight spaces of input histories such that $\Events{\Theta} \cap \Events{\Theta'} = \emptyset$.
The parallel composition $\Theta \cup \Theta'$ and sequential composition $\Theta \seqcomposeSym \Theta'$ are tight.
\end{proposition}
\begin{proof}
See \ref{proof:proposition:parallel-sequential-composition-tight}
\end{proof}

\begin{proposition}
\label{proposition:conditional-sequential-composition-tight}
Let $\Theta$ be a causally complete space of input histories.
Let $(\Theta'_k)_{k \in \max\Ext{\Theta}}$ be a family of causally complete spaces of input histories, with $\Events{\Theta} \cap \Events{\Theta'_k} = \emptyset$ for all $k \in \max\Ext{\Theta}$.
The conditional sequential composition $\Theta \seqcomposeSym \underline{\Theta'}$ is tight.
\end{proposition}
\begin{proof}
See \ref{proof:proposition:conditional-sequential-composition-tight}
\end{proof}

We also show that non-tight spaces arise inevitably when meets of causally complete spaces are considered, even in the simplest case of order-induced spaces (with at least 3 events).
Indeed, the (causally complete) space $\Theta_{3}$ on 3 events which originally sparked our investigation gives an example of such a non-tight meet of order-induced (causally complete) spaces.
\begin{center}
    \begin{tabular}{cc}
    \includegraphics[height=3cm]{svg-inkscape/space-ABC-unique-untight-3-highlighted_svg-tex.pdf}
    &
    \includegraphics[height=3.5cm]{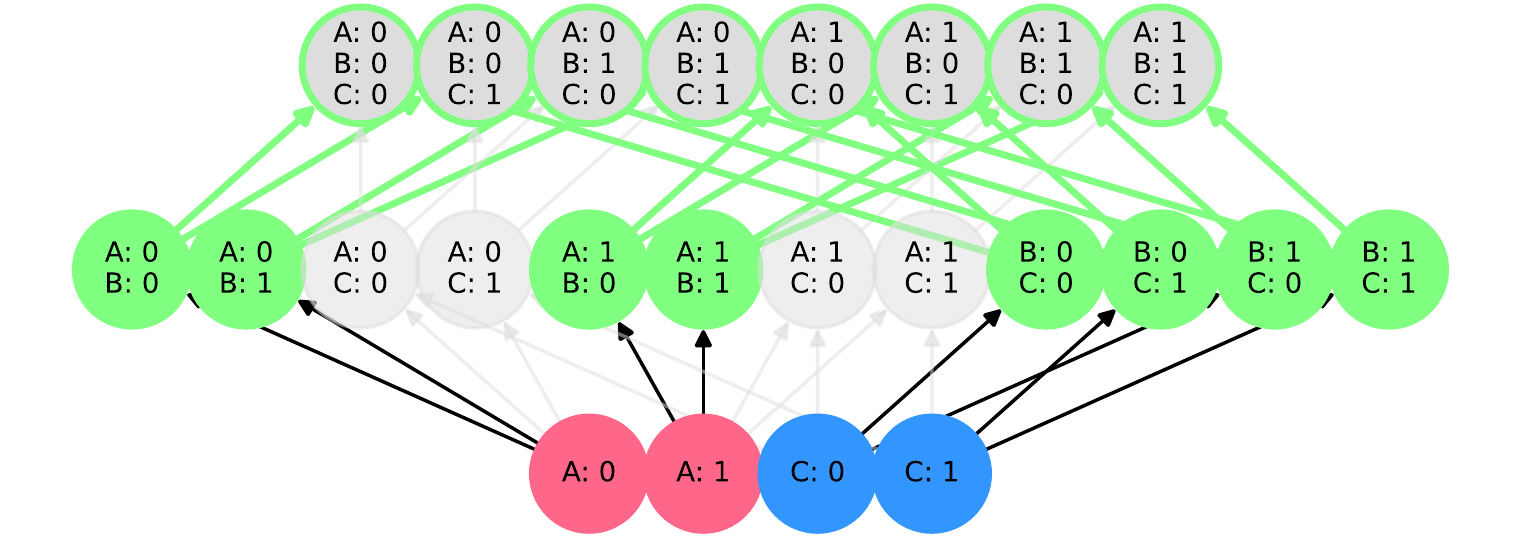}
    \\
    $\Theta_{3}$
    &
    $\Ext{\Theta_{3}}$ with highlights
    \end{tabular}
\end{center}

\begin{proposition}
\label{proposition:hist-space-tight}
Let $\Omega$ be a causal order and let $\underline{I} = \left(I_\omega\right)_{\omega \in \Omega}$ be a family of non-empty input sets. 
The space of input histories $\Theta := \Hist{\Omega, \underline{I}}$ is tight.
\end{proposition}
\begin{proof}
See \ref{proof:proposition:hist-space-tight}
\end{proof}

\begin{theorem}
\label{theorem:non-tight-order-induced}
Let $\Omega$ and $\Omega'$ be two causal orders on the same set of events $E := |\Omega| = |\Omega'|$.
Let $\Theta:=\Hist{\Omega, \underline{I}}$ and $\Theta':=\Hist{\Omega', \underline{I}}$ be the spaces of input histories induced by the two causal orders, for the same family of input sets $\underline{I}$.
The meet $\Theta \wedge \Theta'$ is tight if and only if for all $\omega \in E$ we have $\downset{\omega}_{\Omega}\subseteq\downset{\omega}_{\Omega'}$ or $\downset{\omega}_{\Omega'}\subseteq\downset{\omega}_{\Omega}$.
\end{theorem}
\begin{proof}
See \ref{proof:theorem:non-tight-order-induced}
\end{proof}

Clearly, the hierarchy of causally complete spaces is full of complicated examples.
However, its ``canopy'' is significantly more tranquil, consisting only of the ``causal switch spaces''.
This is what Figure \ref{fig:hierarchy-spaces-ABC} (p.\pageref{fig:hierarchy-spaces-ABC}) shows for the 3-event case and it is consistent with the approach taken by previous literature on indefinite causality \cite{oreshkov2016causal}.

We define causal switch spaces inductively using conditional sequential composition, and then prove that there is a property---namely, the coincidence of input histories with extended input histories---which singles out such inductively-defined conditional sequential compositions, and in particular characterises the causal switch spaces within the hierarchy of causally complete spaces for given events $E$ and input sets $\underline{I}$.
Finally, we prove that the maximal causally complete spaces are all causal switch spaces.

\begin{definition}
Let $E$ be a set of events and $\underline{I} = (I_\omega)_{\omega \in E}$ be a family of non-empty input sets.
The \emph{causal switch spaces} $\CSwitchSpaces{\underline{I}}$ are defined as follows.
If $E = \emptyset$, then $\CSwitchSpaces{\underline{I}} = \emptyset$.
Otherwise, for each $\omega_1 \in E$ we can consider:
\[
\begin{array}{rcl}
\restrict{\underline{I}}{\{\omega_1\}} &=& (I_\omega)_{\omega \in \{\omega_1\}}\\
\restrict{\underline{I}}{E\backslash\{\omega_1\}} &=& (I_\omega)_{\omega \in E\backslash\{\omega_1\}}
\end{array}
\]
Then the set $\CSwitchSpaces{\underline{I}}$ is defined inductively as follows:
\begin{equation}
    \bigcup\limits_{\omega_1 \in E}
    \suchthat{
    \Hist{\{\omega_1\}, \restrict{\underline{I}}{\{\omega_1\}}}
    \seqcomposeSym
    \underline{\Theta}
    }
    {
    \underline{\Theta}
    \in
    \CSwitchSpaces{\restrict{\underline{I}}{E\backslash\{\omega_1\}}}^{I_{\omega_1}}
    }    
\end{equation}
\end{definition}

From the inductive definition, we can immediately derive a recursive formula to compute the number of causal switch spaces on given events $E$ and inputs $\underline{I}$:
\[
\begin{array}{rcl}
\left|\CSwitchSpaces{\underline{I}}\right|
&=& 1 \text{ in the base case where } E = \emptyset
\\
\left|\CSwitchSpaces{\underline{I}}\right|
&=& \sum_{\omega_1 \in E} \left|\CSwitchSpaces{\restrict{\underline{I}}{E\backslash\{\omega_1\}}}\right|^{I_{\omega_1}}
\end{array}
\]
In the special case where $|I_{\omega}|=k \geq 2$ for all $\omega \in E$, the number $S_{n,k}$ of switch spaces as a function of the number $n:=|E|$ of events satisfies the following recursion:
\[
\begin{array}{rcl}
S_{0,k} &:=& 1\\
S_{n,k} &:=& n \left(S_{n-1,k}\right)^k \text{ for } n \geq 1
\end{array}
\]
The recursion above can be solved to yield the following closed form:
\[
S_{n,k} = \prod_{j=1}^{n} j^{k^{n-j}}
\]
It is then easy to see that the number of switch spaces grows more than doubly exponentially:
\[
\log S_{n,k} = \sum_{j=2}^{n} k^{n-j}\log j
\Rightarrow
\frac{k^{n-1}-1}{k-1} \log 2
\leq \log S_{n,k}
\leq \frac{k^{n-1}-1}{k-1} \log n
\]

\begin{theorem}
\label{theorem:hist-equals-exthist}
Let $\Theta$ be a space of input histories such that $\Theta = \Ext{\Theta}$ and such that for all $h \in \Theta$ we have $|\tips{\Theta}{h}| = 1$.
Either $\Theta = \emptyset$ or $\Theta$ takes the following form:
\begin{equation}
    \Theta = \Hist{\{\omega_1\}, \restrict{\underline{\Inputs{\Theta}}}{\{\omega_1\}}}
    \seqcomposeSym \underline{\Theta'}
\end{equation}
where $\omega_1 \in \Events{\Theta}$ is some event and $\underline{\Theta'} = (\Theta'_{i_1})_{i_1 \in \Inputs{\Theta}_{\omega_1}}$ is a family of spaces of input histories such that, for all $i_1 \in \Inputs{\Theta}_{\omega_1}$:
\begin{itemize}
    \item $\omega_1 \notin \Events{\Theta'_{i_1}}$
    \item $\Theta'_{i_1} = \Ext{\Theta'_{i_1}}$
    \item for all $h \in \Theta'_{i_1}$ we have $\tips{\Theta'_{i_1}}{h} = \tip{\Theta}{\{\omega_1:i_1\}\vee h}$
\end{itemize}
That is, we can recursively apply this proposition to each $\Theta'_{i_1}$ in the family.
\end{theorem}
\begin{proof}
See \ref{proof:theorem:hist-equals-exthist}
\end{proof}

\begin{corollary}
\label{corollary:switch-spaces-characterisation}
Let $E$ be a set of events and $\underline{I} = (I_\omega)_{\omega \in E}$ be a family of non-empty input sets.
The switch spaces $\CSwitchSpaces{\underline{I}}$ are exactly the causally complete spaces $\Theta \in \CCSpaces{\underline{I}}$ such that $\Theta = \Ext{\Theta}$.
\end{corollary}
\begin{proof}
See \ref{proof:corollary:switch-spaces-characterisation}
\end{proof}

\begin{theorem}
\label{theorem:causal-canopy}
Let $E$ be a set of events and $\underline{I} = (I_\omega)_{\omega \in E}$ be a family of non-empty input sets.
The maxima of $\CCSpaces{\underline{I}}$ are exactly the causal switch spaces $\CSwitchSpaces{\underline{I}}$. 
\end{theorem}
\begin{proof}
See \ref{proof:theorem:causal-canopy}
\end{proof}

\subsection{The search for causally complete spaces}
\label{subsection:spaces-search}

The algorithm used to find all causally complete spaces on 2 and 3 events is based on the following alternative characterisation of causal completeness.
See \cite{gogioso2022classification} for listing and description of the algorithm.

\begin{theorem}
\label{theorem:causal-completeness-characterisation}
A space of input histories $\Theta$ is causally complete if and only if for every $k \in \Ext{\Theta}$ with $|\dom{k}| \geq 2$ there exists an $\omega \in \dom{k}$ such that $\restrict{k}{\dom{k}\backslash\{\omega\}} \in \Ext{\Theta}$.
\end{theorem}
\begin{proof}
See \ref{proof:theorem:causal-completeness-characterisation}
\end{proof}

An advanced version of the algorithm---using event-input permutation symmetry to reduce the size of the search space---is used in the search for causally complete spaces on 4 events.
See \cite{gogioso2022classification} for listing and description of the advanced algorithm.

After 106 days of computation on an AWS EC2 \texttt{m6g.large} instance---2 virtual CPUs and 8GiB memory on a AWS Graviton2 processor with 64-bit Arm Neoverse cores---the algorithm has discovered 869529223 causally complete spaces on 4 events, divided into 2312000 equivalence classes under event-input permutation symmetry, occupying around 4.85GiB in RAM and 1.07GiB on disk.
To obtain a rough estimate of the current search status and the final number of spaces and equivalence classes, we fit the following 3-parameter power-law to our data:
\[
f(x|m, e, s) = m\left(1-\left(1+\frac{x}{s}\right)^{-e}\right)
\]
Below is the data (solid blue) together with the fitted curve (dashed orange) projected to 180 days, for both spaces (left) and equivalence classes (right).
Each plot includes a constant plot (solid green) of the the estimated final number, i.e. the $m$ parameter of the fitted curve.
\begin{center}
\includegraphics[height=5cm]{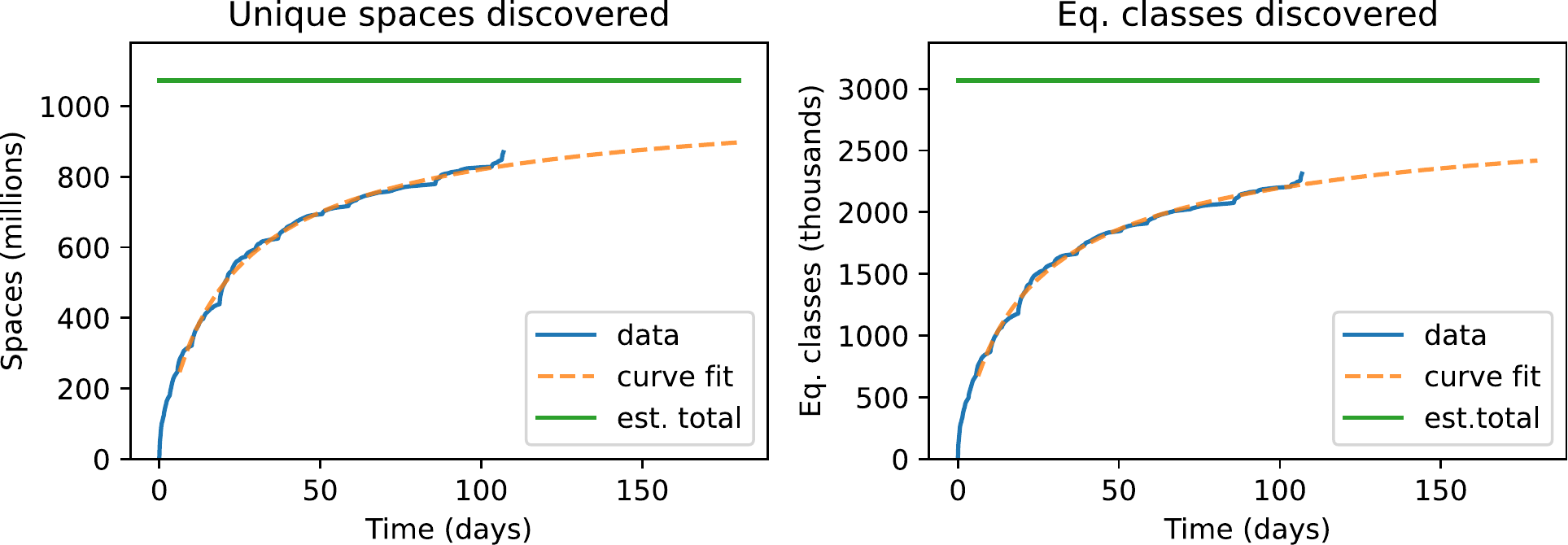}
\end{center}
Based on the fitted trend, we currently predict that there will be around 1 billion causally complete spaces of input histories on 4 events with binary outputs, in around 3 million equivalence classes under event-input permutation.
Whatever the case may be, neither version of the algorithm presented in this work is suitable for a search of all causally complete spaces on 5 events: even a heavily optimised version running on a super-computer would take decades (or more) to complete the task.
The authors have full confidence that significantly more efficient ways to enumerate causally complete spaces will be developed by future literature.

\newpage
\subsection{Proofs for Section \ref{section:spaces}}

\subsubsection{Proof of Proposition \ref{proposition:hist-construction-is-join-prime-subset}}
\label{proof:proposition:hist-construction-is-join-prime-subset}
\begin{proof}
If $h \vee k$ exists and is different from $h$ and $k$, then its domain $\dom{h \vee k}$ must be a union of the causal pasts $\dom{h}$ and $\dom{k}$ for two causally unrelated events---otherwise, we'd have $\dom{h} \subseteq \dom{k}$ or $\dom{k} \subseteq \dom{h}$, implying that $h \vee k = k$ and $h \vee k = h$ respectively.
As the union of the causal pasts of two causally unrelated events, $\dom{h \vee k}$ cannot itself be the causal past of some event, and hence $h \vee k$ cannot an input history in $\Hist{\Omega, \underline{I}}$.
\end{proof}

\subsubsection{Proof of Proposition \ref{proposition:order-induced-space-free-choice}}
\label{proof:proposition:order-induced-space-free-choice}
\begin{proof}
Let $h \in \prod_{\omega \in \Omega} I_\omega$ be a maximal history.
The subset of $\Hist{\Omega, \underline{I}}$ defined by $\mathcal{H}=\suchthat{\restrict{h}{\downset{\omega}}}{\omega \in \dom{h}}$ is compatible by definition, and we have $h = \bigvee \mathcal{H}$ (because $\dom{h} = \cup_{\omega \in \dom{h}} \downset{\omega}$).
\end{proof}

\subsubsection{Proof of Proposition \ref{proposition:induced-spaces-order}}
\label{proof:proposition:induced-spaces-order}
\begin{proof}
We know from the previous section that $\Omega \leq \Xi$ if and only if $\Lsets{\Omega}\supseteq\Lsets{\Xi}$.
Assume that $\Omega \leq \Xi$, let $h \in \ExtHist{\Xi, \underline{I}}$ and consider its domain $\dom{h}$: we have $\dom{h} \in \Lsets{\Xi} \subseteq \Lsets{\Omega}$, and hence $h \in \prod_{\omega \in \dom{h}} I_\omega \subseteq \ExtHist{\Omega, \underline{I}}$.
Conversely, assume that $\ExtHist{\Omega, \underline{I}} \supseteq \ExtHist{\Xi, \underline{I}}$, let $D \in \Lsets{\Xi}$ and consider any $h \in \prod_{\omega \in D}I_\omega$: by assumption, we have $\ExtHist{\Omega, \underline{I}}$, so that $D = \dom{h} \in \Lsets{\Omega}$, proving that $\Lsets{\Omega} \subseteq \Lsets{\Xi}$ and hence that $\Omega \leq \Xi$.
\end{proof}

\subsubsection{Proof of Proposition \ref{proposition:subspace-input-histories-events-inputs}}
\label{proof:proposition:subspace-input-histories-events-inputs}
\begin{proof}
For the event sets $\Events{\Theta'}$ and $\Events{\Theta}$ we have:
\[
\Events{\Theta'}
=
\bigcup_{h \in \Ext{\Theta'}} \dom{h}
\supseteq
\bigcup_{h \in \Ext{\Theta}} \dom{h}
=
\Events{\Theta}
\]
For the input sets $\Inputs{\Theta'}_\omega$ and $\Inputs{\Theta}_\omega$, where $\omega$ is any $\omega \in \Events{\Theta}$, we have:
\[
\begin{array}{rcl}
\Inputs{\Theta'}_\omega
&=&
\suchthat{h_\omega}{h \in \Ext{\Theta'}, \omega \in \dom{h}}
\\
&\supseteq&
\suchthat{h_\omega}{h \in \Ext{\Theta}, \omega \in \dom{h}}
=
\Inputs{\Theta}_\omega
\end{array}
\]
\end{proof}

\subsubsection{Proof of Proposition \ref{proposition:spaces-input-histories-hierarchies}}
\label{proof:proposition:spaces-input-histories-hierarchies}
\begin{proof}
This is all rather straightforward, going through the inclusion order of spaces of extended input histories. Recall that $\Prime{\Ext{\Theta}} = \Theta$ for all $\vee$-prime $\Theta$ and that $\Ext{\Prime{W}} = W$ for all $\vee$-closed $W$.
\begin{itemize}
    \item For the join.
    We have that $\Ext{\Theta}, \Ext{\Theta'} \supseteq \Ext{\Theta} \cap \Ext{\Theta'}$ and $\Ext{\Theta} \cap \Ext{\Theta'}$ is $\vee$-closed, so that $\Theta, \Theta' \leq \Prime{\Ext{\Theta} \cap \Ext{\Theta'}}$.
    Furthermore, if $\Theta, \Theta' \leq \Theta''$, then $\Ext{\Theta}, \Ext{\Theta'} \supseteq \Ext{\Theta''}$, so that $\Ext{\Theta} \cap \Ext{\Theta'} \supseteq \Ext{\Theta''}$ and hence $\Prime{\Ext{\Theta} \cap \Ext{\Theta'}} \leq \Theta''$.
    
    \item For the meet.
    We have that $\Ext{\Theta} \cup \Ext{\Theta'} \supseteq \Ext{\Theta}, \Ext{\Theta'}$ and the $\vee$-closure of $\Ext{\Theta} \cup \Ext{\Theta'}$ is (by definition) its smallest $\vee$-closed superset, so that $\Prime{\text{$\vee$-closure of }\Ext{\Theta} \cup \Ext{\Theta'}} \leq \Theta, \Theta'$.
    Furthermore, if $\Theta'' \leq \Theta, \Theta'$, then $\Ext{\Theta''} \supseteq \Ext{\Theta}, \Ext{\Theta'}$ and $\Ext{\Theta''}$ is $\vee$-closed, so that $\Ext{\Theta''} \supseteq \text{$\vee$-closure of }\Ext{\Theta} \cup \Ext{\Theta'}$ and hence $\Theta'' \leq \Prime{\text{$\vee$-closure of }\Ext{\Theta} \cup \Ext{\Theta'}}$. This proves that $\Theta \wedge \Theta' = \Prime{\text{$\vee$-closure of }\Ext{\Theta} \cup \Ext{\Theta'}}$. The elements added by $\vee$-closure are all not $\vee$-prime by definition, so we also have $\Prime{\text{$\vee$-closure of }\Ext{\Theta} \cup \Ext{\Theta'}} = \Prime{\Ext{\Theta} \cup \Ext{\Theta'}}$, proving our claim.

    \item For the minimum of $\Spaces{\underline{I}}$.
    We have that $\ExtHist{\discrete{E},\underline{I}} = \PFun{\underline{I}}$, so that $\ExtHist{\discrete{E},\underline{I}} \supseteq \Ext{\Theta}$ for all $\Theta$ and hence $\Hist{\discrete{E},\underline{I}} \leq \Theta$.

    \item For $\Spaces{\underline{I}}$ as a full sub-lattice.
    Enough to observe that for a generic $\Theta \in \AllSpaces$ the condition $\Hist{\discrete{E},\underline{I}} \leq \Theta$ implies $\Theta \subseteq \PFun{\underline{I}}$.
    
    \item For the maximum of $\SpacesFC{\underline{I}}$.
    We have that $\Hist{\indiscrete{E}, \underline{I}} = \ExtHist{\indiscrete{E}, \underline{I}} = \prod_{\omega \in E}I_\omega$, so that $\Ext{\Theta} \supseteq \ExtHist{\indiscrete{E},\underline{I}}$ for all $\Theta$ (because of the free-choice condition) and hence $\Theta \leq \Hist{\discrete{E},\underline{I}}$.
\end{itemize}
\end{proof}

\subsubsection{Proof of Proposition \ref{proposition:order-induced-spaces-closed-under-join}}
\label{proof:proposition:order-induced-spaces-closed-under-join}
\begin{proof}
Recall that the join is defined as:
\[ 
    \Hist{\Omega,\underline{I}}\vee\Hist{\Omega',\underline{I}}
    :=
    \Prime{\ExtHist{\Omega,\underline{I}}\cap\ExtHist{\Omega',\underline{I}}}
\]
Because the same inputs are used in both spaces, the extended input histories in the two spaces are entirely determined by the lowersets $\Lsets{\Omega}$ and $\Lsets{\Omega'}$: the common histories are those on common lowersets, i.e. those determined by the intersection $\Lsets{\Omega}\cap\Lsets{\Omega'}$.
However, we know from the previous Section that $\Lsets{\Omega}\cap\Lsets{\Omega'} = \Lsets{\Omega \vee \Omega'}$, from which we conclude that:
\[
    \ExtHist{\Omega,\underline{I}}\cap\ExtHist{\Omega',\underline{I}}
    =
    \ExtHist{\Omega\vee\Omega',\underline{I}}
\]
By taking the prime elements, we obtain our desired statement.
\end{proof}

\subsubsection{Proof of Proposition \ref{proposition:order-induced-spaces-not-closed-under-meet}}
\label{proof:proposition:order-induced-spaces-not-closed-under-meet}
\begin{proof}
Consider $\Omega = \total{\ev{A},\ev{B}}\vee\discrete{\ev{C}}$ and $\Omega' = \discrete{\ev{A}}\vee\total{\ev{C},\ev{B}}$:
\begin{center}
    \begin{tabular}{cc}
    \includegraphics[height=3cm]{svg-inkscape/space-ABC-unique-33-0-highlighted_svg-tex.pdf}
    \hspace{10mm}
    &
    \hspace{10mm}
    \includegraphics[height=3cm]{svg-inkscape/space-ABC-unique-33-3-highlighted_svg-tex.pdf}
    \\
    $\Hist{\Omega,\{0,1\}}$
    \hspace{10mm}
    &
    \hspace{10mm}
    $\Hist{\Omega',\{0,1\}}$
    \end{tabular}
\end{center}
The meet of the causal orders is the discrete space $\Omega \wedge \Omega' = \discrete{\ev{A}, \ev{B}, \ev{C}}$. The meet of the spaces of input histories, on the other hand, is not induced by any causal order (with reference to Figure \ref{fig:hierarchy-spaces-ABC} (p.\pageref{fig:hierarchy-spaces-ABC}), the meet lies in equivalence class 3):
\begin{center}
    \begin{tabular}{cc}
    \includegraphics[height=3cm]{svg-inkscape/space-ABC-unique-33-0-meet-3-highlighted_svg-tex.pdf}
    \hspace{10mm}
    &
    \hspace{10mm}
    \includegraphics[height=3cm]{svg-inkscape/space-ABC-unique-tight-0-highlighted_svg-tex.pdf}
    \\
    $\Hist{\Omega,\{0,1\}}\wedge\Hist{\Omega',\{0,1\}}$
    \hspace{10mm}
    &
    \hspace{10mm}
    $\Hist{\Omega\wedge\Omega',\{0,1\}}$
    \end{tabular}
\end{center}
It is immediately evident that the two spaces are different.
\end{proof}

\subsubsection{Proof of Proposition \ref{proposition:parallel-sequential-composition-spaces}}
\label{proof:proposition:parallel-sequential-composition-spaces}
\begin{proof}
Regarding parallel composition, note that we can never have $h \leq h'$ or $h' \leq h$ for $h \in \Theta$ and $h' \in \Theta'$: as a consequence, all partial functions in $\Theta \cup \Theta'$ are $\vee$-prime, making it a well-defined space of input histories.
Regarding sequential composition, we start by observing that any non-trivial join $h \vee h'$ of two compatible $h, h' \in \Theta$ is neither in $\Theta$ nor in the form $k \vee h'$ for some $k \in \max{\Ext{\Theta}}$ and $h' \in \Theta'$.
If $h \in \Theta$, $k \in \max{\Ext{\Theta}}$ and $h' \in \Theta'$, then the only compatible joins $h \vee (k \vee h')$ are those with $h \leq k$, which are necessarily trivial.
Finally, consider a compatible join $(k\vee h) \vee (k' \vee h')$ for $k, k' \in \max{\Ext{\Theta}}$ and $h, h' \in \Theta'$: compatibility forces $k=k'$, so the join takes the form $k \vee (h \vee h')$, which cannot be in the form $k \vee h''$ for any $h'' \in \Theta'$ other than $h''=h$ or $h''=h'$.
Hence $\Theta \seqcomposeSym \Theta'$ is a well-defined space of input histories.
\end{proof}

\subsubsection{Proof of Proposition \ref{proposition:parallel-sequential-composition-free-choice}}
\label{proof:proposition:parallel-sequential-composition-free-choice}
\begin{proof}
The parallel composition $\Theta \cup \Theta'$ and sequential composition $\Theta \seqcomposeSym \Theta'$ have the same maximal extended input histories:
\[
\begin{array}{rcl}
    \max\Ext{\Theta \cup \Theta'}
    &=&
    \max\Ext{\Theta \seqcomposeSym \Theta'}
    \\
    &=&
    \max{\Ext{\Theta}}\allJoinsSym\max{\Ext{\Theta'}}
    \\
    &=&
    \suchthat{k \vee k'}{k \in \max\Ext{\Theta}, k' \in \max\Ext{\Theta'}}
\end{array}
\]
If $\Theta$ and $\Theta'$ satisfy the free-choice condition, then we have:
\[
    \max{\Ext{\Theta}}\allJoinsSym\max{\Ext{\Theta'}}
    =
    \left(
    \prod\limits_{\omega \in \Events{\Theta}} \Inputs{\Theta}_\omega
    \right)
    \times
    \left(
    \prod\limits_{\omega \in \Events{\Theta'}} \Inputs{\Theta'}_\omega
    \right)
\]
The previous two Observations immediately allow us to conclude.
\end{proof}

\subsubsection{Proof of Proposition \ref{proposition:parallel-sequential-composition-spaces-and-orders}}
\label{proof:proposition:parallel-sequential-composition-spaces-and-orders}
\begin{proof}
For parallel composition, we get:
\[
\fl\begin{array}{rcl}
\Hist{\Omega, \underline{I}}
\cup
\Hist{\Omega', \underline{I}'}
&=&
\bigcup\limits_{\xi \in \Omega}
\prod\limits_{\omega \in \downset{\xi}} I_\omega
\cup
\bigcup\limits_{\xi \in \Omega'}
\prod\limits_{\omega \in \downset{\xi}} I'_\omega
\\
&=&
\bigcup\limits_{\xi \in \Omega \vee \Omega'}
\prod\limits_{\omega \in \downset{\xi}} (\underline{I}\vee\underline{I}')_\omega
\\
&=&
\Hist{\Omega\vee\Omega', \underline{I}\vee\underline{I}'}
\end{array}
\]
For sequential composition, we get:
\[
\fl\begin{array}{rcl}
\Hist{\Omega, \underline{I}}
\seqcomposeSym
\Hist{\Omega', \underline{I}'}
&=&
\bigcup\limits_{\xi \in \Omega}
\prod\limits_{\omega \in \downset{\xi}} I_\omega
\cup\hspace{-2mm}
\bigcup\limits_{\lambda \in \max\Lsets{\Omega}}
\bigcup\limits_{\xi \in \Omega'}
\left(
\prod\limits_{\omega \in \lambda} I_\omega
\right)
\times
\left(
\prod\limits_{\omega' \in \downset{\xi}} I'_{\omega'}
\right)
\\
&=&
\bigcup\limits_{\xi \in \Omega \seqcomposeSym \Omega'}
\prod\limits_{\omega \in \downset{\xi}} (\underline{I}\vee\underline{I}')_\omega
\\
&=&
\Hist{\Omega\seqcomposeSym\Omega', \underline{I}\vee\underline{I}'}
\end{array}
\]
\end{proof}

\subsubsection{Proof of Proposition \ref{proposition:conditional-sequential-composition-spaces}}
\label{proof:proposition:conditional-sequential-composition-spaces}
\begin{proof}
The proof is a straightforward generalisation of that for sequential composition: every instance of $\Theta'$ in the proof for the latter appears in the context of some $k \in \max\Ext{\Theta}$, and it suffices to replace it with $\Theta'_k$.
\end{proof}

\subsubsection{Proof of Proposition \ref{proposition:conditional-sequential-composition-spaces-free-will}}
\label{proof:proposition:conditional-sequential-composition-spaces-free-will}
\begin{proof}
In one direction, presume that all three conditions hold:
\[
\begin{array}{rcl}
\max\Ext{\Theta} &=& \prod\limits_{\omega \in \Events{\Theta}} \Inputs{\Theta}_\omega\\
\max\Ext{\Theta'_k} &=& \prod\limits_{\omega \in E'} I'_\omega\\
\underline{\Inputs{\Theta'_k}} &=& \left(I'_\omega\right)_{\omega \in E'}
\end{array}
\]
where $\underline{I}' = \left(I'_\omega\right)_{\omega \in E'}$ is a fixed family of non-empty input sets.
The maximal extended input histories for the conditional sequential composition take the following form:
\[
\bigcup\limits_{k \in \max\Ext{\Theta}} \suchthat{k \vee k'}{k' \in \max\Ext{\Theta'_k}}
\]
We use the three condition to simplify the expression into a product:
\[
\begin{array}{rl}
&
\bigcup\limits_{k \in \max\Ext{\Theta}} \suchthat{k \vee k'}{k' \in \max\Ext{\Theta'_k}}
\\
=&
\bigcup\limits_{k \in \!\!\prod\limits_{\omega \in \Events{\Theta}}\!\!\!\Inputs{\Theta}_\omega}
\suchthat{k \vee k'}{k' \in \!\prod\limits_{\omega \in E'}\!\! I'_\omega}
\\
=&
\left(\prod\limits_{\omega \in \Events{\Theta}}\!\!\!\Inputs{\Theta}_\omega\right)
\times
\left(\prod\limits_{\omega \in E'}\!\! I'_\omega\right)
\end{array}
\]
Hence the conditional sequential composition satisfies the free-choice condition.
To prove the other direction, we will show that violating any one of the three conditions above results in the conditional sequential composition violating the free-choice condition, i.e. that there is some partial function $\hat{k}$ in the set $P$ below such that $\hat{k} \notin \max\Ext{\Theta \seqcomposeSym \underline{\Theta'}}$:
\[
P
:=
\left(\prod\limits_{\omega \in \Events{\Theta}}\!\!\!\Inputs{\Theta}_\omega\right)
\times
\left(\prod\limits_{\omega \in E'}\!\! I'_\omega\right)
\]
If the space $\Theta$ violates the free-choice condition, then there is some $k \in \!\!\prod_{\omega \in \Events{\Theta}}\!\!\!\Inputs{\Theta}_\omega$ such that $k \notin \max\Ext{\Theta}$: this implies that $\hat{k} \notin \max\Ext{\Theta \seqcomposeSym \underline{\Theta'}}$ for all $\hat{k} \in P$ with $k \leq \hat{K}$, proving that $\Theta \seqcomposeSym \underline{\Theta'}$ violates the free-choice condition.
If the space $\Theta'_k$ violates the free-choice condition, then there is some $k' \in \!\prod_{\omega \in E'}\!\!I'_\omega$ such that $\hat{k} := k\vee k' \notin \max\Ext{\Theta \seqcomposeSym \underline{\Theta'}}$, proving that $\Theta \seqcomposeSym \underline{\Theta'}$ violates the free-choice condition.
If there are $k_1, k_2 \in \max\Ext{\Theta}$ such that $\omega \in \Events{\Theta'_{k_1}}$ and $\omega \notin \Events{\Theta'_{k_1}}$, then $\hat{k} \notin \max\Ext{\Theta \seqcomposeSym \underline{\Theta'}}$ for all $\hat{k} \in P$ with $k_2 \leq \hat{K}$ and $\omega \in \dom{\hat{k}}$, proving that $\Theta \seqcomposeSym \underline{\Theta'}$ violates the free-choice condition.
If there are $k_1, k_2 \in \max\Ext{\Theta}$ such that $i \in \Inputs{\Theta'_{k_1}}_\omega$ and $i \notin \Inputs{\Theta'_{k_2}}_\omega$, then $\hat{k} \notin \max\Ext{\Theta \seqcomposeSym \underline{\Theta'}}$ for all $\hat{k} \in P$ with $k_2 \leq \hat{K}$ and $\hat{k}_{\omega} = i$, proving that $\Theta \seqcomposeSym \underline{\Theta'}$ violates the free-choice condition.
\end{proof}

\subsubsection{Proof of Proposition \ref{proposition:input-history-at-least-one-tip}}
\label{proof:proposition:input-history-at-least-one-tip}
\begin{proof}
Let $h \in \Theta$ be an input history. If we had $\tips{\Theta}{h} = \emptyset$, then we'd have $\bigcup_{k < h}\dom{k} = \dom{h}$, and hence $h = \bigvee_{k < h} k$, which would contradict $\vee$-primality of $h$.
Now let $h \in \Ext{\Theta}$ be an extended input history such that $h \not\in \Theta$. Then $h = \bigvee_{k < h} k$ implies that $\tips{\Theta}{h}=\dom{h}\backslash\bigcup_{k < h}\dom{k} = \emptyset$.
\end{proof}

\subsubsection{Proof of Proposition \ref{proposition:order-induced-space-causal-completeness}}
\label{proof:proposition:order-induced-space-causal-completeness}
\begin{proof}
For all $\omega \in \Omega$ and all $h \in \Theta$ with $\dom{h} = \downset{\omega}$, we must have:
\[
\tips{\Theta}{h}
\;:=\;
\dom{h}\backslash\bigcup_{k < h} \dom{k}
\;=\;
\downset{\omega}\backslash\bigcup_{\omega' \prec \omega} \downset{\omega'}
\;=\;
\causeqcls{\omega}
\]
Hence, $\Theta$ is causally complete iff all tips $\tips{\Theta}{h}$ have size 1, iff all causal equivalence classes $\causeqcls{\omega}$ have size 1, iff $\Omega$ is causally definite.
\end{proof}

\subsubsection{Proof of Proposition \ref{proposition:parallel-sequential-composition-causally-complete}}
\label{proof:proposition:parallel-sequential-composition-causally-complete}
\begin{proof}
We already know that if $\Theta$ and $\Theta'$ both satisfy the free-choice condition, so does their parallel and sequential composition.
First consider $h \in \Theta$.
In both parallel and sequential composition, the input histories $h'$ which satisfy $h' \leq h$ are exactly the $h' \in \Theta$, so that:
\[
\tips{\Theta \cup \Theta'}{h} = \tips{\Theta}{h} = \tips{\Theta \seqcomposeSym \Theta'}{h}
\]
An analogous reasoning applies to $h \in \Theta'$ in the case of parallel composition, which is symmetric in $\Theta$ and $\Theta'$.
Now consider $h = k \vee h'$ in the case of sequential composition, where $k \in \max\Ext{\Theta}$ and $h' \in \Theta'$: it is enough to show that $\tips{\Theta \seqcomposeSym \Theta'}{h} \subseteq \tips{\Theta'}{h}$, because the LHS is guaranteed to have at least one element ($h$ is an input history) and the RHS is guaranteed to have exactly one element ($\Theta'$ is causally complete).
For all $h'' \in \Theta'$, we have $k \vee h'' \in \Theta \seqcomposeSym \Theta'$ and:
\[
\begin{array}{rcl}
\dom{h}\backslash \dom{k \vee h''}
&=&
\dom{k \vee h'} \backslash \dom{k \vee h''}
\\
&=&
\dom{h'} \backslash \dom{h''}
\end{array}
\]
This shows that $\tips{\Theta \seqcomposeSym \Theta'}{h} \subseteq \tips{\Theta'}{h}$, completing our proof.
\end{proof}

\subsubsection{Proof of Proposition \ref{proposition:conditional-sequential-composition-causally-complete}}
\label{proof:proposition:conditional-sequential-composition-causally-complete}
\begin{proof}
We already know that under the assumptions above the conditional sequential composition $\Theta \seqcomposeSym \underline{\Theta'}$ satisfies the free-choice condition.
The remainder of the proof proceeds like that of the previous Proposition in the sequential composition case, where $\Theta'$ is replaced by $\Theta'_k$.
\end{proof}

\subsubsection{Proof of Proposition \ref{proposition:causal-completion-events-inputs}}
\label{proof:proposition:causal-completion-events-inputs}
\begin{proof}
Consider the following subset:
\[
S := \suchthat{
    k \in \Ext{\hat{\Theta}}
}{
    k \in \PFun{\underline{\Inputs{\Theta}}}
} \supseteq \Ext{\Theta}
\]
The set $S$ is $\vee$-closed, hence $\Prime{S} \leq \Theta$.
The set $S$ is also closed downward, so that the tips of a history $h$ in $\Prime{S}$ coincide with the tips of $h$ in $\hat{\Theta}$: this implies that $\Prime{S}$ is causally complete.
Because $S \subseteq \Ext{\hat{\Theta}}$, we have $\Prime{S} \geq \hat{\Theta}$: by maximality of $\hat{\Theta}$ we conclude $\hat{\Theta} = S$, and hence $\Events{\hat{\Theta}} = \Events{\Theta}$ and $\underline{\Inputs{\hat{\Theta}}} = \underline{\Inputs{\Theta}}$.
\end{proof}

\subsubsection{Proof of Theorem \ref{theorem:parallel-composition-causal-completions}}
\label{proof:theorem:parallel-composition-causal-completions}
\begin{proof}
    By Proposition \ref{proposition:parallel-sequential-composition-causally-complete},
    $\hat{\Theta} \cup \hat{\Theta}' \leq \Theta \cup \Theta'$ for all $\hat{\Theta} \in \CausCompl{\Theta}$ and all $\hat{\Theta}' \in \CausCompl{\Theta'}$.
    Let $\Theta'' \leq \Theta \cup \Theta'$ be a maximal causally complete subspace of $\Theta \cup \Theta'$.
    Consider the following subsets:
    \[
    \begin{array}{rcl}
        T &:=& \suchthat{
            k \in \Ext{\Theta''}
        }{
            \dom{k} \subseteq \Events{\Theta}
        } \supseteq \Ext{\Theta}
        \\
        T' &:=& \suchthat{
            k' \in \Ext{\Theta''}
        }{
            \dom{k'} \subseteq \Events{\Theta'}
        } \supseteq \Ext{\Theta'}
    \end{array}
    \]
    The sets $T$ and $T'$ are both $\vee$-closed, whence we have $\Prime{T} \leq \Theta$ and $\Prime{T'} \leq \Theta'$.
    Because $T$ and $T$ are also closed downwards, the tips of a history $h$ in $\Prime{T}$ or in $\Prime{T'}$ coincide with its tips in $\Theta''$: hence, $\Prime{T}$ and $\Prime{T'}$ are both causally complete.
    Now take $\hat{\Theta} \in \CausCompl{\Theta}$ and $\hat{\Theta}' \in \CausCompl{\Theta'}$ such that $\Prime{T} \leq \hat{\Theta}$ and $\Prime{T'} \leq \hat{\Theta}'$:
    \[
        \begin{array}{rcl}
        \Ext{\hat{\Theta} \cup \hat{\Theta}'}
        &=&
        \text{$\vee$-closure of } \Ext{\hat{\Theta}} \cup \Ext{\hat{\Theta}'}
        \\
        &\subseteq&
        \text{$\vee$-closure of } T \cup T'
        \\
        &\subseteq&
        \Ext{\Theta''}
        \end{array}
    \]
    Hence we have $\Theta'' \leq \hat{\Theta} \cup \hat{\Theta}'$, which by maximality of $\Theta''$ implies $\Theta'' = \hat{\Theta} \cup \hat{\Theta}'$.
\end{proof}

\subsubsection{Proof of Theorem \ref{theorem:conditional-sequential-composition-causal-completions}}
\label{proof:theorem:conditional-sequential-composition-causal-completions}
\begin{proof}
    By Proposition \ref{proposition:conditional-sequential-composition-causally-complete},
    $\hat{\Theta} \seqcomposeSym \underline{\hat{\Theta}'} \leq \Theta \seqcomposeSym \underline{\Theta'}$ for all $\hat{\Theta} \in \CausCompl{\Theta}$ and all $\hat{\Theta}'_k \in \CausCompl{\Theta'_k}$.
    Let $\Theta'' \leq \Theta \seqcomposeSym \underline{\Theta'}$ be a maximal causally complete subspace of $\Theta \Theta \seqcomposeSym \underline{\Theta'}$.
    Consider the following subset:
    \[
        T := \suchthat{
            k \in \Ext{\Theta''}
        }{
            \dom{k} \subseteq \Events{\Theta}
        } \supseteq \Ext{\Theta}
    \]
    The set $T$ is $\vee$-closed, whence we have $\Prime{T} \leq \Theta$.
    Because $T$ is also closed downwards, the tips of a history $h$ in $\Prime{T}$ coincide with its tips in $\Theta'$: hence $\Prime{T}$ is causally complete.
    Now take $\hat{\Theta} \in \CausCompl{\Theta}$ such that $\Prime{T} \leq \hat{\Theta}$: we have $\Ext{\Theta''} \supseteq \Ext{\hat{\Theta} \seqcomposeSym T'}$, and hence $\Theta'' \leq \hat{\Theta} \seqcomposeSym T'$.
    By maximality of $\Theta''$, we must have $\Theta'' = \hat{\Theta} \seqcomposeSym T'$, implying $T = \Ext{\hat{\Theta}}$.

    For all $k \in \max\Ext{\Theta}$, consider the following subsets:
    \[
        T'_k := \suchthat{
            k'\;\;
        }{
            \begin{array}{l}
            \dom{k'} \subseteq \Events{\Theta'},\\
            k \vee k' \in \Ext{\Theta''}
            \end{array}
        } \supseteq \Ext{\Theta'_k}
    \]
    The sets $T'_k$ are all $\vee$-closed, whence we have $\Prime{T'_k} \leq \Theta'_k$ for all $k \max\Ext{\Theta}$.
    For every $k' \in T'_K$ and every $k'' \in \Theta''$ such that $k'' \leq k \vee k'$, we have that $k \vee k'' \in \Theta''$ and $k \vee k'' \leq k \vee k'$: the tips of $k \vee k'$ in $\Theta''$ must therefore be the same as the tips of $k'$ in $\Prime{T'_k}$, from which we conclude that $\Prime{T'_k}$ is causally complete.
    For each $k$, take $\hat{\Theta}'_k \in \CausCompl{\Theta'_k}$ such that $\Prime{T'_k} \leq \hat{\Theta}'_k$: by the same reasoning as before, we have $\Theta'' \leq \hat{\Theta} \seqcomposeSym \underline{\hat{\Theta}'}$.
    By maximality of $\Theta''$, we must have $\Theta'' = \hat{\Theta} \seqcomposeSym \underline{\hat{\Theta}'}$, completing our proof.
\end{proof}

\subsubsection{Proof of Proposition \ref{proposition:hierarchy-causally-complete-spaces}}
\label{proof:proposition:hierarchy-causally-complete-spaces}
\begin{proof}
Regarding closure under meet, consider two causally complete spaces $\Theta, \Theta' \in \SpacesFC{\underline{I}}$.
By Proposition \ref{proposition:spaces-input-histories-hierarchies}, the meet $\Theta \wedge \Theta'$ is obtained by taking the $\vee$-prime elements in $\Ext{\Theta} \cup \Ext{\Theta'}$: this means that every input history $h \in \Theta \wedge \Theta'$ (i.e. a $\vee$-prime element in $\Ext{\Theta} \cup \Ext{\Theta'}$) is either an input history in $h \in \Theta$ (i.e. a $\vee$-prime element in $\Ext{\Theta}$) or an input history $h \in \Theta'$ (i.e. a $\vee$-prime element in $\Ext{\Theta'}$). Without loss of generality, assume $h \in \Theta$. We have:
\[
\begin{array}{rcl}
\tips{\Theta\wedge\Theta'}{h}
&=&
\dom{h}\backslash\bigcup_{k \in \Theta\wedge\Theta'\text{ s.t. }k < h} \dom{k}\\
&\subseteq&
\dom{h}\backslash\bigcup_{k \in \Theta\text{ s.t. }k < h} \dom{k}
\;=
\tips{\Theta}{h}
\end{array}
\]
Because $\Theta$ is causally complete, $\tips{\Theta}{h}$ is a singleton, which forces $\tips{\Theta\wedge\Theta'}{h}$ to also be a singleton (by Proposition \ref{proposition:input-history-at-least-one-tip}, $h \in \Theta\wedge\Theta'$ has at least one tip event in $\Theta\wedge\Theta'$).

Regarding the lack of closure under join, Proposition \ref{proposition:order-induced-spaces-closed-under-join} shows that $\Hist{\Omega,\underline{I}}\vee\Hist{\Omega',\underline{I}}=\Hist{\Omega\vee\Omega',\underline{I}}$.
When there are two or more events, we can consider $\Omega = \total{..., \ev{A}, \ev{B}, ...}$ and $\Omega' = \total{..., \ev{A}, \ev{B}, ...}$, so that $\Omega \vee \Omega'$ is an indefinite causal order (with $\ev{A}$ and $\ev{B}$ falling into the same causal equivalence class).
Proposition \ref{proposition:order-induced-space-causal-completeness} proves that $\Hist{\Omega\vee\Omega',\underline{I}}$ is not causally complete, allowing us to conclude that causally complete spaces on two or more events are not closed under join.
\end{proof}

\subsubsection{Proof of Proposition \ref{proposition:parallel-sequential-composition-tight}}
\label{proof:proposition:parallel-sequential-composition-tight}
\begin{proof}
For parallel composition, let $k \in \Ext{\Theta \cup \Theta'}$ and $\omega \in \dom{k}$.
Because events are disjoint, $\omega \in \Events{\Theta}$ and $\omega \in \tips{\Theta \cup \Theta'}{h}$ implies $h \in \Theta$.
Similarly $\omega \in \Events{\Theta'}$ and $\omega \in \tips{\Theta \cup \Theta'}{h}$ implies $h \in \Theta'$.
Hence $\Theta \cup \Theta'$ is tight.
For parallel composition, let $k \in \Ext{\Theta \seqcomposeSym \Theta'}$ and $\omega \in \dom{k}$.
Because events are disjoint, $\omega \in \Events{\Theta}$ and $\omega \in \tips{\Theta \cup \Theta'}{h}$ again implies $h \in \Theta$.
Similarly $\omega \in \Events{\Theta'}$ and $\omega \in \tips{\Theta \cup \Theta'}{h}$ implies $h=k' \vee h'$ with $h' \in \Theta'$ and $k' \in \max\Ext{\Theta}$, with $\omega \in \tips{\Theta'}{h'}$.
Hence $\Theta \seqcomposeSym \Theta'$ is also tight.
\end{proof}

\subsubsection{Proof of Proposition \ref{proposition:conditional-sequential-composition-tight}}
\label{proof:proposition:conditional-sequential-composition-tight}
\begin{proof}
let $k \in \Ext{\Theta \seqcomposeSym \underline{\Theta'}}$ and $\omega \in \dom{k}$.
Because events between $\Theta$ and the $\Theta'_{k'}$ are disjoint, $\omega \in \Events{\Theta}$ and $\omega \in \tips{\Theta \seqcomposeSym \underline{\Theta'}}{h}$ implies $h \in \Theta$.
Now let $\omega \in \cup_{k' \in \max\Ext{\Theta}}\Events{\Theta'}$ and $\omega \in \tips{\Theta \seqcomposeSym \underline{\Theta'}}{h}$.
Necessarily, $k = k' \vee k''$ for some $k' \in \max\Ext{\Theta}$ and some $k'' \in \Ext{\Theta'_{k'}}$, which implies that $h = k' \vee h'$ for the unique $h' \in \Theta'_{k'}$ such that $h' \leq k''$ and $\omega \in \tips{\Theta'_{k'}}{h'}$.
Hence $\Theta \seqcomposeSym \underline{\Theta'}$ is tight.
\end{proof}

\subsubsection{Proof of Proposition \ref{proposition:hist-space-tight}}
\label{proof:proposition:hist-space-tight}
\begin{proof}
Let $k \in \ExtHist{\Omega, \underline{I}}$ and $\omega \in \dom{k}$.
The input history $h := \restrict{k}{\downset{\omega}} \in \Theta$ is the unique input history $h \leq k$ with $\omega \in \tips{\Theta}{h}$.
As a consequence, $\Theta$ is tight.
\end{proof}

\subsubsection{Proof of Theorem \ref{theorem:non-tight-order-induced}}
\label{proof:theorem:non-tight-order-induced}
\begin{proof}
For every input history $h \leq k$ in the meet $\Theta \wedge \Theta'$ we must have that $h \in \Theta$ or $h \in \Theta'$, because the extended input histories in $\Ext{\Theta \wedge \Theta'}$ arise the the compatible joins of input histories in the set $\Theta \cup \Theta'$; for the same reason, we must also have that $\tips{\Theta \wedge \Theta'}{h} \subseteq \tips{\Theta}{h}$ and $\tips{\Theta \wedge \Theta'}{h} \subseteq \tips{\Theta'}{h}$.
Let $k \in \prod_{\omega \in E}I_\omega$, which is a maximal extended input history for $\Theta$, $\Theta'$ and $\Theta \wedge \Theta'$.

In one direction, assume that $\omega \in \tips{\Theta \wedge \Theta'}{h}$ and $\omega \in \tips{\Theta \wedge \Theta'}{h}$ for two distinct input histories $h,h' \leq k$: then $h$ and $h'$ cannot be both in $\Theta$ or both in $\Theta'$, because the two spaces are tight, and without loss of generality we can assume that $h \in \Theta$ and $h' \in \Theta'$.
Since $h = \restrict{k}{\dom{h}}$ and $h' = \restrict{k}{\dom{h'}}$, we must have $\dom{h} \neq \dom{h'}$; since $h$ and $h'$ both have $\omega$ as a tip event, we must furthermore have $\dom{h}\not \subseteq \dom{h'}$ and $\dom{h'} \not\subseteq \dom{h}$.
Because $\dom{h} = \downset{\omega}_{\Omega}$ and $\dom{h'} = \downset{\omega}_{\Omega'}$, we conclude that $\downset{\omega}_{\Omega} \not \subseteq \downset{\omega}_{\Omega'}$ and $\downset{\omega}_{\Omega} \not \subseteq \downset{\omega}_{\Omega'}$.

In the other direction, assume that $\downset{\omega}_{\Omega} \not \subseteq \downset{\omega}_{\Omega'}$ and $\downset{\omega}_{\Omega} \not \subseteq \downset{\omega}_{\Omega'}$ for some $\omega \in E$.
Let $h$ be any input history $h \in \Theta \wedge \Theta'$ such that $h \leq \restrict{k}{\downset{\omega}_{\Omega}}$ and $\omega \in \tips{\Theta \wedge \Theta'}{h}$: one must exist, because $\omega \in \dom{\restrict{k}{\downset{\omega}_{\Omega}}}$; analogously let $h'$ be any input history $h' \in \Theta \wedge \Theta'$ such that $h' \leq \restrict{k}{\downset{\omega}_{\Omega'}}$ and $\omega \in \tips{\Theta \wedge \Theta'}{h'}$.
If it were the case that $h \in \Theta'$, then $\omega \in \tips{\Theta \wedge \Theta'}{h} \subseteq \tips{\Theta'}{h}$ would imply that $\dom{h} = \downset{\omega}_{\Omega'}$: this would contradict the definition of $h \leq \restrict{k}{\downset{\omega}_{\Omega}}$, and hence we must have $h \in \Theta$; analogously, we must have $h' \in \Theta'$.
We conclude that there exist distinct $h, h' \leq k$ such that $\omega \in \tips{\Theta \wedge \Theta'}{h}$ and $\omega \in \tips{\Theta \wedge \Theta'}{h'}$, making $\Theta \wedge \Theta'$ non-tight.
\end{proof}

\subsubsection{Proof of Theorem \ref{theorem:causal-completeness-characterisation}}
\label{proof:theorem:causal-completeness-characterisation}
\begin{proof}
In one direction, assume that for every extended input history $k \in \Ext{\Theta}$ with $|\dom{k}| \geq 2$ there exists an $\omega_{k} \in \dom{k}$ such that $\restrict{k}{\dom{k}\backslash\{\omega_{k}\}} \in \Ext{\Theta}$: this implies $\tips{\Theta}{k} = \{\omega_k\}$ for all input histories $k \in \Theta$ with $|\dom{k}| \geq 2$, making $\Theta$ causally complete (because $\tips{\Theta}{k} = \dom{k}$ always holds when $|\dom{k}| = 1$).

In the other direction, assume that $\Theta$ is causally complete, and let $k \in \Ext{\Theta}$ be any extended input history with $|\dom{k}| \geq 2$.
For every $\omega \in \dom{k}$, we define:
\[
\begin{array}{rcl}
H_{k,\omega}
&:=&
\suchthat{h \in \Ext{\Theta}}{h \leq k, \omega \in \dom{h}}
\\
\bar{H}_{k,\omega}
&:=&
\suchthat{h \in \Ext{\Theta}}{h \leq k, \omega \notin \dom{h}}
\end{array}
\]
By causal completeness, there exists an $\omega_k$ such that $\bar{H}_{k, \omega_k} \neq \emptyset$: otherwise, $k$ is a minimal input history with $|\dom{k}| \geq 2$, contradicting Observation \ref{observation:minimal-input-history-tips}.
For every $\omega' \in \dom{k}\backslash\omega_k$, causal completeness also implies that there exists an $h \in \bar{H}_{k, \omega_k}$ with $\omega' \in \dom{h}$: if this were not the case, then any $h \in H_{k, \omega'}$ of minimal domain size would be a minimal input history with $\{\omega_k, \omega'\} \subseteq \dom{h} = \tips{\Theta}{h}$, contradicting Observation \ref{observation:minimal-input-history-tips}.
As a consequence, $\restrict{k}{\dom{k}\backslash\{\omega_k\}} = \bigvee \bar{H}_{k,\omega_k} \in \Ext{\Theta}$.
\end{proof}

\subsubsection{Proof of Theorem \ref{theorem:hist-equals-exthist}}
\label{proof:theorem:hist-equals-exthist}
\begin{proof}
For convenience, we set $E := \Events{\Theta}$ and $\underline{I} := \underline{\Inputs{\Theta}}$.
Our proof is by induction on the number of events in $E$, and the base case $E = \emptyset$ is trivial.
By assumption, all extended input histories $k \in \Ext{\Theta}$ are input histories $k \in \Theta$.

If $h \in \Theta$ is a non-minimal input history, then we can take the join $h'$ of the input histories strictly below it:
\[
h' := \bigvee \left(\downset{h} \backslash \{h\}\right) \in \Theta
\] 
Ordinarily, we would only be guaranteed that $h'$ is an extended input history, but here it is necessarily also an input history: by $\vee$-primality, it must then be the case that $h' = \restrict{h}{\dom{h}\backslash\tips{\Theta}{h}}$.
Hence, every non-minimal input history $h \in \Theta$ has a unique predecessor in the Hasse diagram for $\Theta$.

If $h, h' \in \Theta$ are two distinct minimal input histories, then $h$ and $h'$ cannot be compatible: otherwise, $h\vee h' \in \Ext{\Theta}=\Theta$ would have two distinct predecessors.
Since all input histories have a single tip event, minimal histories have a single event in their domain: if all distinct minimal histories are to be incompatible, then they must all have the same event in their domain.
As a consequence, minimal input histories take the form $\{\omega_1:i_1\}$ for a unique $\omega_1 \in \Events{\Theta}$ and all $i_1 \in \Inputs{\Theta}_{\omega_1}$.
Furthermore, every non-minimal input history $h \in \Theta$ can be written in the form $\{\omega_1:i_1\} \vee h'$ for a unique $i_1 \in \Inputs{\Theta}_{\omega_1}$ and some partial function $h' \neq \emptyset$ such that $\omega_1 \notin \dom{h'}$.
Hence, $\Theta$ arises as conditional sequential composition:
\[
\Theta = \Hist{\{\omega_1\}, \restrict{\Inputs{\Theta}}{\{\omega_1\}}}
\seqcomposeSym \underline{\Theta'}
\]
where we define the family $\underline{\Theta'} = \left(\Theta'_{i_1}\right)_{i_1 \in \Inputs{\Theta}_{\omega_1}}$ to consist of the following spaces (possibly empty):
\[
\Theta'_{i_1}
:= \suchthat{
    h'
}{
    \{\omega_1:i_1\} \vee h' \in \Theta, h' \neq \emptyset
}
\]
Finally, we have to show that $\Theta'_{i_1}$ our three desired properties.
By construction, $\omega_1 \notin \Events{\Theta'_{i_1}}$.
If $k \in \Ext{\Theta'_{i_1}}$, then we must have $k = \bigvee \mathcal{F}$ for some non-empty set $\mathcal{F} \subseteq \Theta'_{i_1}$ of input histories, so that:
\[
\{\omega_1:i_1\} \vee k
= \bigvee \suchthat{
    \{\omega_1:i_1\} \vee h'   
}{
    h' \in \mathcal{F}
} \in \Theta
\]
This means that $k \in \Theta'_{i_1}$. Hence, $\Theta'_{i_1} = \Ext{\Theta'_{i_1}}$.
Finally, we have that $h \leq h'$ in $\Theta'_{i_1}$ if and only if $\{\omega_1:i_1\} \vee h\leq\{\omega_1:i_1\} \vee h'$ in $\Theta$, so that necessarily:
\[
\tips{\Theta'_{i_1}}{h}
=\tips{\Theta}{\{\omega_1:i_1\} \vee h}
\]
This concludes our proof.
\end{proof}

\subsubsection{Proof of Corollary \ref{corollary:switch-spaces-characterisation}}
\label{proof:corollary:switch-spaces-characterisation}
\begin{proof}
Theorem \ref{theorem:hist-equals-exthist} provides the inductive conditional sequential composition structure: in order to match the definition of switch spaces, all that remains to be shown is that each $\Theta'_{i_1}$ satisfies the free-choice condition:
\[
\max\Ext{\Theta'_{i_1}} = \prod_{\omega \in \Events{\Theta} \backslash \{\omega_1\}} \Inputs{\Theta}_\omega 
\]
Observe that $h \leq h'$ in $\Theta'_{i_1}$ if and only if $\{\omega_1:i_1\} \vee h\leq\{\omega_1:i_1\} \vee h'$ in $\Theta$: hence $h \in \max\Ext{\Theta'_{i_1}} = \max\Theta'_{i_1}$ if and only if $\{\omega_1:i_1\} \vee h \in \max\Ext{\Theta} = \max\Theta$.
Because $\Theta$ satisfies the free-choice condition, we must have:
\[
\{\omega_1:i_1\} \vee h
\in \prod_{\omega \in \Events{\Theta}} \Inputs{\Theta}_\omega 
\]
By removing $\omega_1$ from the domain, we conclude.
\end{proof}

\subsubsection{Proof of Theorem \ref{theorem:causal-canopy}}
\label{proof:theorem:causal-canopy}
\begin{proof}
Let $\Theta \in \CCSpaces{\underline{I}}$ be a non-empty causally complete space.
The main body of this proof will show that there exists an event $\omega_1 \in \Events{\Theta}$ which can be made to ``come first'', i.e. one such that:
\begin{equation}
\label{equation:proof:theorem:causal-canopy-1}
\Theta \leq \Hist{\{\omega_1\}, \Inputs{\Theta}_{\omega_1}}
\seqcomposeSym \underline{\Theta'}    
\end{equation}
where $\Theta'_{\{\omega_1:i_1\}} \in \CCSpaces{\restrict{\underline{I}}{\Events{\Theta}\backslash\{\omega_1\}}}$ is a causally complete space for all $i_1 \in \Inputs{\Theta}_{\omega_1}$.
If $\Events{\Theta} = \{\omega_1\}$, then we are done.
Otherwise, by induction on the number of events we get that $\Theta'_{\{\omega_1:i_1\}} \leq \hat{\Theta}'_{\{\omega_1:i_1\}}$, where $\hat{\Theta}'_{\{\omega_1:i_1\}} \in \CSwitchSpaces{\restrict{\underline{I}}{\Events{\Theta}\backslash\{\omega_1\}}}$ is a causally switch space, for each $i_1 \in \Inputs{\Theta}_{\omega_1}$.
As a consequence, $\Theta$ is a sub-space of a causal switch space:
\[
\Theta \leq \Hist{\{\omega_1\}, \Inputs{\Theta}_{\omega_1}}
\seqcomposeSym \underline{\hat{\Theta}'}
\]
It remains to show is that Equation \ref{equation:proof:theorem:causal-canopy-1} holds for some event $\omega_1 \in \Events{\Theta}$, with all $\Theta'_{\{\omega_1:i_1\}}$ causally complete.

By contradiction, presume that for all $\omega \in \Events{\Theta}$ there exists an input $i_\omega \in \Inputs{\Theta}_{\omega}$ such that $\{\omega: i_\omega\} \notin \Theta$.
Because $\Theta$ satisfies the free-choice condition, the following partial function $k \in \prod_{\omega \in \Events{\Theta}} \Inputs{\Theta}_{\omega}$ must be a maximal extended input history for $\Theta$:
\[
k := \omega \mapsto i_\omega
\]
Now let $h \in \Theta$ be a minimal input history such that $h \leq k$.
Because $\Theta$ is causally complete, Observation \ref{observation:minimal-input-history-tips} implies that $\dom{h} = \{\xi\}$ for some $\xi \in \Events{\Theta}$: hence $\{\xi:i_\xi\} \in \Theta$, contradicting our hypothesis.
Hence, $\exists\omega_1 \in \Events{\Theta}$ such that $\forall i_1 \in \Inputs{\Theta}_{\omega_1}$ we have $\{\omega_1 : i_1\} \in \Theta$.

Let $h_1,...,h_n$ be a total order on the minimal input histories $h \in \Theta$ such that $\omega_1 \notin \dom{h}$ and let $h_{n+1},...,h_{n+m}$ be a total order on the input histories $\{\omega_1: i_1\}$, where $m := |\Inputs{\Theta}_{\omega_1}|$.
Starting at $\Theta_0 := \Theta$, and proceeding by induction in $j$, we create a sequence $(\Theta_j)_{j=0}^n$ of spaces satisfying the following properties:
\begin{enumerate}
    \item for all $j = 1,...,n$, we have that $\Theta_{j-1} \leq \Theta_{j}$.
    \item for all $j = 1,...,n$, we have that $h_{j'} \in \Theta_j$ if and only if $j' < j$.
    \item for all $j = 1,...,n$, we have that $\Theta_{j}$ is causally complete.
\end{enumerate}
At the end of the process, each one of the three properties (i), (ii) and (iii) above implies the corresponding property below: 
\begin{enumerate}
    \item $\Theta = \Theta_0 \leq \Theta_n$
    \item $\Theta_n = \Hist{\{\omega_1\}, \Inputs{\Theta}_{\omega_1}} \seqcomposeSym \underline{\Theta'}$
    \item $\Theta'_{\{\omega_1:i_1\}}$ is causally complete for all $i_1 \in \Inputs{\Theta}_{\omega_1}$.
\end{enumerate}
Given $\Theta_{j-1}$, we define $\Ext{\Theta_{j}}$ by removing $h_j$ and all input histories which don't contain event $\omega_1$ in their domain and don't contain any $h_i$ as a sub-history for $i > j$:
\[
\Ext{\Theta_{j}}
:=
\suchthat{
    k \in \Ext{\Theta_{j-1}}
}{
    \exists j'>j.\; h_{j'} \leq k 
}
\]
Properties (i) and (ii) hold by construction and inductive hypothesis, so it remains to prove that $\Theta_{j}$ is causally complete.

Let $k \in \Ext{\Theta_{j}}$ be an extended input history and let $j'_k > j$ be such that $h_{j'_k} \leq k$.
If $\downset{k}_\Ext{\Theta_{j}} = \downset{k}_\Ext{\Theta_{j-1}}$, then we have $\tips{\Theta_{j}}{k} = \tips{\Theta_{j-1}}{k}$ and hence $|\tips{\Theta_{j}}{k}| \leq 1$ because $\Theta_{j-1}$ is causally complete.
If  $k' \in \downset{k}_\Ext{\Theta_{j-1}}\backslash\downset{k}_\Ext{\Theta_{j}}$, it is possible for $k$ to have gained one or more tip events in the passage from $\Ext{\Theta_{j-1}}$ to $\Ext{\Theta_j}$: we must show that, when this is the case, $k \notin \Theta_{j-1}$ (i.e. $\tips{\Theta_{j-1}}{k} = \emptyset$) and $|\tips{\Theta_{j}}{k}| = 1$.
So, consider a $k$ which has gained tip events, pick one such event $\xi \in \tips{\Theta_{j}}{k}\backslash \tips{\Theta_{j-1}}{k}$ and let $k'_{\xi}$ be a sub-history of $k$ in $\Ext{\Theta_{j-1}}\backslash\Ext{\Theta_j}$ such that $\tip{\Theta_{j-1}}{k'_\xi} = \xi$.
Then we can consider the extended input history $h_{j'_k} \vee k'_{\xi} \in \Ext{\Theta_{j}}$, which satisfies:
\begin{itemize}
    \item $k \geq h_{j'_k} \vee k'_{\xi}$, because $k \geq h_{j'_k}$ and $k \geq k'_{\xi}$.
    \item $\tips{\Theta_j}{h_{j'_k} \vee k'_{\xi}} = \{\xi\}$, because the immediate predecessors of $h_{j'_k} \vee k'_{\xi}$ in $\Ext{\Theta_{j-1}}$ were $h_{j'_k}$ and $k'_{\xi}$, and hence the only predecessor of $h_{j'_k} \vee k'_{\xi}$ is $h_{j'_k}$.
\end{itemize}
Since $\xi \in \tips{\Theta_{j}}{k}$, we must have $k = h_{j'_k} \vee k'_{\xi}$, implying that $|\tips{\Theta_{j}}{k}| = |\{\xi\}| = 1$ as desired.
This completes our proof.
\end{proof}











\ack
Financial support from EPSRC, the Pirie-Reid Scholarship and Hashberg Ltd is gratefully acknowledged.
This publication was made possible through the support of the ID\#62312 grant from the John Templeton Foundation, as part of the project `The Quantum Information Structure of Spacetime' (QISS), https://www.templeton.org/grant/the-quantum-information-structure-ofspacetime-qiss-second-phase.
The opinions expressed in this project/publication are those of the author(s) and do not necessarily reflect the views of the John Templeton Foundation.

\section*{Bibliography}

\bibliographystyle{unsrt}
\bibliography{biblio}

\end{document}